\documentclass{article}

\RequirePackage[utf8]{inputenc}
\RequirePackage{amsfonts, amsmath, amssymb, mathtools, bm, amsthm}
\RequirePackage[colorlinks,citecolor=blue,urlcolor=blue]{hyperref}
\RequirePackage{graphicx}
\RequirePackage[numbers]{natbib}
\RequirePackage{multirow, makecell}
\RequirePackage{geometry}
\RequirePackage[perpage]{footmisc}
\RequirePackage{authblk}
\RequirePackage{subcaption}
\theoremstyle{plain}

\newtheorem{theorem}{Theorem}[section]
\newtheorem{proposition}[theorem]{Proposition}
\newtheorem{lemma}[theorem]{Lemma}
\newtheorem{corollary}[theorem]{Corollary}
\theoremstyle{definition}

\theoremstyle{remark}

\newcommand{\E}{\mathbb{E}}
\newcommand{\Prob}{\mathbb{P}}
\newcommand{\R}{\mathbb{R}}
\newcommand{\Var}{\mathrm{Var}}
\newcommand{\Cov}{\mathrm{Cov}}

\title{Parametric Bootstrap for Fixed Edge-Probability Network Models}
\author[1]{Zhixuan Shao\thanks{zhxshao@ucdavis.edu}}
\author[1]{Can M. Le\thanks{canle@ucdavis.edu}}
\affil[1]{Department of Statistics, University of California, Davis}
\date{}

\begin{document}

\maketitle

\begin{abstract}
This paper studies parametric bootstrap methods for network data, with the goal
of quantifying the uncertainty of network statistics of interest. While existing
network resampling methods primarily focus on count statistics under
node-exchangeable graphon models, we consider more general network statistics,
including local statistics, under the Chung--Lu model without assuming node
exchangeability.
We show that the natural network parametric bootstrap, which first estimates the
network-generating model and then draws bootstrap samples from the estimated
model, generally suffers from bootstrap bias. As a general remedy, we show that
a two-level bootstrap procedure provably reduces this bias. This extends the
classical idea of the iterative bootstrap to the network setting, where
the number of parameters grows with the network size.
Moreover, for many network statistics, the second-level bootstrap provides a way
to construct confidence intervals with higher accuracy. As a by-product of this
analysis, we also obtain a central limit theorem for subgraph counts under the
inhomogeneous Erd\H{o}s--R\'enyi model, which may be of independent interest.
\end{abstract}

\section{Introduction}
Random network models (also known as random graph models) have received continuous attention because of their wide-ranging applications, including social networks (friendship between Facebook users, LinkedIn following, etc.), biological networks (gene network, gene-protein network), information networks (email network, World Wide Web) and many others.
In these application areas, it is often of interest to summarize a network using network statistics \cite{newman2018networks}. Examples include clustering coefficients \cite{wasserman1994social}, subgraph/motif counts \cite{milo2002network, alon2007network}, and node centrality measures \cite{borgatti2005centrality, bonacich1987power}. In particular, these statistics characterize either global properties of the entire network or local properties of individual nodes. For example, the local clustering coefficient of a node measures the probability that two neighbors of the node are connected, while global transitivity gives an indication of the clustering effect in the whole network.

However, comparatively little attention has been paid to assessing the variability of these statistics, with a few exceptions that we will discuss shortly. Quantifying their uncertainty is of utmost importance. Consider the problem of comparing two networks, which is a key question in many biological applications and in social network analysis. A natural approach is to compare the summary statistics of the two networks and ask whether they are significantly different. However, such inference is impossible without knowledge of the underlying variability of the data-generating process.

Among all network statistics, subgraph counts---in particular, the number of occurrences of small patterns such as triangles---play an important role in characterizing random networks; see, e.g., \cite{milo2002network, przytycka2006important, maugis2020testing}. The study of the distribution of subgraph counts in random graphs dates back to \cite{erdos1960evolution}.
Under the Erd\H{o}s--R\'{e}nyi model, separate works \cite{nowicki1988subgraph, rucinski1988small} established the asymptotic normality and Poisson convergence of subgraph counts. The two major techniques they employed---the projection method for U-statistics and the method of moments---have since been frequently used for the asymptotic analysis of subgraph counts in the more general setting of node-exchangeable random graphs, also known as graphon models \cite{bickel2009nonparametric}.
A pioneering work in this direction is \cite{bickel2011method}, which not only established the asymptotic normality of subgraph counts but also provided a powerful tool for frequentist nonparametric network inference \cite{wolfe2013nonparametric, maugis2017statistical}. At a more theoretical level, count functionals may be viewed as network analogues of the moments of a random variable, and knowledge of all population moments can uniquely determine the network model up to weak isomorphism \cite{borgs2010moments}. Recent developments in the asymptotic theory of subgraph counts include CLTs for rooted subgraph counts \cite{maugis2020central} and Edgeworth expansions of network moments \cite{zhang2022edgeworth}.

However, most existing asymptotic distribution theory applies only to subgraph counts, despite the many other network quantities of interest. Even for count statistics alone, important questions remain unresolved. In general graphon settings, the exact asymptotic variance is either unavailable in closed form or infeasible to compute for large-scale networks, and thus must be estimated. To this end, various network resampling methods have been developed to estimate the distributions of network statistics, particularly count statistics.

\subsection{Related Work on Network Resampling Methods}

The first theoretical result for resampling network data was established by the authors of \cite{bhattacharyya2015subsampling}, who show the validity of subsampling for estimating the empirical distribution of count functionals. They prove theoretical properties of their bootstrap variance estimates for count features.
Along similar lines, the work of 
\cite{green2022bootstrapping}
proposes two bootstrap procedures for conducting inference on count functionals: one based on an empirical graphon, which amounts to resampling nodes with replacement, and another based on resampling from a stochastic block model fitted to the observed network. The latter is motivated by the intuition that such a model provides a good approximation to any graphon, provided that the number of communities is chosen to be sufficiently large.
The study in 
\cite{lunde2023subsampling}
considers a procedure that involves subsampling nodes and computing functionals on the induced subgraphs. This procedure is shown to be asymptotically valid under conditions analogous to those in the i.i.d. setting; that is, the subsample size must be $o(n)$, and the functional of interest must converge to a nondegenerate limit distribution.
The authors of \cite{lin2020theoretical} propose a leave-node-out jackknife procedure for networks. Under the sparse graphon model, they show the validity of the network jackknife in the sense that it yields conservative variance estimates for any network functional that is invariant under node permutations. They also establish the consistency of the network jackknife for count functionals.
Motivated by bootstrap methods for U-statistics \cite{bose2018u} and their close connection to Edgeworth expansions \cite{zhang2022edgeworth}, subsequent work \cite{lin2020trading} proposes a new class of multiplier bootstraps for count functionals that achieves higher-order correctness for appropriately sparse graphs.

The work most closely related to ours is \cite{levin2019bootstrapping}, which proposes a two-step bootstrap procedure under the random dot product graph model \cite{young2007random}. The procedure first estimates the latent positions using adjacency spectral embedding \cite{sussman2012consistent}, and then resamples the estimated positions to generate either bootstrap replicates of U-statistics of the latent positions or bootstrap replicates of entire networks. Thus, their method can also bootstrap network statistics that are not expressible as U-statistics. The authors establish consistency for bootstrapping U-statistics of the latent positions.

While U-statistics, especially subgraph counts, form an important class of functionals in network analysis, there are many other network quantities of interest. In particular, many local statistics associated with individual nodes do not take the form of U-statistics. Moreover, most of the aforementioned resampling methods are developed within the framework of graphon models, which assume that the distribution of a random graph is unchanged under permutations of node labels.
However, in practice, the assumption of node exchangeability may not always be desirable, especially when studying local statistics associated with specific nodes. In a network with heterogeneous node properties, we may be particularly interested in certain nodes and want to ensure that each node preserves its own properties in the resampled networks. For example, in a social network, we may naturally expect an influential person, represented by a node with a high centrality measure, to remain influential in the resampled networks.

To preserve node-specific information in the resampled networks, we consider a different setting in which networks are generated from a fixed edge-generating model. Although this may appear to be a simpler setting, most of the aforementioned methods, which assume node exchangeability, are not directly applicable because they also model uncertainty in the latent positions.
Under graphon models, it is typically assumed that uncertainty in the latent positions dominates the Bernoulli randomness of the edges, so that many statistics can be viewed as U-statistics perturbed by asymptotically negligible Bernoulli noise. In contrast, we focus on a setting that is essentially equivalent to assuming that the latent positions are fixed, so that the only randomness comes from the independent Bernoulli edges. This is a key difference between our setting and that of \cite{levin2019bootstrapping}. As a result, we generate bootstrap network replicates without resampling the estimated latent positions.
Figure~\ref{fig:bias_triangle_counts} illustrates the difference between sampling from fixed latent positions, corresponding to a fixed edge-probability matrix $P$, and sampling from random latent positions, corresponding to a random $P$. If nodes are in fact not exchangeable, then imposing node exchangeability introduces unnecessary variability in the sampling distribution.
That said, a potential
downside of viewing networks as having a fixed edge-probability matrix, rather
than as instances generated from a graphon model, is that network statistics
generally depend on the network size. Therefore, additional normalization or
size adjustment may be needed when comparing networks of different sizes.

\section{Preliminaries}
\subsection{Problem Set-Up}
\label{sec:setup}

Consider an undirected, unweighted network $G=(V,E)$, where $V=\{1,2,\ldots,n\}$ is the set of nodes and $E\subset V\times V$ is the set of edges. The network can be represented by a symmetric adjacency matrix $A \in {0,1}^{n\times n}$, with $A_{ij}=1$ if there is an edge between $i$ and $j$, and $A_{ij}=0$ otherwise.
We assume that the observed network $G$ is sampled from a fixed symmetric edge-probability matrix $P \in \R^{n\times n}$ with $P_{ij} \in [0,1]$, so that the entries of $A$ above the diagonal are independent Bernoulli random variables with $\E[A_{ij}] = P_{ij}$. For simplicity, we consider only networks without self-loops, so that $A_{ii}=0$ for all $i\in V$.
This independent-edge generating mechanism is also referred to as the inhomogeneous Erd\H{o}s--R\'{e}nyi random graph model \cite{bollobas2007phase}. 
We emphasize that $P$ is fixed and the only randomness comes from the Bernoulli realizations of $A$ given $P$. For models with random $P$, our method is applicable conditionally on $P$.

We focus on network statistics such as subgraph counts, clustering coefficients, and centrality measures \cite{newman2018networks}. A network statistic can be viewed as a map $T: \{0,1\}^{n\times n} \to \R$ that takes $A$ as input and outputs $T(A)\in\R$. To conduct inference on the statistic $T$, one of our goals is to construct confidence intervals for $\mu=\E[T(A)]$. For clarity, we often write $\mu = \mu(P)$ to indicate the dependence of $\mu$ on $P=\E A$. The population mean $\mu(P)$ is often of greater interest than the observed statistic $T(A_{\mathrm{obs}})$ itself, as it can be regarded as the signal underlying a noisy observation.

Since $P$ is unknown, a natural estimator of $\mu$ is the plug-in estimator $\mu(\widehat{P})$, where $\widehat{P}$ is an estimator of $P$. If $\widehat{P}$ is sufficiently accurate, then we can approximate the distribution of $T(A)$ by that of $T(\widehat{A})$, where $\widehat{A}$ is a network sampled according to $\widehat{P}$.
The procedure is described in Figure~\ref{fig:direct_parametric_bootstrap} and can be viewed as a straightforward network analogue of the classical parametric bootstrap. This generic framework, which we refer to as the network bootstrap, is applicable in principle to any network model together with its corresponding estimation method.

To understand whether and when the network bootstrap procedure shown in Figure~\ref{fig:direct_parametric_bootstrap} works well, a critical question is:
\textit{For a given statistic $T$, how accurate must $\widehat{P}$ be for the bootstrap distribution under $\widehat{P}$ to provide a good approximation to the true distribution under $P$?}

\begin{figure}
    \centering
     \begin{subfigure}[b]{0.4\textwidth}
         \centering
         \includegraphics[width=0.99\textwidth]{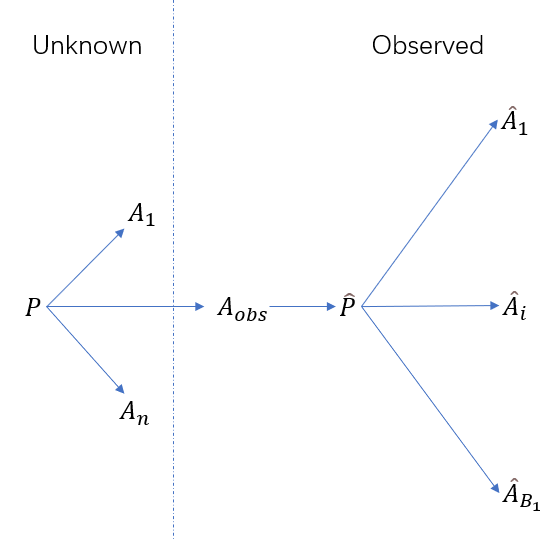}
        \caption{Direct parametric bootstrap}
        \label{fig:direct_parametric_bootstrap}
     \end{subfigure}
     \hfill
     \begin{subfigure}[b]{0.55\textwidth}
         \centering
         \includegraphics[width=0.99\textwidth]{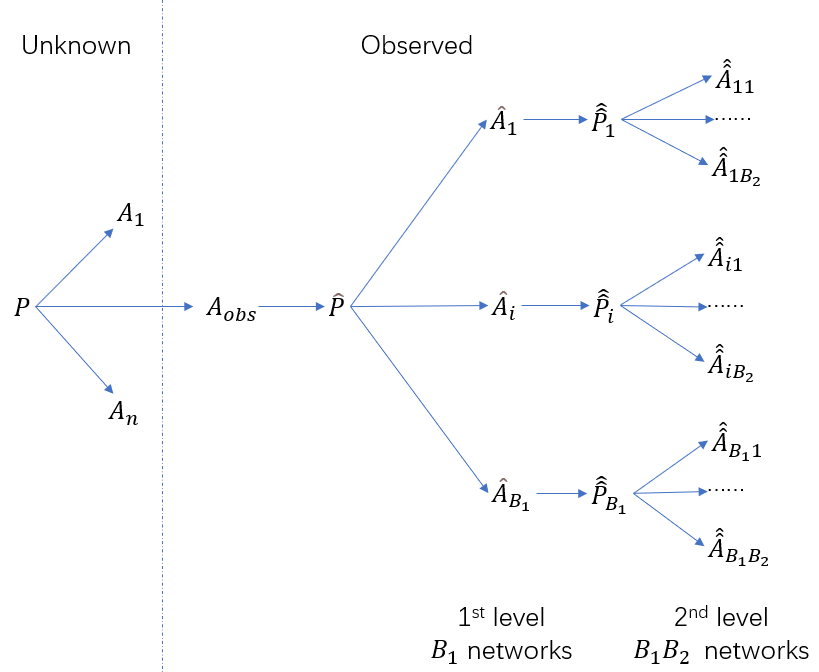}
        \caption{Two-level bootstrap}
        \label{fig:procedure}
     \end{subfigure}
        \caption{The direct parametric bootstrap provides first order approximation: $\mu(\widehat{P}) \approx \mu(P)$. The two-level bootstrap targets second order approximation: $\E_{\widehat{P}}[\mu(\widehat{\widehat{P}})] - \mu(\widehat{P}) \approx \E_P[\mu(\widehat{P})] - \mu(P)$.}
\end{figure}

Providing a general answer to this question is challenging, as it depends not only on the estimation error of $\widehat{P}$, which in turn depends on the structural assumptions imposed on $P$, but also on the specific statistic $T$ being considered. In practice, possible model misspecification, for example due to the choice of the number of blocks or the embedding dimension $K$, is a common cause of inaccurate estimates of $P$. It is reasonable to expect that the bootstrap distribution may fail to provide a good approximation under model misspecification.
However, it turns out that, for some common network models, even when the model is correctly specified, some of the most natural estimators $\widehat{P}$ may still not be accurate enough for this intuitive approach to work well. In fact, even moderate estimation error in $\widehat{P}$ can lead to a substantial discrepancy between the bootstrap distribution and the true distribution.

An example is illustrated in Figure~\ref{fig:bias_triangle_counts} using the triangle count statistic. The figure compares the true distribution of the statistic with bootstrap distributions obtained using different estimates $\widehat{P}$. The observed network is generated from a stochastic block model (SBM) \cite{holland1983stochastic} with $K=3$ communities.
The left panel of Figure~\ref{fig:bias_triangle_counts} shows how model misspecification can affect the bootstrap distribution. As expected, overestimating $K$ generally leads to a relatively milder deviation from the true distribution than underestimating $K$, as it does not necessarily amount to model misspecification. However, the deviation can become more substantial as the overestimation becomes more severe.

In the right panel, we estimate $P$ using three different estimators, $\widehat{P}_{\mathrm{SBM}}$, $\widehat{P}_{\mathrm{DCSBM}}$, and $\widehat{P}_{\mathrm{SVD}}$, with $K=3$; see Section~\ref{section:numerical_expriment}. While all three estimators can be regarded as based on correctly specified models, they have different levels of accuracy.
The bootstrap distribution based on $\widehat{P}_{\mathrm{SBM}}$ recovers the unknown true distribution almost perfectly. However, when using the estimated model $\widehat{P}_{\mathrm{DCSBM}}$ or $\widehat{P}_{\mathrm{SVD}}$, we observe a substantial location shift between the bootstrap distribution and the true distribution. We also report the bootstrap procedure of \cite{levin2019bootstrapping} to
illustrate the difference between our fixed-\(P\) setting and graphon-type
methods in which latent positions are random. Their procedure resamples the
estimated latent positions \(\{\widehat X_i\}_{i=1}^n\) and therefore targets a
distribution that is marginal over latent positions. In contrast, our bootstrap
conditions on the node-specific latent quantities, or equivalently on a fixed
edge-probability matrix. Thus the wider spread of the latent-position bootstrap
reflects a different inferential target, rather than a failure of the method.

This dramatic shift between $\mu(\widehat{P})$ and $\mu(P)$ when bootstrapping
from $\widehat{P}_{\mathrm{DCSBM}}$ or $\widehat{P}_{\mathrm{SVD}}$ is not
merely a consequence of the particular random realization of the observed
network. For subgraph-count statistics, the mean of the bootstrap distribution
generated from $\widehat{P}_{\mathrm{DCSBM}}$ or
$\widehat{P}_{\mathrm{SVD}}$ tends to exhibit positive bias relative to the mean
of the true distribution.
The issue of bootstrap bias, as well as how to address it, is the main focus of
this study. Before turning to this issue, we introduce the working model for
$P$ used throughout much of the paper.

\begin{figure}[tbh]
    \centering
    \begin{minipage}{0.49\textwidth}
    \centering
    \includegraphics[width=1.0\textwidth]{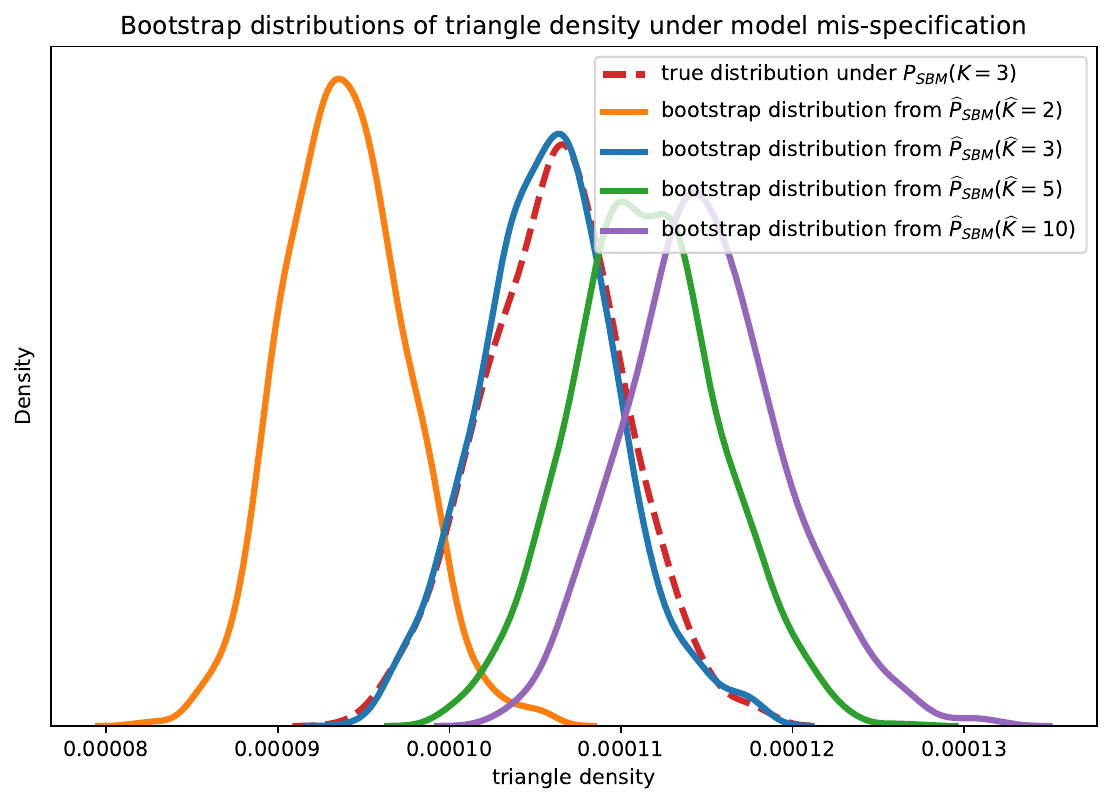}
    \end{minipage}
    \begin{minipage}{0.49\textwidth}
    \centering    \includegraphics[width=1\textwidth]{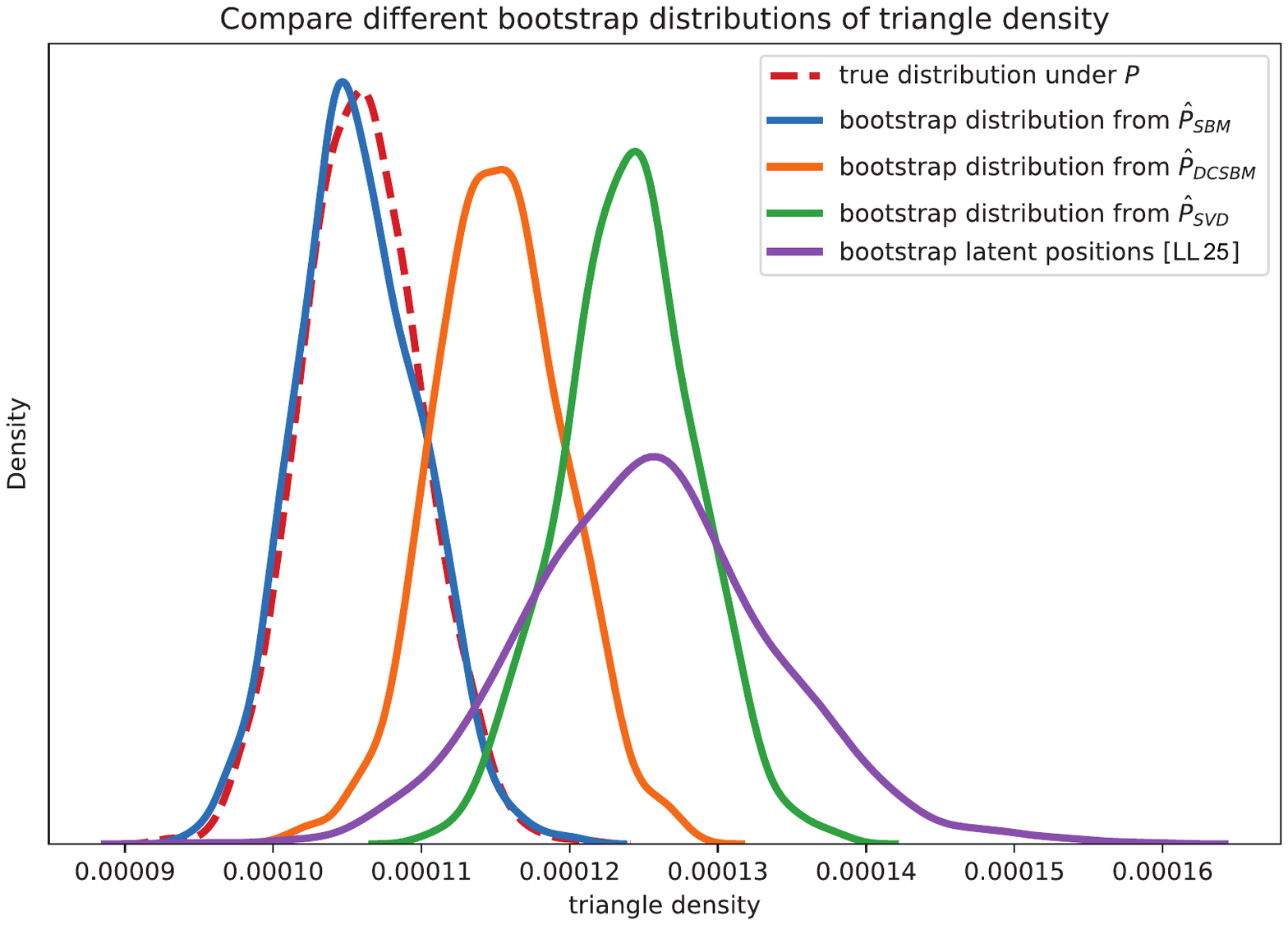}
    \end{minipage}
    \caption{Bootstrap distributions of the triangle density. The red dashed line represents
the true distribution \(F_{P,T}\), and the solid lines represent the bootstrap
distributions \(F_{\widehat P,T}\) based on different estimators
\(\widehat P\). Throughout the paper, dashed lines denote unknown distributions
or quantities, whereas solid lines denote observable or estimable ones. The
observed network is generated from an SBM with \(K=3\) communities. The graphon
bootstrap of \cite{levin2019bootstrapping} is included only to illustrate the
distinction between our fixed edge-probability setting and graphon-type settings; it is
designed for a different inferential target.}     \label{fig:bias_triangle_counts}
\end{figure}

\subsection{Chung-Lu Model}
\label{section:Chung-Lu model}

To introduce the core bias problem without unnecessary complications, we assume in this study that edges are formed independently with probabilities 
\begin{equation*}
    P_{ij} = p \theta_i \theta_j, \quad i \neq j,
\end{equation*}
where $p=p(n)$ is a known parameter controlling the overall edge density, and $\{\theta_i\}_{i=1}^n$ are unknown degree-correction parameters that adjust for individual node degrees. For identifiability, we assume $\sum_{i=1}^n \theta_i = n$. Furthermore, we consider a relatively sparse setting in which $p = o(1)$ and the degree parameters are of constant order as $n\to \infty$. In particular, we assume $n^{-1} \prec p \prec 1$ 
and $c < \min_i \theta_i \le \max_i \theta_i < C$ for some constants $c,C>0$ independent of $n$.
In this context, the notation $x \prec y$ means that $x/y = o(1)$, and we use $x \preccurlyeq y$ to indicate that $x/y = O(1)$.

This model is widely known as the Chung--Lu random graph model \cite{aiello2001random} in the probability literature, where it serves as a fundamental ``null model'' for detecting community structure \cite{newman2006finding}. It generalizes the Erd\H{o}s--R\'{e}nyi model \cite{erdos1960evolution} by allowing degree heterogeneity. It can also be viewed as a special case of the DCSBM \cite{karrer2011stochastic} with only one community.
Finally, we point out that this model is also a special case of the random dot product graph model \cite{young2007random}, with a fixed latent position for each node. The assumption of fixed, rather than random, latent positions distinguishes our work from most other network resampling methods, especially the bootstrap procedure of \cite{levin2019bootstrapping}.

The one-community assumption is made primarily for analytical convenience. It
simplifies the notation and allows us to isolate the main source of bootstrap
bias, namely the effect of estimating the edge-probability matrix. The same
bias phenomenon is expected to arise in degree-corrected stochastic block
models with multiple communities. In particular, when the community labels are
known, or when exact community recovery is possible with probability tending to
one \cite{abbe2018community}, the analysis can be carried out conditional on
the community labels. On this event, the model is a fixed edge-probability
model with node-specific degree parameters and block-level parameters, and the
bootstrap analysis is analogous to the one-community Chung--Lu case.
When community labels are estimated with nonzero error, an additional error
term appears. We expect the resulting approximation to depend on the accuracy
of the label estimator, but controlling this error at the inferential scale
requires a separate analysis. This is beyond the scope of the present work.
Since the Chung--Lu model can be viewed as a one-community DCSBM and serves as
a building block for many degree-heterogeneous network models, the presence of
non-negligible bootstrap bias in this simpler setting suggests that similar
bias phenomena may also arise in richer models and should not
be ignored.

Following the approach of \cite{karrer2011stochastic,zhao2012consistency}, we estimate $P$ by replacing the Bernoulli likelihood with a Poisson likelihood and then using the resulting MLE. This gives
\begin{equation}
\label{eq:Phat_MLE_DCSBM}
    \widehat{P}_{ij} = \widehat{p} \widehat{\theta}_i \widehat{\theta}_j, \quad
\widehat{p} = \frac{S}{n(n-1)}, \quad \widehat{\theta}_i = \frac{n D_i}{S}, \quad D_i=\sum_{j\neq i} A_{ij}, \quad S=\sum_{i=1}^n D_i.    
\end{equation}
Finally, we reiterate that the network bootstrap discussed here is a generic framework applicable, in principle, to any network model together with its corresponding estimation method. In the simulation studies in Section~\ref{section:numerical_expriment}, we consider SBM and DCSBM as underlying models and use several different methods to estimate $P$, yielding different estimators $\widehat{P}$.

\section{Bootstrap Bias of Subgraph Counts}
\label{section:bootstrap_bias}

Given a statistic $T$, our main goal is to construct a confidence interval for its
expectation
\[
\mu := \mu(P) := \E_P[T] := \E_{A\sim P}[T(A)].
\]
For an estimator $\widehat{P}$ of $P$ computed from a single observation $A$,
consider the natural plug-in estimator of $\mu$ given by
\begin{equation*}
\widehat{\mu}\coloneqq\mu(\widehat{P})
  \coloneqq \E_{\widehat{P}}[T]
  \coloneqq \E_{\widehat{A}\sim \widehat{P}}[T(\widehat{A})],
\end{equation*}
which can be approximated with arbitrary  accuracy by bootstrap
sampling $\widehat{A}$ from $\widehat{P}$.
The bootstrap bias of $\widehat{\mu}$ induced by the estimator $\widehat{P}$ is given by
\begin{equation}
    \mathrm{Bias}_P\left(\widehat{\mu}\right)
    \coloneqq \E_P [\widehat{\mu}]- \mu(P).
    \label{eq:true_bias}
\end{equation}

We say that the bias of \(\widehat{\mu}\) is negligible if
\[
   \left|
   \frac{\mathrm{Bias}_P(\widehat{\mu})}
   {\sqrt{\Var_P(\widehat{\mu})}}
   \right|
\]
vanishes asymptotically. This is the relevant criterion for constructing
a valid confidence interval for \(\mu\) based on the distribution of
\(\widehat{\mu}\). A general closeness condition between \(\widehat P\) and
\(P\) is not, by itself, sufficient to guarantee that the bias of
\(\mu(\widehat P)\) is negligible, or that plug-in inference based on
\(\mu(\widehat P)\) is valid. In addition to the closeness of \(\widehat P\) to
\(P\), one must also account for the covariance structure of \(\widehat P\),
which determines the variance scale relevant for inference.

It turns out that, for many statistics, even when the model parameters are
consistently estimated, the bias of $\widehat{\mu}$ can still be non-negligible.
For example, we will see later in
\eqref{eq:Bias_Varmu_ratio_general_subgraph_count} that, for triangle counts,
the bias-to-standard-deviation ratio of $\widehat{\mu}$ is of order
$\Theta(p^{-1/2})$ under the Chung--Lu model.
Figure~\ref{fig:bootstrap_bias} further illustrates the non-negligible biases
of two related statistics: global transitivity and the local clustering
coefficient. We use these two statistics as running examples and revisit them
in Figures~\ref{fig:bias_reduction_plot}, \ref{fig:bootstrap_CI_example},
\ref{fig:CI_global_transitivity}, and~\ref{fig:CI_local_transitivity_high}.
Definitions of these statistics and details of the settings used in these
figures are provided in
Section~\ref{section:simulation_settings_for_example_figures}.


\begin{figure}[h]
\centering
\begin{minipage}{0.49\textwidth}
\centering
\includegraphics[width=1.0\textwidth]{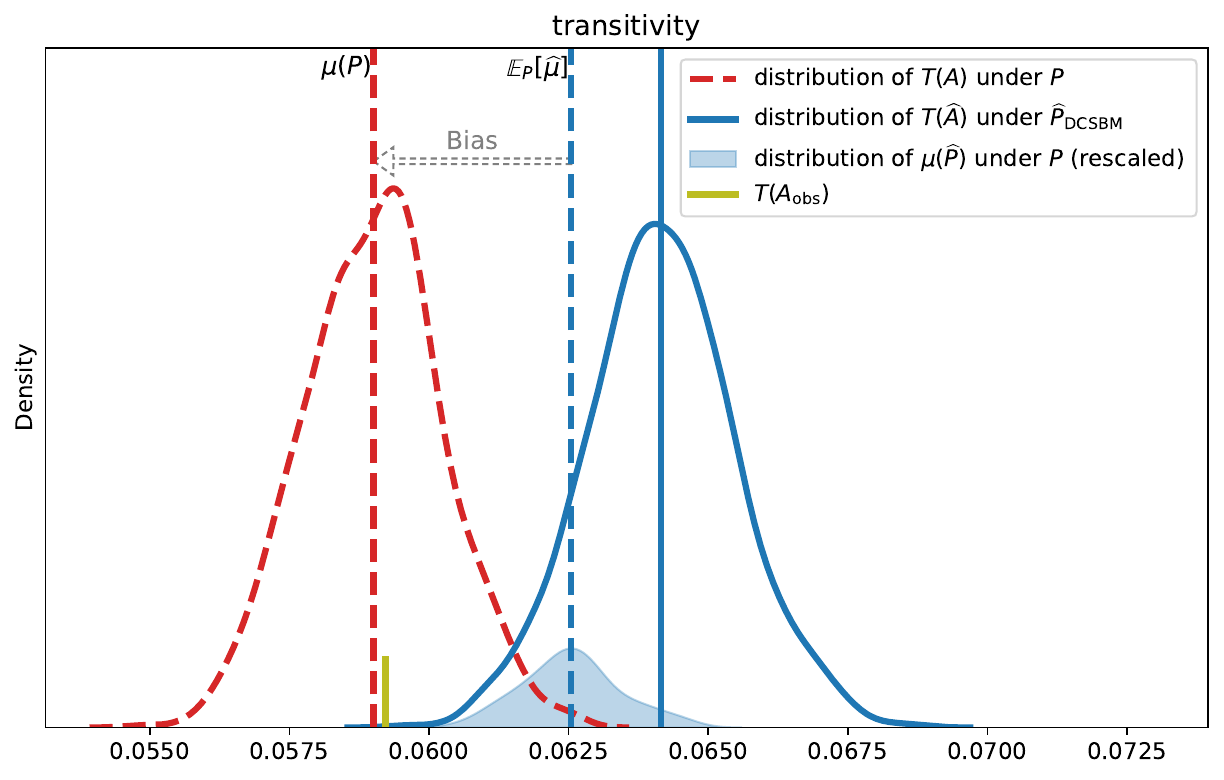}
\end{minipage}
\begin{minipage}{0.49\textwidth}
\centering
\includegraphics[width=1.0\textwidth]{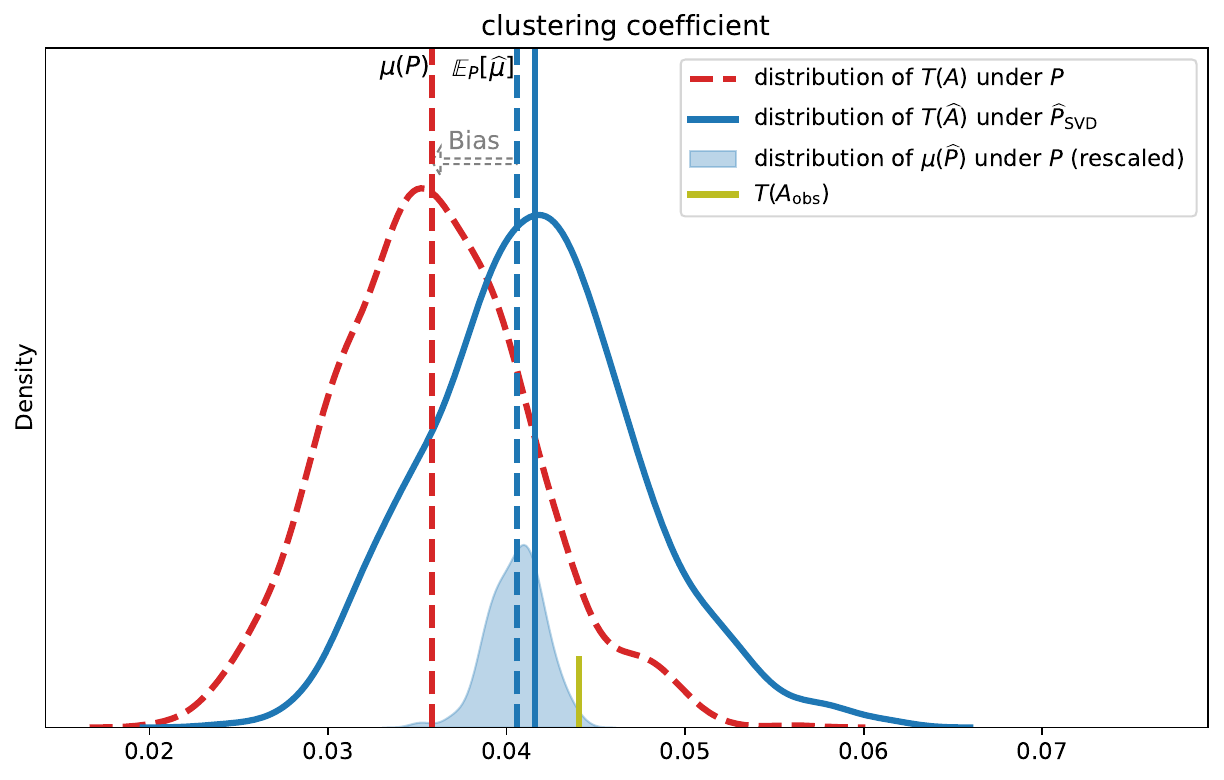}
\end{minipage}
\caption{Bootstrap bias of global transitivity (left) and the local clustering
coefficient of a given node (right). The network is generated from a DCSBM; see
Section~\ref{section:simulation_settings_for_example_figures} for details. The
model is estimated using $\widehat{P}_{\textrm{DCSBM}}$ on the left and
$\widehat{P}_{\textrm{SVD}}$ on the right.
The unknown true distribution and its mean $\mu(P)$ are shown by the red dashed
curves and vertical lines, while the bootstrap distribution and its mean
$\mu(\widehat{P})$ are shown by the blue solid curves and vertical lines. The
blue shaded area at the bottom shows the unknown distribution of
$\mu(\widehat{P})$, rescaled for illustration, and the blue dashed vertical line
represents the unknown value $\E_P[\widehat{\mu}]$. The distance between the two
dashed vertical lines corresponding to $\E_P[\widehat{\mu}]$ and $\mu(P)$ is the bootstrap bias.}
\label{fig:bootstrap_bias}
\end{figure}

To gain further insight into this issue, we derive the asymptotic behavior of
the bootstrap bias for general subgraph-count statistics. These statistics form
a canonical and rich class of network summaries, analogous to moments of random
variables in the classical setting \cite{bickel2011method}. This makes them a
natural class of statistics for studying how the network distribution under the
true model $P$ differs from the distribution under the estimated model
$\widehat P$.
We show that the bootstrap bias is of non-negligible order, arising from
the covariance structure of the estimated entries of $\widehat P$. For this analysis,
let $R$ be a fixed connected graph, meaning that any two vertices of $R$ are
joined by a path. Suppose that $R$ does not depend on $n$, and let $V(R)$ and
$E(R)$ denote its vertex set and edge set, respectively. Denote
\[
v=|V(R)|,
\qquad
e=|E(R)|.
\]
Let \(A\) be the adjacency matrix of the observed graph on vertex set
\([n]\), and let \(T_R(A)\) denote the number of copies of \(R\) in \(A\),
counted modulo automorphisms of \(R\). It can be written as
\begin{equation}
\label{eq:global count}
T_R(A)
=
\frac{1}{|\operatorname{Aut}(R)|}
\sum_{\phi:V(R)\hookrightarrow[n]}
\prod_{\{a,b\}\in E(R)}
A_{\phi(a)\phi(b)},    
\end{equation}
where the sum is over injective maps \(\phi:V(R)\hookrightarrow[n]\), and
\(\operatorname{Aut}(R)\) is the automorphism group of \(R\). The factor
\(|\operatorname{Aut}(R)|^{-1}\) removes the over-counting caused by relabeling the
vertices of \(R\). 
For example, when $R=\Delta$ is the complete graph with three vertices and
three edges, the triangle count $T_{\Delta}(A)$ can be simplified as
\begin{equation}
    T_{\Delta}(A) = \sum_{i<j<k} A_{ij} A_{ik} A_{jk},
    \label{eq:triangle_count_T}
\end{equation}
and its expectation is
\begin{equation}
    \mu _{\Delta}(P) = \E_P [T_{\Delta}(A)]  =  \sum_{i<j<k} P_{ij} P_{ik} P_{jk}.   \label{eq:triangle_count_expectation_general}
\end{equation}

We also study the bias of the rooted subgraph counts considered by \cite{maugis2020central}. 
For a rooted graph $R$
with root vertex $r$, the number of non-automorphic copies of $R$ rooted at a
fixed node $i$ can be written as 
\begin{equation}
T_{R,r}^{(i)}(A) =
\frac{1}{|\operatorname{Aut}(R,r)|}
\sum_{\substack{\phi:V(R)\hookrightarrow[n] \\ \phi(r)=i}} \ 
\prod_{\{a,b\}\in E(R)}
A_{\phi(a)\phi(b)},    
\label{eq:rooted count}
\end{equation}
where $\mathrm{Aut}(R,r):=\{\varphi\in\mathrm{Aut}(R):\varphi(r)=r\}$ is the root-preserving automorphism group. 
For example, the count of triangles rooted at node
$i$ is simply the number of triangles containing node $i$, given by
\begin{equation*}
    T_{\Delta}^{(i)}(A)
    =
    \sum_{\substack{j<k\\ j\neq i,\ k\neq i}}
    A_{ij}A_{ik}A_{jk}.
\end{equation*}
Besides being interesting objects to study in their own right, rooted subgraph
counts also play an important role in variance calculations: the variance of
estimators of global subgraph counts depends on expectations of these rooted
subgraph counts; for details, see Appendix~\ref{sec:variance}. Analogously to
the global subgraph counts, we denote
\begin{equation}
\label{eq:rooted_estimator_and_bias}
\mu_{R,r}^{(i)}
\coloneqq
\mu_{R,r}^{(i)}(P)
=
\E_P\left[T_{R,r}^{(i)}(A)\right],
\quad
\widehat{\mu}_{R,r}^{(i)}
\coloneqq
\mu_{R,r}^{(i)}(\widehat{P}),
\quad
\mathrm{Bias}\left(\widehat{\mu}_{R,r}^{(i)}\right)
\coloneqq
\E_P\left[\widehat{\mu}_{R,r}^{(i)}\right]
-
\mu_{R,r}^{(i)}.
\end{equation}

The following proposition gives the order of the bootstrap bias under the Chung--Lu model, using the maximum likelihood estimator in
\eqref{eq:Phat_MLE_DCSBM}. Its proof is given in Appendix~\ref{sec:proofs of main results}.

\begin{proposition}[Bootstrap bias of subgraph counts]
\label{prop:bias_general_subgraph_count_DCSBM}
Assume that the network is sampled from the Chung--Lu model described in
Section~\ref{section:Chung-Lu model}, with $n^{-1}\prec p\prec 1$. For any fixed
connected subgraph $R$ with $v\ge 3$ vertices and $e\ge 2$ edges, suppose that
$P$ is estimated by $\widehat{P}_{\textrm{MLE}}$ in
\eqref{eq:Phat_MLE_DCSBM}. Then, for the rooted count of $(R,r)$ at node $i$ defined
in \eqref{eq:rooted count},
\begin{equation}   \mathrm{Bias}_P\left(\widehat{\mu}_{R,r}^{(i)}\right)
\asymp n^{v-2}p^{e-1},
\label{eq:bias_order_general_rooted_subgraph_count}
\end{equation}
whereas for the global count of $R$ defined in
\eqref{eq:global count},
\begin{equation}    \mathrm{Bias}_P(\widehat{\mu}_R)
\asymp n^{v-1}p^{e-1}.   \label{eq:bias_order_general_subgraph_count}
\end{equation}
The constants implicit in the $\Theta(\cdot)$ notation in
\eqref{eq:bias_order_general_rooted_subgraph_count} and
\eqref{eq:bias_order_general_subgraph_count} only depend on $R$ and
$\max_{1\le i\le n} \theta_i$.
\end{proposition} 

As a special case, suppose that the true model $P$ is the Erd\H{o}s--R\'enyi
model, that is, $\theta_i\equiv 1$. Then, for triangle counts, the bootstrap
biases simplify to
\begin{equation}   \mathrm{Bias}_P\left(\widehat{\mu}_{\Delta,1}^{(i)}\right)
    \asymp np^2,
    \qquad
\mathrm{Bias}_P(\widehat{\mu}_{\Delta})
    \asymp n^2p^2.
    \label{eq:bootstrap_bias_global_triangle_count_ER_estimated_by_DCSBM}
\end{equation}

To determine whether the bootstrap bias is negligible, we compare
$\mathrm{Bias}_P(\widehat{\mu}_R)$ with
$\left(\Var_P(\widehat{\mu}_R)\right)^{1/2}$. The following lemma gives the
order of this variance. Its proof is given in Appendix~\ref{sec:proofs of main results}.

\begin{lemma}[Variances of plug-in estimators of subgraph means] \label{lemma:variance_of_hat_mu_R}
Under the same setting as in
Proposition~\ref{prop:bias_general_subgraph_count_DCSBM}, for the rooted count
of $R$ with root $r$, we have
\begin{equation}
    \Var_P(\widehat{\mu}_{R,r}^{(i)})
    =
    \Theta(n^{2v-3}p^{2e-1}),
    \label{eq:variance_of_hat_mu_R_i}
\end{equation}
and, for the global count of $R$,
\begin{equation}
    \Var_P(\widehat{\mu}_R)
    =
    \Theta(n^{2v-2}p^{2e-1}).
    \label{eq:variance_of_hat_mu_R}
\end{equation}
\end{lemma}

Combining Lemma~\ref{lemma:variance_of_hat_mu_R} with
Proposition~\ref{prop:bias_general_subgraph_count_DCSBM}, we obtain, for the
global subgraph count,
\begin{equation}   \frac{\mathrm{Bias}_P(\widehat{\mu}_R)}
   {\sqrt{\Var_P(\widehat{\mu}_R)}}
   =
   \Theta(p^{-1/2}),
   \label{eq:Bias_Varmu_ratio_general_subgraph_count}
\end{equation}
which is clearly non-negligible. This verifies that the location shift observed
in Figure~\ref{fig:bootstrap_bias} is not due to random fluctuation.
On the other hand, for the rooted subgraph count,
\begin{equation}
   \frac{\mathrm{Bias}_P(\widehat{\mu}_{R,r}^{(i)})}
   {\sqrt{\Var_P(\widehat{\mu}_{R,r}^{(i)})}}
   =
   \Theta\left((np)^{-1/2}\right),
   \label{eq:Bias_Varmu_ratio_general_rooted_subgraph_count}
\end{equation}
which converges to zero, but possibly as slowly as $(np)^{-1/2}$. As a result,
for relatively sparse finite-size networks, we may still observe a non-negligible
amount of bootstrap bias.

So far, we have discussed the bias problem under the Chung--Lu model when the
maximum likelihood estimator $\widehat{P}_{\textrm{MLE}}$ defined in
\eqref{eq:Phat_MLE_DCSBM} is used. Recall that, in
Figure~\ref{fig:bias_triangle_counts}, the triangle count also exhibits
significant bootstrap bias under SVD estimation. We now consider a different
estimator $\widehat{P}$ and show that this problem is not specific to the MLE.

Let $\theta=(\theta_1,\ldots,\theta_n)^\top$ and
$\widetilde{P}=p{\theta}{\theta}^\top$. Then $\widetilde{P}$ has rank
one, and $P_{ij}=\widetilde{P}_{ij}$ for all $i\neq j$. Therefore, a rank-one
approximation of $A$ of the form
$\widehat{P}_{\textrm{SVD}}=\widehat{\lambda}_1\widehat{v}_1\widehat{v}_1^\top$,
where $\widehat{\lambda}_1=\lambda_1(A)$ is the largest eigenvalue of $A$ and
$\widehat{v}_1$ is the corresponding eigenvector, provides another natural
estimator of $P$.
We argue that similar bootstrap bias also arises for networks sampled from
$\widehat{P}_{\textrm{SVD}}$. To illustrate the magnitude of this bias, we
consider triangle counts under the Erd\H{o}s--R\'enyi model as a specific
example. The following proposition shows that the resulting bootstrap bias is
not negligible. Its proof is given in
Appendix~\ref{sec:proofs of main results}.

\begin{proposition}[Bootstrap bias of triangle count under SVD estimation]
\label{prop:bias_triangle_count_SVD}
Let \(A\sim G(n,p)\) be an Erd\H{o}s--R\'enyi graph with fixed edge probability 
\(p\in(0,\varepsilon)\) and \(\varepsilon<1/2\) is a sufficiently
small constant.  Let
$
\widehat P_{\mathrm{SVD}}
:=
\widehat\lambda_1\widehat v_1\widehat v_1^\top
$
be the rank-one truncated SVD estimator of the edge-probability matrix, where
\(\widehat\lambda_1=\lambda_1(A)\) is the largest eigenvalue of \(A\), and
\(\widehat v_1\) is a corresponding unit eigenvector.
Then the bootstrap bias of triangle count \eqref{eq:triangle_count_T} is non-negligible,
\begin{equation*}
    \frac{\mathrm{Bias}_P\left(\mu_{\Delta}(\widehat{P}_{\textrm{SVD}})\right)}{\sqrt{\Var_P\left(\mu_{\Delta}(\widehat{P}_{\textrm{SVD}})\right)}} \asymp p^{-1/2}.
\end{equation*}
\end{proposition}



The non-negligible bootstrap bias identified here is closely related to familiar
phenomena in high-dimensional inference. Most directly, it resembles the
incidental-parameter problem: when the number of nuisance parameters grows with
the sample size, estimating these nuisance parameters can induce bias in
downstream target estimators
\cite{neyman1948consistent,lancaster2000incidental}. In the Chung--Lu model,
the \(n\) node-specific parameters play an analogous role. Although the network contains $O(n^2)$ dyads, each node-specific parameter is
estimated primarily from its incident edges, so the node-parameter estimation
errors are of order $(np)^{-1/2}$.
When a nonlinear network statistic is evaluated at the estimated
edge-probability matrix, these errors enter through second-order terms and can
accumulate, producing a bias that is non-negligible relative to the standard
deviation of the statistic. This perspective also connects our results to
high-dimensional smooth functional estimation, where plug-in estimators of
nonlinear functionals of high-dimensional parameters often require explicit
bias correction because second-order terms need not be negligible
\cite{koltchinskii2022estimation}. Our contribution is different in that we
study how estimating these nuisance quantities affects bootstrap inference for
network statistics. The proposed two-level bootstrap acts as a
simulation-based bias correction, analogous in spirit to bias correction and
debiasing methods in high-dimensional inference
\cite{hahn2004jackknife,chernozhukov2018double}.

\section{Bias Reduction via the Double Bootstrap}
\label{section:bias_correction}

Given the existence of bootstrap bias, a natural question is how to correct it.
For each statistic $T$ of interest, one possible approach is to derive an
analytic formula for $\mathrm{Bias}_P(\widehat{\mu}_T)$ as a function of $P$ and
then estimate it by replacing $P$ with its estimator $\widehat{P}$. Although
this case-by-case approach works for small subgraphs $R$, such as edges or
triangles, it quickly becomes impractical for more complex statistics.

A more general and principled solution is to \textit{implicitly} estimate the
unknown bias $\mathrm{Bias}_P(\widehat{\mu})$ using a plug-in estimator
$\widehat{\mathrm{Bias}}_P(\widehat{\mu})$ obtained through an additional level
of bootstrap. Specifically, we independently sample $B$ networks
$\{\widehat{A}_1,\ldots,\widehat{A}_B\}$ from the estimated model $\widehat{P}$.
For each bootstrap network $\widehat{A}_b\sim\widehat{P}$, $1\le b\le B$, we
obtain an estimate $\widehat{\widehat{P}}_b$ of $\widehat{P}$ using the same
method used to estimate $P$ from $A$.

For each $\widehat{\widehat{P}}_b$, we similarly define
\begin{equation*}
    \widehat{\widehat{\mu}}_b
    \coloneqq \mu(\widehat{\widehat{P}}_b)
    \coloneqq \E_{\widehat{\widehat{P}}_b}[T]
    \coloneqq
    \E\left[
        T(\widehat{\widehat{A}})
        \mid
        \widehat{\widehat{A}}\sim \widehat{\widehat{P}}_b
    \right].
\end{equation*}
Evaluating $\widehat{\widehat{\mu}}_b$ by Monte Carlo requires another level of
bootstrap networks sampled from $\widehat{\widehat{P}}_b$. This two-level
bootstrap procedure is illustrated in Figure~\ref{fig:procedure}.
We then approximate the true bias in \eqref{eq:true_bias} by the estimated
bootstrap bias
 \begin{equation*}    \widehat{\mathrm{Bias}}_{\widehat{P}} \coloneqq
     \E_{\widehat{P}}[\widehat{\widehat{\mu}}] - \widehat{\mu} = \E_{\widehat{P}}[ \mu(\widehat{\widehat{P}})] - \mu(\widehat{P}) \label{eq:estimated_bias} 
     \approx \frac{1}{B}\sum_{b=1}^B \widehat{\widehat{\mu}}_b - \mu(\widehat{P}).
 \end{equation*}
Using the approximation
\begin{equation*} 
    \E_{\widehat{P}}[\mu(\widehat{\widehat{P}})] - \mu(\widehat{P})
    \approx
    \E_P[\mu(\widehat{P})] - \mu(P),
\end{equation*}
we can reduce the bias by using an additive correction term. In
particular, the bias-corrected point estimate is
\begin{equation*}
    \widehat{\mu} - \widehat{\mathrm{Bias}}_{\widehat{P}}
    =
    2\widehat{\mu} - \E_{\widehat{P}}[\widehat{\widehat{\mu}}].
\end{equation*}
This bias-reduction method is inspired by the idea of iterating the bootstrap
principle \cite{hall1988bootstrap, hall2013bootstrap}. The idea was first
introduced by \cite{efron1983estimating} and later developed and unified in
Hall's textbook \cite{hall1988bootstrap}. In the classical i.i.d. setting, each bootstrap
iteration reduces the error by a factor of at least $n^{-1/2}$.

Note that the remaining bias after correction is given by
\begin{equation*}
    \Delta \mathrm{Bias}
    \coloneqq
    \widehat{\mathrm{Bias}}_{\widehat{P}} - \mathrm{Bias}_P,
\end{equation*}
and remaining bias after correction for rooted subgraph counts $\Delta \mathrm{Bias}^{(i)}$ is defined analogously. 
The next proposition shows that the second-level bootstrap reduces the order of
the bias by a factor of $(np)^{-1/2}$ for rooted counts and by a factor of
$(np)^{-1}$ for global counts. This result essentially extends the
classical iterative bootstrap to the Chung--Lu setting, where the number of
parameters is not fixed but grows at rate $n$. The proof is given in
Appendix~\ref{sec:proofs of main results}.

\begin{figure}
\centering
\begin{minipage}{.49\textwidth}
\centering
\includegraphics[width=1.0\textwidth]{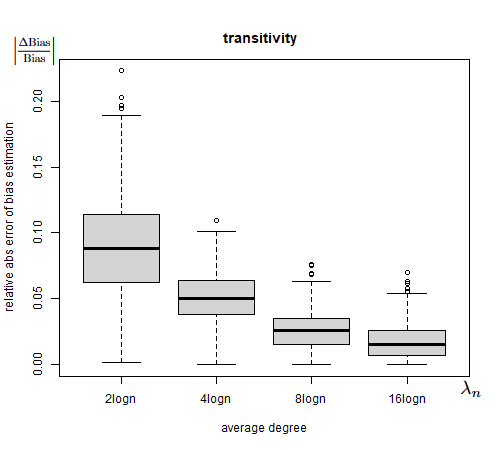}
\end{minipage}
\begin{minipage}{.49\textwidth}
\centering
\includegraphics[width=1.0\textwidth]{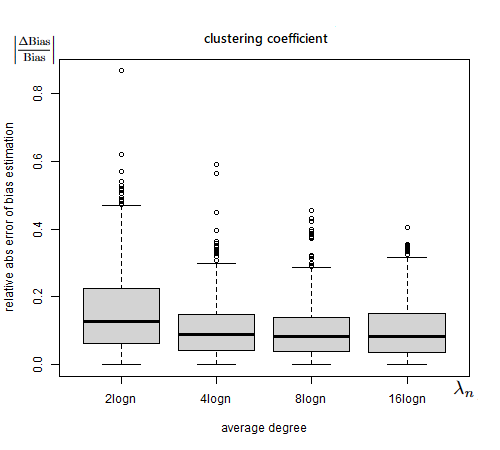}
\end{minipage}
\caption{Bias reduction for global transitivity (left) and the clustering coefficient
(right) as the average degree $\lambda_n$ increases. The $y$-axis shows
$\frac{|\Delta \mathrm{Bias}|}{\mathrm{Bias}}$. The model is estimated by the
DCSBM on the left and by truncated SVD on the right. See
Section~\ref{section:simulation_settings_for_example_figures} for the model
settings.
}
\label{fig:bias_reduction_plot}
\end{figure}

\begin{proposition}[Bias correction via the second-level bootstrap] 
\label{prop:bias_correction_general_subgraph}
Assume that the network is sampled from the Chung--Lu model described in
Section~\ref{section:Chung-Lu model}, with $n^{-1}\prec p\prec 1$. For any
fixed connected subgraph $R$ with $v\ge 3$ vertices and $e\ge 2$ edges, suppose
that $P$ is estimated by $\widehat{P}_{\textrm{MLE}}$ in
\eqref{eq:Phat_MLE_DCSBM}. Further suppose that
$\mathrm{Bias}_P(\widehat{\mu}_R)$ and
$\mathrm{Bias}_P(\widehat{\mu}_{R,r}^{(i)})$ are estimated by
$\widehat{\mathrm{Bias}}_{\widehat{P}}$ and
$\widehat{\mathrm{Bias}}_{\widehat{P}}^{(i)}$, respectively, and that the
additive bias correction is applied. Then, for the rooted count of $R$ at node
$i$ defined in \eqref{eq:rooted count},
\begin{equation*}
    \frac{\Delta \mathrm{Bias}_{R,r}^{(i)}}   {\mathrm{Bias}_P(\widehat{\mu}_{R,r}^{(i)})}
=O_{\Prob}\left((np)^{-1/2}\right),
\end{equation*}
whereas for the global count of $R$ defined in
\eqref{eq:global count},
\begin{equation*}
    \frac{\Delta \mathrm{Bias}_R}   {\mathrm{Bias}_P(\widehat{\mu}_R)}
=O_{\Prob}\left((np)^{-1}\right).
\end{equation*}
\end{proposition}

Combining Proposition~\ref{prop:bias_correction_general_subgraph} with
\eqref{eq:Bias_Varmu_ratio_general_subgraph_count} and
\eqref{eq:Bias_Varmu_ratio_general_rooted_subgraph_count}, the reduced bias is
negligible relative to the standard deviation:
\begin{equation*}
    \frac{\Delta \mathrm{Bias}_{R,r}^{(i)}}
    {\sqrt{\Var_P(\widehat{\mu}_{R,r}^{(i)})}}
    =
    O_{\Prob}\left((np)^{-1}\right),
    \qquad
    \frac{\Delta \mathrm{Bias}_R}
    {\sqrt{\Var_P(\widehat{\mu}_R)}}
    =
    O_{\Prob}\left(n^{-1}p^{-3/2}\right).
\end{equation*}
These ratios converge to zero faster than the corresponding ratios for the
original bootstrap bias.

Figure~\ref{fig:bias_reduction_plot} demonstrates the empirical effectiveness
of bias reduction via the second-level bootstrap. In particular, the relative
remaining bias $\frac{|\Delta \mathrm{Bias}|}{\mathrm{Bias}}$ is much smaller
than one and clearly decreases toward zero as the average degree $\lambda_n$
increases. Figure~\ref{fig:bootstrap_CI_example} further shows that the
estimated bias $\widehat{\mathrm{Bias}}_{\widehat{P}}$ provides a good
approximation to the unknown bias $\mathrm{Bias}_P$.

\section{Confidence Intervals for Expected Subgraph Counts}
\label{section:confidence_intervals}

As mentioned in Section~\ref{sec:setup}, one of our main goals is to construct
confidence intervals for expected network statistics, such as subgraph counts
$\mu(P)$. For this purpose, it is natural to use the plug-in estimate
$\mu(\widehat{P})$ and its distribution to make inference about $\mu(P)$.
However, Section~\ref{section:bootstrap_bias} shows that directly using the
distribution of $\mu(\widehat{P})$ can be problematic, as the bias may grow
faster than its standard deviation, resulting in invalid confidence intervals.
To address this problem, we showed in Section~\ref{section:bias_correction}
that the bootstrap bias can be corrected via the second-level bootstrap.
Equipped with these insights, we now proceed to construct confidence intervals
for $\mu(P)$, again focusing on subgraph-count statistics because they form a
canonical and rich class of network summaries and because the insights obtained
in the previous sections apply naturally to them.

\subsection{Confidence Intervals Based on the Distributions of Statistics}
\label{section:CI_based_on_T}

Before using $\mu_R(\widehat P)$ to construct confidence intervals for
$\mu_R(P)$, we first consider a simpler baseline approach. Since
\[
\mu_R(P)=\mathbb E_P[T_R(A)],
\]
one can instead base inference on the sampling distribution of the observed
subgraph count $T_R$. Although the finite-sample distribution of $T_R$ is
generally unknown, its asymptotic distribution is available from the following
result. As discussed in Sections~\ref{section:CI_based_on_mu} and \ref{section:variance_comparison}, while this direct count-based approach is
conceptually simpler, its accuracy depends on how the variance of $T_R$ compares
with the variance of the plug-in quantity $\mu_R(\widehat P)$.

\begin{proposition}[Asymptotic normality of subgraph counts]
\label{prop:asymptotic_normality_triangle_count}
Consider a network sampled from the inhomogeneous Erd\H{o}s--R\'enyi model with
independent edges and edge probabilities $P_{ij}$ satisfying
$cp\le P_{ij}\le Cp$ for some constants $c,C>0$. For any fixed subgraph $R$,
the following statements hold.

\medskip
\noindent
(i) If $n^{-1/m_R}\prec p\preccurlyeq 1$, where 
\[
    m_R=\max\left\{\frac{e(H)}{v(H)}: H\subseteq R, \ e(H)\ge 1, \ H \text{ has no isolated vertices}\right\},
\]
then
\[
    \frac{T_R-\E[T_R]}{\sqrt{\Var(T_R)}}
    \xrightarrow{D} N(0,1).
\]

\medskip
\noindent
(ii) For fixed vertices $r\in V(R)$ and $i\in[n]$, if $    n^{-1/m_{R,r}}\prec p\preccurlyeq 1$, where
\[
    m_{R,r}
    =
    \max\left\{\frac{e(H)}{v(H\setminus\{r\})}: H\subseteq R, \ e(H)\ge 1, \ H \text{ has no isolated vertices}\right\},
\]
then
\[
    \frac{T_{R,r}^{(i)}-\E[T_{R,r}^{(i)}]}
    {\sqrt{\Var(T_{R,r}^{(i)})}}
    \xrightarrow{D}
    N(0,1).
\]
\end{proposition}

Proposition~\ref{prop:asymptotic_normality_triangle_count} can be viewed as an extension of the result of
\cite{rucinski1988small}, which establishes asymptotic normality for subgraph
counts under the inhomogeneous Erd\H{o}s--R\'enyi model, a model that contains
the Chung--Lu model as a special case. The proof is given in
Appendix~\ref{sec:proofs of main results}.

It is worth mentioning that \cite{zhang2022edgeworth} also establishes
analogous asymptotic normality results for network moments; see also
\cite{ZhangXiaReducedCount2022} for results on reduced subgraph counts. While
their main focus is on the marginal distribution of subgraph counts under a
graphon model, their asymptotic normality result is also applicable to the
conditional-on-$P$ setting, as implied by their Lemma 3.1(b); see the discussion
section of \cite{zhang2022edgeworth}.

However, we highlight several distinctions between
Proposition~\ref{prop:asymptotic_normality_triangle_count} and their
Lemma 3.1(b). First, our result applies to the general inhomogeneous
Erd\H{o}s--R\'enyi model under the assumption that the edge probabilities are of
comparable magnitudes, whereas their result requires $P$ to satisfy certain
regularity conditions, which hold with high probability when $P$ is sampled from
a non-degenerate graphon model. Second, for many subgraphs, our sparsity
condition is weaker. For example, for triangles, Lemma 3.1(b) of
\cite{zhang2022edgeworth} requires $p\succ n^{-2/3}$, whereas
Proposition~\ref{prop:asymptotic_normality_triangle_count} only requires
$p\succ n^{-1}$ for the global count. Finally, while their result applies to
global counts, our result applies to both global and rooted counts.
 
We also note that several bootstrap results for subgraph counts under
exchangeable random graph or graphon models impose conditions that depend on
the cycle structure of the subgraph
\cite{bhattacharyya2015subsampling, green2022bootstrapping, zhang2022edgeworth}.
Proposition~\ref{prop:asymptotic_normality_triangle_count} imposes a condition
of a similar spirit: the required network density depends on the structure of
the subgraph \(R\). However, our condition is expressed through the densest subgraph of \(R\), and we do not impose separate assumptions
depending on whether \(R\) is cyclic or acyclic.
The reason is that our setting differs from the graphon setting. Under graphon
or exchangeable random graph models, randomness in the latent positions can
lead to degeneracy of the associated \(U\)-statistic kernels, and the limiting
behavior may depend on which Hoeffding projection is the first nonzero one. In
contrast, we work conditionally on a fixed edge-probability matrix \(P\). In
this conditional setting, the relevant dependence structure is generated by
shared edges among copies of \(R\). Therefore, the variance and higher
cumulants are governed by the finite collection of edge-overlap subgraphs of
\(R\), and the cumulant argument requires the corresponding overlap scale to
diverge. Thus, the condition is expressed through the dominant edge-overlap
structure of \(R\), rather than through a separate restriction related to
cycles.

To apply Proposition~\ref{prop:asymptotic_normality_triangle_count}, we need
consistent estimates of the variances of $T_R$ and $T_{R,r}^{(i)}$. Fortunately,
the following result shows that these variances can be estimated using the
second-level bootstrap. The proof is given in
Appendix~\ref{sec:proofs of main results}.

\begin{proposition}[Variance approximation for subgraph counts]
\label{prop:variance_approximation_general}
Assume that the network is sampled from the Chung--Lu model described in
Section~\ref{section:Chung-Lu model}, with $n^{-1}\prec p\prec 1$. For any
fixed connected subgraph $R$ with $v\ge 3$ vertices and $e\ge 2$ edges, suppose
that $P$ is estimated by $\widehat{P}_{\textrm{MLE}}$ in
\eqref{eq:Phat_MLE_DCSBM}. Further suppose that $\Var_P(T_R)$ is estimated by
$\Var_{\widehat{P}}(T_R)$. Then, for the rooted count of $R$ at node $i$ defined
in \eqref{eq:rooted count},
\begin{equation*}
    \frac{\Var_{\widehat{P}}(T_R^{(i)})-\Var_P(T_R^{(i)})}{\Var_P(T_R^{(i)})}
=O_{\Prob}\left((np)^{-1/2}\right),
\end{equation*}
whereas for the global count of $R$ defined in
\eqref{eq:global count},
\begin{equation*}
    \frac{\Var_{\widehat{P}}(T_R)-\Var_P(T_R)}{\Var_P(T_R)}
=O_{\Prob}\left((np)^{-1}\right).
\end{equation*}
\end{proposition}

We are now ready to construct confidence intervals for $\mu_R(P)$ based on the
single observation $T=T_R$, focusing on two-sided intervals with nominal coverage
$1-\alpha$; the confidence intervals for rooted counts are analogous and are
therefore omitted. Using Propositions~\ref{prop:asymptotic_normality_triangle_count}
and \ref{prop:variance_approximation_general}, we consider the following
symmetric normal confidence interval for statistics that are asymptotically
normal:
\begin{equation}
    T(A_{\textrm{obs}})
    \pm
    \Phi^{-1}\left(1-\frac{\alpha}{2}\right)
    \sqrt{\Var_{\widehat{P}}(T)},
    \label{eq:CI_obs_sym}
\end{equation}
where $\Var_P(T)$ is approximated by $\Var_{\widehat{P}}(T)$.

For network statistics that are not asymptotically normal, or whose convergence
to normality is slow, one may instead use the following asymmetric, nonnormal
confidence interval:
\begin{equation}
    \left[
    T(A_{\textrm{obs}})
    -
    \left\{
        F^{-1}_{\widehat{P},T}\left(1-\frac{\alpha}{2}\right)
        -
        \mu(\widehat{P})
    \right\},
    \ 
    T(A_{\textrm{obs}})
    -
    \left\{
        F^{-1}_{\widehat{P},T}\left(\frac{\alpha}{2}\right)
        -
        \mu(\widehat{P})
    \right\}
    \right],
    \label{eq:CI_obs_asym}
\end{equation}
where $F_{\widehat{P},T}$ denotes the distribution of $T(\widehat{A})$ under
$\widehat{P}$. Recall that $F_{P,T}$, the distribution of $T(A)$ under $P$, is
unknown. If $F_{\widehat{P},T}$ has approximately the same shape as $F_{P,T}$,
possibly up to a location shift, then
\begin{equation*}
\Prob \left\{
\begin{array}{c}
    \E_P[T(A)]
    +
    \left[
        F^{-1}_{\widehat{P},T}\left(\frac{\alpha}{2}\right)
        -
        \E_{\widehat{P}}[T(\widehat{A})]
    \right]
    \le T(A)
    \\
    \le
    \E_P[T(A)]
    +
    \left[
        F^{-1}_{\widehat{P},T}\left(1-\frac{\alpha}{2}\right)
        -
        \E_{\widehat{P}}[T(\widehat{A})]
    \right]
\end{array}
\right\}
\approx 1-\alpha.
\end{equation*}

Both the normal-type interval in \eqref{eq:CI_obs_sym} and the nonnormal-type
interval in \eqref{eq:CI_obs_asym} are centered at the single observation
$T(A_{\textrm{obs}})$, with their widths determined by the estimated variability
of its distribution. In some cases, these intervals can be quite wide due to the
large variance of a single observation.
Fortunately, in such cases, we can construct a more informative confidence
interval based on the distribution of $\mu(\widehat{P})$, which we discuss in
the next section.

We also note that there is a computational cost difference between using the normal approximation and using the empirical bootstrap. When an asymptotic normal approximation is available, the
bootstrap samples are used mainly to estimate the variance, and the quantiles
are obtained analytically from the normal distribution. In this case, a few
hundred bootstrap samples are often sufficient for a stable variance estimate.
By contrast, empirical bootstrap confidence intervals require estimating tail
quantiles directly from the bootstrap distribution, which typically requires
many more bootstrap samples, often on the order of thousands. This distinction
is especially important for the two-level bootstrap, whose computational cost
scales with the product of the numbers of first- and second-level bootstrap
replicates. Thus, when the normal approximation is accurate, it provides a much
more computationally efficient implementation of the bias-corrected bootstrap;
the empirical bootstrap remains useful when the sampling distribution is
noticeably non-Gaussian or asymmetric.

\begin{figure}
    \centering    \includegraphics[width=0.49\textwidth]{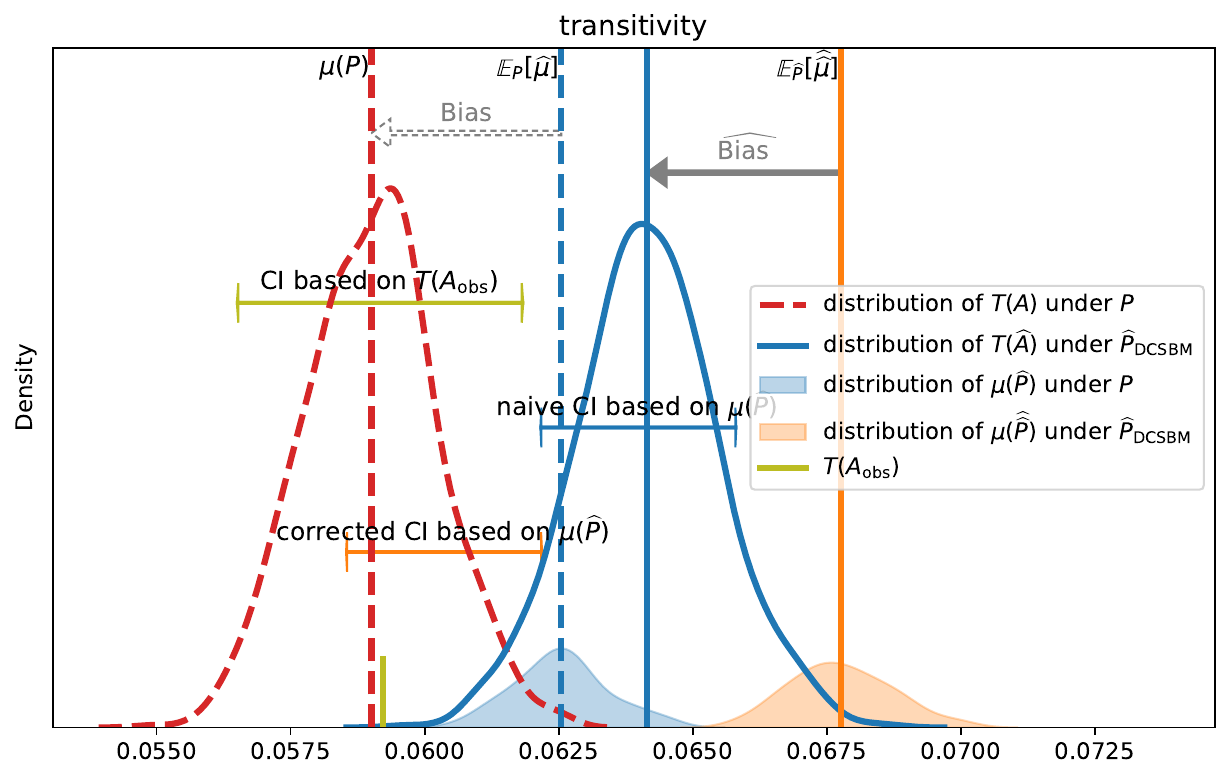}
    \includegraphics[width=0.49\textwidth]{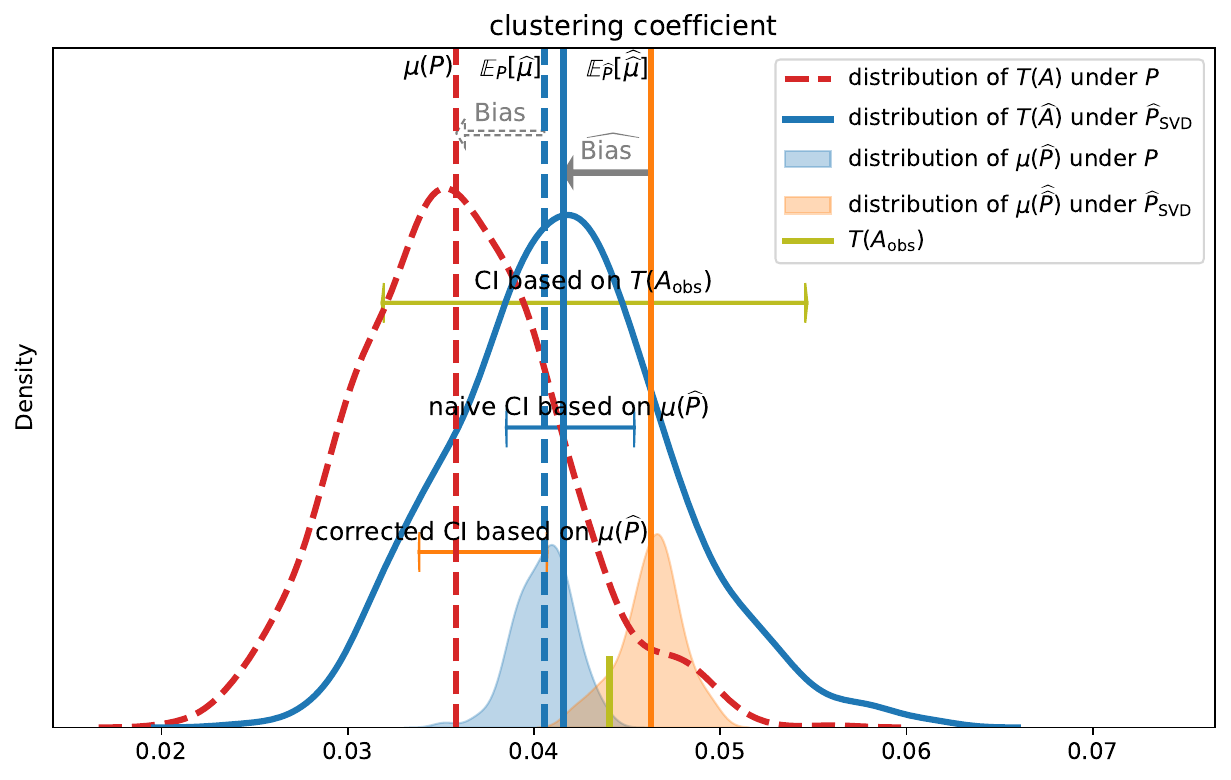}
    \caption{Three types of confidence intervals for global transitivity (left) and the
clustering coefficient of a given node (right), under the same settings as in
Figure~\ref{fig:bootstrap_bias}; see
Section~\ref{section:simulation_settings_for_example_figures} for details.
The orange shaded area shows the conditional distribution of
$\widehat{\widehat{\mu}}$ given $\widehat{P}$, treating $\widehat{P}$ as the
true model, and the orange vertical line marks
$\E_{\widehat{P}}[\widehat{\widehat{\mu}}]$. The distance between this orange
line and the blue line at $\mu(\widehat{P})$ is the estimated bias
$\widehat{\mathrm{Bias}}$, indicated by the gray arrow. The green, blue, and
brown intervals are intended to cover the red dashed line representing $\mu(P)$.
    }
    \label{fig:bootstrap_CI_example}
\end{figure}

\subsection{Confidence Intervals Based on Distributions of Bootstrap Means}
\label{section:CI_based_on_mu}

We now use the asymptotic distribution of $\widehat{\mu}_R=\mu_R(\widehat{P})$ to
construct confidence intervals for $\mu_R=\mu_R(P)$. This approach is justified by
the following results. We first show that $\widehat{\mu}_R$ is asymptotically
normal. The proof is given in Appendix~\ref{sec:proofs of main results}.

\begin{proposition}[Asymptotic normality of bootstrap means]
\label{prop:normal_convergence_expectation_of_subgraph_count}
Let $A$ be a network sampled from the Chung--Lu model with
$n^{-1}\prec p \prec 1$. Consider the conditional expectation
$\widehat{\mu}_R=\E_{\widehat{P}}[T_R]$ of the subgraph count of $R$ given
$\widehat{P}$. Then, as $n\to\infty$,
\begin{equation*}
    \frac{\widehat{\mu}_R-\E[\widehat{\mu}_R]}
    {\sqrt{\Var(\widehat{\mu}_R)}}
    \xrightarrow{D}
    \mathcal{N}(0,1).    
\end{equation*}
The same conclusion holds for rooted subgraph counts. 
\end{proposition}

To apply Proposition~\ref{prop:normal_convergence_expectation_of_subgraph_count},
we need a consistent estimates of the variances of $\widehat{\mu}_R$ and $\widehat{\mu}_{R,r}^{(i)}$. This estimate
is obtained from the second-level bootstrap, as shown in the following
proposition. The proof is given in Appendix~\ref{sec:proofs of main results}.

\begin{proposition}[Consistency of bootstrap variance estimation] 
\label{prop:variance_approximation_mu_hat}
Assume network $A$ is sampled from the Chung-Lu model in Section~\ref{section:Chung-Lu model} with $n^{-1} \prec p \prec 1$.
For any fixed connected subgraph $R$ with $v\ge 3$ nodes and $e \ge2$ edges, if we estimate $P$ using $\widehat{P}_{\textrm{MLE}}$ in \eqref{eq:Phat_MLE_DCSBM},
and estimate $\Var_P(\widehat{\mu}_R)$ and $\Var_P(\widehat{\mu}_R^{(i)})$ by $ \Var_{\widehat{P}} (\widehat{\widehat{\mu}}_R)$ and $\Var_{\widehat{P}}(\widehat{\widehat{\mu}}_R^{(i)})$, respectively, then we have
    \begin{equation*}
\frac{\Var_{\widehat{P}}(\widehat{\widehat{\mu}}_R^{(i)})-\Var_P(\widehat{\mu}_R^{(i)})}{\Var_P(\widehat{\mu}_R^{(i)})} =  O \left((np)^{-1/2}\right),
\qquad        \frac{\Var_{\widehat{P}}(\widehat{\widehat{\mu}}_R)-\Var_P(\widehat{\mu}_R)}{\Var_P(\widehat{\mu}_R)} = O \left((np)^{-1/2}\right).
    \end{equation*} 
\end{proposition}

Using Propositions~\ref{prop:normal_convergence_expectation_of_subgraph_count}
and \ref{prop:variance_approximation_mu_hat}, we are now ready to construct
confidence intervals for $\mu(P)$ based on
$\widehat{\mu}=\mu(\widehat{P})$. The confidence intervals for rooted counts are
analogous and are therefore omitted.

\subsubsection{Confidence Intervals without Bias Correction}
\label{section:CI_naive}

For comparison, we first construct intervals directly for
$\E_P[\widehat{\mu}]$, ignoring the difference between
$\E_P[\widehat{\mu}]$ and $\mu(P)$. This yields the following normal and
nonnormal confidence intervals:
\begin{equation}
    \widehat{\mu}
    \pm
    \Phi^{-1}\left(1-\frac{\alpha}{2}\right)
    \sqrt{\Var_{\widehat{P}}(\widehat{\widehat{\mu}})},
    \label{eq:CI_mu_naive_sym}
\end{equation}
and
\begin{equation}
    \left[
    \widehat{\mu}
    -
    \left\{
        F^{-1}_{\widehat{P},\widehat{\widehat{\mu}}}
        \left(1-\frac{\alpha}{2}\right)
        -
        \E_{\widehat{P}}[\widehat{\widehat{\mu}}]
    \right\},
    \quad
    \widehat{\mu}
    -
    \left\{
        F^{-1}_{\widehat{P},\widehat{\widehat{\mu}}}
        \left(\frac{\alpha}{2}\right)
        -
        \E_{\widehat{P}}[\widehat{\widehat{\mu}}]
    \right\}
    \right].
    \label{eq:CI_mu_naive_asym}
\end{equation}
In practice, we can replace $\E_{\widehat{P}}[\widehat{\widehat{\mu}}]$ by
$B^{-1}\sum_{b=1}^B \widehat{\widehat{\mu}}_b$. The corresponding Monte Carlo
approximation error can be made arbitrarily small by taking $B$ sufficiently
large.

\subsubsection{Bias-Corrected Confidence Intervals} 
\label{section:CI_corrected}

We shift the intervals in Section~\ref{section:CI_naive} by the additive
correction term
\[
    \E_{\widehat{P}}[\mu(\widehat{\widehat{P}})]-\mu(\widehat{P}),
\]
obtaining the following bias-corrected intervals:
\begin{equation}
    2\widehat{\mu}
    -
    \E_{\widehat{P}}[\widehat{\widehat{\mu}}]
    \pm
    \Phi^{-1}\left(1-\frac{\alpha}{2}\right)
    \sqrt{\Var_{\widehat{P}}(\widehat{\widehat{\mu}})},
    \label{eq:CI_mu_corrected_sym}
\end{equation}
and
\begin{equation}
    \left[
    2\widehat{\mu}
    -
    F^{-1}_{\widehat{P},\widehat{\widehat{\mu}}}
    \left(1-\frac{\alpha}{2}\right),
    \quad
    2\widehat{\mu}
    -
    F^{-1}_{\widehat{P},\widehat{\widehat{\mu}}}
    \left(\frac{\alpha}{2}\right)
    \right].
    \label{eq:CI_mu_corrected_asym}
\end{equation}
Since the bias correction only shifts the location, these intervals have the
same widths as the uncorrected intervals in Section~\ref{section:CI_naive}.

So far, we have proposed three types of confidence intervals: intervals based on
the distribution of $T$ in Section~\ref{section:CI_based_on_T}, intervals based
on the distribution of $\widehat{\mu}$ in Section~\ref{section:CI_naive}, and
their bias-corrected versions in Section~\ref{section:CI_corrected}. We
illustrate the construction and comparison of these intervals in
Figure~\ref{fig:bootstrap_CI_example}, using the same settings as in
Figure~\ref{fig:bootstrap_bias}.
Intervals based on the distribution of $T(A)$ in \eqref{eq:CI_obs_asym} are
shown in green. They are centered at the point estimate $T(A_{\textrm{obs}})$,
and their widths are determined by the variance of the bootstrap distribution of
$T$ (the blue solid curve). Intervals based on the distribution of
$\widehat{\mu}=\mu(\widehat{P})$ are shown in blue and brown, depending on whether bias
correction is applied. The naive intervals in \eqref{eq:CI_mu_naive_asym} are
centered at $\widehat{\mu}$, whereas the bias-corrected intervals in
\eqref{eq:CI_mu_corrected_asym} are centered at
$2\widehat{\mu}-\E_{\widehat{P}}[\widehat{\widehat{\mu}}]$.
These two types of intervals have the same width, characterized by the variance
of $\widehat{\widehat{\mu}}$ conditional on $\widehat{P}$ (the brown shaded
density). For both statistics, the green interval successfully covers the true
value $\mu(P)$ (the red dashed vertical line), but its width is unnecessarily
large compared with those of the other two intervals, especially for the local
clustering coefficient. The naive intervals in blue, on the other hand, miss the
target in both examples, whereas the bias-corrected intervals, with equally
narrow widths, are shifted to cover the true value correctly.

\subsection{Variance Comparison}
\label{section:variance_comparison}

For some network statistics, such as node degrees, we have
$\mu(\widehat{P})=\E_{\widehat{P}}[T(\widehat{A})]=T(A_{\mathrm{obs}})$.
In such cases, intervals based on $\widehat{\mu}$ are essentially the same as
those based on $T$.
However, for many other statistics, the expectation can be estimated more
accurately by using additional information from the network, resulting in
$\Var_P(\widehat{\mu})\ll \Var_P(T)$. For these statistics, the additional
computational cost of an extra bootstrap level may be worthwhile because it
leads to narrower confidence intervals, as discussed in
Section~\ref{section:CI_based_on_mu}.
Under the Chung--Lu model, we fully characterize this comparison for subgraph
count statistics, using the previously derived orders of
$\Var_P(\widehat{\mu}_R)$ and $\Var_P(T_R)$.

\begin{proposition}[Variance comparison for subgraph counts]
\label{prop:variance_comparison_general}
Assume that the network $A$ is sampled from the Chung--Lu model described in
Section~\ref{section:Chung-Lu model}. Let $R$ be a fixed connected subgraph, and
let $m_R$ and $m_{R,r}$ be as defined in
Proposition~\ref{prop:asymptotic_normality_triangle_count}. Then the following
statements hold.

\medskip
\noindent
(i) If $\max\{n^{-1}, n^{-1/m_R}\}\prec p \preccurlyeq 1$, then 
\begin{equation*}
    \Var_P(\widehat{\mu}_R) \preccurlyeq \Var_P(T_R).
\end{equation*}
Moreover, $\Var_P(\widehat{\mu}_R)\asymp \Var_P(T_R)$ holds precisely in the following
cases: either $e(R)=1$ and $v(R)=2$, or $e(R)\ge 2$ and
$$
p\gtrsim \max_{\substack{H\subseteq R, \ e(H)\ge 2\\
H \text{ has no isolated vertices} }} n^{-\frac{v(H)-2}{e(H)-1}}.
$$

\medskip
\noindent
(ii) For fixed vertices $r\in V(R)$ and $i\in[n]$, if $\max\{n^{-1}, n^{-1/m_{R,r}}\}\prec p \preccurlyeq 1$, then 
\begin{equation*}
    \Var_P(\widehat{\mu}_{R,r}^{(i)}) \preccurlyeq \Var_P(T_{R,r}^{(i)}).
\end{equation*}
Moreover, $\Var_P(\widehat{\mu}_{R,r}^{(i)}) \asymp
\Var_P(T_{R,r}^{(i)})$ holds precisely in the following cases: either
$e(R)=1$ and $v(R)=2$, or $e(R)\ge 2$ and
$$
p\gtrsim \max_{\substack{H\subseteq R, \ e(H)\ge 2\\
H \text{ has no isolated vertices} }} n^{-\frac{v(H\setminus\{r\})-1}{e(H)-1}}.
$$
\end{proposition}

We end this section with the following discussion. In practice, for an arbitrary
statistic $T$, how can one know \textit{a priori} that
$\Var_P(\widehat{\mu})\ll \Var_P(T)$, and hence that the additional
computational cost of an extra bootstrap level is worthwhile? After all, for a
general statistic $T$, a theoretical comparison between
$\Var_P(\widehat{\mu})$ and $\Var_P(T)$ may be tedious or even infeasible.
Moreover, asymptotic results such as
$\Var_P(\widehat{\mu})\prec \Var_P(T)$ in the examples above do not necessarily
imply that $\Var_P(\widehat{\mu})$ is substantially smaller than $\Var_P(T)$ at
a finite sample size $n$.

In most cases, there is no general answer beyond the intuition that
$\mu(\widehat{P})$ should be a better point estimate when $\widehat{P}$ exploits
more information from the network than is reflected in $T(A_{\mathrm{obs}})$.
Nevertheless, there is a practical solution: whenever the computational cost of
the two-level bootstrap is affordable, we can construct intervals based on
$\widehat{\mu}$ and empirically compare their widths with those of intervals
based on $T(A)$.

\subsection{Computational Cost}

The two-level bootstrap procedure in Figure~\ref{fig:procedure} can be
computationally intensive for large-scale networks. Drawing a network $A$ from
a given model $P$ requires $\Omega(n^2)$ operations, while the computational
cost of estimating $\widehat{P}$ from an observed network $A$ is typically
dominated by the SVD step, namely the computation of the first several
eigenvalues and eigenvectors. For sparse networks with average degree
$\lambda_n$, this step requires $\Omega(\lambda_n n^2)$ operations
\cite{cullum2002lanczos}.
Finally, the cost of evaluating $T(A)$ depends on the statistic $T$. For
example, evaluating the global triangle count $T_{\Delta}$ requires
$\Omega(\lambda_n^2 n)$ operations. Some other statistics, such as betweenness
centrality and closeness centrality, can be even more computationally expensive.
Therefore, the total computational cost is
$\Omega(B_1(n^2+\mathcal{T}))$ for the one-level bootstrap and
$\Omega(B_1B_2(n^2+\mathcal{T})+B_1n^2\lambda_n)$ for the two-level bootstrap,
where $\mathcal{T}$ denotes the cost of evaluating $T(A)$. The latter cost is at
least $\Omega(B_1B_2n^2)$, since one must generate a total of $B_1B_2$
replicate networks.

Nevertheless, some acceleration is possible. Most local statistics do not need
to be evaluated on the entire network; for the same reason, it may be
sufficient to resample subnetworks of much smaller sizes. More importantly, as
with other resampling methods, the bootstrap can be easily parallelized.

\section{Numerical Experiment}
\label{section:numerical_expriment}
\begin{figure}[h]
\centering
\begin{minipage}{0.49\textwidth}
\centering
\includegraphics[width=1.0\textwidth]{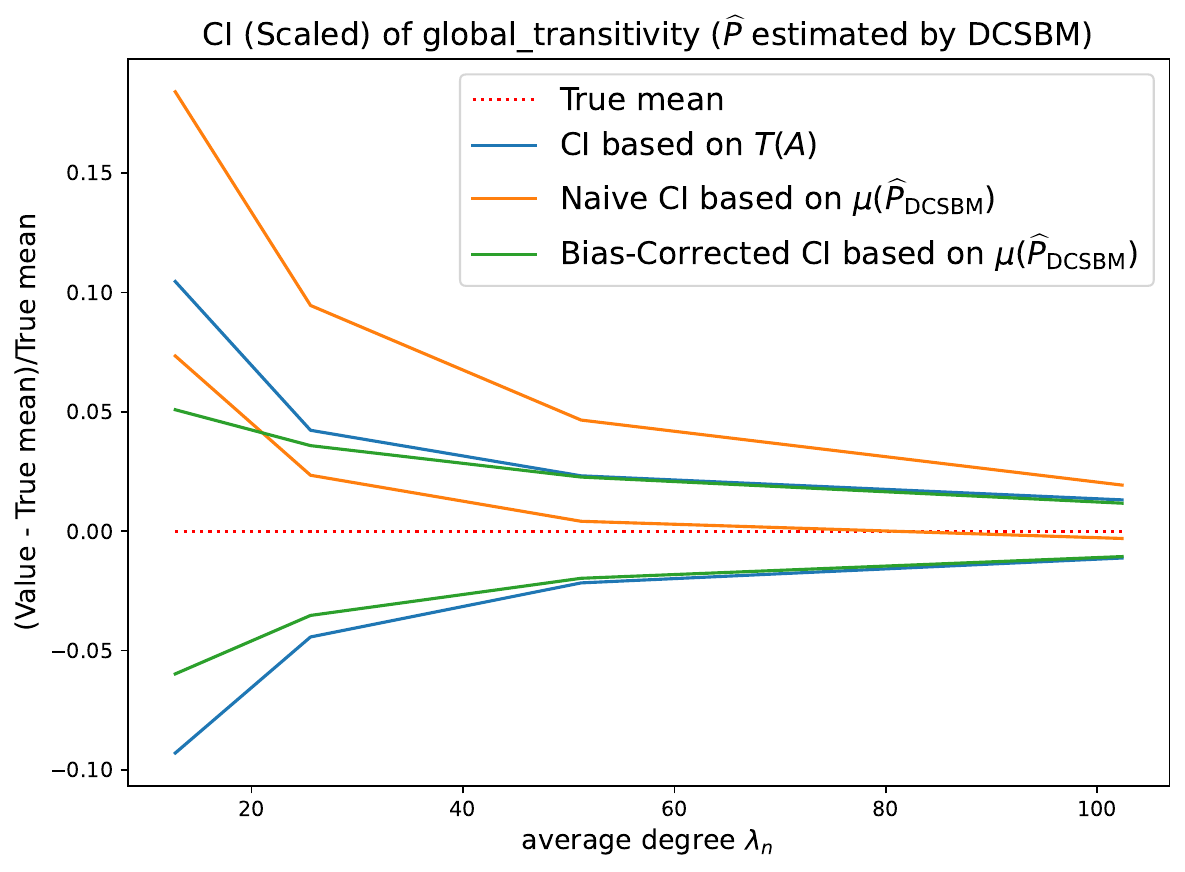}
\end{minipage}
\begin{minipage}{0.49\textwidth}
\centering
\includegraphics[width=1.0\textwidth]{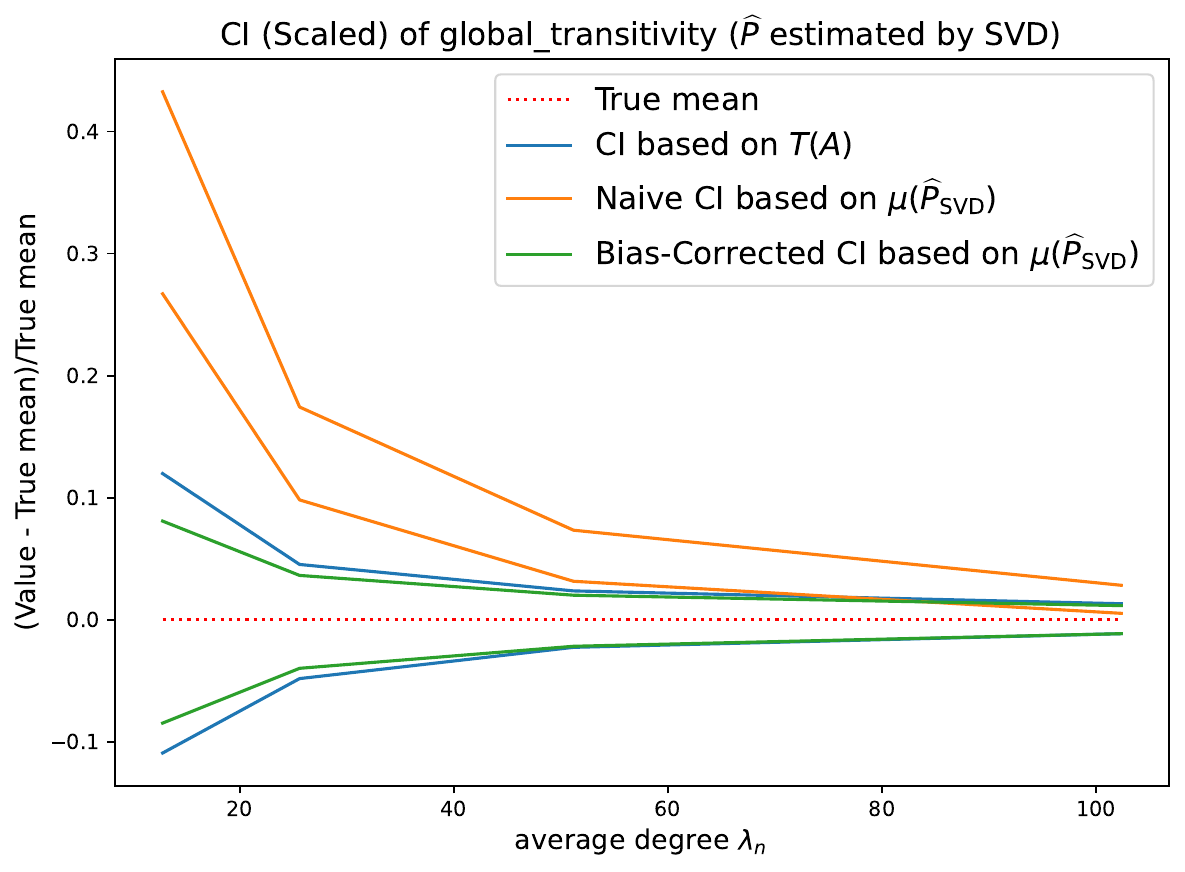}
\end{minipage}
\begin{minipage}{0.49\textwidth}
\centering
\includegraphics[width=1.0\textwidth]{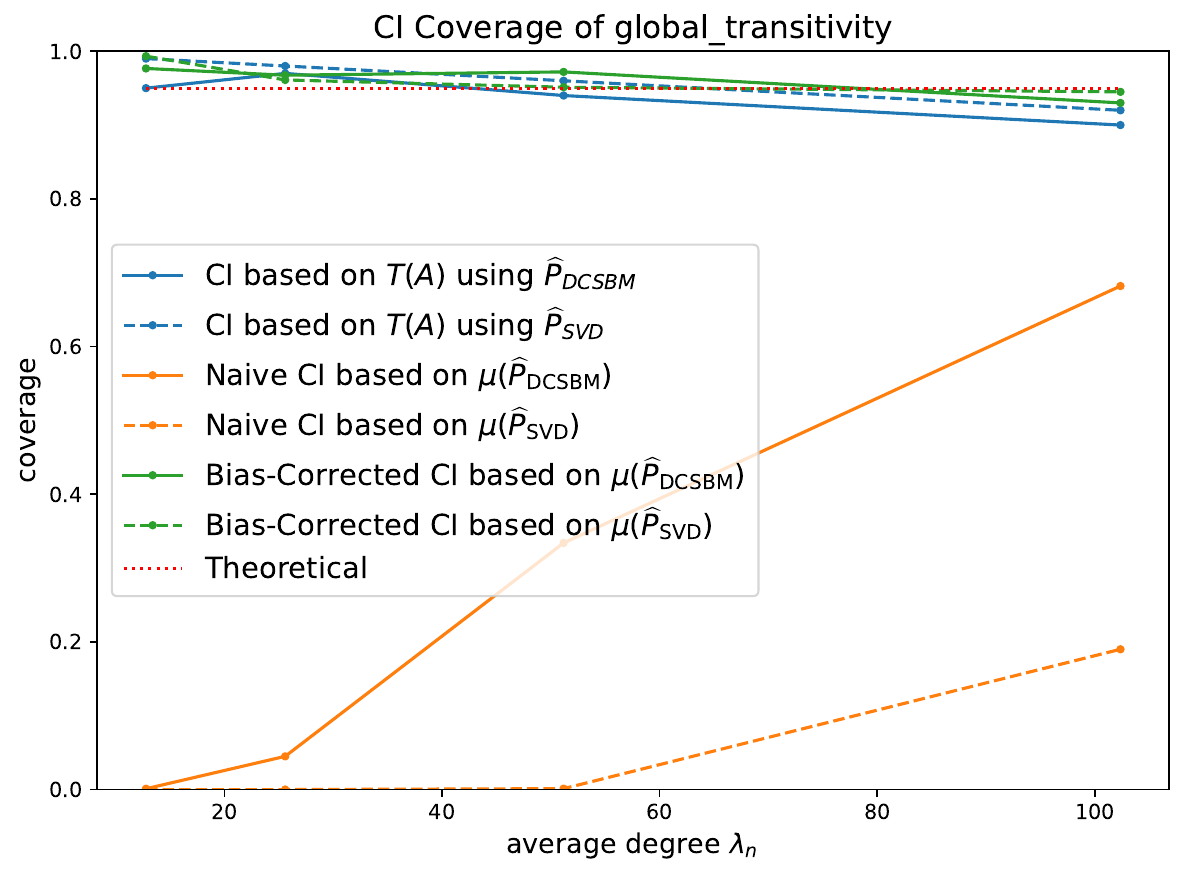}
\end{minipage}
\begin{minipage}{0.49\textwidth}
\centering
\includegraphics[width=1.0\textwidth]{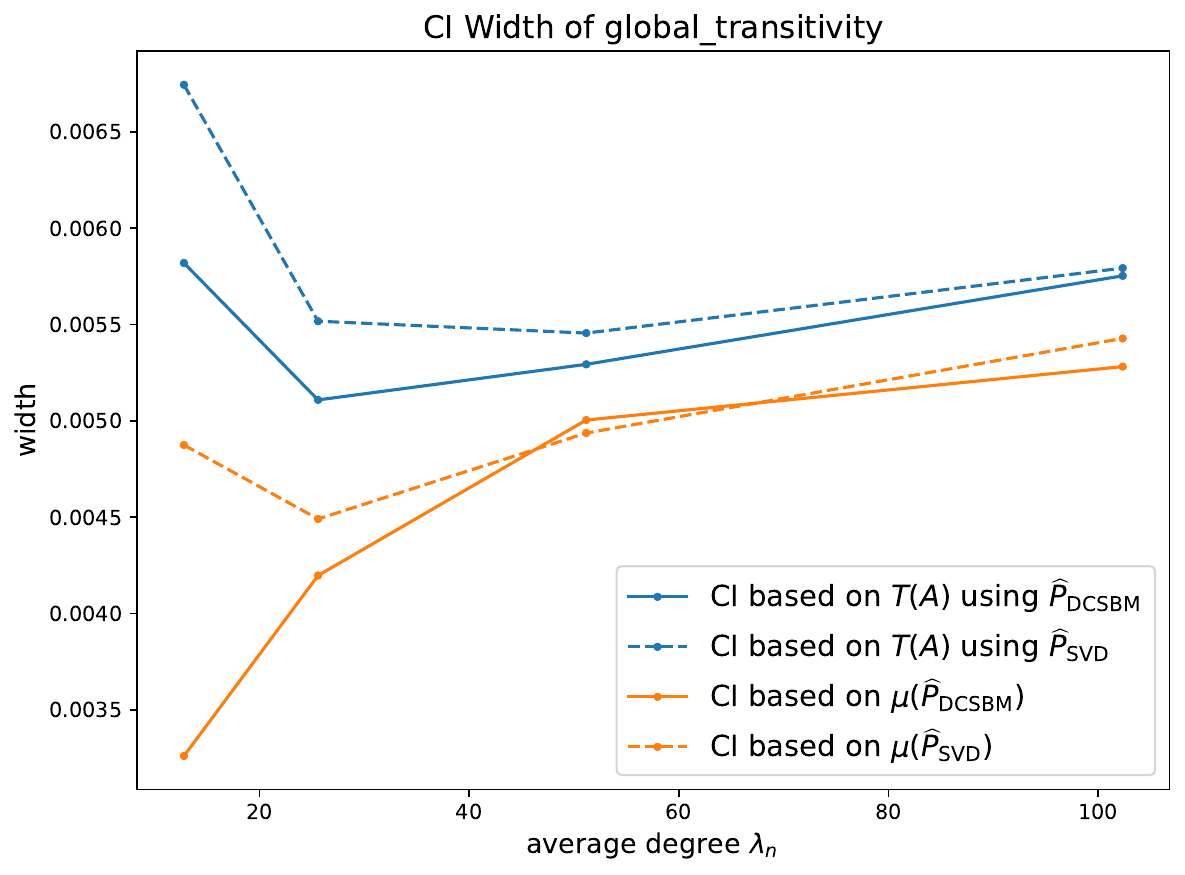}
\end{minipage}
\caption{Coverage and width results for global transitivity. The fixed true model is a
DCSBM with $\theta$ sampled from a uniform distribution. We fix $n=600$ and
adjust $\rho$ to obtain different expected degrees $\lambda_n$. The top two
panels show the average results for three types of confidence intervals,
computed from 1000 samples $\{A_b\}_{b=1}^{1000}\sim P$ and estimated using
$\widehat{P}_{\mathrm{DCSBM}}$ and $\widehat{P}_{\mathrm{SVD}}$, respectively.
For each displayed interval, the lower and upper endpoints are first averaged
over the 1000 repetitions and then centered and scaled by the true mean
$\mu(P)$ for clearer visualization. The bottom-left panel shows how the coverage
of each type of interval changes with $\lambda_n$, while the bottom-right panel
shows the original interval widths, without scaling.}
\label{fig:CI_global_transitivity}
\end{figure}

\begin{figure}[h]
\centering
\begin{minipage}{0.49\textwidth}
\centering
\includegraphics[width=1.0\textwidth]{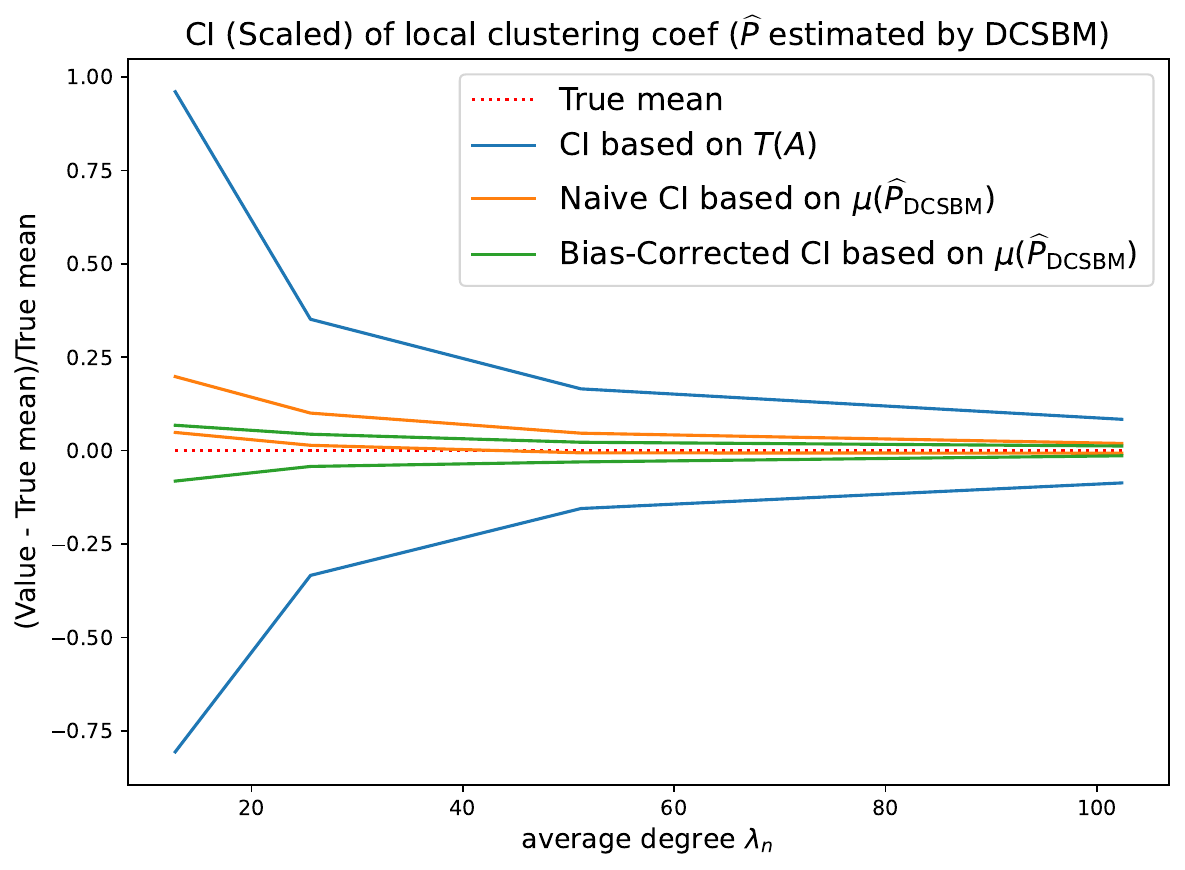}
\end{minipage}
\begin{minipage}{0.49\textwidth}
\centering
\includegraphics[width=1.0\textwidth]{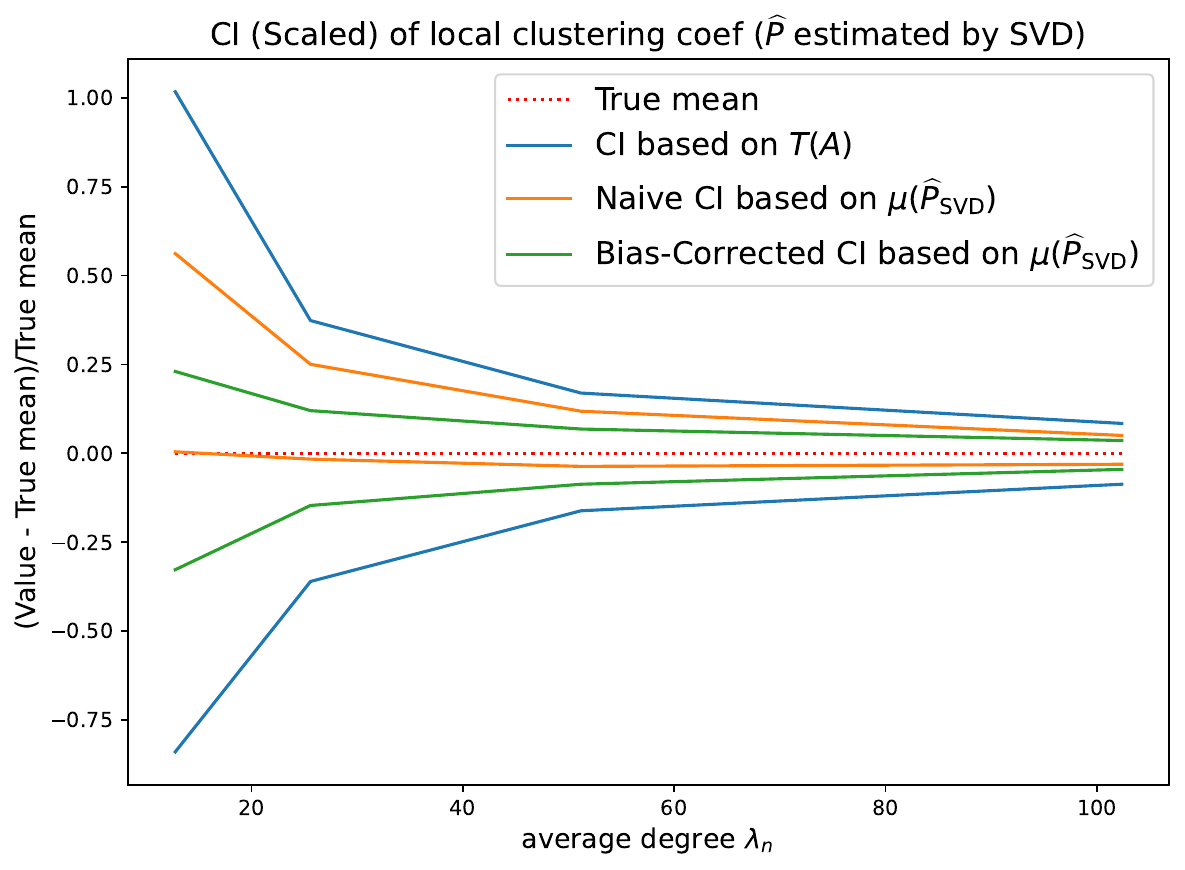}
\end{minipage}
\begin{minipage}{0.49\textwidth}
\centering
\includegraphics[width=1.0\textwidth]{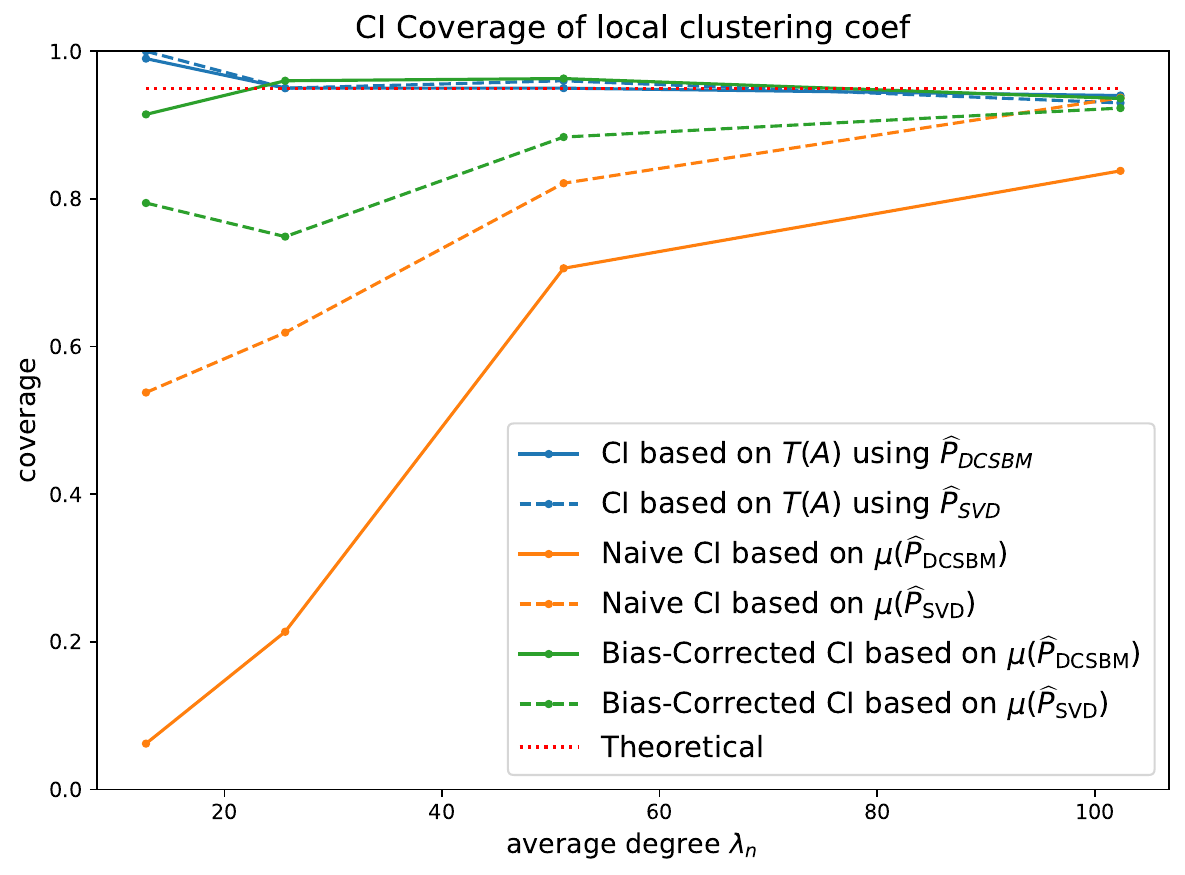}
\end{minipage}
\begin{minipage}{0.49\textwidth}
\centering
\includegraphics[width=1.0\textwidth]{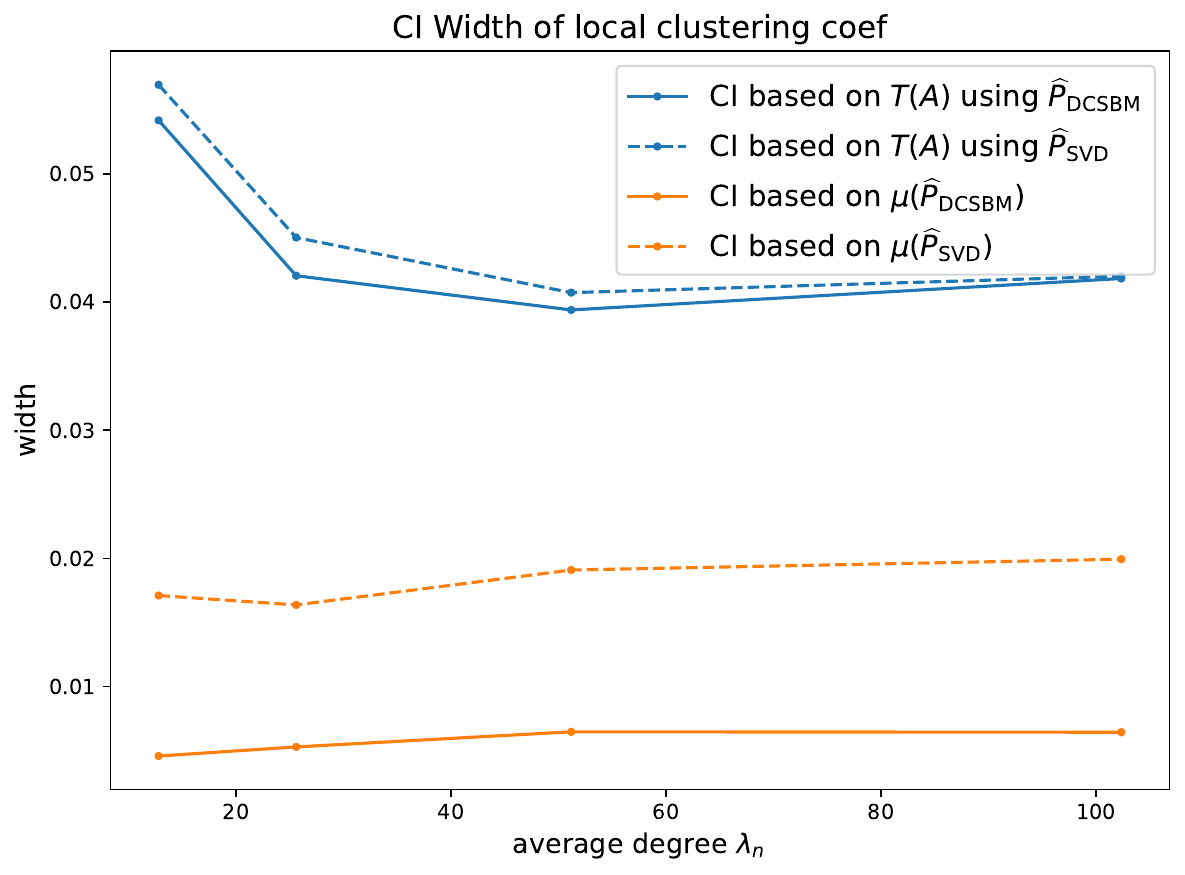}
\end{minipage}
\caption{Coverage and width results for the clustering coefficient of a certain node.
}
\label{fig:CI_local_transitivity_high}
\end{figure}
In this section, we first introduce the simulation settings for the examples in
Figures~\ref{fig:bootstrap_bias}, \ref{fig:bias_reduction_plot}, and
\ref{fig:bootstrap_CI_example}. We then study the performance of different
types of bootstrap confidence intervals in terms of coverage and width.

We construct the fixed edge-probability matrix $P$ as an instance of the DCSBM
with $K=3$ communities and $n=|V|=600$ nodes. Let
\begin{equation*}
    B =
    \left(\begin{array}{ccc}
       5 & 1 & 1 \\
       1 & 5 & 1 \\
       1 & 1 & 5
    \end{array}\right)
\end{equation*}
be the block-wise connectivity matrix.
For each node $i$, let $g_i\in\{1,2,3\}$ be its community label. We divide the
nodes evenly into three communities, so that
\begin{equation*}
    g_i =
    \left\{
    \begin{array}{ll}
    1, & 1\le i\le n/3, \\
    2, & n/3 < i\le 2n/3, \\
    3, & 2n/3 < i\le n.
    \end{array}
    \right.
\end{equation*}
Let $\{\theta_1,\ldots,\theta_n\}$ be an i.i.d. sample from either
(i) the uniform distribution $\mathrm{Unif}[0.2,1]$ or
(ii) a power-law distribution with density $p(x)\propto x^{-5}$, equivalently a
Pareto distribution with parameter $\alpha=4$. Once drawn, we fix
$\{\theta_i\}_{i=1}^n$ and view them as unknown parameters.
We construct the model $P$ as
\begin{equation}
    P_{ij} = \rho B_{g_ig_j} \theta_i \theta_j, \quad i \ne j, \quad \textrm { and } P_{ii} = 0 \textrm{ for all } i, \label{eq:simulation_model_DCSBM}
\end{equation} 
where we include an additional parameter $\rho$ to control the sparsity level.

Once the model $P$ is \textit{fixed} with the selected
$\{\theta_i\}_{i=1}^n$ and $\rho$, we draw the observed network
$A_{\textrm{obs}}$ from this model. We then construct confidence intervals for
the statistic or statistics of interest based on this single observed network
$A_{\textrm{obs}}$.
Throughout the study, the first-level bootstrap draws $B_1=1000$ replicate
networks $\widehat{A}$ from $\widehat{P}$, and the second-level bootstrap draws
$B_2=1000$ replicate networks $\widehat{\widehat{A}}$ from each
$\widehat{\widehat{P}}_b$. The coverage and width of the confidence intervals
are calculated by repeating this procedure 1000 times.

Two approaches are used to estimate $P$: the MLE and truncated SVD. The former
requires estimating the node labels $\{g_i\}_{i=1}^n$, for which we use
regularized spectral clustering \cite{amini2013pseudo}. All other parameters of
the DCSBM, namely $B$ and $\{\theta_i\}_{i=1}^n$, are then estimated by the
profile MLE. For the truncated SVD estimator, we use the true rank $K=3$ of
$P$. We denote the two estimators of $P$ by
$\widehat{P}_{\textrm{DCSBM}}$ and $\widehat{P}_{\textrm{SVD}}$, respectively.
In particular, we view $\widehat{P}_{\textrm{DCSBM}}$ as an extension of the
Chung--Lu estimator $\widehat{P}_{\textrm{MLE}}$ in
\eqref{eq:Phat_MLE_DCSBM} to the case of multiple communities. Based on our
theoretical results for the Chung--Lu model, we expect similar bias problems and
similar effects from bias correction.

Throughout the paper, we mainly use two network statistics as running examples.
The first is transitivity, a global statistic that measures the ratio of the
number of triangles to the number of connected triples in the graph. The second
is the clustering coefficient \eqref{eq:clustering_coef_Ti} of a fixed node
(the node with index $i=1$), which is a local statistic.
Results for several other common network statistics are reported in tables in
Appendix~\ref{section:Additional Simulation Results}.

\subsection{Running Examples in Figures~\ref{fig:bootstrap_bias},
\ref{fig:bias_reduction_plot}, and~\ref{fig:bootstrap_CI_example}.}
\label{section:simulation_settings_for_example_figures}

We first provide more details about the simulation settings used in
Figures~\ref{fig:bootstrap_bias}, \ref{fig:bias_reduction_plot}, and
\ref{fig:bootstrap_CI_example}, where transitivity and the clustering
coefficient serve as two running examples. For this purpose, we first recall
that global transitivity is defined as the ratio of the number of triangles to
the total number of connected triples in the graph,
\begin{equation*}
    T_{CL}(A) \coloneqq \frac{\sum_i \left(\sum_{j<k,\, j\neq i,\, k \neq i} A_{ij}A_{ik}A_{jk}\right)}{\sum_i \left(\sum_{j<k,\, j\neq i,\, k \neq i} A_{ij}A_{ik}\right)} = \frac{3 \times \textrm{\# triangles}}{\textrm{\# connected triples}}.
\end{equation*}
Meanwhile, the local clustering coefficient of a node $i$ with $d_i\ge 2$ is defined as
the ratio of the number of triangles containing node $i$ to the number of
connected triples centered at node $i$:
\begin{equation}
    T_{\mathrm{CL}}^{(i)}(A) \coloneqq \frac{\sum_{j<k,\, j\neq i,\, k \neq i} A_{ij}A_{ik}A_{jk}}{\sum_{j<k,\, j\neq i,\, k \neq i} A_{ij}A_{ik}} = \frac{ \textrm{\# triangles containing node } i}{ \textrm{\# connected triples centered around node } i} .\label{eq:clustering_coef_Ti}
\end{equation}

In both examples, we use the aforementioned DCSBM with $K=3$, where
$\{\theta_i\}_{i=1}^n$ is drawn from $\mathrm{Unif}[0.2,1]$. We use
$\widehat{P}_{\textrm{DCSBM}}$ for the transitivity example and
$\widehat{P}_{\textrm{SVD}}$ for the clustering-coefficient example.
In Figure~\ref{fig:bootstrap_bias} and its follow-up
Figure~\ref{fig:bootstrap_CI_example}, we choose $\rho$ so that
$\lambda_n=2\log n\approx 12.8$ for the transitivity example and
$\lambda_n=4\log n$ for the clustering-coefficient example. These sparsity
levels are chosen to make the main phenomena clearly visible in the figures.
In Figure~\ref{fig:bias_reduction_plot}, we vary $\rho$ so that the average
expected degree $\lambda_n$ takes the values $2\log n$ $(\approx 12.8)$,
$4\log n$, $8\log n$, and $16\log n$.

In Figure~\ref{fig:bootstrap_bias}, the bootstrap bias, represented by the gray
dashed arrow, appears to be non-negligible in both examples. Also, the shape and
width of the bootstrap distribution of $T$ under $\widehat{P}$ (blue solid
curve) resemble those of the true distribution of $T$ under $P$ (red dashed
curve), despite the location shift. In addition, the distribution of
$\widehat{\mu}$ (blue shaded area) is more concentrated than the distribution of
$T$, especially in the local-clustering-coefficient example, which is consistent
with our analysis in Section~\ref{section:variance_comparison}.

In Figure~\ref{fig:bias_reduction_plot}, the remaining bias after correction is
much smaller than the original bias, as the ratio
$|\Delta \mathrm{Bias}|/\mathrm{Bias}$ is typically much smaller than one. In
addition, this ratio decreases toward zero on average as $\lambda_n$ increases,
which empirically suggests that results such as
Proposition~\ref{prop:bias_correction_general_subgraph} may also hold for many
other subgraph-related statistics.

In Figure~\ref{fig:bootstrap_CI_example}, the estimated bias, represented by
the gray solid arrow, approximates the true bias well in both examples. The
conditional distribution of $\widehat{\widehat{\mu}}$ in the second-level
bootstrap (orange shaded area) also approximates the unknown distribution of
$\widehat{\mu}$ well. Both distributions are more concentrated than the
bootstrap distribution of $T$, especially in the local-clustering-coefficient
example.
As a result, the bias-corrected interval clearly outperforms the other two
types: it covers the true value $\mu(P)$ while also having a shorter width than
the interval based on the distribution of $T$.

Lastly, we reiterate that all previous propositions in our analysis are
restricted to subgraph counts under the Chung--Lu model, with $\widehat{P}$
given by the MLE in \eqref{eq:Phat_MLE_DCSBM}. Under these conditions, the
simulation results for subgraph-count statistics align well with the theory, and
we omit their figures here.
It is of greater practical interest to examine whether the same conclusions
extend to other network statistics, such as the two examples considered here.
Although the propositions may not be directly applicable to these examples---for
instance, because $\widehat{P}_{\mathrm{SVD}}$ is used instead of
$\widehat{P}_{\mathrm{MLE}}$ in one case---we nevertheless observe bootstrap
bias. Empirically, the proposed bootstrap framework and bias-correction approach
continue to perform reasonably well in these settings.

\subsection{Coverage and Width of Bootstrap Confidence Intervals}

Continuing with transitivity and the clustering coefficient, in
Figures~\ref{fig:CI_global_transitivity} and
\ref{fig:CI_local_transitivity_high}, we further examine how the coverage and
width of the confidence intervals evolve as the network density increases.
Again, we choose different values of $\rho$ so that the average expected degree
$\lambda_n$ takes the values $2\log n$ $(\approx 12.8)$, $4\log n$, $8\log n$,
and $16\log n$.

In terms of width, we clearly see that, for both statistics, intervals based on
the distribution of $\widehat{\mu}$ are narrower than those based on the
distribution of $T$. For global transitivity, this contrast is more visible when
the network is sparse and becomes less pronounced as the network grows denser.
This is intuitively explained by the strong connection between transitivity and
triangle counts, together with the fact that
$\Var_P(\widehat{\mu}_{\Delta}) \prec \Var_P(T_{\Delta})$ when
$p \prec n^{-1/2}$, whereas
$\Var_P(\widehat{\mu}_{\Delta}) \asymp \Var_P(T_{\Delta})$ when
$n^{-1/2}\prec p\prec 1$.
On the other hand, for the local clustering coefficient, the contrast in width
remains sharp across all sparsity levels, especially when the more accurate
$\widehat{P}_{\textrm{DCSBM}}$ is used.

In terms of coverage, intervals based on the distribution of $T$ provide
coverage close to the nominal level $1-\alpha$. The bias-corrected intervals
based on $\widehat{\mu}$ also perform reasonably well, except for the interval
based on $\mu(\widehat{P}_{\textrm{SVD}})$ for the local clustering coefficient.
This exception is likely due to less accurate estimation of the local parameters
around that node. The naive intervals without correction, however, severely
under-cover for both statistics.
Naturally, for the naive intervals, we would expect the coverage to approach the
nominal level $1-\alpha$ as $\lambda_n$ increases. We do observe increasing
coverage for the naive intervals, but it remains far from the nominal level.
Even at $\lambda_n=16\log n\approx 100$, some intervals still have coverage well
below $80\%$. This again suggests the necessity of bias correction, especially
when the network is sparse.

\subsubsection{Results for Additional Network Statistics}
Extensive coverage and width results for various commonly used network
statistics are reported in
Appendix~\ref{section:Additional Simulation Results}. These statistics include
global statistics, such as triangle density, transitivity, and the assortativity
coefficient by degree, as well as local statistics, such as rooted triangle
counts, clustering coefficients, betweenness scores, and closeness scores.

There are a few additional points worth noting. First, the bias-corrected
intervals almost always have better, or at least comparable, coverage than their
uncorrected counterparts. Second, intervals based on $\mu(\widehat{P})$ are
sometimes, though not always, much narrower than those based on $T(A)$, as
discussed in Section~\ref{section:variance_comparison}. Third, the DCSBM
estimator $\widehat{P}_{\textrm{DCSBM}}$ generally yields better confidence
intervals, with better coverage and shorter widths, than
$\widehat{P}_{\textrm{SVD}}$, naturally because it provides a more accurate
approximation to $P$.
Lastly, in Table~\ref{tab:SBM_2logn_SBM}, we carry out an additional experiment
to examine the behavior of bootstrap intervals when $\widehat{P}$ nearly
recovers $P$. We generate networks from a three-block SBM and estimate $P$ using
the SBM MLE, without estimating $\{\theta_i\}_{i=1}^n$. Thus, when the estimator
$\widehat{P}_{\textrm{SBM}}$ has very little error, all three types of intervals
have satisfactory coverage. In particular, no obvious bias issue is present, as
the uncorrected intervals behave similarly to the corrected ones. In addition,
because $\mu(\widehat{P}_{\textrm{SBM}})$ has small variance, the advantage of
using intervals based on $\widehat{\mu}$ becomes even more pronounced. For some
local statistics, these intervals can be only $1/10$ to $1/5$ as wide as those
based on $T$.

The results in Tables~\ref{tab:DCSBM_powerlaw_2logn_DCSBM}, \ref{tab:DCSBM_powerlaw_2logn_SVD}, and \ref{tab:DCSBM_uniform_2logn_DCSBM} also show that some statistics, especially average
path length and betweenness, are more difficult for all methods considered. This
is not surprising. Unlike subgraph counts or degree-based quantities,
shortest-path-based statistics are highly non-smooth functions of the adjacency
matrix. Adding or deleting a single edge can change many shortest paths at once,
alter which vertices lie on geodesics, or substantially change graph distances.
Consequently, their finite-sample sampling distributions can be more sensitive
to small perturbations of the graph and may be less well approximated by either
a normal approximation or an empirical bootstrap, particularly in sparse or
moderate-size networks. The two-level bootstrap reduces plug-in bias, but it
does not by itself eliminate the intrinsic instability of such path-based
functionals. We therefore view the results for average path length, closeness,
and betweenness as a useful diagnostic: they indicate that additional
regularity, larger network size, or statistic-specific methods may be needed
for reliable inference for shortest-path-based statistics.

\subsection{Real Data Application}

We first apply the bootstrap procedure to a relatively small network,
Zachary's karate club network \cite{zachary1977information}, shown in
Figure~\ref{fig:karate}. The network has \(|V|=34\) nodes and \(|E|=78\)
edges. The data consist of a single social network of individuals who later
split into two factions after a political conflict within the karate club. The
conflict arose between Mr. Hi, the instructor, and John A., the administrator.
The nodes represent individual karate club members, and the edges represent
interactions between members outside the club.

Our goal is to illustrate how the proposed bootstrap procedure can be used to
construct confidence intervals for expected network statistics and to assess
whether observed network features are typical under the DCSBM. We estimate
\(P\) under a DCSBM using the known community labels. For global structural
statistics, we consider the transitivity ratio, triangle density, and degree
assortativity coefficient. For local statistics, we consider vertex
betweenness, defined as the number of shortest paths passing through each
vertex.

\begin{figure}[tbh]
    \centering
    \includegraphics[
        width=0.5\textwidth,
        trim={60pt 60pt 60pt 60pt},
        clip
    ]{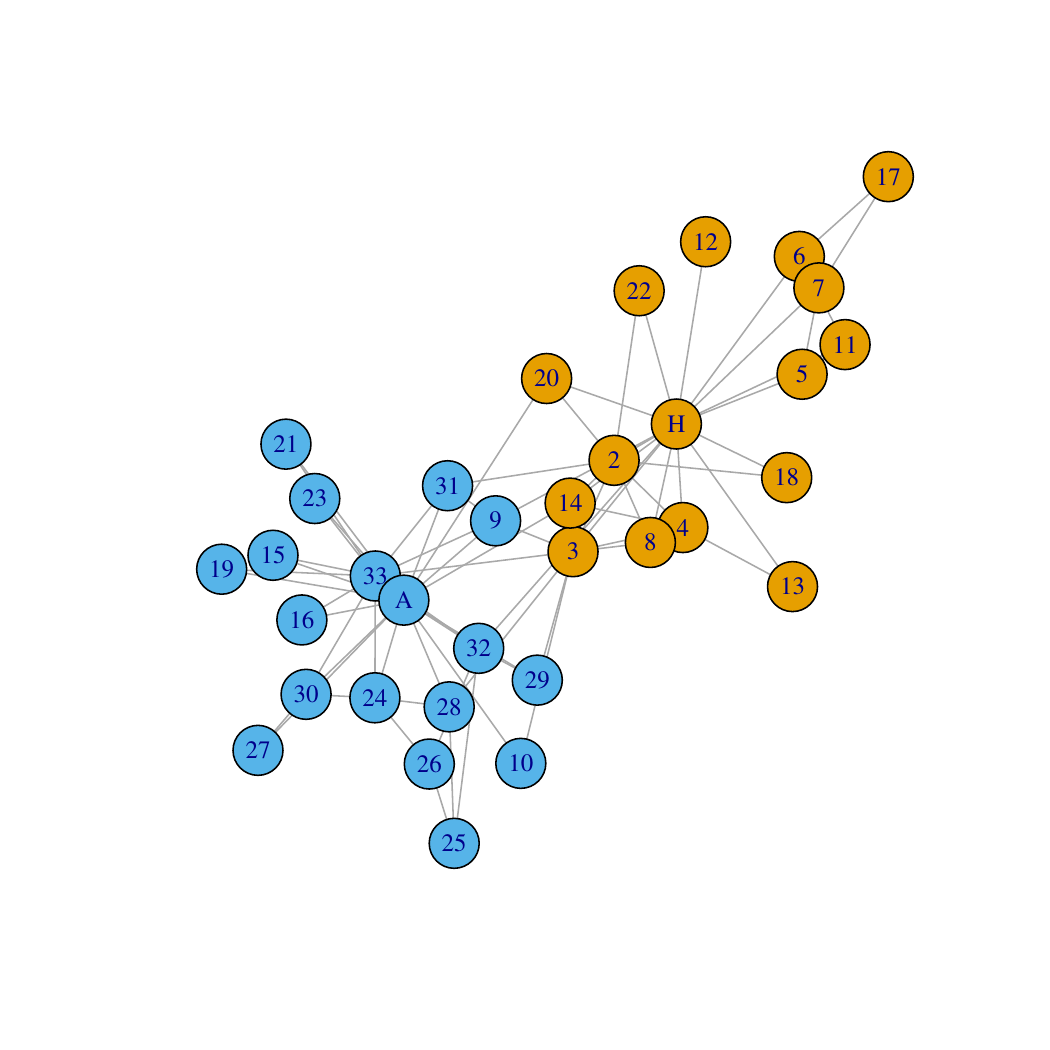}
    \caption{Karate club network with $|V|=34$ nodes and $|E|=78$ edges. Colors represent
factions.}
    \label{fig:karate}
\end{figure}

We first consider three global statistics. The number of local statistics is
large, and their nature is different, so we analyze them separately. For the
three global statistics, we construct simultaneous bias-corrected 95\%
confidence intervals for their population means using \eqref{eq:CI_mu_corrected_sym}, with Bonferroni correction. The resulting
confidence intervals are reported in Table~\ref{tab:karate_CI_global_DCSBM}.
If we assume that the network is sampled from a DCSBM, the observed degree
assortativity coefficient is lower than what is expected under the fitted
model, whereas the observed transitivity and triangle density are consistent
with their fitted population means. To further assess whether the observed
statistics themselves are typical under the fitted model, we also compute simultaneous predictive intervals for the realized network statistics
\(T(A_{\mathrm{obs}})\). These intervals represent the range of values expected for a newly sampled
network from the fitted model. They are centered at the bias-corrected estimated
means, with widths determined by the estimated standard deviation of the
realized statistic $T(A)$ under the fitted model.
As shown in Table~\ref{tab:karate_CI_global_DCSBM}, the observed values of
transitivity and triangle density lie well inside their predictive intervals.
Thus, neither statistic appears unusual under the fitted model. For degree
assortativity, however, the pointwise predictive interval does not contain the
observed value, indicating that the observed degree assortativity is unusually
low at the pointwise 95\% level. This deviation provides empirical evidence
that the fitted DCSBM does not fully capture the disassortative mixing structure
of the observed network. This deviation is consistent with Figure~\ref{fig:karate}, where the network exhibits a
star-like structure: Mr. Hi and John A. are hubs, and many of their neighbors are not connected
to each other. Since count statistics can be viewed as moments of network models, this type of
analysis may provide a useful tool for assessing model goodness of fit.
\begin{table}[ht]
\centering
\begin{tabular}{|l|c|c|c|}
\hline
Statistic
& Observed
& Mean CI
& Predictive interval \\
\hline
Transitivity
& \(0.2557\)
& \((0.1769,\ 0.3454)\)
& \((0.1904,\ 0.3320)\) \\
Triangle density
& \(0.0075\)
& \((0.0029, \ 0.0108)\)
& \((0.0031,\ 0.0106)\) \\
Degree assortativity
& \(-0.4756\)
& \((-0.4231,\ -0.2178)\)
& \((-0.4712,\ -0.1697)\) \\
\hline
\end{tabular}
\caption{Bonferroni-adjusted confidence intervals and predictive intervals for
the three global statistics on the karate club
network, estimated by the DCSBM. Values in the ``Observed'' column represent the
statistics computed from the observed network.}
\label{tab:karate_CI_global_DCSBM}
\end{table}

We perform a similar analysis for the local participation coefficient, which
measures how evenly the edges of a vertex are distributed across the two
factions and therefore reflects cross-faction mixing in the karate club. For
the two-faction case considered here, the participation coefficient of vertex
\(i\) is defined as
$
1-(d_{i,1}^2+d_{i,2}^2)/d_i^2
$, 
where \(d_i\) is the degree of vertex \(i\), and \(d_{i,1}\) and \(d_{i,2}\)
are the numbers of edges from \(i\) to the two communities. If \(d_i=0\), we
set \(\operatorname{Part}(i)=0\). The participation coefficient lies between
\(0\) and \(1/2\): it is \(0\) when all edges of \(i\) are connected to a
single community, and it is \(1/2\) when the edges of \(i\) are evenly split
between the two communities.
The local participation results suggest that the fitted DCSBM captures most of
the observed cross-faction mixing patterns. In particular, neither Mr. Hi nor
John A. appears unusual in terms of participation coefficient: their observed
values are close to the fitted-model predictive distribution. Node 10 has the
largest standardized discrepancy, with observed participation equal to \(0.5\),
compared with a corrected fitted-model mean of approximately \(0.093\). This
vertex is unusual at the pointwise level, suggesting unusually balanced ties
across factions. However, it is not flagged after Bonferroni correction over
all 34 vertices. Thus, the participation-coefficient analysis provides mild
exploratory evidence of local cross-faction mixing, but no local participation
statistic is flagged as unusual after simultaneous correction. As an
illustration, we report the values for Mr. Hi, John A., and node 10 in
Table~\ref{tab:karate_selected_participation}.

\begin{table}[ht]
\centering
\begin{tabular}{|l|c|c|c|}
\hline
ID
& Observed
& Mean CI
& Predictive interval \\
\hline
Mr. Hi
& \(0.2188\)
& \((0.0292,\ 0.5215)\)
& \((-0.0272,\ 0.5778)\) \\
Node 10
& \(0.5000\)
& \((-0.1461,\ 0.3315)\)
& \((-0.4715,\ 0.6570)\) \\
John A.
& \(0.2076\)
& \((0.0727,\ 0.5022)\)
& \((-0.0427,\ 0.6177)\) \\
\hline
\end{tabular}
\caption{Selected participation coefficients in the karate club network. The
mean confidence intervals and predictive intervals are Bonferroni-adjusted over
the \(34\) local participation statistics.}
\label{tab:karate_selected_participation}
\end{table}

To investigate the scalability and computational cost of the two-level bootstrap
procedure, we turn to a larger network. The political blogs dataset
\cite{adamic2005political} records hyperlinks between political blogs shortly
before the 2004 U.S. presidential election. Following common practice, we
consider its giant component, which contains 1222 nodes and 16714 edges. We
consider five global statistics and construct confidence intervals for their
population means; see Table~\ref{tab:polotical_blog_CI_global_DCSBM}. The model
is estimated using $\widehat{P}_{\mathrm{SVD}}$. Using eight multiprocessing
threads, the two-level bootstrap takes about 12 minutes on our server
(Intel(R) Xeon(R) CPU E5-2699 v3 @ 2.30GHz).

\begin{table}[!tbh]
    \centering
    \begin{tabular}{|c|c|c|}
    \hline
        Statistic & Observed & Mean CI \\
    \hline
        Transitivity & 0.2260 & $(0.2128,\ 0.2313)$ \\
    \hline
        Triangle density & $0.0003$ & $(0.00026,\ 0.00028)$ \\
    \hline
        Assortativity by degree & $-0.2213$ & $(-0.0547,\ -0.04657)$ \\
    \hline
        Average path length & 2.738 & $(2.710,\ 2.765)$ \\
    \hline
        Diameter & $8.00$ & $(5.70,\ 5.82)$ \\
    \hline
    \end{tabular}
        \caption{Simultaneous confidence intervals for global statistics on the political blog network.}
\label{tab:polotical_blog_CI_global_DCSBM}
\end{table}


\newpage

\begin{appendix}
\section{Additional Simulation Results}
\label{section:Additional Simulation Results}

\begin{table}[!tbh]
    \centering
    \begin{tabular}{|c|c|c|c|c|c|c|}
    \hline 
    statistic & \makecell[c]{Based on\\$T(A_{\mathrm{obs}})$\\asymmetric\\\eqref{eq:CI_obs_asym}} & \makecell[c]{Based on\\$T(A_{\mathrm{obs}})$\\symmetric\\\eqref{eq:CI_obs_sym}} & \makecell[c]{Based on\\$\mu(\widehat{P})$\\naive\\asymmetric\\\eqref{eq:CI_mu_naive_asym}} & \makecell[c]{Based on\\$\mu(\widehat{P})$\\naive\\symmetric\\\eqref{eq:CI_mu_naive_sym}} & \makecell[c]{Based on\\$\mu(\widehat{P})$\\corrected\\asymmetric\\\eqref{eq:CI_mu_corrected_asym}} & \makecell[c]{Based on\\$\mu(\widehat{P})$\\corrected\\symmetric\\\eqref{eq:CI_mu_corrected_sym}} \\\hline
    \makecell[c]{average\\degree} & \makecell[c]{94.9\% \\ (8.42e-01)} & \makecell[c]{95.7\% \\ (8.50e-01)} & \makecell[c]{96.5\% \\ (8.22e-01)} & \makecell[c]{96.8\% \\ (8.29e-01)} & \makecell[c]{93.9\% \\ (8.22e-01)} & \makecell[c]{94.4\% \\ (8.29e-01)} \\\hline
    \makecell[c]{triangle\\density} & \makecell[c]{96.3\% \\ (6.35e-06)} & \makecell[c]{96.6\% \\ (6.18e-06)} & \makecell[c]{26.7\% \\ (5.93e-06)} & \makecell[c]{34.5\% \\ (7.02e-06)} & \makecell[c]{94.1\% \\ (5.93e-06)} & \makecell[c]{96.8\% \\ (7.02e-06)} \\\hline
    \makecell[c]{rooted\\triangle\\count} & \makecell[c]{97.5\% \\ (6.94e+01)} & \makecell[c]{97.3\% \\ (7.21e+01)} & \makecell[c]{88.9\% \\ (6.09e+01)} & \makecell[c]{90.2\% \\ (6.23e+01)} & \makecell[c]{96.0\% \\ (6.09e+01)} & \makecell[c]{96.9\% \\ (6.23e+01)} \\\hline
    \makecell[c]{transitivity} & \makecell[c]{93.3\% \\ (6.83e-03)} & \makecell[c]{93.0\% \\ (6.55e-03)} & \makecell[c]{21.2\% \\ (5.09e-03)} & \makecell[c]{28.5\% \\ (5.59e-03)} & \makecell[c]{96.3\% \\ (5.09e-03)} & \makecell[c]{96.9\% \\ (5.59e-03)} \\\hline
    \makecell[c]{clustering\\coefficient} & \makecell[c]{92.2\% \\ (2.44e-02)} & \makecell[c]{91.9\% \\ (2.34e-02)} & \makecell[c]{16.3\% \\ (4.49e-03)} & \makecell[c]{19.8\% \\ (4.76e-03)} & \makecell[c]{92.1\% \\ (4.49e-03)} & \makecell[c]{92.7\% \\ (4.76e-03)} \\\hline
    \end{tabular}
        \caption{Coverage and width of different types of intervals. The coverage rate is shown
as a percentage, and the average width is shown in parentheses below. The true
model is a Chung--Lu model, equivalently a one-block DCSBM, with $\theta_i$
generated from the uniform distribution. The average degree is
$\lambda_n=2\log n\approx 12.79$. The model is estimated by
$\widehat{P}_{\mathrm{MLE}}$ in \eqref{eq:Phat_MLE_DCSBM}.}
\label{tab:ChungLu_2logn_DCSBM_p_unknown}
\end{table}

\begin{table}[!t]
    \centering
    \begin{tabular}{|c|c|c|c|c|c|c|}
    \hline 
    statistic & \makecell[c]{Based on\\$T(A_{\mathrm{obs}})$\\asymmetric\\\eqref{eq:CI_obs_asym}} & \makecell[c]{Based on\\$T(A_{\mathrm{obs}})$\\symmetric\\\eqref{eq:CI_obs_sym}} & \makecell[c]{Based on\\$\mu(\widehat{P})$\\naive\\asymmetric\\\eqref{eq:CI_mu_naive_asym}} & \makecell[c]{Based on\\$\mu(\widehat{P})$\\naive\\symmetric\\\eqref{eq:CI_mu_naive_sym}} & \makecell[c]{Based on\\$\mu(\widehat{P})$\\corrected\\asymmetric\\\eqref{eq:CI_mu_corrected_asym}} & \makecell[c]{Based on\\$\mu(\widehat{P})$\\corrected\\symmetric\\\eqref{eq:CI_mu_corrected_sym}} \\\hline
    \makecell[c]{average\\degree} & \makecell[c]{92.2\% \\ (8.89e-01)} & \makecell[c]{85.8\% \\ (7.85e-01)} & \makecell[c]{89.0\% \\ (8.85e-01)} & \makecell[c]{90.6\% \\ (8.96e-01)} & \makecell[c]{92.2\% \\ (8.85e-01)} & \makecell[c]{89.0\% \\ (8.96e-01)} \\\hline
    \makecell[c]{triangle\\density} & \makecell[c]{96.8\% \\ (6.52e-06)} & \makecell[c]{98.4\% \\ (6.30e-06)} & \makecell[c]{38.6\% \\ (6.04e-06)} & \makecell[c]{62.4\% \\ (7.17e-06)} & \makecell[c]{94.8\% \\ (6.04e-06)} & \makecell[c]{96.4\% \\ (7.17e-06)} \\\hline
    \makecell[c]{rooted\\triangle\\count} & \makecell[c]{92.4\% \\ (7.40e+01)} & \makecell[c]{96.4\% \\ (7.95e+01)} & \makecell[c]{98.4\% \\ (7.11e+01)} & \makecell[c]{98.4\% \\ (7.70e+01)} & \makecell[c]{98.4\% \\ (7.11e+01)} & \makecell[c]{98.4\% \\ (7.70e+01)} \\\hline
    \makecell[c]{transitivity} & \makecell[c]{98.4\% \\ (8.22e-03)} & \makecell[c]{96.4\% \\ (7.72e-03)} & \makecell[c]{4.4\% \\ (5.05e-03)} & \makecell[c]{6.8\% \\ (5.89e-03)} & \makecell[c]{93.2\% \\ (5.05e-03)} & \makecell[c]{94.8\% \\ (5.89e-03)} \\\hline
    \makecell[c]{assortativity\\by degree} & \makecell[c]{91.4\% \\ (5.16e-02)} & \makecell[c]{93.0\% \\ (5.35e-02)} & \makecell[c]{88.8\% \\ (1.25e-02)} & \makecell[c]{89.8\% \\ (1.29e-02)} & \makecell[c]{89.0\% \\ (1.25e-02)} & \makecell[c]{91.0\% \\ (1.29e-02)} \\\hline
    \makecell[c]{average\\path\\length} & \makecell[c]{89.0\% \\ (6.38e-02)} & \makecell[c]{87.4\% \\ (5.66e-02)} & \makecell[c]{83.8\% \\ (5.98e-02)} & \makecell[c]{84.6\% \\ (6.01e-02)} & \makecell[c]{84.2\% \\ (5.98e-02)} & \makecell[c]{89.0\% \\ (6.01e-02)} \\\hline
    \makecell[c]{closeness\\score} & \makecell[c]{91.4\% \\ (4.26e-02)} & \makecell[c]{94.6\% \\ (4.33e-02)} & \makecell[c]{96.4\% \\ (3.47e-02)} & \makecell[c]{98.4\% \\ (3.81e-02)} & \makecell[c]{94.8\% \\ (3.47e-02)} & \makecell[c]{96.2\% \\ (3.81e-02)} \\\hline
    \makecell[c]{betweenness\\score} & \makecell[c]{92.8\% \\ (6.59e-02)} & \makecell[c]{94.8\% \\ (7.57e-02)} & \makecell[c]{90.8\% \\ (6.89e-02)} & \makecell[c]{98.0\% \\ (7.10e-02)} & \makecell[c]{90.8\% \\ (6.89e-02)} & \makecell[c]{98.0\% \\ (7.10e-02)} \\\hline
    \makecell[c]{clustering\\coefficient} & \makecell[c]{94.4\% \\ (2.56e-02)} & \makecell[c]{92.8\% \\ (2.45e-02)} & \makecell[c]{20.2\% \\ (6.98e-03)} & \makecell[c]{23.0\% \\ (7.40e-03)} & \makecell[c]{98.4\% \\ (6.98e-03)} & \makecell[c]{98.4\% \\ (7.40e-03)} \\\hline
    \end{tabular}
        \caption{Coverage and width of different types of intervals. The coverage rate is shown
as a percentage, and the average width is shown in parentheses below. The true
model is a three-block DCSBM with $\theta_i$ generated from the power-law
distribution. The average degree is $\lambda_n=2\log n\approx 12.79$. The model
is estimated by $\widehat{P}_{\mathrm{DCSBM}}(\widehat{K}=3)$.}
\label{tab:DCSBM_powerlaw_2logn_DCSBM}
\end{table}

\begin{table}[!t]
    \centering
\begin{tabular}{|c|c|c|c|}
\hline 
statistic & \makecell[c]{Based on\\$T(A_{\mathrm{obs}})$\\symmetric\\\eqref{eq:CI_obs_sym}} & \makecell[c]{Based on\\$\mu(\widehat{P})$\\naive\\symmetric\\\eqref{eq:CI_mu_naive_sym}} & \makecell[c]{Based on\\$\mu(\widehat{P})$\\corrected\\symmetric\\\eqref{eq:CI_mu_corrected_sym}} \\\hline
\makecell[c]{average\\degree} & \makecell[c]{90.0\% \\ (7.80e-01)} & \makecell[c]{96.2\% \\ (1.11e+00)} & \makecell[c]{96.2\% \\ (1.11e+00)} \\\hline
\makecell[c]{triangle\\density} & \makecell[c]{97.6\% \\ (7.53e-06)} & \makecell[c]{4.0\% \\ (1.11e-05)} & \makecell[c]{98.4\% \\ (1.11e-05)} \\\hline
\makecell[c]{rooted\\triangle\\count} & \makecell[c]{98.4\% \\ (9.70e+01)} & \makecell[c]{74.4\% \\ (9.77e+01)} & \makecell[c]{96.2\% \\ (9.74e+01)} \\\hline
\makecell[c]{transitivity} & \makecell[c]{95.8\% \\ (8.64e-03)} & \makecell[c]{0.0\% \\ (8.01e-03)} & \makecell[c]{98.4\% \\ (8.01e-03)} \\\hline
\makecell[c]{assortativity\\by degree} & \makecell[c]{92.0\% \\ (5.01e-02)} & \makecell[c]{44.2\% \\ (1.30e-02)} & \makecell[c]{72.4\% \\ (1.30e-02)} \\\hline
\makecell[c]{average\\path\\length} & \makecell[c]{94.9\% \\ (5.83e-02)} & \makecell[c]{77.6\% \\ (6.71e-02)} & \makecell[c]{80.9\% \\ (6.71e-02)} \\\hline
\makecell[c]{closeness\\score} & \makecell[c]{98.7\% \\ (6.32e-02)} & \makecell[c]{98.7\% \\ (6.50e-02)} & \makecell[c]{99.3\% \\ (6.50e-02)} \\\hline
\makecell[c]{betweenness\\score} & \makecell[c]{99.3\% \\ (7.60e-02)} & \makecell[c]{99.3\% \\ (9.33e-02)} & \makecell[c]{98.7\% \\ (9.33e-02)} \\\hline
\makecell[c]{clustering\\coefficient} & \makecell[c]{96.9\% \\ (2.64e-02)} & \makecell[c]{10.0\% \\ (1.50e-02)} & \makecell[c]{83.1\% \\ (1.50e-02)} \\\hline
\end{tabular}
    \caption{Coverage and width of different types of intervals. The true model is a
three-block DCSBM with $\theta_i$ generated from the power-law distribution. The
average degree is $\lambda_n=2\log n\approx 12.79$. The model is estimated by
$\widehat{P}_{\mathrm{SVD}}(\widehat{K}=3)$.}
\label{tab:DCSBM_powerlaw_2logn_SVD}
\end{table}

\begin{table}[!htbp]
    \centering
    \begin{tabular}{|c|c|c|c|c|c|c|}
    \hline 
    statistic & \makecell[c]{Based on\\$T(A_{\mathrm{obs}})$\\asymmetric\\\eqref{eq:CI_obs_asym}} & \makecell[c]{Based on\\$T(A_{\mathrm{obs}})$\\symmetric\\\eqref{eq:CI_obs_sym}} & \makecell[c]{Based on\\$\mu(\widehat{P})$\\naive\\asymmetric\\\eqref{eq:CI_mu_naive_asym}} & \makecell[c]{Based on\\$\mu(\widehat{P})$\\naive\\symmetric\\\eqref{eq:CI_mu_naive_sym}} & \makecell[c]{Based on\\$\mu(\widehat{P})$\\corrected\\asymmetric\\\eqref{eq:CI_mu_corrected_asym}} & \makecell[c]{Based on\\$\mu(\widehat{P})$\\corrected\\symmetric\\\eqref{eq:CI_mu_corrected_sym}} \\\hline
    \makecell[c]{average\\degree} & \makecell[c]{95.2\% \\ (8.60e-01)} & \makecell[c]{93.2\% \\ (7.92e-01)} & \makecell[c]{93.2\% \\ (8.58e-01)} & \makecell[c]{93.2\% \\ (8.74e-01)} & \makecell[c]{95.2\% \\ (8.58e-01)} & \makecell[c]{93.2\% \\ (8.74e-01)} \\\hline
    \makecell[c]{triangle\\density} & \makecell[c]{96.6\% \\ (4.74e-06)} & \makecell[c]{98.4\% \\ (4.94e-06)} & \makecell[c]{24.8\% \\ (4.58e-06)} & \makecell[c]{32.0\% \\ (5.43e-06)} & \makecell[c]{93.4\% \\ (4.58e-06)} & \makecell[c]{98.4\% \\ (5.43e-06)} \\\hline
    \makecell[c]{rooted\\triangle\\count} & \makecell[c]{88.0\% \\ (1.45e+01)} & \makecell[c]{95.0\% \\ (1.62e+01)} & \makecell[c]{94.4\% \\ (1.46e+01)} & \makecell[c]{98.0\% \\ (1.59e+01)} & \makecell[c]{86.8\% \\ (1.35e+01)} & \makecell[c]{94.4\% \\ (1.53e+01)} \\\hline
    \makecell[c]{transitivity} & \makecell[c]{95.2\% \\ (5.70e-03)} & \makecell[c]{98.6\% \\ (5.95e-03)} & \makecell[c]{0.0\% \\ (3.18e-03)} & \makecell[c]{0.0\% \\ (3.97e-03)} & \makecell[c]{95.0\% \\ (3.18e-03)} & \makecell[c]{98.4\% \\ (3.97e-03)} \\\hline
    \makecell[c]{assortativity\\by degree} & \makecell[c]{96.6\% \\ (5.85e-02)} & \makecell[c]{95.0\% \\ (6.17e-02)} & \makecell[c]{44.2\% \\ (5.91e-03)} & \makecell[c]{53.8\% \\ (6.17e-03)} & \makecell[c]{59.4\% \\ (5.91e-03)} & \makecell[c]{67.2\% \\ (6.17e-03)} \\\hline
    \makecell[c]{average\\path\\length} & \makecell[c]{95.2\% \\ (6.03e-02)} & \makecell[c]{93.2\% \\ (5.83e-02)} & \makecell[c]{88.2\% \\ (5.55e-02)} & \makecell[c]{88.8\% \\ (5.67e-02)} & \makecell[c]{82.0\% \\ (5.55e-02)} & \makecell[c]{85.8\% \\ (5.67e-02)} \\\hline
    \makecell[c]{closeness\\score} & \makecell[c]{98.4\% \\ (4.63e-02)} & \makecell[c]{100.0\% \\ (5.13e-02)} & \makecell[c]{99.4\% \\ (4.50e-02)} & \makecell[c]{100.0\% \\ (4.72e-02)} & \makecell[c]{98.0\% \\ (4.50e-02)} & \makecell[c]{99.4\% \\ (4.72e-02)} \\\hline
    \makecell[c]{betweenness\\score} & \makecell[c]{89.2\% \\ (1.65e-01)} & \makecell[c]{98.4\% \\ (2.06e-01)} & \makecell[c]{92.6\% \\ (2.03e-01)} & \makecell[c]{96.8\% \\ (2.14e-01)} & \makecell[c]{86.4\% \\ (2.03e-01)} & \makecell[c]{94.8\% \\ (2.15e-01)} \\\hline
    \makecell[c]{clustering\\coefficient} & \makecell[c]{98.4\% \\ (5.17e-02)} & \makecell[c]{98.4\% \\ (5.27e-02)} & \makecell[c]{2.0\% \\ (4.35e-03)} & \makecell[c]{2.6\% \\ (5.00e-03)} & \makecell[c]{94.6\% \\ (4.35e-03)} & \makecell[c]{99.2\% \\ (5.00e-03)} \\\hline
    \end{tabular}
        \caption{Coverage and width of different types of intervals. True model is a 3-block DCSBM with $\theta_i$ generated from the uniform distribution. Average degree $\lambda_n = 2 \log n \approx 12.79$. Model estimated by $\widehat{P}_{\mathrm{DCSBM}}(\widehat{K}=3)$.}
\label{tab:DCSBM_uniform_2logn_DCSBM}
\end{table}

\begin{table}[!t]
    \centering
    \begin{tabular}{|c|c|c|c|c|c|c|}
    \hline 
    statistic & \makecell[c]{Based on\\$T(A_{\mathrm{obs}})$\\asymmetric\\\eqref{eq:CI_obs_asym}} & \makecell[c]{Based on\\$T(A_{\mathrm{obs}})$\\symmetric\\\eqref{eq:CI_obs_sym}} & \makecell[c]{Based on\\$\mu(\widehat{P})$\\naive\\asymmetric\\\eqref{eq:CI_mu_naive_asym}} & \makecell[c]{Based on\\$\mu(\widehat{P})$\\naive\\symmetric\\\eqref{eq:CI_mu_naive_sym}} & \makecell[c]{Based on\\$\mu(\widehat{P})$\\corrected\\asymmetric\\\eqref{eq:CI_mu_corrected_asym}} & \makecell[c]{Based on\\$\mu(\widehat{P})$\\corrected\\symmetric\\\eqref{eq:CI_mu_corrected_sym}} \\\hline
    \makecell[c]{average\\degree} & \makecell[c]{94.2\% \\ (8.02e-01)} & \makecell[c]{94.9\% \\ (7.94e-01)} & \makecell[c]{93.4\% \\ (8.01e-01)} & \makecell[c]{96.5\% \\ (8.57e-01)} & \makecell[c]{93.4\% \\ (8.01e-01)} & \makecell[c]{95.7\% \\ (8.57e-01)} \\\hline
    \makecell[c]{triangle\\density} & \makecell[c]{99.2\% \\ (4.98e-06)} & \makecell[c]{99.2\% \\ (4.85e-06)} & \makecell[c]{40.3\% \\ (4.76e-06)} & \makecell[c]{34.6\% \\ (5.01e-06)} & \makecell[c]{94.5\% \\ (4.76e-06)} & \makecell[c]{97.7\% \\ (5.01e-06)} \\\hline
    \makecell[c]{rooted\\triangle\\count} & \makecell[c]{86.9\% \\ (1.38e+01)} & \makecell[c]{94.2\% \\ (1.59e+01)} & \makecell[c]{94.8\% \\ (1.42e+01)} & \makecell[c]{96.8\% \\ (1.52e+01)} & \makecell[c]{89.7\% \\ (1.33e+01)} & \makecell[c]{96.0\% \\ (1.46e+01)} \\\hline
    \makecell[c]{transitivity} & \makecell[c]{96.8\% \\ (5.73e-03)} & \makecell[c]{98.8\% \\ (5.79e-03)} & \makecell[c]{4.8\% \\ (3.79e-03)} & \makecell[c]{3.5\% \\ (3.96e-03)} & \makecell[c]{85.8\% \\ (3.79e-03)} & \makecell[c]{86.9\% \\ (3.96e-03)} \\\hline
    \makecell[c]{assortativity\\by degree} & \makecell[c]{90.5\% \\ (5.98e-02)} & \makecell[c]{92.8\% \\ (6.23e-02)} & \makecell[c]{78.0\% \\ (4.70e-02)} & \makecell[c]{100.0\% \\ (5.24e-02)} & \makecell[c]{54.5\% \\ (4.70e-02)} & \makecell[c]{71.8\% \\ (5.24e-02)} \\\hline
    \makecell[c]{average\\path\\length} & \makecell[c]{94.9\% \\ (5.81e-02)} & \makecell[c]{95.7\% \\ (5.93e-02)} & \makecell[c]{89.4\% \\ (5.37e-02)} & \makecell[c]{93.7\% \\ (5.67e-02)} & \makecell[c]{84.3\% \\ (5.37e-02)} & \makecell[c]{88.6\% \\ (5.67e-02)} \\\hline
    \makecell[c]{closeness\\score} & \makecell[c]{86.5\% \\ (7.76e-02)} & \makecell[c]{90.9\% \\ (8.57e-02)} & \makecell[c]{79.4\% \\ (5.67e-02)} & \makecell[c]{89.7\% \\ (5.96e-02)} & \makecell[c]{79.4\% \\ (5.67e-02)} & \makecell[c]{84.5\% \\ (5.96e-02)} \\\hline
    \makecell[c]{betweenness\\score} & \makecell[c]{98.0\% \\ (5.35e-02)} & \makecell[c]{97.5\% \\ (5.22e-02)} & \makecell[c]{96.8\% \\ (4.95e-02)} & \makecell[c]{96.3\% \\ (5.22e-02)} & \makecell[c]{96.0\% \\ (4.95e-02)} & \makecell[c]{96.3\% \\ (5.22e-02)} \\\hline
    \makecell[c]{clustering\\coefficient} & \makecell[c]{95.2\% \\ (4.85e-02)} & \makecell[c]{95.2\% \\ (5.21e-02)} & \makecell[c]{58.3\% \\ (5.94e-03)} & \makecell[c]{73.4\% \\ (6.10e-03)} & \makecell[c]{79.7\% \\ (5.94e-03)} & \makecell[c]{79.2\% \\ (6.10e-03)} \\\hline
    \end{tabular}
        \caption{Coverage and width of different types of intervals. The true model is a
three-block DCSBM with $\theta_i$ generated from the uniform distribution. The
average degree is $\lambda_n=2\log n\approx 12.79$. The model is estimated by
$\widehat{P}_{\mathrm{DCSBM}}(\widehat{K}=2)$.}
\label{tab:DCSBM_uniform_2logn_DCSBM_K=2}
\end{table}

\begin{table}[!t]
    \centering
    \begin{tabular}{|c|c|c|c|c|c|c|}
    \hline 
    statistic & \makecell[c]{Based on\\$T(A_{\mathrm{obs}})$\\asymmetric\\\eqref{eq:CI_obs_asym}} & \makecell[c]{Based on\\$T(A_{\mathrm{obs}})$\\symmetric\\\eqref{eq:CI_obs_sym}} & \makecell[c]{Based on\\$\mu(\widehat{P})$\\naive\\asymmetric\\\eqref{eq:CI_mu_naive_asym}} & \makecell[c]{Based on\\$\mu(\widehat{P})$\\naive\\symmetric\\\eqref{eq:CI_mu_naive_sym}} & \makecell[c]{Based on\\$\mu(\widehat{P})$\\corrected\\asymmetric\\\eqref{eq:CI_mu_corrected_asym}} & \makecell[c]{Based on\\$\mu(\widehat{P})$\\corrected\\symmetric\\\eqref{eq:CI_mu_corrected_sym}} \\\hline
    \makecell[c]{average\\degree} & \makecell[c]{95.2\% \\ (8.06e-01)} & \makecell[c]{95.4\% \\ (7.94e-01)} & \makecell[c]{93.4\% \\ (8.04e-01)} & \makecell[c]{95.4\% \\ (8.28e-01)} & \makecell[c]{95.2\% \\ (8.04e-01)} & \makecell[c]{95.5\% \\ (8.28e-01)} \\\hline
    \makecell[c]{triangle\\density} & \makecell[c]{98.6\% \\ (4.75e-06)} & \makecell[c]{99.4\% \\ (4.77e-06)} & \makecell[c]{24.5\% \\ (4.45e-06)} & \makecell[c]{18.9\% \\ (4.63e-06)} & \makecell[c]{94.4\% \\ (4.45e-06)} & \makecell[c]{96.4\% \\ (4.63e-06)} \\\hline
    \makecell[c]{rooted\\triangle\\count} & \makecell[c]{88.5\% \\ (1.40e+01)} & \makecell[c]{94.0\% \\ (1.56e+01)} & \makecell[c]{94.1\% \\ (1.45e+01)} & \makecell[c]{96.4\% \\ (1.58e+01)} & \makecell[c]{90.2\% \\ (1.37e+01)} & \makecell[c]{95.1\% \\ (1.52e+01)} \\\hline
    \makecell[c]{transitivity} & \makecell[c]{96.0\% \\ (5.73e-03)} & \makecell[c]{98.5\% \\ (5.78e-03)} & \makecell[c]{0.6\% \\ (3.60e-03)} & \makecell[c]{0.9\% \\ (3.79e-03)} & \makecell[c]{81.9\% \\ (3.60e-03)} & \makecell[c]{86.6\% \\ (3.79e-03)} \\\hline
    \makecell[c]{assortativity\\by degree} & \makecell[c]{92.4\% \\ (6.04e-02)} & \makecell[c]{92.4\% \\ (6.21e-02)} & \makecell[c]{56.9\% \\ (4.60e-02)} & \makecell[c]{65.6\% \\ (4.98e-02)} & \makecell[c]{68.5\% \\ (4.60e-02)} & \makecell[c]{73.9\% \\ (4.98e-02)} \\\hline
    \makecell[c]{average\\path\\length} & \makecell[c]{96.1\% \\ (5.75e-02)} & \makecell[c]{96.1\% \\ (5.89e-02)} & \makecell[c]{88.1\% \\ (5.40e-02)} & \makecell[c]{90.0\% \\ (5.56e-02)} & \makecell[c]{86.1\% \\ (5.40e-02)} & \makecell[c]{88.6\% \\ (5.56e-02)} \\\hline
    \makecell[c]{closeness\\score} & \makecell[c]{87.5\% \\ (7.54e-02)} & \makecell[c]{92.1\% \\ (8.50e-02)} & \makecell[c]{83.4\% \\ (5.91e-02)} & \makecell[c]{91.1\% \\ (6.28e-02)} & \makecell[c]{82.5\% \\ (5.91e-02)} & \makecell[c]{86.2\% \\ (6.28e-02)} \\\hline
    \makecell[c]{betweenness\\score} & \makecell[c]{96.5\% \\ (5.20e-02)} & \makecell[c]{96.9\% \\ (5.12e-02)} & \makecell[c]{97.0\% \\ (5.11e-02)} & \makecell[c]{97.2\% \\ (5.29e-02)} & \makecell[c]{97.6\% \\ (5.11e-02)} & \makecell[c]{97.9\% \\ (5.29e-02)} \\\hline
    \makecell[c]{clustering\\coefficient} & \makecell[c]{96.0\% \\ (5.06e-02)} & \makecell[c]{96.2\% \\ (5.26e-02)} & \makecell[c]{73.8\% \\ (7.01e-03)} & \makecell[c]{73.0\% \\ (7.26e-03)} & \makecell[c]{69.4\% \\ (7.01e-03)} & \makecell[c]{71.5\% \\ (7.26e-03)} \\\hline
    \end{tabular}
        \caption{Coverage and width of different types of intervals. The true model is a
three-block DCSBM with $\theta_i$ generated from the uniform distribution. The
average degree is $\lambda_n=2\log n\approx 12.79$. The model is estimated by
$\widehat{P}_{\mathrm{DCSBM}}(\widehat{K}=5)$.}
\label{tab:DCSBM_uniform_2logn_DCSBM_K=5}
\end{table}

\begin{table}[!t]
    \centering
    \begin{tabular}{|c|c|c|c|c|c|c|}
    \hline 
    statistic & \makecell[c]{Based on\\$T(A_{\mathrm{obs}})$\\asymmetric\\\eqref{eq:CI_obs_asym}} & \makecell[c]{Based on\\$T(A_{\mathrm{obs}})$\\symmetric\\\eqref{eq:CI_obs_sym}} & \makecell[c]{Based on\\$\mu(\widehat{P})$\\naive\\asymmetric\\\eqref{eq:CI_mu_naive_asym}} & \makecell[c]{Based on\\$\mu(\widehat{P})$\\naive\\symmetric\\\eqref{eq:CI_mu_naive_sym}} & \makecell[c]{Based on\\$\mu(\widehat{P})$\\corrected\\asymmetric\\\eqref{eq:CI_mu_corrected_asym}} & \makecell[c]{Based on\\$\mu(\widehat{P})$\\corrected\\symmetric\\\eqref{eq:CI_mu_corrected_sym}} \\\hline
    \makecell[c]{average\\degree} & \makecell[c]{86.0\% \\ (8.01e-01)} & \makecell[c]{92.0\% \\ (7.93e-01)} & \makecell[c]{86.0\% \\ (8.00e-01)} & \makecell[c]{94.0\% \\ (8.60e-01)} & \makecell[c]{86.0\% \\ (8.00e-01)} & \makecell[c]{92.0\% \\ (8.60e-01)} \\\hline
    \makecell[c]{triangle\\density} & \makecell[c]{98.0\% \\ (4.97e-06)} & \makecell[c]{98.0\% \\ (4.91e-06)} & \makecell[c]{44.0\% \\ (4.72e-06)} & \makecell[c]{30.0\% \\ (5.12e-06)} & \makecell[c]{92.0\% \\ (4.72e-06)} & \makecell[c]{96.0\% \\ (5.12e-06)} \\\hline
    \makecell[c]{rooted\\triangle\\count} & \makecell[c]{90.0\% \\ (1.44e+01)} & \makecell[c]{94.0\% \\ (1.69e+01)} & \makecell[c]{92.0\% \\ (1.51e+01)} & \makecell[c]{96.0\% \\ (1.66e+01)} & \makecell[c]{92.0\% \\ (1.43e+01)} & \makecell[c]{96.0\% \\ (1.62e+01)} \\\hline
    \makecell[c]{transitivity} & \makecell[c]{98.0\% \\ (5.72e-03)} & \makecell[c]{100.0\% \\ (5.89e-03)} & \makecell[c]{2.0\% \\ (3.70e-03)} & \makecell[c]{6.0\% \\ (4.05e-03)} & \makecell[c]{78.0\% \\ (3.70e-03)} & \makecell[c]{84.0\% \\ (4.05e-03)} \\\hline
    \makecell[c]{assortativity\\by degree} & \makecell[c]{92.0\% \\ (6.05e-02)} & \makecell[c]{96.0\% \\ (6.22e-02)} & \makecell[c]{44.0\% \\ (4.19e-02)} & \makecell[c]{46.0\% \\ (4.32e-02)} & \makecell[c]{74.0\% \\ (4.19e-02)} & \makecell[c]{78.0\% \\ (4.32e-02)} \\\hline
    \makecell[c]{average\\path\\length} & \makecell[c]{92.0\% \\ (5.76e-02)} & \makecell[c]{92.0\% \\ (5.89e-02)} & \makecell[c]{82.0\% \\ (5.33e-02)} & \makecell[c]{82.0\% \\ (5.64e-02)} & \makecell[c]{80.0\% \\ (5.33e-02)} & \makecell[c]{82.0\% \\ (5.64e-02)} \\\hline
    \makecell[c]{closeness\\score} & \makecell[c]{92.0\% \\ (8.04e-02)} & \makecell[c]{92.0\% \\ (8.71e-02)} & \makecell[c]{88.0\% \\ (5.91e-02)} & \makecell[c]{94.0\% \\ (6.42e-02)} & \makecell[c]{88.0\% \\ (5.91e-02)} & \makecell[c]{92.0\% \\ (6.42e-02)} \\\hline
    \makecell[c]{betweenness\\score} & \makecell[c]{98.0\% \\ (5.06e-02)} & \makecell[c]{100.0\% \\ (5.28e-02)} & \makecell[c]{100.0\% \\ (4.89e-02)} & \makecell[c]{100.0\% \\ (5.13e-02)} & \makecell[c]{100.0\% \\ (4.89e-02)} & \makecell[c]{100.0\% \\ (5.13e-02)} \\\hline
    \makecell[c]{clustering\\coefficient\\high} & \makecell[c]{98.0\% \\ (4.78e-02)} & \makecell[c]{98.0\% \\ (4.99e-02)} & \makecell[c]{72.0\% \\ (7.55e-03)} & \makecell[c]{80.0\% \\ (7.87e-03)} & \makecell[c]{76.0\% \\ (7.55e-03)} & \makecell[c]{80.0\% \\ (7.87e-03)} \\\hline
    \end{tabular}
        \caption{Coverage and width of different types of intervals. The true model is a
three-block DCSBM with $\theta_i$ generated from the uniform distribution. The
average degree is $\lambda_n=2\log n\approx 12.79$. The model is estimated by
$\widehat{P}_{\mathrm{DCSBM}}(\widehat{K}=7)$.}
\label{tab:DCSBM_uniform_2logn_DCSBM_K=7}
\end{table}

\begin{table}[!htbp]
    \centering
\begin{tabular}{|c|c|c|c|}
\hline 
statistic & \makecell[c]{Based on\\$T(A_{\mathrm{obs}})$\\symmetric\\\eqref{eq:CI_obs_sym}} & \makecell[c]{Based on\\$\mu(\widehat{P})$\\naive\\symmetric\\\eqref{eq:CI_mu_naive_sym}} & \makecell[c]{Based on\\$\mu(\widehat{P})$\\corrected\\symmetric\\\eqref{eq:CI_mu_corrected_sym}} \\\hline
\makecell[c]{average\\degree} & \makecell[c]{92.0\% \\ (7.90e-01)} & \makecell[c]{96.0\% \\ (9.12e-01)} & \makecell[c]{94.0\% \\ (9.12e-01)} \\\hline
\makecell[c]{triangle\\density} & \makecell[c]{100.0\% \\ (6.23e-06)} & \makecell[c]{2.0\% \\ (8.82e-06)} & \makecell[c]{98.0\% \\ (8.82e-06)} \\\hline
\makecell[c]{rooted\\triangle\\count} & \makecell[c]{96.0\% \\ (1.95e+01)} & \makecell[c]{98.0\% \\ (2.84e+01)} & \makecell[c]{98.0\% \\ (2.26e+01)} \\\hline
\makecell[c]{transitivity} & \makecell[c]{100.0\% \\ (6.99e-03)} & \makecell[c]{0.0\% \\ (6.94e-03)} & \makecell[c]{100.0\% \\ (6.94e-03)} \\\hline
\makecell[c]{assortativity\\by degree} & \makecell[c]{92.0\% \\ (6.02e-02)} & \makecell[c]{8.6\% \\ (6.95e-03)} & \makecell[c]{24.4\% \\ (6.95e-03)} \\\hline
\makecell[c]{average\\path\\length} & \makecell[c]{92.0\% \\ (5.85e-02)} & \makecell[c]{90.0\% \\ (5.47e-02)} & \makecell[c]{72.0\% \\ (5.47e-02)} \\\hline
\makecell[c]{closeness\\score} & \makecell[c]{92.0\% \\ (8.75e-02)} & \makecell[c]{94.4\% \\ (6.92e-02)} & \makecell[c]{91.2\% \\ (6.92e-02)} \\\hline
\makecell[c]{betweenness\\score} & \makecell[c]{94.0\% \\ (2.14e-01)} & \makecell[c]{96.0\% \\ (2.40e-01)} & \makecell[c]{92.0\% \\ (2.41e-01)} \\\hline
\makecell[c]{clustering\\coefficient} & \makecell[c]{100.0\% \\ (5.40e-02)} & \makecell[c]{43.8\% \\ (1.88e-02)} & \makecell[c]{94.0\% \\ (1.88e-02)} \\\hline
\end{tabular}
    \caption{Coverage and width of different types of intervals. The true model is a
three-block DCSBM with $\theta_i$ generated from the uniform distribution. The
average degree is $\lambda_n=2\log n\approx 12.79$. The model is estimated by
$\widehat{P}_{\mathrm{SVD}}(\widehat{K}=3)$.}
\label{tab:DCSBM_uniform_2logn_SVD}
\end{table}

\begin{table}[!t]
    \centering
    \begin{tabular}{|c|c|c|c|c|c|c|}
    \hline 
    statistic & \makecell[c]{Based on\\$T(A_{\mathrm{obs}})$\\symmetric\\(\ref{eq:CI_obs_sym})} & \makecell[c]{Based on\\$\mu(\widehat{P})$\\naive\\symmetric\\(\ref{eq:CI_mu_naive_sym})} & \makecell[c]{Based on\\$\mu(\widehat{P})$\\corrected\\symmetric\\(\ref{eq:CI_mu_corrected_sym})} \\\hline
    \makecell[c]{average\\degree} & \makecell[c]{96.0\% \\ (7.86e-01)} & \makecell[c]{91.4\% \\ (9.30e-01)} & \makecell[c]{97.2\% \\ (9.30e-01)} \\\hline
    \makecell[c]{triangle\\density} & \makecell[c]{100.0\% \\ (5.80e-06)} & \makecell[c]{7.8\% \\ (7.70e-06)} & \makecell[c]{98.0\% \\ (7.70e-06)} \\\hline
    \makecell[c]{rooted\\triangle\\count} & \makecell[c]{96.6\% \\ (1.80e+01)} & \makecell[c]{99.1\% \\ (2.36e+01)} & \makecell[c]{98.5\% \\ (1.97e+01)} \\\hline
    \makecell[c]{transitivity} & \makecell[c]{100.0\% \\ (6.72e-03)} & \makecell[c]{0.0\% \\ (6.21e-03)} & \makecell[c]{97.5\% \\ (6.21e-03)} \\\hline
    \makecell[c]{assortativity\\by degree} & \makecell[c]{91.8\% \\ (6.14e-02)} & \makecell[c]{4.9\% \\ (5.32e-03)} & \makecell[c]{12.9\% \\ (5.32e-03)} \\\hline
    \makecell[c]{average\\path\\length} & \makecell[c]{96.6\% \\ (6.00e-02)} & \makecell[c]{84.2\% \\ (5.75e-02)} & \makecell[c]{80.6\% \\ (5.75e-02)} \\\hline
    \makecell[c]{closeness\\score} & \makecell[c]{92.9\% \\ (9.00e-02)} & \makecell[c]{92.0\% \\ (6.76e-02)} & \makecell[c]{88.6\% \\ (6.76e-02)} \\\hline
    \makecell[c]{betweenness\\score} & \makecell[c]{97.4\% \\ (5.02e-02)} & \makecell[c]{96.3\% \\ (5.27e-02)} & \makecell[c]{97.4\% \\ (5.27e-02)} \\\hline
    \makecell[c]{clustering\\coefficient\\high} & \makecell[c]{97.2\% \\ (5.51e-02)} & \makecell[c]{71.2\% \\ (1.74e-02)} & \makecell[c]{92.0\% \\ (1.74e-02)} \\\hline
    \end{tabular}
        \caption{Coverage and width of different types of intervals. The true model is a
three-block DCSBM with $\theta_i$ generated from the uniform distribution. The
average degree is $\lambda_n=2\log n\approx 12.79$. The model is estimated by
$\widehat{P}_{\mathrm{SVD}}(\widehat{K}=2)$.}
\label{tab:DCSBM_uniform_2logn_SVD_K=2}
\end{table}

\begin{table}[!htbp]
    \centering
    \begin{tabular}{|c|c|c|c|c|c|c|}
    \hline 
    statistic & \makecell[c]{Based on\\$T(A_{\mathrm{obs}})$\\symmetric\\(\ref{eq:CI_obs_sym})} & \makecell[c]{Based on\\$\mu(\widehat{P})$\\naive\\symmetric\\(\ref{eq:CI_mu_naive_sym})} & \makecell[c]{Based on\\$\mu(\widehat{P})$\\corrected\\symmetric\\(\ref{eq:CI_mu_corrected_sym})} \\\hline
    \makecell[c]{average\\degree} & \makecell[c]{94.2\% \\ (8.01e-01)} & \makecell[c]{37.0\% \\ (9.81e-01)} & \makecell[c]{92.0\% \\ (9.81e-01)} \\\hline
    \makecell[c]{triangle\\density} & \makecell[c]{100.0\% \\ (7.19e-06)} & \makecell[c]{0.0\% \\ (1.13e-05)} & \makecell[c]{97.8\% \\ (1.13e-05)} \\\hline
    \makecell[c]{rooted\\triangle\\count} & \makecell[c]{99.0\% \\ (2.06e+01)} & \makecell[c]{100.0\% \\ (3.24e+01)} & \makecell[c]{99.0\% \\ (2.45e+01)} \\\hline
    \makecell[c]{transitivity} & \makecell[c]{100.0\% \\ (7.29e-03)} & \makecell[c]{0.0\% \\ (7.32e-03)} & \makecell[c]{97.2\% \\ (7.32e-03)} \\\hline
    \makecell[c]{assortativity\\by degree} & \makecell[c]{89.5\% \\ (5.85e-02)} & \makecell[c]{4.8\% \\ (8.89e-03)} & \makecell[c]{30.8\% \\ (8.89e-03)} \\\hline
    \makecell[c]{average\\path\\length} & \makecell[c]{95.0\% \\ (5.55e-02)} & \makecell[c]{48.2\% \\ (5.13e-02)} & \makecell[c]{63.5\% \\ (5.13e-02)} \\\hline
    \makecell[c]{closeness\\score} & \makecell[c]{93.0\% \\ (8.37e-02)} & \makecell[c]{90.8\% \\ (7.06e-02)} & \makecell[c]{94.0\% \\ (7.06e-02)} \\\hline
    \makecell[c]{betweenness\\score} & \makecell[c]{97.8\% \\ (4.93e-02)} & \makecell[c]{86.2\% \\ (5.21e-02)} & \makecell[c]{97.5\% \\ (5.21e-02)} \\\hline
    \makecell[c]{clustering\\coefficient} & \makecell[c]{98.5\% \\ (5.71e-02)} & \makecell[c]{6.0\% \\ (1.88e-02)} & \makecell[c]{86.8\% \\ (1.88e-02)} \\\hline
    \end{tabular}
        \caption{Coverage and width of different types of intervals. The true model is a
three-block DCSBM with $\theta_i$ generated from the uniform distribution. The
average degree is $\lambda_n=2\log n\approx 12.79$. The model is estimated by
$\widehat{P}_{\mathrm{SVD}}(\widehat{K}=5)$.}
\label{tab:DCSBM_uniform_2logn_SVD_K=5}
\end{table}

\begin{table}[!htbp]
    \centering
    \begin{tabular}{|c|c|c|c|c|c|c|}
    \hline 
    statistic & \makecell[c]{Based on\\$T(A_{\mathrm{obs}})$\\symmetric\\(\ref{eq:CI_obs_sym})} & \makecell[c]{Based on\\$\mu(\widehat{P})$\\naive\\symmetric\\(\ref{eq:CI_mu_naive_sym})} & \makecell[c]{Based on\\$\mu(\widehat{P})$\\corrected\\symmetric\\(\ref{eq:CI_mu_corrected_sym})} \\\hline
    \makecell[c]{average\\degree} & \makecell[c]{95.0\% \\ (8.34e-01)} & \makecell[c]{0.0\% \\ (9.83e-01)} & \makecell[c]{95.0\% \\ (9.83e-01)} \\\hline
    \makecell[c]{triangle\\density} & \makecell[c]{100.0\% \\ (8.91e-06)} & \makecell[c]{0.0\% \\ (1.27e-05)} & \makecell[c]{13.0\% \\ (1.24e-05)} \\\hline
    \makecell[c]{rooted\\triangle\\count} & \makecell[c]{100.0\% \\ (2.38e+01)} & \makecell[c]{99.8\% \\ (4.75e+01)} & \makecell[c]{100.0\% \\ (3.15e+01)} \\\hline
    \makecell[c]{transitivity} & \makecell[c]{99.8\% \\ (7.26e-03)} & \makecell[c]{0.0\% \\ (5.66e-03)} & \makecell[c]{0.0\% \\ (5.66e-03)} \\\hline
    \makecell[c]{assortativity\\by degree} & \makecell[c]{91.2\% \\ (5.53e-02)} & \makecell[c]{1.8\% \\ (9.97e-03)} & \makecell[c]{14.8\% \\ (9.97e-03)} \\\hline
    \makecell[c]{average\\path\\length} & \makecell[c]{93.0\% \\ (4.88e-02)} & \makecell[c]{0.0\% \\ (4.10e-02)} & \makecell[c]{84.5\% \\ (4.10e-02)} \\\hline
    \makecell[c]{closeness\\score} & \makecell[c]{90.8\% \\ (7.93e-02)} & \makecell[c]{81.8\% \\ (7.20e-02)} & \makecell[c]{90.8\% \\ (7.20e-02)} \\\hline
    \makecell[c]{betweenness\\score} & \makecell[c]{95.2\% \\ (4.84e-02)} & \makecell[c]{66.5\% \\ (5.39e-02)} & \makecell[c]{95.8\% \\ (5.39e-02)} \\\hline
    \makecell[c]{clustering\\coefficient} & \makecell[c]{98.5\% \\ (5.47e-02)} & \makecell[c]{0.0\% \\ (1.65e-02)} & \makecell[c]{59.2\% \\ (1.65e-02)} \\\hline
    \end{tabular}
        \caption{Coverage and width of different types of intervals. The true model is a
three-block DCSBM with $\theta_i$ generated from the uniform distribution. The
average degree is $\lambda_n=2\log n\approx 12.79$. The model is estimated by
$\widehat{P}_{\mathrm{SVD}}(\widehat{K}=10)$.}
\label{tab:DCSBM_uniform_2logn_SVD_K=10}
\end{table}

\begin{table}[!htbp]
    \centering
    \begin{tabular}{|c|c|c|c|c|c|c|}
    \hline 
    statistic & \makecell[c]{Based on\\$T(A_{\mathrm{obs}})$\\asymmetric\\\eqref{eq:CI_obs_asym}} & \makecell[c]{Based on\\$T(A_{\mathrm{obs}})$\\symmetric\\\eqref{eq:CI_obs_sym}} & \makecell[c]{Based on\\$\mu(\widehat{P})$\\naive\\asymmetric\\\eqref{eq:CI_mu_naive_asym}} & \makecell[c]{Based on\\$\mu(\widehat{P})$\\naive\\symmetric\\\eqref{eq:CI_mu_naive_sym}} & \makecell[c]{Based on\\$\mu(\widehat{P})$\\corrected\\asymmetric\\\eqref{eq:CI_mu_corrected_asym}} & \makecell[c]{Based on\\$\mu(\widehat{P})$\\corrected\\symmetric\\\eqref{eq:CI_mu_corrected_sym}} \\\hline
    \makecell[c]{average\\degree} & \makecell[c]{94.5\% \\ (9.48e-01)} & \makecell[c]{86.2\% \\ (7.92e-01)} & \makecell[c]{92.5\% \\ (9.45e-01)} & \makecell[c]{92.5\% \\ (8.87e-01)} & \makecell[c]{94.5\% \\ (9.45e-01)} & \makecell[c]{90.2\% \\ (8.87e-01)} \\\hline
    \makecell[c]{triangle\\density} & \makecell[c]{93.5\% \\ (3.78e-06)} & \makecell[c]{93.5\% \\ (3.56e-06)} & \makecell[c]{88.5\% \\ (2.77e-06)} & \makecell[c]{90.5\% \\ (2.74e-06)} & \makecell[c]{92.5\% \\ (2.77e-06)} & \makecell[c]{90.2\% \\ (2.74e-06)} \\\hline
    \makecell[c]{rooted\\triangle\\count} & \makecell[c]{90.8\% \\ (5.06e+00)} & \makecell[c]{93.8\% \\ (6.51e+00)} & \makecell[c]{96.8\% \\ (7.74e-01)} & \makecell[c]{96.8\% \\ (7.94e-01)} & \makecell[c]{94.5\% \\ (7.74e-01)} & \makecell[c]{94.5\% \\ (7.94e-01)} \\\hline
    \makecell[c]{transitivity} & \makecell[c]{96.8\% \\ (5.83e-03)} & \makecell[c]{98.8\% \\ (5.87e-03)} & \makecell[c]{89.2\% \\ (2.39e-03)} & \makecell[c]{91.5\% \\ (2.51e-03)} & \makecell[c]{94.5\% \\ (2.39e-03)} & \makecell[c]{94.8\% \\ (2.51e-03)} \\\hline
    \makecell[c]{assortativity\\by degree} & \makecell[c]{95.0\% \\ (6.71e-02)} & \makecell[c]{100.0\% \\ (7.51e-02)} & \makecell[c]{93.0\% \\ (3.20e-02)} & \makecell[c]{91.8\% \\ (3.24e-02)} & \makecell[c]{89.0\% \\ (3.20e-02)} & \makecell[c]{91.0\% \\ (3.24e-02)} \\\hline
    \makecell[c]{average\\path\\length} & \makecell[c]{92.5\% \\ (6.75e-02)} & \makecell[c]{88.2\% \\ (5.72e-02)} & \makecell[c]{94.5\% \\ (6.75e-02)} & \makecell[c]{92.5\% \\ (6.32e-02)} & \makecell[c]{92.2\% \\ (6.75e-02)} & \makecell[c]{89.2\% \\ (6.32e-02)} \\\hline
    \makecell[c]{closeness\\score} & \makecell[c]{97.0\% \\ (5.60e-02)} & \makecell[c]{94.8\% \\ (5.58e-02)} & \makecell[c]{92.5\% \\ (8.92e-03)} & \makecell[c]{94.8\% \\ (8.99e-03)} & \makecell[c]{94.2\% \\ (8.92e-03)} & \makecell[c]{93.5\% \\ (8.99e-03)} \\\hline
    \makecell[c]{betweenness\\score} & \makecell[c]{96.5\% \\ (6.25e-02)} & \makecell[c]{98.8\% \\ (6.81e-02)} & \makecell[c]{92.2\% \\ (1.11e-02)} & \makecell[c]{89.0\% \\ (1.07e-02)} & \makecell[c]{92.0\% \\ (1.11e-02)} & \makecell[c]{89.0\% \\ (1.07e-02)} \\\hline
    \makecell[c]{clustering\\coefficient} & \makecell[c]{88.8\% \\ (5.36e-02)} & \makecell[c]{100.0\% \\ (7.16e-02)} & \makecell[c]{97.8\% \\ (4.36e-03)} & \makecell[c]{98.0\% \\ (4.51e-03)} & \makecell[c]{90.0\% \\ (4.36e-03)} & \makecell[c]{90.8\% \\ (4.51e-03)} \\\hline
    \end{tabular}
        \caption{Coverage and width of different types of intervals. The true model is a
three-block SBM. The average degree is $\lambda_n=2\log n\approx 12.79$. The
model is estimated by $\widehat{P}_{\mathrm{SBM}}(\widehat{K}=3)$.}
\label{tab:SBM_2logn_SBM}
\end{table}

\begin{table}[ht]
    \centering
    \begin{tabular}{|c|c|c|c|c|c|c|}
    \hline 
    statistic & \makecell[c]{Based on\\$T(A_{\mathrm{obs}})$\\asymmetric\\(\ref{eq:CI_obs_asym})} & \makecell[c]{Based on\\$T(A_{\mathrm{obs}})$\\symmetric\\(\ref{eq:CI_obs_sym})} & \makecell[c]{Based on\\$\mu(\widehat{P})$\\naive\\asymmetric\\(\ref{eq:CI_mu_naive_asym})} & \makecell[c]{Based on\\$\mu(\widehat{P})$\\naive\\symmetric\\(\ref{eq:CI_mu_naive_sym})} & \makecell[c]{Based on\\$\mu(\widehat{P})$\\corrected\\asymmetric\\(\ref{eq:CI_mu_corrected_asym})} & \makecell[c]{Based on\\$\mu(\widehat{P})$\\corrected\\symmetric\\(\ref{eq:CI_mu_corrected_sym})} \\\hline
    \makecell[c]{average\\degree} & \makecell[c]{85.7\% \\ (7.82e-01)} & \makecell[c]{89.8\% \\ (7.96e-01)} & \makecell[c]{85.7\% \\ (7.82e-01)} & \makecell[c]{91.8\% \\ (8.26e-01)} & \makecell[c]{85.7\% \\ (7.82e-01)} & \makecell[c]{91.8\% \\ (8.26e-01)} \\\hline
    \makecell[c]{triangle\\density} & \makecell[c]{84.5\% \\ (2.97e-06)} & \makecell[c]{87.3\% \\ (3.16e-06)} & \makecell[c]{17.8\% \\ (2.38e-06)} & \makecell[c]{24.5\% \\ (2.44e-06)} & \makecell[c]{20.0\% \\ (2.38e-06)} & \makecell[c]{25.8\% \\ (2.44e-06)} \\\hline
    \makecell[c]{rooted\\triangle\\count} & \makecell[c]{85.3\% \\ (5.04e+00)} & \makecell[c]{93.7\% \\ (6.66e+00)} & \makecell[c]{88.5\% \\ (1.63e+00)} & \makecell[c]{90.2\% \\ (1.44e+00)} & \makecell[c]{88.2\% \\ (1.63e+00)} & \makecell[c]{90.2\% \\ (1.44e+00)} \\\hline
    \makecell[c]{transitivity} & \makecell[c]{90.7\% \\ (5.04e-03)} & \makecell[c]{89.8\% \\ (5.28e-03)} & \makecell[c]{0.0\% \\ (2.34e-03)} & \makecell[c]{0.0\% \\ (2.39e-03)} & \makecell[c]{0.0\% \\ (2.34e-03)} & \makecell[c]{0.0\% \\ (2.39e-03)} \\\hline
    \makecell[c]{assortativity\\by degree} & \makecell[c]{98.2\% \\ (7.51e-02)} & \makecell[c]{98.2\% \\ (7.64e-02)} & \makecell[c]{92.7\% \\ (4.59e-02)} & \makecell[c]{94.7\% \\ (4.91e-02)} & \makecell[c]{92.8\% \\ (4.59e-02)} & \makecell[c]{94.8\% \\ (4.91e-02)} \\\hline
    \makecell[c]{average\\path\\length} & \makecell[c]{86.8\% \\ (5.11e-02)} & \makecell[c]{89.8\% \\ (5.53e-02)} & \makecell[c]{88.2\% \\ (5.14e-02)} & \makecell[c]{88.2\% \\ (5.53e-02)} & \makecell[c]{88.2\% \\ (5.14e-02)} & \makecell[c]{88.3\% \\ (5.53e-02)} \\\hline
    \makecell[c]{closeness\\score} & \makecell[c]{97.2\% \\ (5.85e-02)} & \makecell[c]{96.8\% \\ (5.78e-02)} & \makecell[c]{91.7\% \\ (1.11e-02)} & \makecell[c]{92.3\% \\ (1.07e-02)} & \makecell[c]{92.5\% \\ (1.11e-02)} & \makecell[c]{92.7\% \\ (1.07e-02)} \\\hline
    \makecell[c]{betweenness\\score} & \makecell[c]{91.3\% \\ (5.09e-02)} & \makecell[c]{92.2\% \\ (5.31e-02)} & \makecell[c]{88.2\% \\ (1.08e-02)} & \makecell[c]{94.5\% \\ (1.05e-02)} & \makecell[c]{90.5\% \\ (1.08e-02)} & \makecell[c]{95.0\% \\ (1.05e-02)} \\\hline
    \makecell[c]{clustering\\coefficient} & \makecell[c]{87.3\% \\ (5.41e-02)} & \makecell[c]{98.3\% \\ (6.62e-02)} & \makecell[c]{89.3\% \\ (9.87e-03)} & \makecell[c]{88.3\% \\ (8.78e-03)} & \makecell[c]{86.2\% \\ (9.87e-03)} & \makecell[c]{88.2\% \\ (8.78e-03)} \\\hline
    \end{tabular}
        \caption{Coverage and width of different types of intervals. The true model is a
three-block SBM. The average degree is $\lambda_n=2\log n\approx 12.79$. The
model is estimated by $\widehat{P}_{\mathrm{SBM}}(\widehat{K}=2)$.}
    \label{tab:SBM_2logn_SBM_K=2}
\end{table}

\begin{table}[ht]
    \centering
    \begin{tabular}{|c|c|c|c|c|c|c|}
    \hline 
    statistic & \makecell[c]{Based on\\$T(A_{\mathrm{obs}})$\\asymmetric\\(\ref{eq:CI_obs_asym})} & \makecell[c]{Based on\\$T(A_{\mathrm{obs}})$\\symmetric\\(\ref{eq:CI_obs_sym})} & \makecell[c]{Based on\\$\mu(\widehat{P})$\\naive\\asymmetric\\(\ref{eq:CI_mu_naive_asym})} & \makecell[c]{Based on\\$\mu(\widehat{P})$\\naive\\symmetric\\(\ref{eq:CI_mu_naive_sym})} & \makecell[c]{Based on\\$\mu(\widehat{P})$\\corrected\\asymmetric\\(\ref{eq:CI_mu_corrected_asym})} & \makecell[c]{Based on\\$\mu(\widehat{P})$\\corrected\\symmetric\\(\ref{eq:CI_mu_corrected_sym})} \\\hline
    \makecell[c]{average\\degree} & \makecell[c]{91.7\% \\ (8.13e-01)} & \makecell[c]{90.4\% \\ (7.93e-01)} & \makecell[c]{91.9\% \\ (8.19e-01)} & \makecell[c]{91.7\% \\ (8.44e-01)} & \makecell[c]{92.4\% \\ (8.19e-01)} & \makecell[c]{91.0\% \\ (8.44e-01)} \\\hline
    \makecell[c]{triangle\\density} & \makecell[c]{93.1\% \\ (3.60e-06)} & \makecell[c]{95.9\% \\ (3.85e-06)} & \makecell[c]{91.4\% \\ (3.76e-06)} & \makecell[c]{92.7\% \\ (4.15e-06)} & \makecell[c]{96.6\% \\ (3.76e-06)} & \makecell[c]{96.6\% \\ (4.15e-06)} \\\hline
    \makecell[c]{rooted\\triangle\\count} & \makecell[c]{83.6\% \\ (5.55e+00)} & \makecell[c]{95.1\% \\ (7.38e+00)} & \makecell[c]{86.7\% \\ (5.17e+00)} & \makecell[c]{95.1\% \\ (6.44e+00)} & \makecell[c]{84.7\% \\ (4.88e+00)} & \makecell[c]{90.7\% \\ (6.15e+00)} \\\hline
    \makecell[c]{transitivity} & \makecell[c]{90.6\% \\ (5.38e-03)} & \makecell[c]{96.7\% \\ (6.00e-03)} & \makecell[c]{90.4\% \\ (5.11e-03)} & \makecell[c]{95.3\% \\ (5.60e-03)} & \makecell[c]{95.3\% \\ (5.11e-03)} & \makecell[c]{96.7\% \\ (5.60e-03)} \\\hline
    \makecell[c]{assortativity\\by degree} & \makecell[c]{95.9\% \\ (6.91e-02)} & \makecell[c]{97.1\% \\ (7.19e-02)} & \makecell[c]{55.9\% \\ (6.11e-02)} & \makecell[c]{63.4\% \\ (6.37e-02)} & \makecell[c]{93.9\% \\ (6.11e-02)} & \makecell[c]{95.0\% \\ (6.37e-02)} \\\hline
    \makecell[c]{average\\path\\length} & \makecell[c]{92.0\% \\ (5.61e-02)} & \makecell[c]{89.9\% \\ (5.59e-02)} & \makecell[c]{91.7\% \\ (5.63e-02)} & \makecell[c]{92.4\% \\ (5.80e-02)} & \makecell[c]{91.7\% \\ (5.63e-02)} & \makecell[c]{92.4\% \\ (5.80e-02)} \\\hline
    \makecell[c]{closeness\\score} & \makecell[c]{93.1\% \\ (5.92e-02)} & \makecell[c]{96.7\% \\ (5.88e-02)} & \makecell[c]{85.0\% \\ (3.57e-02)} & \makecell[c]{100.0\% \\ (4.04e-02)} & \makecell[c]{74.0\% \\ (3.57e-02)} & \makecell[c]{95.9\% \\ (4.04e-02)} \\\hline
    \makecell[c]{betweenness\\score} & \makecell[c]{88.6\% \\ (5.23e-02)} & \makecell[c]{92.1\% \\ (5.46e-02)} & \makecell[c]{84.4\% \\ (3.93e-02)} & \makecell[c]{95.9\% \\ (4.27e-02)} & \makecell[c]{75.0\% \\ (3.93e-02)} & \makecell[c]{88.3\% \\ (4.27e-02)} \\\hline
    \makecell[c]{clustering\\coefficient} & \makecell[c]{91.3\% \\ (5.69e-02)} & \makecell[c]{97.1\% \\ (7.09e-02)} & \makecell[c]{87.1\% \\ (2.04e-02)} & \makecell[c]{89.9\% \\ (2.15e-02)} & \makecell[c]{87.1\% \\ (2.04e-02)} & \makecell[c]{88.6\% \\ (2.15e-02)} \\\hline
    \end{tabular}
        \caption{Coverage and width of different types of intervals. The true model is a
three-block SBM. The average degree is $\lambda_n=2\log n\approx 12.79$. The
model is estimated by $\widehat{P}_{\mathrm{SBM}}(\widehat{K}=5)$.}
    \label{tab:SBM_2logn_SBM_K=5}
\end{table}

\begin{table}[ht]
    \centering
    \begin{tabular}{|c|c|c|c|c|c|c|}
    \hline 
    statistic & \makecell[c]{Based on\\$T(A_{\mathrm{obs}})$\\asymmetric\\(\ref{eq:CI_obs_asym})} & \makecell[c]{Based on\\$T(A_{\mathrm{obs}})$\\symmetric\\(\ref{eq:CI_obs_sym})} & \makecell[c]{Based on\\$\mu(\widehat{P})$\\naive\\asymmetric\\(\ref{eq:CI_mu_naive_asym})} & \makecell[c]{Based on\\$\mu(\widehat{P})$\\naive\\symmetric\\(\ref{eq:CI_mu_naive_sym})} & \makecell[c]{Based on\\$\mu(\widehat{P})$\\corrected\\asymmetric\\(\ref{eq:CI_mu_corrected_asym})} & \makecell[c]{Based on\\$\mu(\widehat{P})$\\corrected\\symmetric\\(\ref{eq:CI_mu_corrected_sym})} \\\hline
    \makecell[c]{average\\degree} & \makecell[c]{89.7\% \\ (8.00e-01)} & \makecell[c]{90.4\% \\ (7.93e-01)} & \makecell[c]{97.3\% \\ (2.52e+00)} & \makecell[c]{97.3\% \\ (2.43e+00)} & \makecell[c]{93.9\% \\ (2.52e+00)} & \makecell[c]{98.0\% \\ (2.43e+00)} \\\hline
    \makecell[c]{triangle\\density} & \makecell[c]{97.1\% \\ (4.06e-06)} & \makecell[c]{96.6\% \\ (4.14e-06)} & \makecell[c]{98.7\% \\ (1.82e-05)} & \makecell[c]{100.0\% \\ (1.96e-05)} & \makecell[c]{97.9\% \\ (1.56e-05)} & \makecell[c]{100.0\% \\ (1.94e-05)} \\\hline
    \makecell[c]{rooted\\triangle\\count} & \makecell[c]{78.6\% \\ (5.80e+00)} & \makecell[c]{95.1\% \\ (7.57e+00)} & \makecell[c]{89.6\% \\ (6.56e+00)} & \makecell[c]{95.9\% \\ (8.80e+00)} & \makecell[c]{84.1\% \\ (5.95e+00)} & \makecell[c]{94.6\% \\ (8.18e+00)} \\\hline
    \makecell[c]{transitivity} & \makecell[c]{95.4\% \\ (5.98e-03)} & \makecell[c]{96.7\% \\ (6.18e-03)} & \makecell[c]{98.7\% \\ (1.61e-02)} & \makecell[c]{98.7\% \\ (1.55e-02)} & \makecell[c]{98.0\% \\ (1.61e-02)} & \makecell[c]{97.3\% \\ (1.55e-02)} \\\hline
    \makecell[c]{assortativity\\by degree} & \makecell[c]{94.3\% \\ (6.52e-02)} & \makecell[c]{97.1\% \\ (6.85e-02)} & \makecell[c]{14.4\% \\ (1.12e-01)} & \makecell[c]{88.0\% \\ (1.06e-01)} & \makecell[c]{93.1\% \\ (1.12e-01)} & \makecell[c]{96.6\% \\ (1.06e-01)} \\\hline
    \makecell[c]{average\\path\\length} & \makecell[c]{91.3\% \\ (5.59e-02)} & \makecell[c]{90.6\% \\ (5.63e-02)} & \makecell[c]{96.6\% \\ (1.52e-01)} & \makecell[c]{97.3\% \\ (1.44e-01)} & \makecell[c]{94.6\% \\ (1.52e-01)} & \makecell[c]{97.3\% \\ (1.44e-01)} \\\hline
    \makecell[c]{closeness\\score} & \makecell[c]{93.9\% \\ (5.77e-02)} & \makecell[c]{96.0\% \\ (6.04e-02)} & \makecell[c]{89.9\% \\ (5.20e-02)} & \makecell[c]{98.0\% \\ (5.51e-02)} & \makecell[c]{80.7\% \\ (5.20e-02)} & \makecell[c]{92.7\% \\ (5.51e-02)} \\\hline
    \makecell[c]{betweenness\\score} & \makecell[c]{91.9\% \\ (5.39e-02)} & \makecell[c]{92.9\% \\ (5.54e-02)} & \makecell[c]{91.9\% \\ (5.26e-02)} & \makecell[c]{95.3\% \\ (5.40e-02)} & \makecell[c]{87.6\% \\ (5.26e-02)} & \makecell[c]{92.4\% \\ (5.40e-02)} \\\hline
    \makecell[c]{clustering\\coefficient} & \makecell[c]{87.7\% \\ (5.72e-02)} & \makecell[c]{97.1\% \\ (7.53e-02)} & \makecell[c]{96.0\% \\ (2.84e-02)} & \makecell[c]{99.4\% \\ (2.87e-02)} & \makecell[c]{94.0\% \\ (2.83e-02)} & \makecell[c]{96.0\% \\ (2.87e-02)} \\\hline
    \end{tabular}
        \caption{Coverage and width of different types of intervals. The true model is a
three-block SBM. The average degree is $\lambda_n=2\log n\approx 12.79$. The
model is estimated by $\widehat{P}_{\mathrm{SBM}}(\widehat{K}=10)$.}
    \label{tab:SBM_2logn_SBM_K=10}
\end{table}

\clearpage

\section{Expectations of Subgraph Counts}

Let \(R\) be a fixed simple graph, or motif, with vertex set \(V(R)\), edge set
\(E(R)\), and
\[
v=|V(R)|,
\qquad
e=|E(R)|.
\]
For \(a\in V(R)\), we denote by $r_a$ the degree of $a$ in $R$. It follows that $\sum_{a\in V(R)} r_a = 2e$. The goal of this section is to compute the biases of $\mu_R(\widehat{P})$ and $\mu_{R,r}^{(i)}(\widehat P)$. 

\paragraph{Global subgraph counts.} 
Recall from \eqref{eq:global count} that the
number of non-automorphic copies of \(R\) in the observed graph with adjacency matrix \(A\) is 
\[
T_R(A)
=
\frac{1}{|\operatorname{Aut}(R)|}
\sum_{\phi:V(R)\hookrightarrow[n]}
\prod_{\{a,b\}\in E(R)}
A_{\phi(a)\phi(b)}.
\]
Taking expectation of $T_R(A)$ under the Chung--Lu model gives
\[
\mu_R(P)
:=
\mathbb E_P[T_R(A)]
=
\frac{1}{|\operatorname{Aut}(R)|}
\sum_{\phi:V(R)\hookrightarrow[n]}
\prod_{\{a,b\}\in E(R)}
P_{\phi(a)\phi(b)}.
\]
Substituting \(P_{ij}=p\theta_i\theta_j\), we obtain
\[
\prod_{\{a,b\}\in E(R)}
P_{\phi(a)\phi(b)}
=
\prod_{\{a,b\}\in E(R)}
p\theta_{\phi(a)}\theta_{\phi(b)}
=
p^e
\prod_{\{a,b\}\in E(R)}
\theta_{\phi(a)}^{r_a}.
\]
Indeed, each edge contributes one factor \(p\), giving \(p^e\), and each occurrence
of a motif vertex \(a\) as an endpoint contributes one factor
\(\theta_{\phi(a)}\). Since vertex \(a\) is incident to \(r_a\) edges in \(R\), the
factor \(\theta_{\phi(a)}\) appears \(r_a\) times.
Therefore
\begin{equation}
\label{eq:mu P}
\mu_R(P)
=
\frac{p^e}{|\operatorname{Aut}(R)|}
\sum_{\phi:V(R)\hookrightarrow[n]}
\prod_{a\in V(R)}
\theta_{\phi(a)}^{r_a}.
\end{equation}

When $p$ and $\theta_i$ are unknown, we estimate them by
\begin{equation}
\label{eq: main parameter estimators}
\hat{p} = \frac{S}{n(n-1)}, \quad \hat{\theta}_i = \frac{n D_i}{S}, \quad D_i=\sum_{j\neq i} A_{ij}, \quad S=\sum_{i=1}^n D_i.    
\end{equation}
When $S=0$, which occurs with exponentially small probability, both numerator and denominator of $\hat{\theta}_i$ are zero. In that case, we set $\hat{\theta}_i=1$ for all $i\in[n]$. 
Since \(\sum_{a\in V(R)}r_a=2e\), plugging these estimates in the formula of $\mu_R(P)$, we get 
\begin{equation}
\label{eq:mu Phat}
\mu_R(\widehat P)
=
\frac{1}{|\operatorname{Aut}(R)|}
\left(\frac{n}{n-1}\right)^e
S^{-e}
\sum_{\phi:V(R)\hookrightarrow[n]}
\prod_{a\in V(R)}
D_{\phi(a)}^{r_a},    
\end{equation}
if $S\ge 1$, and $\mu_R(\widehat{P})=0$ if $S=0$.

\paragraph{Rooted subgraph counts.}
Recall from \eqref{eq:rooted count} that, for a rooted subgraph $(R,r)$, the
number of rooted copies of $(R,r)$ at node $i\in[n]$ is
$$
T_{R,r}^{(i)}(A) =
\frac{1}{|\operatorname{Aut}(R,r)|}
\sum_{\substack{\phi:V(R)\hookrightarrow[n] \\ \phi(r)=i}} \ 
\prod_{\{a,b\}\in E(R)}
A_{\phi(a)\phi(b)}.    
$$
Similarly to the case of global subgraph counts, the expectation of
$T_{R,r}^{(i)}(A)$ is given by
\begin{equation}
\label{eq:rooted mu P}
\mu_{R,r}^{(i)}(P)
=
\frac{p^e}{|\operatorname{Aut}(R,r)|}
\sum_{\substack{\phi:V(R)\hookrightarrow[n] \\ \phi(r)=i}} \ 
\prod_{a\in V(R)}
\theta_{\phi(a)}^{r_a},
\end{equation}
and its plugin estimator is
\begin{equation}
\label{eq:rooted mu Phat}
\mu_{R,r}^{(i)}(\widehat P)
=
\frac{1}{|\operatorname{Aut}(R,r)|}
\left(\frac{n}{n-1}\right)^e
S^{-e}
\sum_{\substack{\phi:V(R)\hookrightarrow[n] \\ \phi(r)=i}} \
\prod_{a\in V(R)}
D_{\phi(a)}^{r_a}.    
\end{equation}

\paragraph{Falling factorials.}
The main obstacle is the dependence among the factors
in $\prod_{a\in V(R)}D_{\phi(a)}^{r_a}$. Instead of working directly with
$D_{\phi(a)}^{r_a}$, it is more convenient to work with their falling
factorials, which essentially allows us to bypass this dependence. 

For each $i$ and $r$, denote the falling factorial of $D_i$ of power $r$ by  
$$
(D_i)_r
:=
D_i(D_i-1)\cdots(D_i-r+1).
$$
Since $(D_i)_r$ counts the number of ordered $r$-tuples of distinct neighbors of node $i$, we have   
\[
(D_i)_r
=
\sum_{\substack{x_1,\ldots,x_r\neq i\\
x_1,\ldots,x_r\ \mathrm{distinct}}}
A_{ix_1}\cdots A_{ix_r},
\]
where the sum is over all distinct indices $x_1,\ldots,x_r\in[n]\setminus\{i\}$.
By a slight abuse of terminology, we refer to $x_1,\ldots,x_r$ as neighbors of
$i$. For distinct vertices $i_1,\ldots,i_m$ and nonnegative integers
$r_1,\ldots,r_m$, we have
\begin{equation}
\label{eq:falling degree product}  \prod_{\ell=1}^m ( D_{i_\ell})_{r_\ell}
=
\sum_{\substack{x_{\ell,s}\neq i_\ell\\
x_{\ell,1},\ldots,x_{\ell,r_\ell}\ \mathrm{distinct} \\  
1\leq \ell\leq m,\ 1\leq s\leq r_\ell}}
\prod_{\ell=1}^m
\prod_{s=1}^{r_\ell}
A_{i_\ell x_{\ell,s}},  
\end{equation}
where the sum is over neighbor assignments
\[
x_{\ell,s}\neq i_\ell,
\qquad
1\leq \ell\leq m,
\qquad
1\leq s\leq r_\ell,
\]
with the restriction that $x_{\ell,1},\ldots,x_{\ell,r_\ell}$ are distinct for each \(\ell\). This is the key formula for studying $\mu_R(\widehat{P})$.

\subsection{Mixed Falling Factorials of Node Degrees}
In this section, we fix distinct vertices $i_1,\ldots,i_m$, nonnegative
integers $r_1,\ldots,r_m$, and an integer $h$. The first lemma shows that
$(D_{i_\ell})_{r_\ell}$, $1\le \ell\le m$, are essentially independent for our
purposes.  

\begin{lemma}[Mixed falling factorial]
\label{lemma:mixed falling factorial bound}
Assume that the network is generated from the Chung-Lu model. Then
\begin{equation*}
        \mathbb E\left[
\prod_{\ell=1}^m
(D_{i_\ell})_{r_\ell}
\right]=
\prod_{\ell=1}^m
(\E D_{i_\ell})^{r_\ell}
\left[
1
-
\sum_{\ell=1}^m
\binom{r_\ell}{2}
d_{i_\ell}^{-2}\sum_{x\neq i_\ell} P_{i_\ell x}^2
+
O\left(\frac{1}{n^2p}\right)
\right],
\end{equation*}
with the convention that $\binom{r_\ell}{2}=0$ if $r_\ell<2$.
\end{lemma}

\begin{proof}[Proof of Lemma~\ref{lemma:mixed falling factorial bound}]
From \eqref{eq:falling degree product}, we have 
\[
\E\prod_{\ell=1}^m ( D_{i_\ell})_{r_\ell}
=
\sum_{\substack{x_{\ell,s}\neq i_\ell\\
x_{\ell,1},\ldots,x_{\ell,r_\ell}\ \mathrm{distinct} \\  
1\leq \ell\leq m,\ 1\leq s\leq r_\ell}}
\E\prod_{\ell=1}^m
\prod_{s=1}^{r_\ell}
A_{i_\ell x_{\ell,s}}.
\]
Since $\E D_i=\sum_{x\neq i}P_{ix}$, 
we have a similar expansion
\[
\prod_{\ell=1}^m (\E D_{i_\ell})^{r_\ell}
=
\sum_{\substack{x_{\ell,s}\neq i_\ell\\
1\leq \ell\leq m,\ 1\leq s\leq r_\ell}}
\prod_{\ell=1}^m
\prod_{s=1}^{r_\ell}
P_{i_\ell x_{\ell,s}}.
\]
This is simply the expansion of a product of unrestricted sums, which includes assignments
where the same node \(i_\ell\) chooses the same neighbor more than once.
By contrast, the falling factorial \((D_{i_\ell})_{r_\ell}\) only allows distinct
neighbors for node \(i_\ell\). 

We now compare $\E\prod_{\ell=1}^m ( D_{i_\ell})_{r_\ell}$ with 
$\prod_{\ell=1}^m (\E D_{i_\ell})^{r_\ell}$.
The difference between these two quantities is due to two factors:
(i) the summation is taken over different sets of indices, and
(ii) in general,
\[
\E\left[
\prod_{\ell=1}^m
\prod_{s=1}^{r_\ell}
A_{i_\ell x_{\ell,s}}
\right]
\neq
\prod_{\ell=1}^m
\prod_{s=1}^{r_\ell}
\E A_{i_\ell x_{\ell,s}}
=
\prod_{\ell=1}^m
\prod_{s=1}^{r_\ell}
P_{i_\ell x_{\ell,s}},
\]
because some factors $A_{i_\ell x_{\ell,s}}$ may appear more than once in $\prod_{\ell=1}^m
\prod_{s=1}^{r_\ell}
A_{i_\ell x_{\ell,s}}$.
It remains to show that the issues caused by (i) and (ii) are negligible. To this end, we denote
\[
\Psi :=
\sum_{\substack{x_{\ell,s}\neq i_\ell\\
x_{\ell,1},\ldots,x_{\ell,r_\ell}\ \mathrm{distinct}, \  
1\leq \ell\leq m,\ 1\leq s\leq r_\ell}}
\prod_{\ell=1}^m
\prod_{s=1}^{r_\ell}
P_{i_\ell x_{\ell,s}}.
\]
We will show that the effect of (ii) is negligible, so
\(
\Psi \approx \E\prod_{\ell=1}^m (D_{i_\ell})_{r_\ell},
\)
and that the effect of (i) is also negligible, so
\(
\Psi \approx \prod_{\ell=1}^m (\E D_{i_\ell})^{r_\ell}
\). 
The triangle inequality then implies the claim of the lemma.

\paragraph{Reciprocal-edge collisions.} We first address the difference caused by (ii). 
If all $\{i_\ell,x_{\ell,s}\}$, $1\le \ell\le m$ and $1\le s\le r_\ell$, are distinct sets then $\E\left[
\prod_{\ell=1}^m
\prod_{s=1}^{r_\ell}
A_{i_\ell x_{\ell,s}}
\right]
=
\prod_{\ell=1}^m
\prod_{s=1}^{r_\ell}
\E A_{i_\ell x_{\ell,s}}
$   
due to the independence of undirected edges. Since all indices $x_{\ell,s}$, $1\le s\le r_\ell$, are distinct for each fixed $\ell$, if
\[
\{i_\ell,x_{\ell,s}\}
=
\{i_k,x_{k,t}\}
\]
then
\[
x_{\ell,s}=i_k,
\qquad
x_{k,t}=i_\ell.
\]
Thus the only repeated-undirected-edge pattern is a reciprocal choice:
node \(i_\ell\) chooses \(i_k\), and node \(i_k\) chooses \(i_\ell\). In that case
the product contains $A_{i_\ell i_k}A_{i_k i_\ell}
=
A_{i_\ell i_k}^2
=
A_{i_\ell i_k}$.

Fix two pairs of indices \((\ell,s)\) and \((k,t)\), with \(\ell\neq k\). Let
\(\mathcal C_{\ell s,kt}\) denote the total contribution to $\E\prod_{\ell=1}^m ( D_{i_\ell})_{r_\ell}$ of all assignments satisfying $x_{\ell,s}=i_k$ and $x_{k,t}=i_\ell$. 
We claim that
\begin{equation}\label{eq:Clskt}
\mathcal C_{\ell s,kt}
\leq
C
P_{i_\ell i_k}
(\E D_{i_\ell})^{r_\ell-1}
(\E D_{i_k})^{r_k-1}
\prod_{a\neq \ell,k}
(\E D_{i_a})^{r_a}.    
\end{equation}
Indeed, after forcing the reciprocal pair, the repeated edge contributes one
factor $A_{i_\ell i_k}$, and
the remaining indices are bounded by ordinary degree powers:
\begin{eqnarray*}
\mathcal C_{\ell s,kt}
&=&\E\Bigg[A_{i_\ell i_k}\sum_{\substack{x_{\ell,s'}\neq i_\ell, \ x_{k,t'}\neq i_k\\
x_{\ell,1},\ldots,x_{\ell,r_\ell}\ \mathrm{distinct} \\
x_{k,1},\ldots,x_{k,r_k}\ \mathrm{distinct}}}
\prod_{s'\neq s}
A_{i_\ell x_{\ell,s'}}
\prod_{t'\neq t}
A_{i_k x_{k,t'}}
\prod_{a\neq \ell,k}
D_{i_a}^{r_a}
\Bigg],
\end{eqnarray*}
where the sum is over all distinct $x_{\ell,s'}\in [n]\setminus\{x_{\ell,s}=i_k\}$ with $ s'\in[m]\setminus \{s\}$ and distinct $x_{k,t'}\in [n]\setminus\{x_{k,t}=i_\ell\}$ with $t'\in[m]\setminus\{t\}$. Without the constraints that $x_{\ell,s'}\neq i_k$ and $x_{k,t'}\neq i_\ell$, we have
$$
A_{i_\ell i_k}\sum_{\substack{x_{\ell,s'}\neq i_\ell, \ x_{k,t'}\neq i_k\\
\mathrm{distinct} \ x_{\ell,s'}, \ s'\in[m]\setminus\{s\} \\
\mathrm{distinct} \ x_{\ell,s'}, \ t'\in[m]\setminus\{t\}}}
\prod_{s'\neq t}
A_{i_\ell x_{\ell,s'}}
\prod_{t'\neq t}
A_{i_k x_{k,t'}} = A_{i_\ell i_k}(D_{i_\ell})_{(r_\ell-1)}
(D_{i_k})_{(r_k-1)}.
$$
Due to the constraints, the right-hand side serves as an upper bound. Since $(D_i)_r< D_i^r$, we have 
\begin{eqnarray*}
\mathcal C_{\ell s,kt}
&\leq&
\mathbb E\left[
A_{i_\ell i_k}
D_{i_\ell}^{r_\ell-1}
D_{i_k}^{r_k-1}
\prod_{a\neq \ell,k}
D_{i_a}^{r_a}
\right].
\end{eqnarray*}
Let \(e_0=\{i_\ell,i_k\}\) and denote by \(D_j^{(-e_0)}\) the degree of vertex \(j\)
with the edge \(e_0\) removed. Conditional on \(A_{i_\ell i_k}=1\), the remaining
edges are independent and
\[
D_{i_\ell}=D_{i_\ell}^{(-e_0)}+1,
\qquad
D_{i_k}=D_{i_k}^{(-e_0)}+1.
\]
Therefore, the right-hand side of the bound above equals
\[
P_{i_\ell i_k}
\,
\mathbb E\left[
(D_{i_\ell}^{(-e_0)}+1)^{r_\ell-1}
(D_{i_k}^{(-e_0)}+1)^{r_k-1}
\prod_{a\neq \ell,k}
(D_{i_a}^{(-e_0)})^{r_a}
\right].
\]
By Lemma~\ref{lemma:degree moment bound}, for each integer \(r\), we have
\[
\mathbb E\left[(D_i^{(-e_0)}+1)^r\right]
\leq
C_r (\E D_i)^r
\]
for some constant $C_r$ only depending on $r$.
Applying the generalized Holder's inequality, that is,  $\E|\prod_{i=1}^a X_i|\le \prod_{i=1}^a(\E|X_i|^a)^{1/a}$, we obtain
\[
\mathbb E\left[
(D_{i_\ell}^{(-e_0)}+1)^{r_\ell-1}
(D_{i_k}^{(-e_0)}+1)^{r_k-1}
\prod_{a\neq \ell,k}
(D_{i_a}^{(-e_0)})^{r_a}
\right]
\leq
C
(\E D_{i_\ell})^{r_\ell-1}
(\E D_{i_k})^{r_k-1}
\prod_{a\neq \ell,k}
(\E D_{i_a})^{r_a},
\]
and the claim \eqref{eq:Clskt} is proved.

We now compare the bound in \eqref{eq:Clskt} with the leading term $\prod_{a=1}^m (\E D_{i_a})^{r_a}$. 
Since $P_{i_\ell i_k}=p\theta_{i_\ell}\theta_{i_k}\leq Cp
$ and $\E D_i=p\theta_i(n-\theta_i)\asymp np$, we have
\[
\frac{
P_{i_\ell i_k}
(\E D_{i_\ell})^{r_\ell-1}
(\E D_{i_k})^{r_k-1}
\prod_{a\neq \ell,k}(\E D_{i_a})^{r_a}
}{
\prod_{a=1}^m (\E D_{i_a})^{r_a}
}
=
\frac{P_{i_\ell i_k}}{\E D_{i_\ell} \E D_{i_k}}
\leq
\frac{C}{n^2p}.
\]
There are only $\sum_{\ell<k} r_\ell r_k=O(1)$
possible pairs of indices $(\ell,s)$ and $(k,t)$ for which $x_{\ell,s}=i_k$ and $x_{k,t}=i_\ell$, because \(r_\ell\) and $m$ are fixed. Therefore the total
contribution of all reciprocal-edge collisions to $\E\prod_{\ell=1}^m (D_{i_\ell})_{r_\ell}$ is bounded by
\[
\frac{C}{n^2p}
\prod_{\ell=1}^m(\E D_{i_a})^{r_\ell}.
\]

The same argument can be used to control the total contribution of all
reciprocal-edge collisions to $\Psi$. This case is slightly simpler because it
does not involve taking expectations. Moreover, unlike the previous case, where
the repeated edge contributes a factor $A_{i_\ell i_k}$ because
$A_{i_\ell i_k}^2=A_{i_\ell i_k}$, here the repeated edge contributes a factor
$P_{i_\ell i_k}^2$ to $\Psi$. Therefore, due to the extra factor of $p$ in
$P_{i_\ell i_k}^2$, the total contribution of all reciprocal-edge collisions to
$\Psi$ is bounded by
\[
\frac{C}{n^2}
\prod_{\ell=1}^m(\E D_{i_\ell})^{r_\ell}.
\]
Together, these bounds give
\begin{equation}
\label{eq:comparing expected product of falling degree}
    \E\prod_{\ell=1}^m (D_{i_\ell})_{r_\ell}
=
\Psi
+
O\left(\frac{1}{n^2p}\right)
\prod_{\ell=1}^m(\E D_{i_\ell})^{r_\ell}.
\end{equation}

\paragraph{Summation over different sets of indices.}
We now address the difference caused by (i), showing that \(
\Psi \approx \prod_{\ell=1}^m (\E D_{i_\ell})^{r_\ell}
\) as the total contribution of all terms with non-distinct indices $x_{\ell,1},\ldots,x_{\ell,r_\ell}$ is small.

For each \(r\ge 0\), define the ordered distinct-neighbor sum
\[
T_i(r)
:=
\sum_{\substack{x_1,\ldots,x_r\neq i\\
x_1,\ldots,x_r\ \mathrm{distinct}}}
\prod_{s=1}^r P_{ix_s}.
\]
With this notation, we have $\Psi
=
\prod_{\ell=1}^m
T_{i_\ell}(r_\ell)
$. To analyze $\Psi$, 
we compare $T_i(r)$ with its unrestricted counterpart
\[
d_i^r = \Bigg(\sum_{x\neq i}P_{ix}\Bigg)^r = \sum_{x_1,\ldots,x_r\neq i}
\prod_{s=1}^r P_{ix_s}.
\]
Thus \(T_i(r)\) is obtained from \(d_i^r\) by removing all tuples containing
repeated coordinates. Using inclusion-exclusion, we have
\[
T_i(r)
=
d_i^r
-
\mathbf{1}_{r\ge 2}\sum_{1\le a<b\le r} \ 
\sum_{\substack{x_1,\ldots,x_r\neq i\\
x_a=x_b}} \ 
\prod_{s=1}^r P_{ix_s}
+
\text{terms with at least two repeated constraints}.
\]
For each $r\ge 2$ and each pair $a<b$,
\[
\sum_{\substack{x_1,\ldots,x_r\neq i\\
x_a=x_b}} \ 
\prod_{s=1}^r P_{ix_s}
=
\sum_{x\neq i}P_{ix}^2
\Bigg(\sum_{y\neq i}w_y\Bigg)^{r-2}
=
d_i^{r-2}\sum_{x\neq i}P_{ix}^2.
\]
Since there are $\binom{r}{2}$ possible pairs $a<b$, the first correction term
in the formula for $T_i(r)$ is
\[
\sum_{1\le a<b\le r} \ 
\sum_{\substack{x_1,\ldots,x_r\neq i\\
x_a=x_b}} \ 
\prod_{s=1}^r P_{ix_s}
 = \binom r2 d_i^{r-2}\sum_{x\neq i}P_{ix}^2.
\]
The remaining terms come from intersections of at least two repeated-pair
events. There are two basic types. First, two repeated pairs may overlap; for
example, $x_a=x_b$ and $x_a=x_c$, resulting in a pattern with three repeated
coordinates. Since the contribution of the remaining coordinates is at most
$d_i^{r-3}\mathbf{1}_{\{r\ge 3\}}$, the total contribution of this type is at
most
\[
    \mathbf{1}_{\{r\ge 3\}}d_i^{r-3}
    \sum_{x\neq i}P_{ix}^3 = \mathbf{1}_{\{r\ge 3\}}O\left(\frac{d_i^{r}}{n^2}\right).
\]
Second, there may be two distinct repeated pairs; for example,
$x_a=x_b\neq x_c=x_d$. The contribution of this type is at most
\[
    \mathbf{1}_{\{r\ge 4\}}d_i^{r-4}
    \Bigg(\sum_{x\neq i}P_{ix}^2\Bigg)^2 =     \mathbf{1}_{\{r\ge 4\}}O\left(\frac{d_i^r}{n^2}\right).
\]
Therefore, we have the relative expansion
\begin{equation*}
T_i(r)
=
d_i^r
\left[
1
-
\mathbf{1}_{r\ge 2}\binom r2d_i^{-2}\sum_{x\neq i}P_{ix}^2
+
O\left(n^{-2}\right)
\right].
\end{equation*}
Applying this to each fixed root \(i_\ell\) and expanding the product, we obtain
\begin{equation}
\label{eq:comparing Psi and product of expected degree}    
\Psi =
\prod_{\ell=1}^m d_{i_\ell}^{r_\ell}
\left[
1
-
\sum_{\ell=1}^m
\mathbf{1}_{r_\ell\ge 2}\binom{r_\ell}{2}
d_{i_\ell}^{-2}\sum_{x\neq i_\ell} P_{i_\ell x}^2
+
O(n^{-2})
\right].
\end{equation}
Together with \eqref{eq:comparing expected product of falling degree}
, this implies the claim of the lemma. 
\end{proof}


Next, we show that multiplying $\prod_{\ell=1}^m
(D_{i_\ell})_{r_\ell}$ by a fixed power of the total degree $S^h$ changes the mixed falling factorial only through the deterministic factor $(\E S)^h$.
This is done in two steps. First, we show that multiplying
$\prod_{\ell=1}^m (D_{i_\ell})_{r_\ell}$ by a fixed power of the centered error
$S-\E S$ yields a lower-order term. We then use this result to prove the claim.

Denote $E = A-P$. For an undirected edge \(e=\{i,j\}\), write $A_e=A_{ij}$, $P_e=P_{ij}$, and $E_e=A_e-P_e$. 
Then
\[
S - \E S = 2\sum_e E_e,
\]
where the sum is over unordered edges \(e=\{i,j\}\), \(i<j\). We have the following bound.

\begin{lemma}[Mixed falling factorial and error moment]
\label{lemma:Z^kX bound}
For every fixed integer \(k\geq 1\),
\[
\left|
\mathbb E\left[
(S-\E S)^k \prod_{\ell=1}^m
(D_{i_\ell})_{r_\ell}
\right]
\right|
\leq
C_k
(\E S)^{\lfloor k/2\rfloor}
\prod_{\ell=1}^m
(\E D_{i_\ell})^{r_\ell}.
\]    
\end{lemma}

\begin{proof}[Proof of Lemma~\ref{lemma:Z^kX bound}]
Recall from \eqref{eq:falling degree product} that 
\[
\prod_{\ell=1}^m ( D_{i_\ell})_{r_\ell}
=
\sum_{\substack{x_{\ell,s}\neq i_\ell\\
x_{\ell,1},\ldots,x_{\ell,r_\ell}\ \mathrm{distinct} \\  
1\leq \ell\leq m,\ 1\leq s\leq r_\ell}}
\prod_{\ell=1}^m
\prod_{s=1}^{r_\ell}
A_{i_\ell x_{\ell,s}},
\]
where the sum is over all neighbor assignments $\eta$ satisfying the condition in the formula.   
For each such \(\eta\), let \(H(\eta)\) be the set of
distinct undirected edges $\{i_\ell,x_{\ell,s}\}$ selected by $\eta$. Since $A_e^t=A_e$ for any positive integer $t$, we have
$$
\prod_{\ell=1}^m ( D_{i_\ell})_{r_\ell}
=
\sum_{\eta}
\prod_{e\in H(\eta)}
A_e, \qquad
\E\prod_{\ell=1}^m ( D_{i_\ell})_{r_\ell}
=
\sum_{\eta}
\prod_{e\in H(\eta)}
P_e.
$$
We first prove the following bound for a fixed finite edge set \(H\):
\begin{equation}
\label{eq:Z^kX fix H}
\left|
\mathbb E\left[
(S-\E S)^k\left(\prod_{e\in H}A_e\right)
\right]
\right|
\leq
C_k
(\E S)^{\lfloor k/2\rfloor}
\prod_{e\in H}P_e.    
\end{equation}
Expanding \((S-\E S)^k\), we get
\[
(S-\E S)^k
=
2^k
\sum_{e_1,\ldots,e_k}
E_{e_1}\cdots E_{e_k}.
\]
Therefore
\[
\mathbb E\left[
(S-\E S)^k\left(\prod_{e\in H}A_e\right)
\right]
=
2^k
\sum_{e_1,\ldots,e_k}
\mathbb E\left[
E_{e_1}\cdots E_{e_k}
\left(\prod_{e\in H}A_e\right)
\right].
\]
For each tuple $(e_1,\ldots,e_k)$, if an edge $e\notin H$ appears exactly once
among $e_1,\ldots,e_k$, then the corresponding expectation is zero, because the
entries of $E$ above the diagonal are independent and mean-zero. Thus, for a
nonzero contribution, every edge outside $H$ that appears in the tuple
$(e_1,\ldots,e_k)$ must appear at least twice. For such a tuple,
let \(M=M(e_1,\ldots,e_k)\) be the set of distinct edges outside \(H\) that appear among
\(e_1,\ldots,e_k\) and $s(e)\ge 0$ be the multiplicity of $e$ in the tuple. Since every edge in \(M\) appears at least twice, $|M|\leq \lfloor k/2\rfloor$ and
$s(e)\ge 2$ for each $e\in M$. The independence of $\{E_e:e\in M\cup H\}$ implies  
$$
\E\left[
E_{e_1}\cdots E_{e_k}
\left(\prod_{e\in H}A_e\right)
\right] = \prod_{e\in M} \E \left[E_e^{s(e)}\right]
\prod_{e\in H} \E \left[E_e^{s(e)}A_e\right]\le
\prod_{e\in M}P_e\cdot\prod_{e\in H}P_e.
$$
For the last inequality above, we use the following elementary bounds: $\mathbb E[|E_e|^sA_e]\leq\mathbb E[A_e]=P_e$ if $s\ge 0$ and $\mathbb E[|E_e|^s]\leq \mathbb E[|E_e|^2] = P_e(1-P_e)\le P_e$ 
if \(s\geq2\), because $|E_e|\le 1$. Since $|M|\leq \lfloor k/2\rfloor$, this gives
\begin{eqnarray*}
\sum_{e_1,\ldots,e_k}
\mathbb E\left[
E_{e_1}\cdots E_{e_k}
\left(\prod_{e\in H}A_e\right)
\right]
&\le& \prod_{e\in H}P_e\cdot\sum_{e_1,\ldots,e_k}\prod_{e\in M(e_1,\ldots,e_k)}P_e \\
&\le& \prod_{e\in H}P_e \cdot\sum_{t=1}^{\lfloor k/2\rfloor} \sum_{e_1,\ldots,e_t}\prod_{i=1}^t P_{e_i}\\
&\le& C_k \prod_{e\in H}P_e \cdot(\E S)^{\lfloor k/2\rfloor},
\end{eqnarray*}
for some $C_k$ only depending on $k$, which proves \eqref{eq:Z^kX fix H}. It follows from \eqref{eq:Z^kX fix H}, \eqref{eq:comparing expected product of falling degree}, and \eqref{eq:comparing Psi and product of expected degree} that
\begin{eqnarray*}
\E \left[(S-\E S)^k \prod_{\ell=1}^m ( D_{i_\ell})_{r_\ell}\right]
\le 
C_k(\E S)^{\lfloor k/2\rfloor}\sum_{\eta}
\prod_{e\in H(\eta)}
P_e\le C_k'
(\E S)^{\lfloor k/2\rfloor} \prod_{\ell=1}^m
(\E D_{i_\ell})^{r_\ell}.
\end{eqnarray*}
The proof is complete. 
\end{proof}

Using Lemma~\ref{lemma:Z^kX bound}, we obtain the following result. 

\begin{lemma}[Mixed falling factorial and edge count]
\label{lemma:mixed falling factorial and edge count}
For every integer \(h\in\mathbb{Z}\), we have 
\[
\mathbb E\left[
S^h \prod_{\ell=1}^m
(D_{i_\ell})_{r_\ell}
\right]
=
(\E S)^h
\prod_{\ell=1}^m
(\E D_{i_\ell})^{r_\ell}
\left[
1
-
\sum_{\ell=1}^m
\binom{r_\ell}{2}
d_{i_\ell}^{-2}\sum_{x\neq i_\ell} P_{i_\ell x}^2
+
O\left(\frac{1}{n^2p}\right)
\right],
\]    
with the convention that $\binom{r_\ell}{2}=0$ when $r_\ell<2$.
\end{lemma}

\begin{proof}[Proof of Lemma~\ref{lemma:mixed falling factorial and edge count}]

We first suppose \(h\geq0\). Expanding $S^h = (\E S + S-\E S)^h$ gives
\begin{eqnarray*}
\mathbb E\left[
S^h\prod_{\ell=1}^m ( D_{i_\ell})_{r_\ell}
\right]
&=&
(\E S)^h\mathbb E \left[\prod_{\ell=1}^m ( D_{i_\ell})_{r_\ell}\right]
+
\sum_{k=1}^h
\binom hk
(\E S)^{h-k}
\mathbb E\left[
(S-\E S)^k \prod_{\ell=1}^m
(D_{i_\ell})_{r_\ell}
\right]
\end{eqnarray*}
Using Lemma~\ref{lemma:Z^kX bound}, the sum on the right-hand side is bounded by 
\begin{equation*}
    \sum_{k=1}^h C_k
(\E S)^{h-k+\lfloor k/2\rfloor}
\prod_{\ell=1}^m
(\E D_{i_\ell})^{r_\ell}
\le C_h (\E S)^{h-1}\prod_{\ell=1}^m
(\E D_{i_\ell})^{r_\ell}.
\end{equation*}
Thus, by Lemma~\ref{lemma:mixed falling factorial bound} and the fact that $\E S \asymp n^2p$, we have
\begin{eqnarray*}
\mathbb E\left[
S^h\prod_{\ell=1}^m ( D_{i_\ell})_{r_\ell}
\right]
=
(\E S)^h\prod_{\ell=1}^m ( \E D_{i_\ell})^{r_\ell}
\left[
1
-
\sum_{\ell=1}^m
\binom{r_\ell}{2}
d_{i_\ell}^{-2}\sum_{x\neq i_\ell} P_{i_\ell x}^2
+
O\left(\frac{1}{n^2p}\right)
\right].
\end{eqnarray*}

Now suppose \(h<0\). Consider the event
\[
\mathcal E
=
\{|S-\E S|\leq \E S/2\}.
\]
Since \(S/2\) is a sum of independent Bernoulli variables with mean \(\E S/2\), by
Bernstein's inequality, we have
$\mathbb P(\mathcal E)
\ge 1 - \exp(-c n^2p)$. 
On the event \(\mathcal E\), we have \(|S-\E S|/\E S\leq1/2\), so using the Taylor expansion of the function
\(x\mapsto(1+x)^h\), we get
\[
S^h
=
(\E S)^h
\left(1+\frac{S-\E S}{\E S}\right)^h
=
(\E S)^h
\left[
1
+
h\frac{S-\E S}{\E S}
+
O_h\left(\frac{(S-\E S)^2}{(\E S)^2}\right)
\right].
\]
Therefore,
\begin{eqnarray*}
\mathbb E\left[
S^h\prod_{\ell=1}^m ( D_{i_\ell})_{r_\ell}\mathbf{1}_{\mathcal{E}}
\right] 
&=& (\E S)^h \E\left[
\prod_{\ell=1}^m ( D_{i_\ell})_{r_\ell}\mathbf{1}_{\mathcal{E}}
\right]
+ h(\E S)^{h-1} \E\left[(S-\E S)
\prod_{\ell=1}^m ( D_{i_\ell})_{r_\ell}\mathbf{1}_{\mathcal{E}}
\right] \\
&+& O_h\left((\E S)^{h-2}\E\left[(S-\E S)^2
\prod_{\ell=1}^m ( D_{i_\ell})_{r_\ell}\mathbf{1}_{\mathcal{E}}
\right]\right).
\end{eqnarray*}
We now bound the three terms on the right-hand side above. Since $\mathcal{E}$ occurs with probability at least $1-\exp(-cn^2p)$ and all terms in the expectation sign are polynomial in $n$, it follows that removing $\mathbf{1}_{\mathcal{E}}$ only introduces an error of order $\exp(-cn^2p)$. Thus,
\[
E\left[
\prod_{\ell=1}^m ( D_{i_\ell})_{r_\ell}\mathbf{1}_{\mathcal{E}}
\right] = E\left[
\prod_{\ell=1}^m ( D_{i_\ell})_{r_\ell}
\right]  + \exp(-cn^2p).
\]
By Lemma~\ref{lemma:Z^kX bound}, 
\begin{eqnarray*}
\left|\E\Big[(S-\E S)
\prod_{\ell=1}^m ( D_{i_\ell})_{r_\ell}\mathbf{1}_{\mathcal{E}}\Big]\right| 
&=&
\left|\E\Big[(S-\E S)
\prod_{\ell=1}^m ( D_{i_\ell})_{r_\ell}\Big] + \exp(-cn^2p)\right| \\
&\le& 
C\E\left[\prod_{\ell=1}^m ( D_{i_\ell})_{r_\ell}\right] + \exp(-cn^2p).    
\end{eqnarray*}
Similarly, by Lemma~\ref{lemma:Z^kX bound}, 
\begin{eqnarray*}
\left|\E\Big[(S-\E S)^2
\prod_{\ell=1}^m ( D_{i_\ell})_{r_\ell}\mathbf{1}_{\mathcal{E}}\Big]\right| 
\le 
C(\E S)\E\left[\prod_{\ell=1}^m ( D_{i_\ell})_{r_\ell}\right] + \exp(-cn^2p).    
\end{eqnarray*}
Putting these together, we obtain
\begin{eqnarray*}
\mathbb E\left[
S^h\prod_{\ell=1}^m ( D_{i_\ell})_{r_\ell}
\right] 
&=& \mathbb E\left[S^h\prod_{\ell=1}^m ( D_{i_\ell})_{r_\ell}
\right] + O(\exp(-cn^2p))\\
&=&
(\E S)^h E\left[
\prod_{\ell=1}^m ( D_{i_\ell})_{r_\ell}
\right]\left\{1+O\left(\frac{1}{n^2p}\right)\right\}.
\end{eqnarray*}
The claim of the lemma then follows from Lemma~\ref{lemma:mixed falling factorial bound}.
\end{proof}

\begin{lemma}[Degree moment bound]
\label{lemma:degree moment bound}
    Let $A$ be the adjacency matrix of a random network generated from the Chung--Lu model. Then for each nonnegative integer $s$, we have 
    $$\mathbb E D_i^s \le C_s \left(\E D_i\right)^s$$ for some constant $C_s$ only depending on $s$.
\end{lemma}
\begin{proof}[Proof of Lemma~\ref{lemma:degree moment bound}]
Denote by $(D_i)_r = D_i(D_i-1)\cdots(D_i-r+1)$ the falling factorial. Since $(D_i)_r$ counts the ordered-$r$ tuples of distinct neighbors of $i$,
$$
(D_i)_r = \sum_{x_1,...,x_r \ \text{distinct}} A_{ix_1}\cdots A_{i x_r}.
$$
This implies
\begin{eqnarray*}
\E(D_i)_r &=& \sum_{x_1,...,x_r \ \text{distinct}} \E[A_{ix_1}\cdots A_{i x_r}] = \sum_{x_1,...,x_r \ \text{distinct}} P_{ix_1}\cdots P_{i x_r}\\
&\le& \sum_{x_1,...,x_r} P_{ix_1}\cdots P_{i x_r} = \left(\sum_x P_{ix}\right)^r = \left(\E D_i\right)^r.
\end{eqnarray*}
Now, we expand $D_i^s$ using the falling factorials and Stirling numbers of the second kind
\[
D_i^s=\sum_{r=0}^s \left\{ {s\atop r} \right\}(D_i)_r.
\]
By the bound for $\E(D_i)_r$ above,
this gives 
\[
\mathbb E D_i^s
\leq
\sum_{r=0}^s \left\{ {s\atop r} \right\}\left(\E D_i\right)^r
\leq
C_s \left(\E D_i\right)^s.
\]
The lemma is proved.
\end{proof}

\subsection{Mixed Moments of Node Degrees and Edge Count}

Equipped with Lemma~\ref{lemma:mixed falling factorial and edge count}, we are
now ready to compute the mixed moments of the node degrees and the edge count,
which are key quantities for understanding $\E[\mu_R(\widehat{P})]$ according
to \eqref{eq:mu Phat}.

\begin{lemma}[Mixed moments of  node degrees and edge count]
\label{lemma:mixed moment of node degrees and edge count}
For every integer \(h\in\mathbb{Z}\) and a vector $\mathbf{r}=(r_1,...,r_m)$ of positive integers with $s(\mathbf{r}):=\sum_{\ell=1}^m r_\ell$,  we have 
\[
\mathbb E\left[
S^h \prod_{\ell=1}^m
D_{i_\ell}^{r_\ell}
\right]
= (\E S)^h\prod_{\ell=1}^m
(\E D_{i_\ell})^{r_\ell}\left[1+O\left(\frac{1}{(np)^2}\right)\right] + \mathrm{Bias},
\]
where 
\begin{eqnarray*}
\mathrm{Bias} &=& 
(\E S)^h\prod_{k=1}^m
(\E D_{i_k})^{r_k}\sum_{\ell=1}^m
\binom{r_\ell}{2}(\E D_{i_\ell})^{-1}\left[1-
(\E D_{i_\ell})^{-1}\sum_{x\neq i_\ell} P_{i_\ell x}^2\right],
\end{eqnarray*}
with the convention that $\binom{r_\ell}{2}=0$ when $r_\ell<2$.
\end{lemma}

\begin{proof}[Proof of Lemma~\ref{lemma:mixed moment of node degrees and edge count}]
Expanding ordinary node degree powers into falling factorials gives
\[
D_i^r
=
\sum_{k=0}^r
\left\{ {r\atop k} \right\}
(D_i)_k,
\]
where \(\left\{ {r\atop k} \right\}\) is a Stirling number of the second kind. Let $\mathbf{r}=(r_1,\ldots,r_m)$ and $\mathbf{k}=(k_1,\ldots,k_m)$, and
write $\mathbf{k}\le \mathbf{r}$ to denote component-wise inequalities. Expanding $D_{i_\ell}^{r_\ell}$ for each $\ell\in[m]$ and multiplying them together give
\[
\mathbb E\left[
S^{h}
\prod_{\ell=1}^mD_{i_\ell}^{r_\ell}
\right]
=
\sum_{\mathbf k\leq \mathbf r}
C_{\mathbf r,\mathbf k}
\,
\mathbb 
E\left[
S^{h}
\prod_{\ell=1}^m
(D_{i_\ell})_{k_\ell}
\right], \qquad 
C_{\mathbf r,\mathbf k}
:=
\prod_{\ell=1}^m
\left\{ {r_\ell\atop k_\ell} \right\}.
\]
Applying the estimates for $\mathbb E\left[
S^{h}
\prod_{\ell=1}^m (D_{i_\ell})_{k_\ell}
\right]
$ from Lemma~\ref{lemma:mixed falling factorial and edge count}, we get
\[
\mathbb E\left[
S^{h}
\prod_{\ell=1}^m D_{i_\ell}^{r_\ell}
\right]
=
(\E S)^{h}\sum_{\mathbf k\leq\mathbf r}
C_{\mathbf r,\mathbf k}
\prod_{\ell=1}^m
(\E D_{i_\ell})^{k_\ell}
\left[
1
-
\sum_{\ell=1}^m
\binom{k_\ell}{2}
d_{i_\ell}^{-2}\sum_{x\neq i_\ell} P_{i_\ell x}^2
+
O\left(\frac{1}{n^2p}\right)
\right].
\]
The leading term in the sum corresponds to $\mathbf{k}=\mathbf{r}$ and is given by
$$
(\E S)^h\prod_{\ell=1}^m
(\E D_{i_\ell})^{r_\ell}\asymp n^{2h+s(\mathbf{r})}p^{h+s(\mathbf{r})}.
$$
The leading bias terms come from the second-order term in the case
$\mathbf{k}=\mathbf{r}$, namely
$$
-(\E S)^h\prod_{\ell=1}^m
(\E D_{i_\ell})^{r_\ell}\sum_{\ell=1}^m
\binom{r_\ell}{2}
d_{i_\ell}^{-2}\sum_{x\neq i_\ell} P_{i_\ell x}^2\asymp n^{2h+s(\mathbf{r})-1}p^{h+s(\mathbf{r})}, 
$$
together with the leading terms corresponding to those $\mathbf{k}$ satisfying
$s(\mathbf{k})=s(\mathbf{r})-1$. They occur when
exactly one exponent drops by one:
$k_\ell=r_\ell-1$ for some $\ell\in[m]$ and $k_k=r_k$ for all $k\neq \ell$.
Since
\[
\left\{ {r_\ell\atop r_\ell-1} \right\}
=
\binom{r_\ell}{2},
\]
the total contribution of the leading bias terms with $s(\mathbf{k})=s(\mathbf{r})-1$ is
\[
(\E S)^h\prod_{k=1}^m
(\E D_{i_k})^{r_k}\sum_{\ell=1}^m
\binom{r_\ell}{2}(\E D_{i_\ell})^{-1} \asymp n^{2h+s(\mathbf{r})-1}p^{h+s(\mathbf{r})-1}.
\]
The second-order terms corresponding to those $\mathbf{k}$ with
$s(\mathbf{k})=s(\mathbf{r})-1$, together with the terms satisfying
$s(\mathbf{k})\le s(\mathbf{r})-2$, are of relative order $O((np)^{-2})$
compared with the leading term.
\end{proof}

\subsection{Plug-in Estimators for Global Subgraph Means}
\label{sec:mu hat for global mean}
We start with the global counts. From \eqref{eq:mu Phat} we have
$$
\E \left[\mu_R(\widehat{P})\right] = C_R
\sum_{\phi:V(R)\hookrightarrow[n]}
\E\left[ S^{-e}\prod_{a\in V(R)} 
D_{\phi(a)}^{r_a}\right].
$$
We apply Lemma~\ref{lemma:mixed moment of node degrees and edge count} to calculate the sum on the right-hand side. To that end, we use the following notation:
\begin{eqnarray*}
d_i &=& \E D_i =  p\theta_i(n-\theta_i)=:p \eta_i(\theta),\\ 
s(d) &=&\sum_{i=1}^n d_i = p\Big(n^2-\sum_{i=1}^n\theta_i^2\Big)=:p M(\theta)=:n(n-1)p_{\theta}, \\ 
\sum_{j\neq i}P_{ij}^2 &=& p^2\theta_i^2\sum_{j\neq i}\theta_j^2=:p^2\theta_i^2\sigma_i(\theta).
\end{eqnarray*}
Recall $\mu_R(P)$ from \eqref{eq:mu P}.
With the notation above and Lemma~\ref{lemma:mixed moment of node degrees and edge count}, we obtain 
\begin{equation}
\label{eq:bias decomposition global}
\mathrm{Bias}(\mu_R(\widehat{P}))=p^{e}\left(\Psi_1-\Psi_2\right)+p^{e-1}\Psi_3+O(n^{v-2}p^{e-2}),    
\end{equation} 
where $\Psi_1,\Psi_2,\Psi_3$ are functions of $\theta=(\theta_1,...,\theta_n)$ given by
\begin{eqnarray*}
\Psi_1 &=& \frac{1}{|\mathrm{Aut}(R)|}
\sum_{\phi:V(R)\hookrightarrow[n]} \left[\left(\frac{n}{n-1}\right)^e M(\theta)^{-e}\prod_{a\in V(R)}
\eta_{\phi(a)}(\theta)^{r_{\phi(a)}}-\prod_{a\in V(R)}
\theta_{\phi(a)}^{r_{\phi(a)}}\right],\\
\Psi_2 &=& \frac{1}{|\mathrm{Aut}(R)|}
\sum_{\phi:V(R)\hookrightarrow[n]} \left(\frac{n}{n-1}\right)^e M(\theta)^{-e}
\sum_{a\in V(R)}
\binom{r_a}{2}\eta_{{\phi(a)}}^{r_a-2}(\theta)\theta_{\phi(a)}^2\sigma_{\phi(a)}(\theta)
\prod_{b\neq a}
\eta_{\phi(b)}^{r_b}(\theta), \\
\Psi_3 &=& \frac{1}{|\mathrm{Aut}(R)|}
\sum_{\phi:V(R)\hookrightarrow[n]} \left(\frac{n}{n-1}\right)^e M(\theta)^{-e}
\sum_{a\in V(R)}
\binom{r_a}{2}\eta_{\phi(a)}^{r_a-1}(\theta)
\prod_{b\neq a}
\eta_{\phi(b)}^{r_b}(\theta).
\end{eqnarray*}
The following lemma provides the derivative bounds of these terms, which are needed for the analysis of bias correction. 

\begin{lemma}[Derivative bounds]
\label{lemma:derivative bounds for Psi_k}
Assume that there exist constants $C,c>0$ such that $c\leq \theta_i\leq C$ for all $i$. 
Then uniformly over all $i,j\in [n]$ and $1\le k\le 3$, we have    
\begin{eqnarray*}
\Psi_k = O(n^{v-1}), \qquad
    \partial_i\Psi_k(\theta) = O(n^{v-2}), \qquad
    \partial_{ij}\Psi_k =  O(n^{v-2})\mathbf{1}_{\{i=j\}} + O(n^{v-3}). 
\end{eqnarray*}
\end{lemma}

\begin{proof}[Proof of Lemma~\ref{lemma:derivative bounds for Psi_k}]   
It follows from the assumption 
$c\leq \theta_i\leq C$ for all $i$ that
\[
\eta(\theta_i)\asymp n,
\qquad
M(\theta)\asymp n^2,
\qquad
\sigma_i(\theta)\asymp n.
\]
Also, $\partial_i \eta_i(\theta)=n-2\theta_i=O(n)$ and $\partial_{ii}\eta_i(\theta)=-2$. Thus, for any non-negative integer \(r\),
\[
\partial_i \eta_i(\theta)^r=O(n^r),
\qquad
\partial_{ii}\eta_i(\theta)^r=O(n^r).
\]
Furthermore,
\[
\partial_i M(\theta)^{-e}=O(n^{-2e-2}),
\qquad
\partial_{ij}M(\theta)^{-e}=O\left(n^{-2e-4}+\mathbf{1}_{\{i=j\}}n^{-2e-2}\right).
\]

\paragraph{Bounds for $\Psi_1$.}
For each embedding \(\phi\), write
\[
U_\phi(\theta)
:=
\left(\frac{n}{n-1}\right)^e M(\theta)^{-e}
\prod_{a\in V(R)}
\eta_{\phi(a)}(\theta)^{r_a},
\qquad
V_\phi(\theta)
:=
\prod_{a\in V(R)}
\theta_{\phi(a)}^{r_a}.
\]
Then $U_\phi(\theta)
=
V_\phi(\theta)L_\phi(\theta),$
where
\[
L_\phi(\theta)
=
\left(\frac{n}{n-1}\right)^e
\left(1-\frac1{n^2}\sum_{j=1}^n\theta_j^2\right)^{-e}
\prod_{a\in V(R)}
\left(1-\frac{\theta_{\phi(a)}}{n}\right)^{r_a}.
\]
Since \(R\) is fixed and \(\theta_i\asymp 1\),
\[
L_\phi(\theta)-1=O(n^{-1}).
\]
Moreover,
\[
\partial_i(L_\phi(\theta)-1)
=
O(n^{-1})\mathbf 1_{\{i\in\phi(V(R))\}}
+
O(n^{-2}),
\]
and for $i\neq j$,
\[
\partial_{ij}\{L_\phi(\theta)-1\}
=
O(n^{-2})\mathbf 1_{\{i,j\in\phi(V(R))\}}
+
O(n^{-3})
\left[
\mathbf 1_{\{i\in\phi(V(R))\}}
+
\mathbf 1_{\{j\in\phi(V(R))\}}
\right]
+
O(n^{-4}),
\]
with the analogous diagonal bound
\[
\partial_{ii}\{L_\phi(\theta)-1\}
=
O(n^{-2})\mathbf 1_{\{i\in\phi(V(R))\}}
+
O(n^{-2}).
\]
Since $U_\phi-V_\phi
=
V_\phi(L_\phi-1)$
and \(V_\phi=O(1)\), each embedding contributes \(O(n^{-1})\) to
\(\Psi_1\). Hence
\[
\Psi_1(\theta)=\frac{1}{|\mathrm{Aut}(R)|}
\sum_{\phi:V(R)\hookrightarrow[n]}[U_\theta - V_\theta]=O(n^{v-1}).
\]
For the first derivative,
\[
\partial_i (U_\phi-V_\phi)
=
\partial_i V_\phi(L_\phi-1)
+
V_\phi\partial_i(L_\phi-1),
\]
the local term $\partial_i V_\phi$ forces one vertex of the embedding to equal $i$, leaving $v-1$ free vertices, and carries the factor
$L_\phi-1=O(n^{-1})$. The global term $\partial_i(L_\phi-1)$ either forces one
index to be $i$ with a factor $O(n^{-1})$, or does not force any vertex but
contributes $O(n^{-2})$ per embedding. Therefore,
\[
\partial_i\Psi_1(\theta)=O(n^{v-2}).
\]
For second derivatives, if \(i\neq j\), the largest local contribution forces
two distinct embedding vertices and keeps the \(O(n^{-1})\) finite-\(n\) gain,
giving \(O(n^{v-3})\). Terms involving a global derivative are no larger. Thus
\[
\partial_{ij} \Psi_1(\theta)=O(n^{v-3}),
\qquad i\neq j.
\]
For the diagonal derivative, both derivatives may force only one vertex, so
\[
\partial_{ii} \Psi_1(\theta)=O(n^{v-2}).
\]

\paragraph{Bounds for $\Psi_3$.}
Each summand in $\Psi_3$ has size
\[
M(\theta)^{-e}
\eta_{\phi(a)}^{r_a-1}(\theta)
\prod_{b\neq a}
\eta_{\phi(b)}^{r_b}(\theta)
=
O(n^{-1}).
\]
There are \(O(n^v)\) embeddings, so
\[
\Psi_3(\theta)=O(n^{v-1}).
\]
For the first derivative, if the derivative hits a local factor
\(\eta_{\phi(b)}^{r_b}\), then one vertex of the embedding is forced and the
summand remains \(O(n^{-1})\). This gives \(O(n^{v-2})\). If the derivative
hits \(M^{-e}\), then the summand gains a relative factor \(O(n^{-2})\), giving
at most \(O(n^{v-3})\). Thus
\[
\partial_i \Psi_3(\theta)=O(n^{v-2}).
\]
For \(i\neq j\), two local derivatives force two distinct embedding vertices,
leaving \(v-2\) free vertices and keeping the \(O(n^{-1})\) summand size; hence
\[
\partial_{ij}\Psi_3(\theta)=O(n^{v-3}).
\]
For \(i=j\), the two derivatives may force only one vertex, and therefore
\[
\partial_{ii}\Psi_3(\theta)=O(n^{v-2}).
\]

\paragraph{Bounds for $\Psi_2$.}
For simplicity, denote 
$$K_i(\theta)
:=
\eta_i(\theta)^{-2}\theta_i^2\sigma_i(\theta)=\frac{\sigma_i(\theta)}{(n-\theta_i)^2}, \quad C_\phi(\theta)
:=
\left(\frac{n}{n-1}\right)^e M(\theta)^{-e}
\prod_{a\in V(R)}
\eta_{\phi(a)}^{r_a}(\theta).$$ 
A direct calculation shows that 
$$
K_i(\theta) = O(n^{-1}), \qquad C_\phi(\theta) = O(1).
$$
With this notation, 
\[
\Psi_2(\theta) = 
\frac{1}{|\operatorname{Aut}(R)|}
\sum_{\phi:V(R)\hookrightarrow[n]}
C_\phi(\theta)
\sum_{a\in V(R)}
\binom{r_a}{2}
K_{\phi(a)}(\theta).
\]
Since each embedding
again contributes \(O(n^{-1})\) to $\Psi_2$, we have
\[
\Psi_2(\theta)=O(n^{v-1}).
\]

We next record derivative bounds for \(K_i\). Using the formula for $K_i$, we see that uniformly in \(u\), we have
\[
\partial_u K_i(\theta)=O(n^{-2}).
\]
For second derivatives, if \(u=v=i\), then
\[
\partial_{ii}K_i(\theta)
=
6\sigma_i(\theta)(n-\theta_i)^{-4}
=
O(n^{-3}).
\]
If \(u=i\) and \(v\neq i\), then
\[
\partial_{iv}K_i(\theta)
=
4\theta_v(n-\theta_i)^{-3}
=
O(n^{-3}).
\]
If \(u=v\neq i\), then
\[
\partial_{uu}K_i(\theta)
=
2(n-\theta_i)^{-2}
=
O(n^{-2}).
\]
Finally, if \(u\neq v\) and neither index equals \(i\), then
\[
\partial_{uv}K_i(\theta)=0.
\]
In summary, we have
\[
\partial_{uv}K_i(\theta)
=
O(n^{-3})\mathbf 1_{\{u=v=i\}}
+
O(n^{-2})
\left[
\mathbf 1_{\{u=i,\ v\neq i\}}
+
\mathbf 1_{\{v=i,\ u\neq i\}}
+
\mathbf 1_{\{u=v\neq i\}}
\right].
\]

Next, we bound the derivatives of the weight factor
\(C_\phi\). Its first derivatives satisfy 
\[
\partial_u C_\phi(\theta)
=
O(1)\mathbf 1_{\{u\in \phi(V(R))\}}
+
O(n^{-2}).
\]
The first term comes from differentiating a local factor
\(\eta_{\phi(a)}^{r_a}(\theta)\). The second term comes from differentiating the global denominator \(M(\theta)^{-e}\).
Similarly, for \(u\neq v\),
\[
\partial_{uv}C_\phi(\theta)
=
O(1)\mathbf 1_{\{u,v\in \phi(V(R))\}}
+
O(n^{-2})
\left[
\mathbf 1_{\{u\in\phi(V(R))\}}
+
\mathbf 1_{\{v\in\phi(V(R))\}}
\right]
+
O(n^{-4}),
\]
and for the diagonal derivative,
\[
\partial_{uu}C_\phi(\theta)
=
O(1)\mathbf 1_{\{u\in\phi(V(R))\}}
+
O(n^{-2}).
\]

We now consider $\Psi_2$. 
For the first derivative,
\[
\partial_u\Psi_2(\theta)
=
\frac{1}{|\operatorname{Aut}(R)|}
\sum_{\phi}
\left[
\partial_u C_\phi(\theta)
\sum_{a\in V(R)}
\binom{r_a}{2}K_{\phi(a)}(\theta)
+
C_\phi(\theta)
\sum_{a\in V(R)}
\binom{r_a}{2}\partial_uK_{\phi(a)}(\theta)
\right].
\]
If \(\partial_u\) hits \(C_\phi\), then either \(u\in\phi(V(R))\), leaving
\(O(n^{v-1})\) embeddings, and \(K_{\phi(a)}=O(n^{-1})\), giving
\[
O(n^{v-1})O(n^{-1})=O(n^{v-2}),
\]
or the derivative hits only the global denominator, giving \(O(n^{-2})\) per
embedding and no forced vertex, contributing
\[
O(n^v)O(n^{-2})O(n^{-1})=O(n^{v-3}),
\]
which is smaller.
If \(\partial_u\) hits \(K_{\phi(a)}\), then
\(\partial_uK_{\phi(a)}=O(n^{-2})\), giving
\[
O(n^{v})O(n^{-2})=O(n^{v-2}).
\]
Therefore
\[
\partial_u\Psi_2(\theta)=O(n^{v-2}).
\]

For the off-diagonal second derivative, \(u\neq v\), we apply the product rule
to \(C_\phi K_{\phi(a)}\). There are three types of terms.
First, if both derivatives hit \(C_\phi\), then we get
$\partial_{uv}C_\phi(\theta)\,K_{\phi(a)}(\theta)$. 
The largest case is \(u,v\in\phi(V(R))\), leaving \(O(n^{v-2})\) embeddings.
Since \(K_{\phi(a)}=O(n^{-1})\), this contributes
\[
O(n^{v-2})O(n^{-1})=O(n^{v-3}).
\]
Second, if one derivative hits \(C_\phi\) and the other hits \(K_{\phi(a)}\), we
obtain 
$\partial_u C_\phi(\theta)\,\partial_vK_{\phi(a)}(\theta)$. 
If \(u\in\phi(V(R))\) then at least one vertex is forced and since
\(\partial_vK_{\phi(a)}=O(n^{-2})\), the total contribution is
\[
O(n^{v-1})O(n^{-2})=O(n^{v-3}).
\]
Other cases in which $u\in\phi(V(R))$ contribute lower-order terms. Third, if both derivatives hit \(K_{\phi(a)}\), then we get $C_\phi(\theta)\,\partial_{uv}K_{\phi(a)}(\theta)$. 
For \(u\neq v\), the only nonzero cases occur when one of \(u,v\) equals
\(\phi(a)\), in which case one vertex is forced and the derivative is
\(O(n^{-2})\) or smaller. This contributes
\[
O(n^{v-1})O(n^{-2})=O(n^{v-3}).
\]
If neither derivative equals \(\phi(a)\), then the mixed derivative is zero. 
Therefore
\[
\partial_{uv}\Psi_2(\theta)=O(n^{v-3}),
\qquad u\neq v.
\]

For the diagonal second derivative, \(\partial_{uu}\Psi_2\), two derivatives may
force only one embedding vertex. This is why the diagonal bound is larger. The
largest contributions are
$\partial_{uu}C_\phi(\theta)\,K_{\phi(a)}(\theta)$
with \(u\in\phi(V(R))\), giving
\[
O(n^{v-1})O(n^{-1})=O(n^{v-2}),
\]
and $C_\phi(\theta)\,\partial_{uu}K_{\phi(a)}(\theta)$, giving
\[
O(n^v)O(n^{-2})=O(n^{v-2}).
\]
Thus
\[
\partial_{uu}\Psi_2(\theta)=O(n^{v-2}).
\]
The proof is complete. 
\end{proof}

With the derivative bounds in Lemma~\ref{lemma:derivative bounds for Psi_k}, we
are ready to show that the functions $\Psi_k$ are stable under input
perturbations.

\begin{lemma}[Stability of $\Psi_k$]
\label{lemma:stability of Psi_k}
Let $\widehat{\theta}$ be the estimate of $\theta$ given by \eqref{eq: main parameter estimators}. Then
\begin{equation*}
\label{eq:H-alpha-stability-final}
\Psi_k(\widehat\theta)-\Psi_k(\theta)
=
O\left(n^{v-2}p^{-1}\right),
\qquad k=1,2,3.
\end{equation*}
\end{lemma}

\begin{proof}[Proof of Lemma~\ref{lemma:stability of Psi_k}]
Taylor expansion gives
\[
\Psi_k(\widehat\theta)-\Psi_k(\theta)
=
\sum_{i=1}^n
\partial_i\Psi_k(\theta)(\widehat\theta_i-\theta_i)
+
\frac12
\sum_{i,j=1}^n
\partial_{ij}\Psi_k(\bar\theta)
(\widehat\theta_i-\theta_i)(\widehat\theta_j-\theta_j),
\]
where \(\bar\theta\) lies between \(\theta\) and \(\widehat\theta\). For the linear terms, by Lemma~\ref{lemma:derivative bounds for Psi_k}, we have $\partial_i\Psi_k:=\partial_i\Psi_k(\theta)=O(n^{v-2})$. Regarding $\hat{\theta}$, recall from \eqref{eq: main parameter estimators} that uniformly in \(i\), 
\begin{equation}
\label{eq:theta hat decomposition} \widehat\theta_i-\theta_i
=
\left(
\frac{nD_i}{S}-\frac{nd_i}{s(d)}
\right)
+
\left(
\frac{nd_i}{s(d)}-\theta_i
\right) = \left(
\frac{nD_i}{S}-\frac{nd_i}{s(d)}
\right)
+ O(n^{-1}).   
\end{equation}
Note that
\[
\sum_{i=1}^n  
\left(
\frac{nD_i}{S}-\frac{nd_i}{s(d)}
\right)\partial_i\Psi_k
=
\frac{n}{S}
\sum_{i=1}^n (D_i-d_i)\partial_i\Psi_k
-
\frac{n}{S\,s(d)}
\left(\sum_{i=1}^n d_i\partial_i \Psi_k \right)
[S-s(d)],
\]
and
\[
\sum_{i=1}^n (D_i-d_i)\partial_i\Psi_k
=
\sum_{i<j}(\partial_i\Psi_k+\partial_j\Psi_k)(A_{ij}-P_{ij}).
\]
By Bernstein's inequality,
\[
\sum_{i=1}^n (D_i-d_i)\partial_i\Psi_k
=
O\left(n^{v-1}p^{1/2}\right), \qquad
S-s(d)
=
O(np^{1/2}).
\]
It follows that 
\(S\asymp s(d)\asymp n^2p\), and since $\sum_{i=1}^n  d_i\partial_i\Psi_k=O(n^vp)$,
\[
\sum_{i=1}^n (\widehat\theta_i-\theta_i)\partial_i\Psi_k
=
O\left(n^{v-2}p^{-1/2}\right)
+
O(n^{v-2})
=
O\left(n^{v-2}p^{-1/2}\right).
\]

We now address the the quadratic term. From \eqref{eq:theta hat decomposition}, we get
\[
\widehat\theta_i-\theta_i
=
O\left(\frac{D_i-d_i}{np}\right)
+
O\left(\frac{|S-s(d)|}{n^2p}\right)
+
O(n^{-1}),
\]
and standard degree concentration gives $\sum_{i=1}^n(D_i-d_i)^2
= O(n^2p)$. 
Therefore
\[
\sum_{i=1}^n(\widehat\theta_i-\theta_i)^2
=
O(p^{-1}).
\]
From Lemma~\ref{lemma:derivative bounds for Psi_k}, up to a constant factor, the quadratic term is bounded by
\[
n^{v-2}\sum_{i=1}^n(\widehat\theta_i-\theta_i)^2
+
n^{v-3}
\sum_{i\neq j}
|\widehat\theta_i-\theta_i|
|\widehat\theta_j-\theta_j|.
\]
By Cauchy's inequality,
\[
\sum_{i\neq j}
|\widehat\theta_i-\theta_i|
|\widehat\theta_j-\theta_j|
\leq
n\sum_{i=1}^n(\widehat\theta_i-\theta_i)^2.
\]
Thus the quadratic term is $O\left(n^{v-2}p^{-1}\right)$ and 
the proof is complete.    
\end{proof}

With Lemma~\ref{lemma:stability of Psi_k}, we are now ready to bound
the error of the bias estimate obtained from the second-level bootstrap.

\begin{lemma}[Bias correction by bootstrap]
\label{lemma:bias correction by bootstrap}
We have 
$$
\mathrm{Bias}(\mu_R(\widehat{\widehat{P}}))-\mathrm{Bias}(\mu_R(\widehat{P})) = O(n^{v-2}p^{e-2}).
$$
\end{lemma}

\begin{proof}[Proof of Lemma~\ref{lemma:bias correction by bootstrap}]
Recall from \eqref{eq:bias decomposition global} that the bias of $\mu_R(\widehat{P})$ in estimating $\mu_R(P)$ is given by 
$$
\mathrm{Bias}(\mu_R(\widehat{P}))=p^{e}\left(\Psi_1(\theta)-\Psi_2(\theta)\right)+p^{e-1}\Psi_3(\theta)+O(n^{v-2}p^{e-2}).    
$$
Conditioning on $\widehat{P}$,
our estimate of this bias using the second layer bootstrap is 
$$
\mathrm{Bias}(\mu_R(\widehat{\widehat{P}}))=\widehat{p}^{e}\left(\Psi_1(\widehat{\theta})-\Psi_2(\widehat{\theta})\right)+\widehat{p}^{e-1}\Psi_3(\widehat{\theta})+O(n^{v-2}\widehat{p}^{e-2}).    
$$
Therefore, 
\begin{eqnarray*}
\mathrm{Bias}(\mu_R(\widehat{\widehat{P}}))-\mathrm{Bias}(\mu_R(\widehat{P}))
&=&
(\widehat p^{\,e-1}-p^{e-1})\Psi_3(\theta)
+
\widehat p^{\,e-1}
\left[\Psi_3(\widehat\theta)-\Psi_3(\theta)\right]
\nonumber\\
&+&
\widehat p^{\,e}
\left[\Psi_1(\widehat{\theta})-\Psi_1(\theta)+\Psi_2(\theta)-\Psi_2(\widehat{\theta})\right]\nonumber\\
&+&
(\widehat p^{\,e}-p^e)\left[\Psi_1(\theta)-\Psi_2(\theta)\right] +O(n^{v-2}p^{e-2}).
\end{eqnarray*}
By Bernstein's inequality for the edge count, we have $\widehat p-p
= \left(n^{-1}p^{1/2}\right)$. 
Therefore,
\[
\widehat p^{\,e-1}-p^{e-1}
=
O\left(n^{-1}p^{e-3/2}\right),
\qquad
\widehat p^{\,e}-p^e
=
O\left(n^{-1}p^{e-1/2}\right).
\]
From Lemma~\ref{lemma:derivative bounds for Psi_k} and Lemma~\ref{lemma:stability of Psi_k}, we obtain
$$
\mathrm{Bias}(\mu_R(\widehat{\widehat{P}}))-\mathrm{Bias}(\mu_R(\widehat{P})) = O(n^{v-2}p^{e-2}).
$$
The proof is complete.     
\end{proof}


\subsection{Plugin Estimators for Rooted Subgraph Means}

We now repeat the same bias-correction argument in Section~\ref{sec:mu hat for global mean} for rooted subgraph counts.
Let \((R,o)\) be a fixed rooted graph, where \(o\in V(R)\) is the distinguished
root. Recall from \eqref{eq:rooted mu P} that
\begin{equation*}
\mu_{R,o}^{(i)}(P)
=
\frac{p^e}{|\operatorname{Aut}(R,o)|}
\sum_{\substack{\phi:V(R)\hookrightarrow[n] \\ \phi(o)=i}} \ 
\prod_{a\in V(R)}
\theta_{\phi(a)}^{r_a},
\end{equation*}
and its plugin estimator from \eqref{eq:rooted mu Phat} is
\begin{equation*}
\mu_{R,o}^{(i)}(\widehat P)
=
\frac{1}{|\operatorname{Aut}(R,o)|}
\left(\frac{n}{n-1}\right)^e
S^{-e}
\sum_{\substack{\phi:V(R)\hookrightarrow[n] \\ \phi(o)=i}} \
\prod_{a\in V(R)}
D_{\phi(a)}^{r_a}.    
\end{equation*}
We apply Lemma~\ref{lemma:mixed moment of node degrees and edge count} to
calculate the expectation of $\mu_{R,o}^{(i)}(\widehat P)$, using the notation
defined in Section~\ref{sec:mu hat for global mean}. Analogously to the bias
decomposition for global counts in \eqref{eq:bias decomposition global}, we have
\begin{equation}
\label{eq:bias decomposition rooted}
\mathrm{Bias}\left(\mu_{R,o}^{(i)}(\widehat{P})\right)=p^{e}\left(\Psi_1-\Psi_2\right)+p^{e-1}\Psi_3+O(n^{v-3}p^{e-2}),    
\end{equation} 
where $\Psi_{1i},\Psi_{2i},\Psi_{3i}$ are functions of $\theta=(\theta_1,...,\theta_n)$ given by
\begin{eqnarray*}
\Psi_{1i} &=& \frac{1}{|\mathrm{Aut}(R,o)|}
\sum_{\substack{\phi:V(R)\hookrightarrow[n] \\ \phi(o)=i}} \left[\left(\frac{n}{n-1}\right)^e M(\theta)^{-e}\prod_{a\in V(R)}
\eta_{\phi(a)}(\theta)^{r_{\phi(a)}}-\prod_{a\in V(R)}
\theta_{\phi(a)}^{r_{\phi(a)}}\right],\\
\Psi_{2i} &=& \frac{1}{|\mathrm{Aut}(R,o)|}
\sum_{\substack{\phi:V(R)\hookrightarrow[n] \\ \phi(o)=i}} \left(\frac{n}{n-1}\right)^e M(\theta)^{-e}
\sum_{a\in V(R)}
\binom{r_a}{2}\eta_{{\phi(a)}}^{r_a-2}(\theta)\theta_{\phi(a)}^2\sigma_{\phi(a)}(\theta)
\prod_{b\neq a}
\eta_{\phi(b)}^{r_b}(\theta), \\
\Psi_{3i} &=& \frac{1}{|\mathrm{Aut}(R,o)|}
\sum_{\substack{\phi:V(R)\hookrightarrow[n] \\ \phi(o)=i}} \left(\frac{n}{n-1}\right)^e M(\theta)^{-e}
\sum_{a\in V(R)}
\binom{r_a}{2}\eta_{\phi(a)}^{r_a-1}(\theta)
\prod_{b\neq a}
\eta_{\phi(b)}^{r_b}(\theta),
\end{eqnarray*}
with the convention that terms involving $\binom{r_a}{2}$ are set to zero when
$r_a<2$.

\paragraph{Magnitude bounds.} The argument for bounding $\Psi_{ki}$ is the same as that used to bound
$\Psi_k$ in the proof of Lemma~\ref{lemma:derivative bounds for Psi_k}, except
that the rooted embedding sum has only $O(n^{v-1})$ terms. Therefore,
\begin{equation}
\label{eq:rooted-Psi-size}
\Psi_{ki}(\theta)=O(n^{v-2}),
\qquad
k=1,2,3.
\end{equation}

\paragraph{Derivative bounds.}
The derivative bounds for $\Psi_{ki}$ differ from the global case because the
root is fixed at \(i\). If a derivative is taken with respect to the root
coordinate \(\theta_i\), it does not reduce the number of free embedding vertices.
If a derivative is taken with respect to a non-root coordinate \(\theta_u\),
\(u\neq i\), then a non-root vertex is forced to equal \(u\), reducing the
number of free vertices by one.
Consequently, uniformly on the regular set of $\theta$ with \(\theta_j\asymp 1\) for all $j\in[n]$,
\begin{equation}
\label{eq:rooted-Psi-first-derivative}
\partial_u\Psi_{ki}(\theta)
=
O(n^{v-2})\mathbf 1_{\{u=i\}}
+
O(n^{v-3})\mathbf 1_{\{u\neq i\}}.
\end{equation}
For second derivatives,
\begin{equation}
\label{eq:rooted-Psi-second-derivative}
\partial_{uv}\Psi_{ki}(\theta)
=
\begin{cases}
O(n^{v-2}), & u=v=i,\\
O(n^{v-3}), & u=v\neq i,\\
O(n^{v-3}), & u\neq v \text{ and } \{u,v\}\cap\{i\}\neq\emptyset,\\
O(n^{v-4}), & u\neq v \text{ and } u,v\neq i.
\end{cases}
\end{equation}
These bounds follow from the same vertex-product counting as in the global
case. The key point is that each \(\Psi_{ki}\) carries a finite-\(n\) factor
\(n^{-1}\). For \(\Psi_{1i}\), this factor is 
\[
\left(\frac{n}{n-1}\right)^eM(\theta)^{-e}\prod_{a\in V(R)} \eta_{\phi(a)}^{r_a} - \prod_{a\in V(R)} \theta_{\phi(a)}^{r_a}(\theta). 
\]
For \(\Psi_{2i}\), it is the distinct-neighbor correction $\eta_j(\theta)^{-2}\theta_j^2\sigma_j(\theta)=O(n^{-1})$, and for \(\Psi_{3i}\), it is the inverse expected-degree factor $\eta_j(\theta)^{-1}=O(n^{-1})$. Counting the number of free vertices in the node embedding yields
\eqref{eq:rooted-Psi-second-derivative}.

\paragraph{Stability of $\Psi_{ki}$.}
We claim that, for each \(k=1,2,3\),
\begin{equation}
\label{eq:rooted-Psi-stability}
\Psi_{ki}(\widehat\theta)
-
\Psi_{ki}(\theta)
= O
\left(n^{v-5/2}p^{-1/2}\right).
\end{equation}
To see this, recall from \eqref{eq:theta hat decomposition} that
\begin{equation*}
\widehat\theta_u-\theta_u
=
\left(
\frac{nD_u}{S}-\frac{nd_u}{s(d)}
\right)
+
\left(
\frac{nd_u}{s(d)}-\theta_u
\right) = \left(
\frac{nD_u}{S}-\frac{nd_u}{s(d)}
\right)
+ O(n^{-1}).   
\end{equation*}
Taylor expansion gives
\[
\Psi_{ki}(\widehat\theta)-\Psi_{ki}(\theta)
=
\sum_{u=1}^n
\partial_u\Psi_{ki}(\theta)(\widehat\theta_u-\theta_u)
+
\frac12
\sum_{u,v=1}^n
\partial_{uv}\Psi_{ki}(\bar\theta)(\widehat\theta_u-\theta_u)(\widehat\theta_v-\theta_v),
\]
where \(\bar\theta\) lies between \(\theta\) and \(\widehat\theta\).
We first bound the linear part. Using
\eqref{eq:rooted-Psi-first-derivative} and the fact that
$\widehat{\theta}_u-\theta_u=O((np)^{-1/2})$, the linear term corresponding to
$u=i$ is bounded by
\[
\partial_i\Psi_{ki}(\theta)(\widehat{\theta}_i-\theta_i)
=
O(n^{v-5/2}p^{-1/2}).
\]
For the terms with $u\neq i$, again by \eqref{eq:rooted-Psi-first-derivative}, we have  
\(
\partial_u\Psi_{ki}(\theta)=O(n^{v-3}).
\)
Therefore Bernstein's inequality yields 
\[
\sum_{u\neq i}
\partial_u\Psi_{ki}(\theta)(\widehat{\theta}_u-\theta_u)
=
O\left(n^{v-3}p^{-1/2}\right).
\]
Combining this bound with the term corresponding to $u=i$, we obtain
\[
\sum_{u=1}^n
\partial_u\Psi_{ki}(\theta)(\widehat\theta_u-\theta_u)
=O \left(n^{v-5/2}p^{-1/2}\right).
\]
For the quadratic part, we use \eqref{eq:rooted-Psi-second-derivative} and the following bounds
\[
(\widehat\theta_u-\theta_u)^2
=
O\left((np)^{-1}\right),
\qquad
\sum_{u=1}^n(\widehat\theta_u-\theta_u)^2
= O(p^{-1}).
\]
The root diagonal term satisfies
\[
\partial_{ii}\Psi_{ki}(\bar\theta)(\widehat\theta_u-\theta_u)^2
=
O(n^{v-2}(np)^{-1})
=O\left(n^{v-3}p^{-1}\right).
\]
The nonroot diagonal terms satisfy
\[
\sum_{u\neq i}^n
\partial_{uu}\Psi_{ki}(\bar\theta)(\widehat\theta_u-\theta_u)^2 = 
O\left(n^{v-3}p^{-1}\right).
\]
The root-nonroot off-diagonal terms satisfy, by Cauchy's inequality,
\begin{eqnarray*}
\sum_{u\neq i}^n
\partial_{iu}\Psi_{ki}(\bar\theta)(\widehat\theta_i-\theta_i)(\widehat\theta_u-\theta_u) \le O(n^{v-3})|\widehat\theta_i-\theta_i| \left(n\sum_{u=1}^n(\widehat\theta_u-\theta_u)^2\right)^{1/2} = O(n^{v-3}p^{-1}).    
\end{eqnarray*}
Finally, the nonroot--nonroot off-diagonal terms are bounded by
\[
O(n^{v-4})
\sum_{\substack{u\neq v\\ u,v\neq i}}
|\widehat\theta_u-\theta_u||\widehat\theta_v-\theta_v|
\leq
O(n^{v-4})
n\sum_{u=1}^n(\widehat\theta_u-\theta_u)^2
=O\left(n^{v-3}p^{-1}\right).
\]
Therefore, the full quadratic part is $O\left(n^{v-3}p^{-1}\right)=
O\left(n^{v-5/2}p^{-1/2}\right)$, and 
\eqref{eq:rooted-Psi-stability} holds.

\paragraph{Bias correction by second-layer bootstrap.}
Recall from \eqref{eq:bias decomposition rooted} that 
\begin{equation*}
\mathrm{Bias}\left(\mu_{R,o}^{(i)}(\widehat{P})\right)=p^{e}\left(\Psi_1(\theta)-\Psi_2(\theta)\right)+p^{e-1}\Psi_3(\theta)+O(n^{v-3}p^{e-2}).    
\end{equation*}
The second-layer bootstrap estimate of this bias is 
\begin{equation*}
\mathrm{Bias}\left(\mu_{R,o}^{(i)}(\widehat{\widehat{P}})\right)=\widehat{p}^{e}\left(\Psi_1(\widehat\theta)-\Psi_2(\widehat\theta)\right)+\widehat{p}^{e-1}\Psi_3(\widehat\theta)+O(n^{v-3}\widehat{p}^{e-2}),    
\end{equation*}
where $\hat{p}$ and $\hat{\theta}$ are given by \eqref{eq: main parameter estimators}. Since $\widehat{p}\asymp p$ with high probability, the error of this bias estimation is
\begin{eqnarray*}
\mathrm{Bias}\left(\mu_{R,o}^{(i)}(\widehat{\widehat{P}})\right) - \mathrm{Bias}\left(\mu_{R,o}^{(i)}(\widehat{P})\right)
&=& (\widehat p^{e-1}-p^{e-1})
\Psi_{3i}(\theta)
+
\widehat p^{e-1}
\left[
\Psi_{3i}(\widehat\theta)-\Psi_{3i}(\theta)
\right]\\
&+&
\widehat p^e
\left[
\Psi_{1i}(\widehat\theta)-\Psi_{1i}(\theta)
+
\Psi_{2i}(\theta)-\Psi_{2i}(\widehat\theta)
\right]\\
&+&(\widehat p^e-p^e)
\left[
\Psi_{1i}(\theta)
-
\Psi_{2i}(\theta)
\right] + O(n^{v-3}p^{e-2}).
\label{eq:rooted-G-stability-expansion}
\end{eqnarray*}
Since  
\(
\widehat p-p
=
\widetilde O\left(n^{-1}p^{1/2}\right)
\)
by Bernstein's inequality, 
\[
\widehat p^e-p^e
= O\left(n^{-1}p^{e-1/2}\right),
\qquad
\widehat p^{e-1}-p^{e-1}
=O\left(n^{-1}p^{e-3/2}\right).
\]
Using \eqref{eq:rooted-Psi-size} and \eqref{eq:rooted-Psi-stability}, we conclude
\begin{equation}
\label{eq:rooted-G-stability-final}
\mathrm{Bias}\left(\mu_{R,o}^{(i)}(\widehat{\widehat{P}})\right) - \mathrm{Bias}\left(\mu_{R,o}^{(i)}(\widehat{P})\right)
=
O\left(n^{v-5/2}p^{e-3/2}\right).
\end{equation}

\section{Variances of Plugin Estimators of Subgraph Means}
\label{sec:variance}

We now calculate the variance of $\mu_R(\widehat{P})$. To simplify the
presentation, we postpone the treatment of rooted subgraph counts until
Section~\ref{sec:rooted count variance}. For this purpose, it is convenient to
view $\mu_R(\widehat{P})$ in two equivalent ways: as a function of the
edge-error matrix $E=A-\E A=A-P$, and as a function of the degree vector $D$.
The former viewpoint allows us to apply the Hoeffding decomposition to
$\mu_R(\widehat{P})$ and approximate its variance using the first-order term.
The latter viewpoint is convenient for calculating the coefficients of this
approximation and for showing that the higher-order terms are negligible.

Recall from \eqref{eq:mu Phat} that
\begin{equation*}
\mu_R(\widehat P)
=
\frac{1}{|\operatorname{Aut}(R)|}
\left(\frac{n}{n-1}\right)^e
S^{-e}
\sum_{\phi:V(R)\hookrightarrow[n]}
\prod_{a\in V(R)}
D_{\phi(a)}^{r_a}.
\end{equation*}
To view $\mu_R(\widehat P)$ as a function of the vector of node degrees, denote
\[
D:=(D_1,\ldots,D_n), \qquad
d:=\E D=(d_1,\ldots,d_n), \qquad
X:=D-d.
\]
Consider the function $F:\mathbb{R}^n\to\mathbb{R}$ defined by
\begin{equation}
\label{eq:F definition}
    F(x)
    =
    \frac{1}{|\operatorname{Aut}(R)|}
    \left(\frac{n}{n-1}\right)^e
    s(x)^{-e}
    \sum_{\phi:V(R)\hookrightarrow[n]}
    \prod_{a\in V(R)}
    x_{\phi(a)}^{r_a},
\end{equation}
if $s(x):=\sum_{i=1}^n x_i\neq 0$ and $F(x)=0$ if $s(x)=0$. Thus, $\mu_R(\widehat{P})=F(D)$. However, note that $\mu_R(P)\neq F(d)$ in
general. For notational simplicity, we suppress the dependence of $F$ on the
subgraph $R$.

To calculate the variance of $\mu_R(\widehat{P})$, we use the Hoeffding
decomposition with respect to the independent edge components of
$E=A-\E A=A-P$. To connect this decomposition with $F$, let $K_n$ denote the
undirected complete graph on $n$ vertices, with unordered edge set
\[
E(K_n)=\{\{i,j\}:1\le i<j\le n\}. 
\]
For each $\alpha\in E(K_n)$, write 
$$
A_\alpha = A_{ij}, \qquad P_{\alpha} = P_{ij}, \qquad E_\alpha = E_\alpha = A_\alpha - P_\alpha. 
$$
Let \(B\in\mathbb R^{n\times |E(K_n)|}\) be the vertex-edge incidence
matrix, defined by
\[
B_{i\alpha}:=\mathbf 1\{i\in\alpha\}.
\]
Then $(BE)_i:=\sum_{\alpha} B_{i\alpha}E_\alpha = \sum_{j=1}^n E_{ij}=D_i-d_i$, so 
\[
D = d+  BE.
\]
Define the induced edge-error function 
\begin{equation}
\label{eq:untruncated G definition}
G(z):=F(d+Bz),
\qquad z\in\mathbb R^{E(K_n)}.
\end{equation}
Then
\[
\mu_R(\widehat P)=F(D)=G(E).
\]
For each $\alpha\in E(K_n)$, differentiating \(G(z)=F(d+Bz)\) with respect to the
edge-error coordinate \(z_\alpha\) gives
\begin{equation}
\label{eq:first derivarive connection}
\partial_\alpha G(z):=
    \frac{\partial G}{\partial z_\alpha}(z)
=
\nabla F(d+Bz)^\top Be_\alpha
=\sum_{i\in\alpha}F_i(d+Bz).
\end{equation}
In particular, 
\[
\partial_\alpha G(0)
=
\sum_{i\in\alpha}F_i(d).
\]
Similarly, for \(\alpha,\beta,\gamma\in E(K_n)\),
\begin{eqnarray}
\label{eq:second derivative connection}
\partial_\alpha\partial_\beta G(z)
&=&
\sum_{i\in\alpha}\sum_{j\in\beta}F_{ij}(d+Bz),\\
\partial_\alpha\partial_\beta\partial_\gamma G(z)
&=&
\label{eq:third derivative connection}
\sum_{i\in\alpha}
\sum_{j\in\beta}
\sum_{k\in\gamma}
F_{ijk}(d+Bz).
\end{eqnarray}

The main idea of the proof is to expand $G(E)$ in an orthonormal basis for the
$L^2$ space generated by the random variables in the entries of $E$ and to show
that the variance of $G(E)$ is well approximated by the first-order terms, which
admit simple formulas. To establish this approximation, we need to bound the
partial derivatives of $G$ using their relationship with the derivatives of $F$
in \eqref{eq:first derivarive connection},
\eqref{eq:second derivative connection}, and
\eqref{eq:third derivative connection}. Although these bounds hold for typical
degree vectors satisfying $D\approx \E D$, they may be too large for our
purposes when some $D_i$ are close to $n$ or when $s(D)$ is close to zero.
As a technical workaround, we truncate the degree vector and introduce truncated
versions of $F$ and $G$ next.

\subsection{Node Degree Truncation}
\label{sec:truncation global}

Assume the Chung--Lu model with $c_\theta\le \theta_i\le C_\theta$ for all
$1\le i\le n$ and $np/\log n\to\infty$. Then there exist constants
$C=C(c_\theta,C_\theta)>0$ and $c=c(c_\theta,C_\theta)>0$ such that the
regular-degree event
\begin{equation}
\label{eq:regular-degree-event}
\Omega_n
:=
\left\{
    cnp\le D_i\le Cnp
    \text{ for all } i=1,\ldots,n
\right\}
\end{equation}
holds with very high probability; that is, $\mathbb P(\Omega_n^c)$ is smaller
than any fixed negative power of $n$. This follows directly from Bernstein's
inequality and a union bound.

We now define the truncation. Choose a smooth map $\tau:\mathbb{R}\to\mathbb{R}$
such that
\[
\tau(u)=u
\quad\text{for } u\in[cnp,Cnp],
\]
while globally
\[
\frac{cnp}{2}\le \tau(u)\le 2Cnp,
\]
and
\[
|\tau'(u)|\lesssim 1,
\qquad
|\tau''(u)|\lesssim (np)^{-1},
\qquad
|\tau'''(u)|\lesssim (np)^{-2}.
\]
Such a function can be obtained by rescaling a fixed smooth cutoff function. The
bounds on the derivatives of $\tau$ are needed to control the derivatives of the
truncated versions of $F$ and $G$, defined by
\begin{equation}
\label{eq:truncated FG}
\mathcal F(x)
:= 
F\bigl(\tau(x_1),\ldots,\tau(x_n)\bigr),\qquad
\mathcal G(z)
= \mathcal{F}(d+Bz).        
\end{equation}
By definition, $F(D)=\mathcal{F}(D)$ and $G(E)=\mathcal{G}(E)$ conditional on $\Omega_n$. Also, similar to \eqref{eq:first derivarive connection},
\eqref{eq:second derivative connection}, and
\eqref{eq:third derivative connection}, for any $\alpha,\beta,\gamma\in E(K_n)$, we have 
\begin{eqnarray}
\label{eq:localized-first-derivative-connection}
\partial_\alpha \mathcal G(z)
&=&
\sum_{i\in\alpha}\mathcal F_i(d+Bz),\\
\label{eq:localized-second-derivative-connection}
\partial_\alpha\partial_\beta \mathcal G(z)
&=&
\sum_{i\in\alpha}\sum_{j\in\beta}
\mathcal F_{ij}(d+Bz),\\
\label{eq:localized-third-derivative-connection}
\partial_\alpha\partial_\beta\partial_\gamma \mathcal G(z)
&=&
\sum_{i\in\alpha}\sum_{j\in\beta}\sum_{k\in\gamma}
\mathcal F_{ijk}(d+Bz).
\end{eqnarray}
As expected, $G$ and $\mathcal{G}$ have asymptotically
equivalent variances.

\begin{lemma}[Equivalent variances]
\label{lemma:equivalent variances} For any $\kappa>0$, we have
$$
|\Var(F(D))-\Var(\mathcal{F}(D))|\lesssim n^{-\kappa}\big[1+\Var(\mathcal{F}(D))\big].
$$
\end{lemma}

\begin{proof}[Proof of Lemma~\ref{lemma:equivalent variances}]
Recall that $F(D)=\mu_R(\widehat{P})=0$ if $S=\sum_{i=1}^n D_i=0$. When $S\ge 1$, since \(0\le D_i\le n\) and
\(\sum_{a\in V(R)}r_a=2e\), we have the crude bound
\[
F(D)
=
\frac{1}{|\operatorname{Aut}(R)|}
\left(\frac{n}{n-1}\right)^e
S^{-e}
\sum_{\phi:V(R)\hookrightarrow[n]}
\prod_{a\in V(R)}
D_{\phi(a)}^{r_a}
\lesssim n^v n^{2e} = n^{v+2e}.
\]
Similarly, since \(\tau(D_i)\asymp np\) globally, we have 
\(
\mathcal F(D)\lesssim n^{v+2e}.
\)
Denote $Q = F(D)-\mathcal{F(D)}$. Since $F(D)=\mathcal{F}(D)$ conditional on $\Omega_n$, we have $Q = Q\mathbf{1}_{\Omega_n^c}$. Therefore,
$$
\Var(Q)\le \E[Q^2] = \E[Q^2\mathbf{1}_{\Omega_n^c}] \lesssim 2n^{v+2e-\kappa} 
$$
for any $\kappa>0$. 
By By Cauchy--Schwarz inequality,
\begin{eqnarray*}
|\Var(F(D))-\Var(\mathcal{F}(D))| &=& |2\Cov(\mathcal{F}(D),Q)+\Var(Q) \\
&\le& \sqrt{\Var(\mathcal{F}(D))\Var(Q)} + \Var(Q). 
\end{eqnarray*}
The lemma then follows from the bound on $\Var(Q)$ above.
\end{proof}

\subsection{Hoeffding Decomposition}
\label{sec:Hoeffding decomposition global}

We are now ready to apply Hoeffding decomposition to $\mathcal{G}(E)$. 
For each undirected edge $\alpha=\{i,j\}$, define
\[
    Z_\alpha:=\frac{E_\alpha}{\sigma_\alpha}.
\]
Then $\E Z_\alpha=0$ and $\E Z_\alpha^2=1$. Since the edge variables are
independent, the products
\[
    Z_I:=\prod_{\alpha\in I}Z_\alpha, \qquad I\subset E(K_n)
\]
form an orthonormal basis for
the $L^2$ space generated by the entries of $E = A - \E A$. Therefore,
$\mathcal{G}(E)$ admits the Hoeffding decomposition
\[
    \mathcal{G}(E)=\E \mathcal{G}(E)+\sum_{I\neq \emptyset} a_I Z_I,
    \qquad
    a_I:=\E[\mathcal{G}(E)Z_I].
\]
The first-order part is
\[
\sum_{\alpha}a_\alpha Z_\alpha = \sum_{\alpha}\frac{a_\alpha}{\sigma_\alpha} E_\alpha=:\sum_{\alpha}c_\alpha E_\alpha,
\]
where the coefficient $c_\alpha=a_\alpha/\sigma_\alpha = \E[\mathcal{G}(E) E_\alpha]/\sigma_\alpha^2$ satisfies 
\begin{eqnarray}
c_\alpha &=& 
\nonumber
\sigma_\alpha^{-2}\Big\{\E[\mathcal{G}(E)E_\alpha|A_\alpha=1]\cdot\Prob(A_\alpha=1)+\E[\mathcal{G}(E)E_\alpha|A_\alpha=0]\cdot\Prob(A_\alpha=0)\Big\}\\ 
&=& \nonumber
\sigma_\alpha^{-2}\Big\{\E[\mathcal{G}(E)(1-P_\alpha)|A_\alpha=1]\cdot P_\alpha+\E[\mathcal{G}(E)(-P_\alpha)|A_\alpha=0]\cdot (1-P_\alpha)\Big\}\\
&=&\E[\mathcal{G}(E)| A_\alpha=1]-\E[\mathcal{G}(E)| A_\alpha=0] 
\label{eq:c alpha def}
\end{eqnarray}
because $\sigma_\alpha^2 = P_\alpha(1-P_\alpha)$. 
Define the higher-order Hoeffding component by $\mathcal{G}^{(2+)}:=\sum_{|I|\geq2}a_IZ_I$. Then
\begin{equation}
\label{eq:Hoeffding decomposition of mathcal G}
\mathcal{G}(E)-\E \mathcal{G}(E)=\sum_\alpha c_\alpha E_\alpha+\mathcal{G}^{(2+)}
\end{equation}
It follows from the orthogonality of $Z_I$ that 
\begin{equation*}
\label{eq:variance decomposition of mathcal G} 
\mathrm{Var}(\mathcal{G}(E)) = \sum_\alpha c_\alpha^2 \sigma_\alpha^2 + \mathrm{Var}\left(\mathcal{G}^{(2+)}\right).
\end{equation*}
Our goal is to show that the second term is negligible and that, for the first
term, $c_\alpha$ can be approximated by
$\partial_\alpha \mathcal{G}(0)
    =
    \sum_{i\in\alpha}\mathcal{F}_i(d)$. 
To this end, we need to bound the derivatives of $\mathcal{F}$ at $d$, which we
do next.

\subsection{Derivative Bounds}
\label{sec:derivative bound global}

\begin{lemma}[Derivative bounds]
\label{lem:localized-derivative-bounds}
The function 
$\mathcal F$ defined in \eqref{eq:truncated FG} satisfies \(\mathcal F\in C^3(\mathbb R^n)\). Moreover, uniformly over all
\(x\in\mathbb R^n\) and $i,j,k\in [n]$, the following derivative bounds hold:
\begin{eqnarray*}
|\mathcal F_i(x)|
&\lesssim& n^{v-2}p^{e-1},\\
|\mathcal F_{ij}(x)|
&\lesssim& n^{v-2-\nu(i,j)}p^{e-2},\\
|\mathcal F_{ijk}(x)|
&\lesssim&
n^{v-3-\nu(i,j,k)}p^{e-3},
\end{eqnarray*}
where \(\nu(i_1,\ldots,i_m)\) denotes the number of distinct indices among
\(i_1,\ldots,i_m\).
\end{lemma}

\begin{proof}
Recall from \eqref{eq:F definition} that
$$
    F(x)
    =
    \frac{1}{|\operatorname{Aut}(R)|}
    \left(\frac{n}{n-1}\right)^e
    s(x)^{-e}
    \sum_{\phi:V(R)\hookrightarrow[n]}
    \prod_{a\in V(R)}
    x_{\phi(a)}^{r_a},
$$
where $s(x)=\sum_{i=1}^n x_i$. Denote the sum on the right-hand side by \(H\). Since \(H\) is a polynomial and \(s(x)^{-e}\) is smooth on
\(
    \{x\in\R^n:s(x)>0\},
\)
the function \(F\) is \(C^\infty\) on \(\{s\in\R^n: s(x)>0\}\).
By construction, $cnp\le \tau(u)\le Cnp $ for all $u\in\R$.
Therefore, for every \(x\in\mathbb R^n\) and $i\in[n]$,
\[
T_i(x):=\tau(x_i)\asymp np, \qquad s(T(x)):=\sum_{i=1}^n T_i(x)\asymp n^2p.
\]
In particular, \(T(x)=(T_1(x),...,T_n(x))\) always lies in the regular degree region and
\(s(T(x))>0\). Since \(T\in C^3(\mathbb R^n)\) and \(F\) is smooth on the
range of \(T\), we have $\mathcal F=F\circ T\in C^3(\mathbb R^n)$.

For simplicity, we first record the derivative bounds for \(F\) on the regular region: let
\(x\in\R^n\) be any vector satisfying $x_i\asymp np$ for all $i\in[n]$ and $s(x)\asymp np$.
For \(m\le 3\), let \(i_1,\ldots,i_m\in[n]\) and denote
$\nu=|\{i_1,\ldots,i_m\}|$. 
We claim that
\begin{equation}
\label{eq:F-regular-derivative-general}
    \left|
\partial_{i_1}\cdots\partial_{i_m}F(x)
    \right|
    \lesssim
    n^{v-m-\nu}p^{e-m}.
\end{equation}
To prove this, we first bound the derivatives of \(H\). Since \(R\) is fixed
and $\sum_{a\in V(R)}r_a=2e$, each summand in \(H(x)\) has size \(O((np)^{2e})\), and there are
\(O(n^v)\) injective maps. Hence
\[
    |H(x)|\lesssim n^{v+2e}p^{2e}.
\]
More generally, differentiating \(H\) \(m\) times with respect to
\(i_1,\ldots,i_m\) forces the embedding to use the \(\nu=|\{i_1,...,i_m\}|\) distinct vertices
among \(i_1,\ldots,i_m\). Thus \(\nu\) motif vertices are fixed, leaving
\(O(n^{v-\nu})\) choices for the remaining motif vertices. Since each derivative
removes one power of \(np\), we have
\begin{equation}
\label{eq:H-derivative-general}
    \left|
    \partial_{i_1}\cdots\partial_{i_m}H(x)
    \right|
    \lesssim
    n^{v-\nu}(np)^{2e-m} = n^{v+2e-\nu-m}p^{2e-m}.
\end{equation}
Now consider \(F(x)=c_Rs(x)^{-e}H(x)\), where $c_R$ denotes the constant.  A term in an \(m\)-th derivative
of \(F\) is obtained by putting \(r\) derivatives on \(s(x)^{-e}\) and
\(m-r\) derivatives on \(H(x)\), where \(0\le r\le m\). Derivatives of
\(s(x)^{-e}\) of order \(r\) are bounded by a constant multiple of
\[
    s(x)^{-e-r}\asymp (n^2p)^{-e-r} = n^{-2e-2r}p^{-e-r}.
\]
Suppose the \(m-r\) derivatives falling on \(H\) involve \(\nu_H\) distinct
coordinates.  By \eqref{eq:H-derivative-general}, the corresponding factor
from \(H\) is bounded by
\[
 n^{v+2e-\nu_H-m+r}p^{2e-m+r}.
\]
Thus this product-rule term is bounded by
\[
 n^{v-\nu_H-m-r}p^{e-m}.
\]
The total number \(\nu\) of distinct coordinates among \(i_1,\ldots,i_m\)
satisfies $\nu\le \nu_H+r$
because the \(r\) derivatives hitting \(s(y)^{-e}\) can introduce at most
\(r\) additional distinct coordinates. Hence
\[
    n^{v-\nu_H-m-r}p^{e-m}
    \le
    n^{v-m-\nu}p^{e-m}.
\]
This proves \eqref{eq:F-regular-derivative-general}.

We now transfer these bounds to \(\mathcal F(x)=F(T(x))\). Since
\(T(x)\) is always in the regular region, the bounds above apply to
\(F\) and its derivatives evaluated at \(T(x)\). It remains to handle the additional terms arising from the chain rule. 
For first derivatives, 
\[
    \mathcal F_i(x)
    =
    F_i(T(x))\tau'(x_i).
\]
Using $|F_i(T(x))|\lesssim n^{v-2}p^{e-1}$ from \eqref{eq:F-regular-derivative-general} and
$|\tau'(x_i)|\lesssim 1$, 
we obtain
\[
    |\mathcal F_i(x)|
    \lesssim n^{v-2}p^{e-1}.
\]
For second derivatives,
\[
    \mathcal F_{ij}(x)
    =
    F_{ij}(T(x))\tau'(x_i)\tau'(x_j)
    +
    \mathbf 1_{\{i=j\}}F_i(T(x))\tau''(x_i).
\]
If \(i\neq j\), the second term is absent, so from \eqref{eq:F-regular-derivative-general},
\[
    |\mathcal F_{ij}(x)|
    \lesssim
    |F_{ij}(T(x))|
    \lesssim
    n^{v-4}p^{e-2}.
\]
If \(i=j\), then using $|F_i(T(x))|\lesssim n^{v-2}p^{e-1}$ and $|F_{ii}(T(x))|\lesssim n^{v-3}p^{e-2}$ from \eqref{eq:F-regular-derivative-general}, and $|\tau''(x_i)|\lesssim (np)^{-1}$, we get 
\[
    |\mathcal F_{ii}(x)|
    \lesssim
    |F_{ii}(T(x))|
    +
    |F_i(T(x))|\,|\tau''(x_i)| \lesssim
    n^{v-3}p^{e-2}.
\]

It remains to bound the third derivatives.  The chain rule gives
\[
\begin{aligned}
    \mathcal F_{ijk}(x)
    &=
    F_{ijk}(T(x))\tau'(x_i)\tau'(x_j)\tau'(x_k) \\
    &\quad
    +
    \mathbf 1_{\{i=j\}}
    F_{ik}(T(x))\tau''(x_i)\tau'(x_k) \\
    &\quad
    +
    \mathbf 1_{\{i=k\}}
    F_{ij}(T(x))\tau''(x_i)\tau'(x_j) \\
    &\quad
    +
    \mathbf 1_{\{j=k\}}
    F_{ij}(T(x))\tau'(x_i)\tau''(x_j) \\
    &\quad
    +
    \mathbf 1_{\{i=j=k\}}
    F_i(T(x))\tau'''(x_i).
\end{aligned}
\]
We consider the possible index patterns.
If \(i,j,k\) are all distinct, then all indicator terms vanish. Hence by \eqref{eq:F-regular-derivative-general}, 
\[
    |\mathcal F_{ijk}(x)|
    \lesssim
    |F_{ijk}(T(x))|
    \lesssim
    n^{v-6}p^{e-3}.
\]
If exactly two of \(i,j,k\) are equal, say \(i=j\neq k\), then by \eqref{eq:F-regular-derivative-general} and $|\tau''(x_i)|\lesssim q^{-1}=(np)^{-1}$, 
\[
\begin{aligned}
    |\mathcal F_{iik}(x)|
    &\lesssim
    |F_{iik}(T(x))|
    +
    |F_{ik}(T(x))|\,|\tau''(x_i)|\lesssim n^{v-5}p^{e-3}.
\end{aligned}
\]
Finally, if \(i=j=k\), then
\[
    \mathcal F_{iii}(x)
    =
    F_{iii}(T(x))(\tau'(x_i))^3
    +
    3F_{ii}(T(x))\tau'(x_i)\tau''(x_i)
    +
    F_i(T(x))\tau'''(x_i).
\]
Using \eqref{eq:F-regular-derivative-general}, $|\tau'(x_i)|\lesssim 1$, $|\tau''(x_i)|\lesssim q^{-1}=(np)^{-1}$, and $|\tau'''(x_i)|\lesssim q^{-2}=(np)^{-2}$, we get
\[
|\mathcal F_{iii}(x)| \lesssim n^{v-4}p^{e-3}.
\]
The proof is complete.
\end{proof}

\subsection{Approximation of the Coefficients of the First-Order Terms}
\label{sec:first order coefficient approximation global}
Using the derivative bounds from Section~\ref{sec:derivative bound global}, we are now ready to
compare the exact first-order Hoeffding coefficient $c_\alpha$ with the partial
derivative
$\partial_\alpha \mathcal{G}(0)=\partial_\alpha G(0)$.

\begin{lemma}[Approximating linear coefficients of Hoeffding decomposition]\label{lem:first-projection-error}
For every unordered edge $\alpha=\{i,j\}\in E(K_n)$, we have 
\begin{equation}
\label{eq:general c approximation}
\left|c_\alpha- \partial_\alpha \mathcal{G}(0)\right| \le \sup_{z\in\R^{E(K_n)}}|\partial_\alpha^2\mathcal{G}(z)| + \sum_{\beta\neq\alpha}
      \sigma_\beta^2      \sup_{z\in\R^{E(K_n)}}|\partial_\alpha\partial_\beta^2\mathcal{G}(z)|.
\end{equation}
In particular, using the derivative bounds in Lemma~\ref{lem:localized-derivative-bounds}, we obtain
\[
    \left|c_\alpha- \partial_\alpha \mathcal{G}(0)\right|
    \lesssim
    n^{v-3}p^{e-2}.
\]
Consequently,
\[ \sum_\alpha\sigma_\alpha^2\left[c_\alpha- \partial_\alpha \mathcal{G}(0)\right]^2
    \lesssim
    n^{2v-4}p^{2e-3}.
\]
\end{lemma}

\begin{proof}
Let $E_{-\alpha}$ be the matrix obtained from $E$ by setting the $\alpha$ entry of $E$ to zero, and denote by $\E_{-\alpha}$ the expectation over the randomness of $E_{-\alpha}$. By \eqref{eq:c alpha def}, we have 
\[
\begin{aligned}
    c_\alpha
    =\E_{-\alpha}
      \left[
      \mathcal{G}(E_{-\alpha}+(1-P_\alpha)e_\alpha)
      -\mathcal{G}(E_{-\alpha}-P_\alpha e_\alpha)
      \right] 
    =\int_{-P_\alpha}^{1-P_\alpha}
      \E_{-\alpha}
      \left[\partial_\alpha \mathcal{G}(E^{(-\alpha)}+te_\alpha)\right]
      \,d t,
\end{aligned}
\]
where the last equality is the one-dimensional fundamental theorem of
calculus in the $\alpha$ coordinate. Therefore,
\[
    c_\alpha-\partial_\alpha \mathcal{G}(0)
    =\int_{-P_\alpha}^{1-P_\alpha}
      \left[
      \E_{-\alpha}\partial_\alpha \mathcal{G}(E^{(-\alpha)}+te_\alpha)
      -\partial_\alpha \mathcal{G}(0)
      \right]\,d t.
\]
The interval has length one.  It suffices to bound the integrand uniformly in
$t\in[-P_\alpha,1-P_\alpha]$.
We bound the the integrand by splitting it as 
\[
\begin{aligned}
      \Big[\E_{-\alpha}\partial_\alpha \mathcal{G}(E_{-\alpha}+te_\alpha)
      -\partial_\alpha \mathcal{G}(te_\alpha)\Big]      
      +
      \Big[\partial_\alpha \mathcal{G}(te_\alpha)-\partial_\alpha \mathcal{G}(0)\Big].
\end{aligned}
\]
The second term is deterministic.  By the mean-value theorem, 
\[
    |\partial_\alpha \mathcal{G}(te_\alpha)-\partial_\alpha \mathcal{G}(0)|
    \leq  |t|\sup_z|\partial_\alpha^2\mathcal{G}(z)|\le \sup_z|\partial_\alpha^2\mathcal{G}(z)|.
\]
It remains to bound the first term. We apply Lemma~\ref{lem:centered-expectation} to
\(
    \varphi(z):=\partial_\alpha \mathcal{G}(z+te_\alpha)
\),
where $z$ ranges over all coordinates $\beta\neq\alpha$. Since
\(
    \partial_{\beta\beta}\varphi(z)
    =\partial_\alpha\partial_\beta^2\mathcal{G}(z+te_\alpha),
\)
we have
$$
    \left|
      \E_{-\alpha}\partial_\alpha \mathcal{G}(E_{-\alpha}+te_\alpha)
      -\partial_\alpha \mathcal{G}(te_\alpha)
      \right| \leq
      \sum_{\beta\neq\alpha}
      \sigma_\beta^2      \sup_z|\partial_\alpha\partial_\beta^2\mathcal{G}(z)|.
$$
Putting these inequalities together, we obtain \eqref{eq:general c approximation}.

To obtain a concrete bound from \eqref{eq:general c approximation}, we use the
fact that derivatives of $\mathcal{G}$ can be written as sums of derivatives of
$\mathcal{F}$, and that these derivatives are uniformly bounded. By
\eqref{eq:localized-second-derivative-connection} and
Lemma~\ref{lem:localized-derivative-bounds}, the first term on the right-hand
side of \eqref{eq:general c approximation} satisfies
$$
\sup_z|\partial_\alpha^2\mathcal{G}(z)|\le \sup_x \sum_{i\in \alpha}\sum_{j\in \alpha} |\mathcal{F}_{ij}(x)|
    \lesssim n^{v-3}p^{e-2}.
$$
For the second term on the right-hand side of
\eqref{eq:general c approximation}, note that
$\partial_\alpha\partial_\beta^2\mathcal{G}(z)$ can be written as a sum of
third-order partial derivatives of $\mathcal{F}$, according to
\eqref{eq:localized-third-derivative-connection}. These derivatives are
uniformly bounded by Lemma~\ref{lem:localized-derivative-bounds}. Therefore, if
$\beta$ shares a vertex with $\alpha$, then
\[ |\partial_\alpha\partial_\beta^2\mathcal{G}(x)|
    \lesssim n^{v-4}p^{e-3},
\]
and there are $O(n)$ such $\beta$. If $\beta$ is disjoint from $\alpha$, the relevant degree third derivatives involve at least two
distinct vertex indices. Hence,
\[
|\partial_\alpha\partial_\beta^2\mathcal{G}(x)|
    \lesssim n^{v-5}p^{e-3},
\]
and there are $O(n^2)$ such $\beta$.  Since $\sigma_\beta^2\lesssim p$, it follows that 
$$
\sum_{\beta\neq\alpha}
\sigma_\beta^2
\sup_x|\partial_\alpha\partial_\beta^2\mathcal{G}(x)|
\lesssim
p\left[
    n\,n^{v-4}p^{e-3}
    +n^2\,n^{v-5}p^{e-3}
    \right] \\
    \lesssim n^{v-3}p^{e-2}.
$$
Combining the two bounds gives
\[
    |c_\alpha-\partial_\alpha \mathcal{G}(0)|\lesssim n^{v-3}p^{e-2}.
\]
Finally, there are $O(n^2)$ edges and $\sigma_\alpha^2\lesssim p$, so
$$
\sum_\alpha\sigma_\alpha^2(c_\alpha-h_\alpha)^2
    \lesssim
    n^2p\left(n^{v-3}p^{e-2}\right)^2 \\
    =n^{2v-4}p^{2e-3}.
$$
The proof is complete.
\end{proof}

\begin{lemma}[A centered expectation bound]\label{lem:centered-expectation}
Let $Y_1,\ldots,Y_m$ be independent centered random variables with
$\E Y_r^2=\tau_r^2$.  If $\varphi$ is twice continuously differentiable and
has bounded second partial derivatives, then
\[
    \left|\E\varphi(Y)-\varphi(0)\right|
    \leq
    \frac12\sum_{r=1}^m
    \tau_r^2
    \sup_x\left|\frac{\partial^2\varphi(x)}{\partial x_r^2}\right|.
\]
\end{lemma}

\begin{proof}
We use the telescope sum
\[
    \E\varphi(Y)-\varphi(0)
    =\sum_{r=1}^m
     \E\{\varphi(Y^{(r)})-\varphi(Y^{(r-1)})\},
\]
where $Y^{(r)}:=(Y_1,\ldots,Y_r,0,\ldots,0)$.
Condition on $Y_1,\ldots,Y_{r-1}$.  Taylor expanding in the $r$-th coordinate gives
\[
\begin{aligned}
    \varphi(Y^{(r)})-\varphi(Y^{(r-1)})
    =Y_r\partial_r\varphi(Y^{(r-1)})       +\int_0^{Y_r}(Y_r-u)
       \partial_{rr}\varphi(Y^{(r-1)}+ue_r)\,d u.
\end{aligned}
\]
The first term has conditional expectation zero because $Y_r$ is centered and
independent of $Y_1,\ldots,Y_{r-1}$.  The integral is bounded in absolute value
by
\[
    \frac12Y_r^2\sup_x|\partial_{rr}\varphi(x)|.
\]
Taking expectations and summing over $r$ proves the lemma.
\end{proof}

\subsection{Higher-order Variance Bound}
\label{sec:higher-order variance bound global}
We proceed to bound the variance of the higher order term $\mathcal{G}^{(2+)}$. To that end, for two distinct edges $\alpha$ and $\beta$, define the second discrete
difference of $\mathcal{G}=\mathcal{G}(E)$ by
\[
    \Delta_{\alpha\beta}\mathcal{G}
    :=
    \mathcal{G}_{\alpha\beta}^{11}-\mathcal{G}_{\alpha\beta}^{10}-\mathcal{G}_{\alpha\beta}^{01}+\mathcal{G}_{\alpha\beta}^{00},
\]
where $\mathcal{G}_{\alpha\beta}^{ab}$ denotes the value of $\mathcal{G}$ when $A_\alpha=a$ and $A_\beta=b$,
with all other entries of $A$ held fixed. 
The following lemma shows that the second discrete differences control the higher-order variance.

\begin{lemma}[Higher-order variance bound by discrete differences]\label{lem:anova-bound}
We have
\[
    \Var(\mathcal{G}^{(2+)})
    \leq \frac{1}{2}
    \sum_{\alpha\neq\beta}
    \sigma_\alpha^2\sigma_\beta^2
    \|\Delta_{\alpha\beta}\mathcal{G}\|_2^2    \leq
    \frac{1}{2}\sum_{\alpha\neq\beta}
    \sigma_\alpha^2\sigma_\beta^2
    \|\Delta_{\alpha\beta}\mathcal{G}\|_\infty^2.
\]
\end{lemma}

\begin{proof}
It is sufficient to show the first inequality because the second inequality is trivial. Since $\mathcal{G} = \E \mathcal{G} + \sum_{I\neq\emptyset} a_I Z_I$, we first observe how $\Delta_{\alpha\beta}$ acts on a basis vector $Z_I$. Recall that, for $\alpha\neq\beta$, we have
\[
    \Delta_{\alpha\beta}Z
    :=
    Z_{\alpha\beta}^{11}-Z_{\alpha\beta}^{10}-Z_{\alpha\beta}^{01}+Z_{\alpha\beta}^{00},
\]
where $Z_{\alpha\beta}^{ab}$ denotes the value of $Z$ when $A_\alpha=a$ and $A_\beta=b$,
with all other entries of $A$ held fixed.
If
$I$ does not contain both $\alpha$ and $\beta$, then $Z_I$ is independent of at
least one of the two coordinates being differenced, and
\[
    \Delta_{\alpha\beta}Z_I=0.
\]
If $I\supseteq\{\alpha,\beta\}$, then $Z_I=Z_\alpha Z_\beta Z_{I\setminus\{\alpha,\beta\}}$ and 
\[
    \Delta_{\alpha\beta}Z_I=Z_{I\setminus\{\alpha,\beta\}}\Delta_{\alpha\beta}(Z_\alpha Z_\beta).
\]
Since
\begin{eqnarray*}
\Delta_{\alpha\beta}(Z_\alpha Z_\beta) =  \frac{1-P_\alpha}{\sigma_\alpha}\cdot\frac{1-P_\beta}{\sigma_\beta} - \frac{1-P_\alpha}{\sigma_\alpha}\cdot\frac{-P_\beta}{\sigma_\beta} - \frac{-P_\alpha}{\sigma_\alpha}\cdot\frac{1-P_\beta}{\sigma_\beta} + \frac{-P_\alpha}{\sigma_\alpha}\cdot\frac{-P_\beta}{\sigma_\beta} = \frac{1}{\sigma_\alpha\sigma_\beta},   
\end{eqnarray*}
it follows that $\Delta_{\alpha\beta}Z_I
=(\sigma_\alpha\sigma_\beta)^{-1}Z_{I\setminus\{\alpha,\beta\}}$.
Applying this to the Hoeffding expansion gives
\[
    \Delta_{\alpha\beta}\mathcal{G}
    =\frac{1}{\sigma_\alpha\sigma_\beta}
      \sum_{I\supseteq\{\alpha,\beta\}}a_I
      Z_{I\setminus\{\alpha,\beta\}}.
\]
By orthonormality of $Z_I$,
\[
    \sigma_\alpha^2\sigma_\beta^2
    \|\Delta_{\alpha\beta}\mathcal{G}\|_2^2
    =\sum_{I\supseteq\{\alpha,\beta\}}a_I^2.
\]
Summing over all distinct pairs $\alpha\neq\beta$, we obtain
\[
\begin{aligned}
    \sum_{\alpha\neq\beta}
    \sigma_\alpha^2\sigma_\beta^2
    \|\Delta_{\alpha\beta}\mathcal{G}\|_2^2
    =\sum_{\alpha<\beta}
      \sum_{I\supseteq\{\alpha,\beta\}}a_I^2 
    =\sum_{|I|\geq2}|I|(|I|-1)a_I^2\ge 2\sum_{|I|\geq2}a_I^2.
\end{aligned}
\]
Since $\Var(\mathcal{G}^{(2+)})= \sum_{|S|\geq2}a_S^2$, the claim of the lemma follows.
\end{proof}

Next, we bound the discrete differences of $\mathcal{G}=\mathcal{G}(E)$ by the derivatives of $\mathcal{G}$.

\begin{lemma}[Discrete difference as an integrated edge derivative]\label{lem:discrete-integral}
For two distinct edges $\alpha$ and $\beta$, we have
\[
    \Delta_{\alpha\beta}\mathcal{G}
    =\int_0^1\int_0^1
\partial_\alpha\partial_\beta \mathcal{G}\bigl(E_{\alpha\beta}^{00}+se_\alpha+te_\beta\bigr)
      \, ds \, dt.
\]
Consequently,
\[
    |\Delta_{\alpha\beta}\mathcal{G}|
    \leq
    \sup_{z\in\R^{n\times n}}
\left|\partial_\alpha\partial_\beta \mathcal{G}\left(z\right)\right|\le\sup_{x\in\R^n} \sum_{i\in\alpha}\sum_{j\in\beta}|\mathcal{F}_{ij}(x)|.
\]
\end{lemma}

\begin{proof}
Denote $g(s,t):=\mathcal{G}(E_{\alpha\beta}^{00}+se_\alpha+te_\beta)$. Then
\[
    \Delta_{\alpha\beta}\mathcal{G}
    =g(1,1)-g(1,0)-g(0,1)+g(0,0).
\]
By applying the one-dimensional fundamental theorem of calculus first in
$t$ and then in $s$,
\[
\begin{aligned}
g(1,1)-g(1,0)-g(0,1)+g(0,0)
=\int_0^1
      \left\{\frac{\partial g}{\partial t}(1,t)
      -\frac{\partial g}{\partial t}(0,t)\right\}\,d t  
=\int_0^1\int_0^1
      \frac{\partial^2 g}{\partial s\partial t}(s,t)
      \,d s \,d t.
\end{aligned}
\]
Since
\[
    \frac{\partial^2 g}{\partial s\partial t}(s,t)
    =\partial_\alpha\partial_\beta \mathcal{G}(E_{\alpha\beta}^{00}+se_\alpha+te_\beta),
\]
the claim of the lemma follows from \eqref{eq:localized-second-derivative-connection}.
\end{proof}

Equipped with Lemma~\ref{lem:anova-bound}
, we are ready to bound the variance of $\mathcal{G}^{2+}$

\begin{lemma}[Higher-order variance bound]\label{lemma:Wge2}
Under the regular Chung--Lu assumptions,
\[
    \Var\left(\mathcal{G}^{(2+)}\right)\lesssim n^{2v-3}p^{2e-2}.
\]
\end{lemma}

\begin{proof}[Proof of Lemma~\ref{lemma:Wge2}]
By Lemma~\ref{lem:anova-bound} and the fact that $\sigma_\alpha^2=P_\alpha(1-P_\alpha)\lesssim p$, we have
\[
    \Var\left(\mathcal{G}^{(2+)}\right)
    \leq \frac{1}{2}
    \sum_{\alpha\neq\beta}
    \sigma_\alpha^2\sigma_\beta^2
    \|\Delta_{\alpha\beta}\mathcal{G}\|_\infty^2 \lesssim p^2\sum_{\alpha\neq\beta} \|\Delta_{\alpha\beta}\mathcal{G}\|_\infty^2 .
\]
By Lemma~\ref{lem:discrete-integral} and \eqref{eq:localized-second-derivative-connection}, we have 
$$
\|\Delta_{\alpha\beta}\mathcal{G}\|_\infty^2 \le \sup_{z\in\R^{n\times n}}
\left|\partial_\alpha\partial_\beta \mathcal{G}\left(z\right)\right| \le \sup_{x\in\R^n} \sum_{i\in\alpha}\sum_{j\in\beta}
|\mathcal F_{ij}(x)|.
$$
It then follows from 
Lemma~\ref{lem:localized-derivative-bounds} that
\[
\begin{aligned}
    \Var\left(\mathcal{G}^{(2+)}\right)
    &\lesssim
    p^2
    \sum_{\alpha\neq\beta}
    \left(
      n^{v-4}p^{e-2}
      +\mathbf 1_{\{\alpha\cap\beta\neq\varnothing\}}
       n^{v-3}p^{e-2}
    \right)^2.
\end{aligned}
\]
Since there are $O(n^4)$ unordered pairs of disjoint edges and $O(n^3)$ unordered
pairs of distinct edges that share a vertex, 
\[
\begin{aligned}    \Var\left(\mathcal{G}^{(2+)}\right)
    \lesssim
    p^2
    \left[
    n^4(n^{v-4}p^{e-2})^2
    +n^3(n^{v-3}p^{e-2})^2
    \right]
    \le 2n^{2v-3}p^{2e-2}.
\end{aligned}
\]
This proves the claim.
\end{proof}

\subsection{Leading Term of the Variance}

The leading term in the Hoeffding decomposition of $\mu_R(\widehat{P})$ is
$\sum_\alpha c_\alpha E_\alpha$. According to
Section~\ref{sec:first order coefficient approximation global}, the
coefficients $c_\alpha$ can be approximated by $\partial_\alpha \mathcal{G}(0)
    =
    \sum_{i\in\alpha} F_i(d)$.
Therefore, the leading term of the variance of $\mu_R(\widehat{P})$ is
\[
    V_R: = \sum_{i<j} P_{ij}(1-P_{ij})(F_i(d)+F_j(d))^2.
\]
The following lemma gives the order of this term.

\begin{lemma}[Scale of the leading variance]
\label{lemma:scale of V_R}
We have
\[
\sum_{i<j}P_{ij}(1-P_{ij})(F_i(d)+F_j(d))^2 \asymp n^{2v-2}p^{2e-1}.
\]
\end{lemma}

\begin{proof}[Proof of Lemma~\ref{lemma:scale of V_R}]
By Lemma~\ref{lem:localized-derivative-bounds}, we have $F_i(d)\lesssim n^{v-2}p^{e-1}$.  
Therefore,
\[
V_R=\sum_{i<j}P_{ij}(1-P_{ij})(F_i(d)+F_j(d))^2\lesssim
n^2p\left(n^{v-2}p^{e-1}\right)^2
=
n^{2v-2}p^{2e-1}.
\]
To obtain a lower bound for $V_R$, we first note that, since $P_{ij}$ is bounded
away from one, the Cauchy--Schwarz inequality gives
\[
V_{R}
\gtrsim 
\sum_{i<j}P_{ij}\left(F_i(d)+F_j(d)\right)^2\ge \Bigg(\sum_{i<j}P_{ij}\Bigg)^{-1}
\Bigg(\sum_{i<j}P_{ij}(F_i(d)+F_j(d))\Bigg)^2
\]
We use the homogeneity of $F$ to bound the second factor on the right-hand side.
Recall from \eqref{eq:F definition} that, for a fixed subgraph $R$,
\[
F(d)
=
C_R s(d)^{-e}H(x), \qquad C_R = \frac{1}{|\operatorname{Aut}(R)|}
\left(\frac{n}{n-1}\right)^e, 
\qquad H(d)
=
\sum_{\phi:V(R)\hookrightarrow[n]}
\prod_{a\in V(R)}d_{\phi(a)}^{r_a}.
\]
where $s(d)=\sum_{i=1}^n d_i$. 
This formula shows that $F$ is homogeneous of degree $e$, that is, $F(td)=t^eF(d)$ hold for every $t>0$. Differentiating both sides with respect to $t$ and setting $t=1$, we obtain the Euler's identity
\[
\sum_{i=1}^n d_iF_{i}(d)
=
eF(d).
\]
Since $F(d)\asymp n^vp^e$, 
this implies
$$
\sum_{i=1}^n d_iF_{i}(d) \asymp n^vp^e.  
$$
Therefore,
\[
\sum_{i<j}P_{ij}(F_i(d)+F_j(d))
=
\sum_i F_i(d)\sum_{j\neq i}P_{ij}
=
\sum_i d_i F_i(d) 
\asymp n^v p^e.
\]
Since $\sum_{i<j}P_{ij}\asymp n^2p$, it follows that 
\[
V_{R}
\gtrsim (n^2p)^{-1}(n^vp^e)^2
=
n^{2v-2}p^{2e-1}.
\]
This proves the claim of the lemma.
\end{proof}

By Lemma~\ref{lemma:scale of V_R} and the preceding discussion, we obtain
\begin{equation}
\label{eq:variace leading term global}
\Var\left(\mu_R(\widehat{P})\right)= V_R+ O\left(n^{2v-3}p^{2e-2}\right) = 
V_R \left[ 1+ O\left(\frac{1}{np}\right) \right].
\end{equation}
The following lemma shows that $V_R$ can be written
in terms of subgraph means. We will use this result to show that the second
bootstrap step can be used to estimate the variance of $\mu_R(\widehat{P})$.

\begin{lemma}[First derivative as a function of subgraph means]
\label{lemma:first derivarive in terms of rooted motif means} 
We have
$$
F_i(d) = \frac{1+O\left(n^{-1}\right)}{d_i}\sum_{a\in V(R)} \frac{r_a|\operatorname{Aut}(R,a)|}{|\operatorname{Aut}(R)|}  \mu_{R,a}^{(i)}(P) - \frac{e\left[1+O\left(n^{-1}\right)\right]}{s(d)}\mu_R(P). 
$$
\end{lemma}

\begin{proof}[Proof of Lemma~\ref{lemma:first derivarive in terms of rooted motif means}]
Recall from \eqref{eq:F definition} that for a fixed subgraph $R$, 
\[
F(d)
=
C_R s(d)^{-e}H(d), \qquad C_R = \frac{1}{|\operatorname{Aut}(R)|}
\left(\frac{n}{n-1}\right)^e, 
\qquad H(d)
=
\sum_{\phi:V(R)\hookrightarrow[n]}
\prod_{a\in V(R)}d_{\phi(a)}^{r_a}.
\]
where $s(d)=\sum_{i=1}^n d_i$. 
Differentiating $F(d)$ with respect to $d_i$ gives
$$
F_i(d) = C_R s(d)^{-e}\sum_{\phi:V(R)\hookrightarrow[n]}\left[
\sum_{a\in V(R)}
\frac{r_a\mathbf 1\{\phi(a)=i\}}{d_i}
\right]
\prod_{b\in V(R)}d_{\phi(b)}^{r_b}-eC_R s(d)^{-e-1}H(d).
$$
Note that the last term on the right-hand side is $-eF(d)/s(d)$. 
To express $F_i(d)$ in terms of subgraph means, we consider the following
degree-induced Chung--Lu model, whose edge-probability matrix $\widetilde{P}$ is given by
\[
\widetilde{P}_{ij}
:=
\frac{n}{n-1}\frac{d_id_j}{s(d)} = \frac{s(d)}{n(n-1)}\cdot\frac{nd_i}{s(d)}\cdot\frac{nd_j}{s(d)},
\qquad i\neq j.
\]
We claim that $F(d)=\mu_R(\widetilde{P})$.  
Indeed, under this model,  
$$
\widetilde{p}= \frac{s(d)}{n(n-1)}, \qquad \widetilde{\theta}_i = \frac{nd_i}{s(d)}, \qquad 
\widetilde{d}_i=\sum_{j=1}^n \widetilde{P}_{ij}=\frac{nd_i}{n-1}, \qquad s(\widetilde{d})=\sum_{i=1}^n\widetilde{d}_i=\frac{ns(d)}{n-1}.$$ 
Indeed, by \eqref{eq:mu P} and \eqref{eq:F definition}, we have
\begin{equation*}
\mu_{R}(\widetilde{P})
=
\frac{1}{|\operatorname{Aut}(R,r)|}\left(\frac{s(d)}{n(n-1)}\right)^{e}
\sum_{\substack{\phi:V(R)\hookrightarrow[n] \\ \phi(r)=i}} \left(\frac{n}{s(d)}\right)^{2e} 
\prod_{a\in V(R)}
d_{\phi(a)}^{r_a}=F(d).
\end{equation*}
In addition, for a motif vertex \(a\in V(R)\) and a graph vertex \(i\in[n]\), recall from \eqref{eq:rooted mu Phat} that the rooted motif mean under $\widetilde{P}$ is  
\[
\mu_{R,a}^{(i)}(\widetilde{P})
=
\frac{1}{|\operatorname{Aut}(R,a)|}
\left(\frac{n}{n-1}\right)^e s(d)^{-e}\sum_{\substack{\phi:V(R)\hookrightarrow[n] \\ \phi(r)=i}} \ 
\prod_{b\in V(R)}
d_{\phi(b)}^{r_b}.
\]
It follows that
$$
F_i(d) = \frac{1}{d_i}\sum_{a\in V(R)} \frac{r_a|\operatorname{Aut}(R,a)|}{|\operatorname{Aut}(R)|}  \mu_{R,a}^{(i)}(\widetilde{P}) - \frac{e}{s(d)}\mu_R(\widetilde{P}). 
$$
It remains to approximate $\widetilde{P}$ and its associated means by $P$ and
the corresponding means under $P$. Since $P_{ij}=p\theta_i\theta_j$ and $\sum_{i=1}^n \theta_i=n$, we have 
$$
d_i = np\theta_i\left[1+O(n^{-1})\right], \qquad s(d) = n^2p\left[1+O(n^{-1})\right]. 
$$
It follows that $\widetilde{P}_{ij} = P_{ij}(1+O(1/n))$, and hence, 
$$
\mu_{R,a}^{(i)}(\widetilde{P}) = \mu_{R,a}^{(i)}(P)\left[1+O(n^{-1})\right], \qquad \mu_{R}(\widetilde{P}) = \mu_{R}(P)\left[1+O(n^{-1})\right].
$$
The proof is complete. 
\end{proof}



\subsection{Variance Approximation via the Second-Level Bootstrap}
Since the variance of $\mu_R(\widehat{P})$ is not directly available, we estimate
it using the second-level bootstrap estimator
$\Var_{\widehat{P}}(\mu_R(\widehat{\widehat{P}}))$. The following lemma shows
the consistency of this estimator.

\begin{lemma}[Consistency of second-level bootstrap variance estimation]
\label{lemma:variance estimation via second layer boostrap global}  We have 
\begin{equation*}
\Var_{\widehat{P}}((\mu_R(\widehat{\widehat{P}})) =  \Var((\mu_R(\widehat P))\left[1+O\left((np)^{-1/2}\right)\right].    
\end{equation*}
\end{lemma}

\begin{proof}[Proof of Lemma~\ref{lemma:variance estimation via second layer boostrap global}]

Recall from \eqref{eq:variace leading term global} that 
\[
\Var((\mu_R(\widehat P)) = V_R\left[1+O\left(\frac{1}{np}\right)\right], \qquad
V_R
=
\sum_{i<j}
P_{ij}(1-P_{ij})
\left[F_i(d)+F_j(d)\right]^2.
\]
We now study the corresponding second-layer bootstrap estimate of $\Var((\mu_R(\widehat P))$.  Conditioning on the regular-degree event $\Omega_n$ defined in \eqref{eq:regular-degree-event}, we have   
\[
\Var_{\widehat{P}}((\mu_R(\widehat{\widehat{P}})) = \widehat{V}_R\left[1+O\left(\frac{1}{n\widehat{p}}\right)\right], \qquad
\widehat{V}_R
=
\sum_{i<j}
\widehat{P}_{ij}(1-\widehat{P}_{ij})
\left[F_i(\widehat{d}\,)+F_j(\widehat{d}\,)\right]^2,
\]
where $\widehat{d}=(\widehat{d}_1,...,\widehat{d}_n)$, and by \eqref{eq: main parameter estimators},
\begin{equation}
\label{eq:di hat approximation}
\widehat{d}_i = \sum_{j\neq i} \widehat{P}_{ij} =
\frac{n}{n-1}D_i
\left(1-\frac{D_i}{S}\right)
=D_i\left[1+O(n^{-1})\right] = d_i\left[1+O\left((np)^{-1/2}\right)\right].
\end{equation}
Note that, on the event $\Omega_n$, the degree vector is regular and hence
$F=\mathcal{F}$. By Taylor's expansion and the derivative bounds in
Lemma~\ref{lem:localized-derivative-bounds}, we have
\[
\left|F_i(\widehat{d}\,)-F_i(d)\right|
\le
\sum_{k=1}^n
\sup_x |\mathcal F_{ik}(x)|\,\left|\widehat{d}_k-d_k\right|\lesssim n^{v-3}p^{e-2}(np)^{1/2}=n^{v-5/2}p^{e-3/2}.
\]
Denote $h_{ij} := F_i(d)+F_j(d)$ and $\widehat{h}_{ij} := F_i(\widehat{d} \, )+F_j(\widehat{d} \, )$. Then $|h_{ij} - \widehat{h}_{ij}|\lesssim n^{v-5/2}p^{e-3/2}$, and therefore
$$
R_n := \sum_{i<j} P_{ij}(1-P_{ij})\left(h_{ij}-\widehat{h}_{ij}\right)^2 \lesssim pn^2 n^{2v-5}p^{2e-3} = n^{2v-3}p^{2e-2}.
$$
Since
\(
V_R\asymp n^{2v-2}p^{2e-1}
\)
by Lemma~\ref{lemma:scale of V_R}, 
this gives 
$R_n = O(V_R(np)^{-1})$.
By Cauchy--Schwarz inequality,
\[
\Bigg|
\sum_{i<j}
P_{ij}(1-P_{ij})
\left(\widehat{h}_{ij}^2-h_{ij}^2\right)
\Bigg|
\le
2V_R^{1/2}R_n^{1/2}+R_n.
\]
Using $R_n = O(V_R(np)^{-1})$, we get
\begin{equation*}
\label{eq:score-comparison}
\sum_{i<j}
P_{ij}(1-P_{ij})\widehat{h}_{ij}^2
=
V_R\left[1+O((np)^{-1/2})\right].
\end{equation*}
By \eqref{eq: main parameter estimators}, we have
\begin{equation}
\label{eq:Pij approximation}
\widehat P_{ij}(1-\widehat P_{ij})
=
P_{ij}(1-P_{ij})
\left[1+O((np)^{-1/2})\right].
\end{equation}
uniformly over \(i<j\).  Combining these estimates, we obtain
\begin{equation*}
\label{eq:VRstar-VR}
\widehat{V}_R
=
V_R\left[1+O\left((np)^{-1/2}\right)\right],
\end{equation*}
and the claim of the lemma follows.
\end{proof}

\subsection{Rooted Subgraph Counts}
\label{sec:rooted count variance}

In this section we calculate the variance of the plugin estimators of the rooted subgraph means. Recall from \eqref{eq:rooted mu Phat} that
$$
\mu_{R,\rho}^{(i)}(\widehat P)
=
\frac{1}{|\operatorname{Aut}(R,\rho)|}
\left(\frac{n}{n-1}\right)^e
S^{-e}
\sum_{\substack{\phi:V(R)\hookrightarrow[n] \\ \phi(\rho)=i}} \
\prod_{a\in V(R)}
D_{\phi(a)}^{r_a}.    
$$
Similar to \eqref{eq:F definition}, we denote
\begin{equation}
\label{eq:F definition rooted}
    F^{(i)}_\rho(x)
    =
    \frac{1}{|\operatorname{Aut}(R,\rho)|}
    \left(\frac{n}{n-1}\right)^e
    s(x)^{-e}
\sum_{\substack{\phi:V(R)\hookrightarrow[n] \\ \phi(\rho)=i}} \
    \prod_{a\in V(R)}
    x_{\phi(a)}^{r_a},
\end{equation}
In addition, let 
$$
G_\rho^{(i)}(z):=F_\rho^{(i)}(d+Bz),
\qquad z\in\mathbb R^{E(K_n)},
$$
where $B$ is the node-edge incident matrix. 
Similar to the case of global subgraph counts, we have
\begin{eqnarray}
\label{eq:rooted first derivarive connection}
\partial_\alpha G_\rho^{(i)}(z)
&=& \sum_{j\in\alpha}\partial_j F^{(i)}_\rho(d+Bz),\\
\label{eq:rooted second derivative connection}
\partial_\alpha\partial_\beta G_\rho^{(i)}(z)
&=&
\sum_{j\in\alpha}\sum_{k\in\beta}\partial_j\partial_k F_\rho^{(i)}(d+Bz),\\
\label{eq:rooted third derivative connection}
\partial_\alpha\partial_\beta\partial_\gamma G_\rho^{(i)}(z)
&=&
\sum_{j\in\alpha}
\sum_{k\in\beta}
\sum_{\ell\in\gamma}
\partial_j\partial_k\partial_\ell F_\rho^{(i)}(d+Bz).
\end{eqnarray}
To control the regularity of the degree vector, we use the same truncation function $\tau$ as defined in Section~\ref{sec:truncation global}, and define the 
degree-truncated versions of $F_\rho^{(i)}$ and $G_\rho^{(i)}$, as
\begin{equation}
\label{eq:truncated FG}
\mathcal F_\rho^{(i)}(x)
:= 
F_\rho^{(i)}\bigl(\tau(x_1),\ldots,\tau(x_n)\bigr),\qquad
\mathcal G_\rho^{(i)}(z)
= \mathcal{F}_\rho^{(i)}(d+Bz).    
\end{equation}
By definition, $F_\rho^{(i)}(D)=\mathcal{F}_\rho^{(i)}(D)$ and $G_\rho^{(i)}(E)=\mathcal{G}_\rho^{(i)}(E)$ conditional on $\Omega_n$ from \eqref{eq:regular-degree-event}. Also, similar to \eqref{eq:rooted first derivarive connection},
\eqref{eq:rooted second derivative connection}, and
\eqref{eq:rooted third derivative connection}, for any $\alpha,\beta,\gamma\in E(K_n)$, we have 
\begin{eqnarray}
\label{eq:rooted localized-first-derivative-connection}
\partial_\alpha \mathcal{G}_\rho^{(i)}(z)
&=&
\sum_{j\in\alpha}\partial_j\mathcal F_\rho^{(i)}(d+Bz),\\
\label{eq:rooted localized-second-derivative-connection}
\partial_\alpha\partial_\beta \mathcal G_\rho^{(i)}(z)
&=&
\sum_{j\in\alpha}\sum_{k\in\beta}
\partial_j\partial_k\mathcal F_\rho^{(i)}(d+Bz),\\
\label{eq:rooted localized-third-derivative-connection}
\partial_\alpha\partial_\beta\partial_\gamma \mathcal G_\rho^{(i)}(z)
&=&
\sum_{j\in\alpha}\sum_{k\in\beta}\sum_{\ell\in\gamma}
\mathcal \partial_j\partial_k\partial_\ell F_\rho^{(i)}(d+Bz).
\end{eqnarray}
We now summarize the results for rooted subgraph counts in parallel with those for global subgraph counts in the previous sections.

\paragraph{Equivalent variances.}
Since $\Prob(\Omega_n^c)$ decays faster than $n^{-\kappa}$ for every
$\kappa>0$, an analogue of Lemma~\ref{lemma:equivalent variances} holds: 
\begin{equation}
\label{eq:rooted equivalent variances}  
|\Var(F_\rho^{(i)}(D))-\Var(\mathcal{F}_\rho^{(i)}(D))|\lesssim n^{-\kappa}\big[1+\Var(\mathcal{F}_\rho^{(i)}(D))\big].
\end{equation}

\paragraph{Hoeffding decomposition.}
In the case of global counts, the variance calculation starts with the
Hoeffding decomposition in Section~\ref{sec:Hoeffding decomposition global}.
This decomposition holds for any function of $E$ and, in particular, applies to
$\mathcal{G}_\rho^{(i)}(E)$:
\begin{equation}
\label{eq:rooted Hoeffding decomposition}
\mathcal{G}_\rho^{(i)}(E)-\E \mathcal{G}_\rho^{(i)}(E) = \sum_{\alpha}c_\alpha E_\alpha + 
\left(\mathcal{G}_\rho^{(i)}\right)^{(2+)},
\end{equation}
where  
\begin{eqnarray}
c_\alpha = \sigma_\alpha^{-2} \E\left[\mathcal{G}_\rho^{(i)}(E) E_\alpha\right]
= \E\left[\mathcal{G}_\rho^{(i)}(E)| A_\alpha=1\right]-\E\left[\mathcal{G}_\rho^{(i)}(E)| A_\alpha=0\right], 
\label{eq:rooted c alpha def}
\end{eqnarray}
and $\big(\mathcal{G}_\rho^{(i)}\big)^{(2+)}$ denotes the higher-order terms in the
decomposition, which are orthogonal to $E_\alpha$. 
Our goal is to show that the second term is negligible and that, for the first
term, $c_\alpha$ can be approximated by
$\partial_\alpha \mathcal{G}_\rho^{(i)}(0)
    =
\sum_{j\in\alpha}\partial_j\mathcal{F}_\rho^{(i)}(d)$. 
To this end, we need to bound the derivatives of $\mathcal{F}_\rho^{(i)}$ at $d$, which we
do next.

\paragraph{Derivative bounds.} 
Analogously to Lemma~\ref{lem:localized-derivative-bounds} for global counts,
we have the following result for rooted subgraph counts.

\begin{lemma}[Derivative bounds for rooted counts]
\label{lem:rooted-derivative-envelope}
For \(m\le 3\) and \(k_1,\ldots,k_m\in[n]\), denote
\[
\nu_i(k_1,\ldots,k_m)
:=
\left|\{k_1,\ldots,k_m\}\setminus\{i\}\right|.
\]
Then, uniformly over all \(x\in\mathbb R^n\),
\[
\left|
\partial_{k_1}\cdots\partial_{k_m}
\mathcal F_{\rho}^{(i)}(x)
\right|
\lesssim
n^{v-1-m-\nu_i(k_1,\ldots,k_m)}p^{e-m}.
\]
\end{lemma}

\begin{proof}[Proof of Lemma~\ref{lem:rooted-derivative-envelope}]
We first prove the corresponding bound for 
\(F_{\rho}^{(i)}\) defined by \eqref{eq:F definition rooted} on the degree-regular region:  let \(x\in\R^n\) satisfying $x_\ell\asymp np$ for all $\ell\in[n]$ and
$s(x)\asymp n^2p$. 
The rooted polynomial in the formula of $F_\rho^{(i)}$ is
\[
H(x)
=
\sum_{\substack{\phi:V(R)\hookrightarrow[n]\\ \phi(\rho)=i}} \ 
\prod_{a\in V(R)}y_{\phi(a)}^{r_a}.
\]
Since the root \(\rho\) is already fixed to the graph vertex \(i\), there are
only \(v-1\) free motif vertices.  If we differentiate $H$ with respect to
\(k_1,\ldots,k_m\), the number of additional graph vertices that must be fixed
is
\[
\nu_i(k_1,\ldots,k_m)
=
\left|\{k_1,\ldots,k_m\}\setminus\{i\}\right|.
\]
Since each derivative removes one degree factor,
\[
\left|
\partial_{k_1}\cdots\partial_{k_m}
H(x)
\right|
\lesssim
n^{v-1-\nu_i(u_1,\ldots,u_m)}(np)^{2e-m}.
\]
For some constant $C_{R,\rho}$, we can write
\[
F_{\rho}^{(i)}(x)=C_{R,\rho}s(x)^{-e}H(x).
\]
A product-rule term in an \(m\)-th derivative places \(\ell\) derivatives on
\(s(x)^{-e}\) and \(m-\ell\) derivatives on \(H(x)\).  The
\(\ell\)-th derivative of \(s(x)^{-e}\) is \(O((n^2p)^{-e-\ell})\).  If the $m-\ell$
derivatives landing on \(H\) involve \(\nu_H\) distinct
coordinates different from \(i\), then the corresponding term is bounded by
\[
(n^2p)^{-e-\ell}
n^{v-1-\nu_H}(np)^{2e-(m-\ell)}=n^{v-1-m-\nu_H-\ell}p^{e-m}.
\]
Since the \(\ell\) derivatives landing on \(s(x)^{-e}\) can introduce at most
\(\ell\) additional non-root coordinates,
\[
\nu_i(k_1,\ldots,k_m)\le \nu_H+\ell.
\]
Therefore
\[
n^{v-1-m-\nu_H-\ell}p^{e-m}
\le
n^{v-1-m-\nu_i(u_1,\ldots,u_m)}p^{e-m}.
\]
This proves the claimed bound for \(F_{\rho}^{(i)}\) on the regular region.

For the degree-truncated version
\[
\mathcal F_{\rho}^{(i)}(x)
=
F_{\rho}^{(i)}(\tau(x_1),\ldots,\tau(x_n)),
\]
the vector $(\tau(x_1),\ldots,\tau(x_n))$
always lies in the regular region.  The chain rule introduces factors
\(\tau'\), \(\tau''\), and \(\tau'''\).  By construction,
\[
|\tau'|\lesssim 1,
\qquad
|\tau''|\lesssim (np)^{-1},
\qquad
|\tau'''|\lesssim (np)^{-2}.
\]
These factors have exactly the scale needed to preserve the claimed bounds. 
\end{proof}

\paragraph{Approximation of the coefficients of the first-order terms.} Since the coefficients $c_\alpha$ do not admit explicit formulas, we approximate
them by derivatives of $\mathcal{G}_\rho^{(i)}$. We use the bound \eqref{eq:general c approximation} in Lemma~\ref{lem:first-projection-error}, which holds for any smooth function, we have 
\begin{equation*}
\left|c_\alpha- \partial_\alpha \mathcal{G}_\rho^{(i)}(0)\right| \le \sup_{z\in\R^{E(K_n)}}|\partial_\alpha^2\mathcal{G}_\rho^{(i)}(z)| + \sum_{\beta\neq\alpha}
      \sigma_\beta^2      \sup_{z\in\R^{E(K_n)}}|\partial_\alpha\partial_\beta^2\mathcal{G}_\rho^{(i)}(z)|.
\end{equation*}
From \eqref{eq:rooted localized-second-derivative-connection} and Lemma~\ref{lem:rooted-derivative-envelope}, we get 
$$
\sup_{z\in\R^{E(K_n)}}|\partial_\alpha^2\mathcal{G}_\rho^{(i)}(z)| \lesssim n^{v-3}p^{e-2}\mathbf{1}_{\{i\in\alpha\}}+n^{v-4}p^{e-2}\mathbf{1}_{\{i\notin\alpha\}}.
$$
Similarly, from \eqref{eq:rooted localized-third-derivative-connection} and Lemma~\ref{lem:rooted-derivative-envelope}, and the fact that $\sigma_\beta^2\lesssim p$,  we have
$$
\sum_{\beta\neq\alpha}
      \sigma_\beta^2      \sup_{z\in\R^{E(K_n)}}|\partial_\alpha\partial_\beta^2\mathcal{G}_\rho^{(i)}(z)| \lesssim n^{v-3}p^{e-2} \mathbf{1}_{\{i\in\alpha\}} +  n^{v-4}p^{e-2} \mathbf{1}_{\{i\notin\alpha\}}. 
$$
Therefore,
$$
\left|c_\alpha- \partial_\alpha \mathcal{G}_\rho^{(i)}(0)\right| \lesssim n^{v-3}p^{e-2} \mathbf{1}_{\{i\in\alpha\}} +  n^{v-4}p^{e-2} \mathbf{1}_{\{i\notin\alpha\}}. 
$$
It then follows that
\begin{equation}
\label{eq:coeffcient approx rooted}
    \sum_\alpha \sigma_\alpha^2 \left(c_\alpha- \partial_\alpha \mathcal{G}_\rho^{(i)}(0)\right)^2\lesssim n^{2v-5}p^{2e-3}.
\end{equation}

\paragraph{Higher-order variance bound.}
Applying Lemmas~\ref{lem:anova-bound} and \ref{lem:discrete-integral} from
Section~\ref{sec:higher-order variance bound global}, which hold for arbitrary
smooth functions, we obtain
$$
\left(\mathcal{G}_\rho^{(i)}\right)^{(2+)} \le \sum_{\alpha\neq\beta}
    \sigma_\alpha^2\sigma_\beta^2
    \sup_{x\in\R^n} \sum_{k\in\alpha}\sum_{\ell\in\beta}|\partial_k\partial_\ell\mathcal{F}_\rho^{(i)}(x)|.
$$
Applying Lemma~\ref{lem:rooted-derivative-envelope} and splitting the analysis
into cases according to the value of $\nu_i(k,\ell)$, we obtain
$$
\sum_{k\in\alpha}\sum_{\ell\in\beta}|\partial_k\partial_\ell\mathcal{F}_\rho^{(i)}(x)|
\lesssim
p^{e-2}
\left[
n^{v-5}
+
n^{v-4}\mathbf 1\{\alpha\cap\beta\neq\varnothing
\text{ or } i\in\alpha\cup\beta\}
+
n^{v-3}\mathbf 1\{i\in\alpha\cap\beta\}
\right].
$$
Since \(\sigma_\alpha^2\lesssim p\), we count edge pairs as follows:
there are \(O(n^2)\) pairs of edges both incident to \(i\), \(O(n^3)\) pairs
with either exactly one root incidence or a shared non-root endpoint, and
\(O(n^4)\) remaining pairs.  Therefore,
\begin{equation}
\label{eq:higher order variance bound rooted}
\operatorname{Var}\left(\left(\mathcal G_{\rho}^{(i)}\right)^{(2+)}\right)
\lesssim
n^{2v-4}p^{2e-2}.    
\end{equation}

\paragraph{Leading term of the variance.} We first compute  $\sum_\alpha\sigma_\alpha^2(\partial_\alpha\mathcal{G}_\rho^{(i)}(0))^2 $ and its order. 
The leading term of the variance of $\mu_{R,\rho}^{(i)}(\widehat{P})$ will follow. 
To that end, we split $\sum_\alpha\sigma_\alpha^2(\partial_\alpha\mathcal{G}_\rho^{(i)}(0))^2 $ into two part depending on whether $\alpha$ contains $i$:
$$
\sum_\alpha\sigma_\alpha^2(\partial_\alpha\mathcal{G}_\rho^{(i)}(0))^2 = \sum_{\alpha\ni i} \sigma_\alpha^2(\partial_\alpha\mathcal{G}_\rho^{(i)}(0))^2 + \sum_{\alpha\not\ni i} \sigma_\alpha^2(\partial_\alpha\mathcal{G}_\rho^{(i)}(0))^2 =:Q_i+Q_{-i}.
$$
From \eqref{eq:rooted localized-first-derivative-connection}, we have 
$$
\partial_\alpha \mathcal{G}_\rho^{(i)}(0)
=
\sum_{k\in\alpha}\partial_k\mathcal F_\rho^{(i)}(d)=
\sum_{k\in\alpha}\partial_k F_\rho^{(i)}(d).$$
It follows from Lemma~\ref{lem:rooted-derivative-envelope} that, if
$i\notin \alpha$, then
$\partial_\alpha \mathcal{G}_\rho^{(i)}(0)\lesssim n^{v-3}p^{e-1}$. Hence,
\begin{eqnarray*}
Q_{-i} = \sum_{\alpha\not\ni i} \sigma_\alpha^2(\partial_\alpha\mathcal{G}_\rho^{(i)}(0))^2 \lesssim n^2p\left(n^{v-3}p^{e-1}\right)^2=n^{2v-4}p^{2e-1}.    
\end{eqnarray*}
Similarly, $\partial_k F_\rho^{(i)}(d) \lesssim n^{v-3}p^{e-1}$ if $k\ne i$ and $\partial_i F_\rho^{(i)}(d) \lesssim n^{v-2}p^{e-1}$. Therefore,  
\begin{eqnarray*}
Q_i &=& 
\sum_{k\neq i}
P_{ik}(1-P_{ik})\left(\partial_i F_\rho^{(i)}(d)+\partial_k F_\rho^{(i)}(d)\right)^2\\
&=& \left(\partial_i F_\rho^{(i)}(d)\right)^2\sum_{k\neq i}
P_{ik}(1-P_{ik})  + O(n^{2v-4}p^{2e-1}).  
\end{eqnarray*}
It remains to calculate $\partial_i F_\rho^{(i)}(d)$. Recall from \eqref{eq:F definition rooted} that  
\begin{equation*}
    F^{(i)}_\rho(d)
    = C_{R,\rho}
    s(d)^{-e} H(d), \quad 
    C_{R,\rho} =     \frac{1}{|\operatorname{Aut}(R,\rho)|}
    \left(\frac{n}{n-1}\right)^e, \quad H(d) = \sum_{\substack{\phi:V(R)\hookrightarrow[n] \\ \phi(\rho)=i}} \
    \prod_{a\in V(R)}
    d_{\phi(a)}^{r_a}.
\end{equation*}
By the chain rule,
\[
\partial_iF_{\rho}^{(i)}(d)
=
C_{R,\rho}s(d)^{-e}\partial_i H(d) -
\frac{e}{s(d)}F_{\rho}^{(i)}(d).
\]
Since $\phi$ is injective and  $\phi(\rho)=i$,
\[
\partial_i H(d) = \sum_{\substack{\phi:V(R)\hookrightarrow[n] \\ \phi(\rho)=i}} \partial_i
    \prod_{a\in V(R)}
    d_{\phi(a)}^{r_a} = \sum_{\substack{\phi:V(R)\hookrightarrow[n] \\ \phi(\rho)=i}} 
\frac{r_\rho}{d_i}
\prod_{a\in V(R)}d_{\phi(a)}^{r_a}=\frac{r_\rho}{d_i}H(d).
\]
Substituting this into the derivative formula and noting that
$F_\rho^{(i)}(d)\asymp n^{v-1}p^e$, we get
\[
\partial_i F_{\rho}^{(i)}(d)
=
\frac{r_\rho}{d_i}C_{R,\rho}s(d)^{-e}
H(d)
-
\frac{e}{s(d)}\Phi_{\rho}^{(i)}(d) = \left(\frac{r_\rho}{d_i}-\frac{e}{s(d)}\right)F_{\rho}^{(i)}(d)\asymp n^{v-2}p^{e-1}.
\]
Combining \eqref{eq:rooted equivalent variances}, \eqref{eq:rooted Hoeffding decomposition}, \eqref{eq:coeffcient approx rooted}, \eqref{eq:higher order variance bound rooted}, and the bounds for $Q_i$ and $Q_{-i}$, we obtain
\begin{equation}
\label{eq:rooted variance order}
\Var\left(\mu_{R,\rho}^{(i)}(\widehat{P})\right) = \left[\left(\frac{r_\rho}{d_i}-\frac{e}{s(d)}\right)F_{\rho}^{(i)}(d)\right]^2\sum_{k\neq i}P_{ik}(1-P_{ik}) + O(n^{2v-4}p^{2e-2})\asymp n^{2v-3}p^{2e-1}.    
\end{equation}

For the analysis of the second layer of bootstrap, we also express the rooted
degree functional \(F_r^{(i)}(d)\) as a rooted subgraph mean.  Consider the
degree-induced Chung--Lu model, whose edge-probability matrix $\widetilde{P}$ is given by
\[
\widetilde{P}_{ij}
:=
\frac{n}{n-1}\frac{d_id_j}{s(d)} = \frac{s(d)}{n(n-1)}\cdot\frac{nd_i}{s(d)}\cdot\frac{nd_j}{s(d)},
\qquad i\neq j.
\]
Under this model,  
$$
\widetilde{p}= \frac{s(d)}{n(n-1)}, \qquad \widetilde{\theta}_i = \frac{nd_i}{s(d)}, \qquad 
\widetilde{d}_i=\sum_{j=1}^n \widetilde{P}_{ij}=\frac{nd_i}{n-1}, \qquad s(\widetilde{d})=\sum_{i=1}^n\widetilde{d}_i=\frac{ns(d)}{n-1}.$$ 
We claim that $F_\rho^{(i)}(d)=\mu_{R,\rho}^{(i)}(\widetilde{P})$.  
Indeed, by \eqref{eq:rooted mu P} and \eqref{eq:F definition rooted}, we have
\begin{equation*}
\mu_{R,\rho}^{(i)}(\widetilde{P})
=
\frac{1}{|\operatorname{Aut}(R,\rho)|}\left(\frac{s(d)}{n(n-1)}\right)^{e}
\sum_{\substack{\phi:V(R)\hookrightarrow[n] \\ \phi(\rho)=i}} \left(\frac{n}{s(d)}\right)^{2e} 
\prod_{a\in V(R)}
d_{\phi(a)}^{r_a}=F_\rho^{(i)}(d).
\end{equation*}
It remains to approximate $\widetilde{P}$ and its associated means by $P$ and
the corresponding means under $P$. Since $P_{ij}=p\theta_i\theta_j$ and $\sum_{i=1}^n \theta_i=n$, we have 
$$
d_i = np\theta_i\left[1+O(n^{-1})\right], \qquad s(d) = n^2p\left[1+O(n^{-1})\right]. 
$$
It follows that $\widetilde{P}_{ij} = P_{ij}(1+O(1/n))$, and hence, 
$$F_\rho^{(i)}(d)=\mu_{R,\rho}^{(i)}(\widetilde{P}) = \mu_{R,\rho}^{(i)}(P)\left[1+O(n^{-1})\right].$$ 
From \eqref{eq:rooted variance order}, we obtain
\begin{equation}
\label{eq:rooted variance by rooted mean}
\Var\left(\mu_{R,\rho}^{(i)}(\widehat{P})\right) =\left(\frac{r_\rho\mu_{R,\rho}^{(i)}(P)}{d_i}\right)^2\sum_{k\neq i}P_{ik}(1-P_{ik}) + O(n^{2v-4}p^{2e-2}).    
\end{equation}

\paragraph{Variance approximation via the second-level bootstrap.}

Recall from \eqref{eq:rooted variance order} that
\[
\Var\left(\mu_{R,\rho}^{(i)}(\widehat P)\right)
=
V_{R,\rho}^{(i)}
+
O(n^{2v-4}p^{2e-2}),
\]
where
\[
V_{R,\rho}^{(i)}
:=
\left[
\left(
\frac{r_\rho}{d_i}
-
\frac{e}{s(d)}
\right)
F_{\rho}^{(i)}(d)
\right]^2
\sum_{k\neq i}P_{ik}(1-P_{ik}).
\]
Since $F_\rho^{(i)}(d)\asymp n^{v-1}p^e$ by \eqref{eq:F definition rooted}, it follows that $W_{R,\rho}^{(i)}
\asymp n^{2v-3}p^{2e-1}$, and therefore
\[
\Var\left(\mu_{R,\rho}^{(i)}(\widehat P)\right)
=
V_{R,\rho}^{(i)}
\left[
1+O\left(\frac1{np}\right)
\right].
\]
We now analyze the second-level bootstrap analogue. 
Conditioning on $\widehat{P}$ and the degree-regular event $\Omega_n$ defined in \eqref{eq:regular-degree-event}, we have
\[
\Var_{\widehat P}
\left(
\mu_{R,\rho}^{(i)}(\widehat{\widehat P})
\right)
=
\widehat V_{R,\rho}^{(i)}
\left[
1+O_p\left(\frac1{np}\right)
\right],
\]
where
\[
\widehat V_{R,\rho}^{(i)}
:=
\left[
\left(
\frac{r_\rho}{\widehat d_i}
-
\frac{e}{s(\widehat d \,)}
\right)
F_{\rho}^{(i)}(\widehat d \,)
\right]^2
\sum_{k\neq i}\widehat P_{ik}(1-\widehat P_{ik})\asymp n^{2v-3}p^{2e-1}.
\]
It remains to compare \(\widehat W_{R,\rho}^{(i)}\) with
\(W_{R,\rho}^{(i)}\).
From \eqref{eq: main parameter estimators}, \eqref{eq:Pij approximation}, and \eqref{eq:di hat approximation}, 
\begin{equation}
\label{eq:Phat dhat}
\widehat P_{k\ell}(1-\widehat P_{k\ell})
= P_{}(1-P_{k\ell})
\left[1+O((np)^{-1/2})\right], \qquad
\widehat{d}_k =  d_k\left[1+O\left((np)^{-1/2}\right)\right].
\end{equation}
In particular, $s(\widehat d) = \sum_{k=1}^n \widehat{d}_k
s(d)[1+O((np)^{-1/2})]$, and 
\begin{equation*}
\label{eq:rooted-coefficient-comparison}
\frac{r_\rho}{\widehat d_i}
-
\frac{e}{s(\widehat d \,)}
=
\left(
\frac{r_\rho}{d_i}
-
\frac{e}{s(d)}
\right)
\left[1+O((np)^{-1/2})\right].
\end{equation*}
Next, we compare \(F_{\rho}^{(i)}(\widehat d\,)\) and
\(F_{\rho}^{(i)}(d)\). Using Taylor's expansion, Lemma~\ref{lem:rooted-derivative-envelope}, and the bound $\widehat{d}_k-d_k\lesssim (np)^{1/2}$ from \eqref{eq:Phat dhat}, we get
\[
\left|
F_{\rho}^{(i)}(\widehat d\,)
-
F_{\rho}^{(i)}(d)
\right|
\le
\sum_{k=1}^n
\sup_x
\left|
\mathcal F_{\rho,k}^{(i)}(x)
\right|
\left|
\widehat d_k-d_k
\right|\lesssim n^{v-3/2}p^{e-1/2}.
\]
Therefore, 
\begin{equation}
\label{eq:degree and Fi}
\left(
\frac{r_\rho}{\widehat d_i}
-
\frac{e}{s(\widehat d\,)}
\right)
F_{\rho}^{(i)}(\widehat d\,)
=
\left(
\frac{r_\rho}{d_i}
-
\frac{e}{s(d)}
\right)
F_{\rho}^{(i)}(d)\left[1+O((np)^{-1/2})\right].
\end{equation}
Finally, we compare the edge-variance factor. From \eqref{eq:Phat dhat}, 
\begin{equation*}
\sum_{k\neq i}\widehat P_{ik}(1-\widehat P_{ik})
=
\sum_{k\neq i}P_{ik}(1-P_{ik})
\left[1+O((np)^{-1/2})\right].
\end{equation*}
Using \eqref{eq:degree and Fi}, we conclude that
\[
\widehat V_{R,\rho}^{(i)}
=
V_{R,\rho}^{(i)}
\left[1+O((np)^{-1/2})\right],
\]
and therefore
\begin{equation}
\label{eq:variance estimate via second layer bootstrap rooted}
\Var_{\widehat P}
\left(
\mu_{R,\rho}^{(i)}(\widehat{\widehat P})
\right)
=
\Var
\left(
\mu_{R,\rho}^{(i)}(\widehat P)
\right)
\left[
1+O_p((np)^{-1/2})
\right].    
\end{equation}


\subsection{Variance Approximation for Subgraph Counts}

In this section, we calculate the variances of the global count $T_R$ defined
in \eqref{eq:global count} and the rooted count defined in
\eqref{eq:rooted count}. We also show that these variances can be estimated by
the second-level bootstrap. The main observation is that the variance of a
subgraph count can be written as a finite sum of subgraph means, which can in
turn be estimated by the second-level bootstrap.

We start with the global count. Let $\mathcal I_R:=\{\phi:V(R)\hookrightarrow[n]\}$ be the set of injective embeddings of $R$. Recall from \eqref{eq:global count} that  
\[
T_R
=
\frac{1}{|\operatorname{Aut}(R)|}
\sum_{\phi\in\mathcal I_R}I_\phi, \qquad
I_\phi
:=
\prod_{\{a,b\}\in E(R)}
A_{\phi(a)\phi(b)}.
\]
Therefore, the variance of $T_R$ is
\[
\Var_P(T_R)
=
\frac{1}{|\operatorname{Aut}(R)|^2}
\sum_{\phi,\psi\in\mathcal I_R}
\Cov_P(I_\phi,I_\psi).
\]
For a fixed pair \((\phi,\psi)\), let \(E_\phi\) and \(E_\psi\) be the graph
edge sets used by the two embedded copies of $R$. If $E_\phi\cap E_\psi=\emptyset$
then \(I_\phi\) and \(I_\psi\) depend on disjoint sets of independent edges, so the covariance of $I_\phi$ and $I_\psi$ is zero.  
If \(E_\phi\cap E_\psi\neq\varnothing\), then
\[
\E_P[I_\phi I_\psi]
=
\prod_{e\in E_\phi\cup E_\psi}P_e, \qquad \E_P[I_\phi]\E_P[I_\psi]
=
\left(\prod_{e\in E_\phi\cup E_\psi}P_e\right)
\left(\prod_{e\in E_\phi\cap E_\psi}P_e\right).
\]
Therefore
\begin{equation}
\label{eq:covariance-overlap-form-careful}
\Cov_P(I_\phi,I_\psi)
=
\prod_{e\in E_\phi\cup E_\psi}P_e
\left[
1-\prod_{e\in E_\phi\cap E_\psi}P_e
\right].
\end{equation}
Because \(R\) is fixed, there are only finitely many possible isomorphism types
for the union of two copies of \(R\) sharing at least one edge.  Indeed, two
copies of \(R\) have at most \(2v\) vertices before identification, and there
are only finitely many ways to identify vertices of two fixed finite graphs
while preserving injectivity inside each copy.  Let
\[
\mathcal{O}_R := \{\text{overlap type}\}
\]
denote this finite collection of overlap types of two copies of $R$.  For
\(\omega\in\mathcal{O}_R\), let \(R_\omega^{+}\) be the simple union graph and let
\(R_\omega^-\) be the common edge subgraph.  Grouping the covariance sum by overlap type gives
\begin{equation}
\label{eq:global-variance-overlap-decomposition-careful}
\Var_P(T_R)
=
\sum_{\omega\in\mathcal{O}_R}
b_\omega\Psi_\omega(P),
\end{equation}
where \(b_\omega>0\) depends only on \(R\), and
\begin{eqnarray*}
\Psi_\omega(P)
&:=&
\sum_{\varphi:V(R_\omega^+)\hookrightarrow[n]}
\prod_{e\in E(R_\omega^+)}P_{\varphi(e)}
-\sum_{\varphi:V(R_\omega^+)\hookrightarrow[n]}
\prod_{e\in E(R_\omega^+)}P_{\varphi(e)}
\prod_{e\in E(R_\omega^-)}P_{\varphi(e)}\\
&=:&\Psi_\omega^+(P) - \Psi_\omega^-(P).     
\end{eqnarray*}
Here, we use $\varphi(e) =\{\varphi(a),\varphi(b)\}$ if $a,b\in V(R_\omega^+)$. 
We note that $\Psi_\omega^+(P)$ is, up to a constant factor, the expected number of
copies of $R_\omega^+$ and therefore can be handled by
Lemma~\ref{lem:fixed-simple-motif-stability}. The second sum under the
Chung--Lu model can be written as 
$$
\Psi_\omega^-(P) = 
\sum_{\varphi:V(R_\omega^+)\hookrightarrow[n]} p^{e(R_\omega^+)+e(R_\omega^-)}
\prod_{a\in V(R_\omega^+)}\theta_{\varphi(a)}^{r(R_\omega^+,a)+r(R_\omega^-,a)},
$$
where $e(R_\omega^\pm)$ are the number of edges of $R_\omega^\pm$ and
$r(R_\omega^\pm,a)$ are the degrees of $a$ as a vertex of $R_\omega^\pm$. This formula
can be viewed as the expected number of copies of a multigraph $\mathcal{R}_\omega$
obtained from $R_\omega^+$ by adding one extra copy of every edge in $R_\omega^-$.
Lemma~\ref{lem:multigraph-to-simple-lifting} makes this interpretation rigorous
and connects it to the bound in Lemma~\ref{lem:fixed-simple-motif-stability}.

\begin{lemma}[Plug-in stability for fixed simple subgraphs]
\label{lem:fixed-simple-motif-stability}
Let $R$ be a fixed connected simple graph, and let $\mu_R(P)$ and
$\mu_R(\widehat{P})$ denote the global subgraph mean and its plug-in estimator,
defined in \eqref{eq:mu P} and \eqref{eq:mu Phat}, respectively. Then
\[
\mu_R(\widehat P)
=
\mu_R(P)
\left[
1+O\left((np)^{-1}\right)
\right].
\]
Similarly, for fixed vertices $o\in V(R)$ and $i\in[n]$, let
$\mu_{R,o}^{(i)}(P)$ and $\mu_{R,o}^{(i)}(\widehat{P})$ denote the rooted
subgraph mean and its plug-in estimator, defined in \eqref{eq:rooted mu P} and
\eqref{eq:rooted mu Phat}, respectively. Then
\[
\mu_{R,o}^{(i)}(\widehat P)
=
\mu_{R,o}^{(i)}(P)
\left[
1+O_{\mathbb P}\left((np)^{-1/2}\right)
\right].
\]
The bounds are uniform over any fixed finite collection of such simple graphs.
\end{lemma}

\begin{proof}[Proof of Lemma~\ref{lem:fixed-simple-motif-stability}]
We prove the global statement first. Recall from \eqref{eq:F definition}  that
\[
F(x)
=
\frac{1}{|\operatorname{Aut}(R)|}
\left(\frac{n}{n-1}\right)^{e}
s(x)^{-e}
\sum_{\phi:V(H)\hookrightarrow[n]}
\prod_{a\in V(H)}x_{\phi(a)}^{r_a},
\]
where $s(x):=\sum_{\ell=1}^n x_\ell$
and \(r_a\) denotes the degree of  vertex \(a\) in \(H\). Then $\mu_R(\widehat P)=F(D)$, where $D$ is the vector of node degrees. Conditioning on the regular-degree event $\Omega_n$ defined in \eqref{eq:regular-degree-event}, $F(D) = \mathcal{F}(D)$, where $\mathcal{F}$ is the degree-truncated function defined in \eqref{eq:truncated FG}. 

We first bound the random fluctuation \(F(D)-F(d)\). Taylor expansion
around \(d\) gives
\[
F(D)-F(d)=\sum_{i=1}^nF_{i}(d)(D_i-d_i) + 
\int_0^1(1-t)
\sum_{i,j=1}^n
F_{ij}(d+t(D-d))(D_i-d_i)(D_j-d_j)\,dt.
\]
For the linear term, since $D_i-d_i=\sum_{j\neq i}(A_{ij}-P_{ij})$, we have 
\[
\sum_{i=1}^nF_{i}(d)(D_i-d_i)
=
\sum_{i<j}
\left[F_i(d)+F_j(d)\right](A_{ij}-P_{ij}).
\]
Since $F_i(d) = \mathcal{F}_i(d) =O(n^{v-2}p^{e-1})$ by Lemma~\ref{lem:localized-derivative-bounds}, 
Bernstein's inequality for the sum of independent centered Bernoulli variables yields
\[
\sum_{i=1}^nF_{i}(d)(D_i-d_i)
=
O\left(n^{v-1}p^{e-1/2}\right).
\]
For the quadratic Taylor remainder, from Lemma~\ref{lem:localized-derivative-bounds}, we have 
\[
|F_{ii}(x)|
\lesssim n^{v-3}p^{e-2}, \qquad
|F_{ij}(x)|
\lesssim n^{v-4}p^{e-2}, \quad i\neq j.
\]
Therefore, the quadratic remainder is bounded by
\[
n^{v-3}p^{e-2}\sum_i (D_i-d_i)^2
+
n^{v-4}p^{e-2}
\sum_{i\neq j}|(D_i-d_i)(D_j-d_j)|,
\]
which is of order $O(n^{v-1}p^{e-1})$ because $D_i-d_i = O((np)^{1/2})$. Combining the bounds for the linear and quadratic terms, we obtain
$$
F(D) - F(d) = O(n^{v-1}p^{e-1}).
$$

We now bound \(F(d)-\mu_R(P)\). Let $\widetilde{P}$ be a Chung-Lu edge probability matrix with entries
\[
\widetilde{P}_{ij}
:=
\frac{n}{n-1}\frac{d_id_j}{s(d)}.
\]
It is straightforward to verify that $\widetilde{P}_{ij} = P_{ij}[1+O(n^{-1})]$ and $F(d)=\mu_R(\widetilde{P})$. Since $R$ is fixed, it follows that 
\[F(d) = 
\mu_R(\widetilde{P})
=
\mu_R(P)\left[1+O(n^{-1})\right].
\]
Since $\mu_R(\widehat P)=F(D)$ and  $\mu_R(P)\asymp n^{v}p^{e}$, we obtain
\[
\mu_R(\widehat P)
= \mu_R(P)
\left[
1+O\left((np)^{-1}\right)
\right].
\]
The bound for rooted version follows the same argument, using Lemma~\ref{lem:rooted-derivative-envelope}. 
\end{proof}

We next show that fixed multigraph motif means can be reduced to the preceding
simple-graph case.

\begin{lemma}[Reduction of fixed multigraph functionals to simple graph functionals]
\label{lem:multigraph-to-simple-lifting}
Let \(\mathcal H\) be a fixed connected multigraph.  For each unordered pair
\(\{a,b\}\subset V(\mathcal H)\), let \(m_{ab}\) denote the number of edges
connecting \(a\) and \(b\) in \(\mathcal H\). Assume that \(P\) is a
Chung--Lu edge-probability matrix satisfying
\[
P_{ij}=p\theta_i\theta_j,
\qquad
\sum_{i=1}^n\theta_i=n,
\qquad
0<c\leq \theta_i\leq C<\infty.
\]
Consider the multigraph embedding functional
\[
M_{\mathcal H}(P)
:=
\sum_{\varphi:V(\mathcal H)\hookrightarrow[n]}
\prod_{a<b}
P_{\varphi(a)\varphi(b)}^{m_{ab}}.
\]
Construct a lifted simple graph \(H\) from \(\mathcal H\) as follows.  First, include
all vertices of \(\mathcal H\) in \(H\), and include one simple edge
\(\{a,b\}\) in \(H\) for every unordered pair \(\{a,b\}\) with
\(m_{ab}\geq 1\).  Second, for each \(\{a,b\}\) with \(m_{ab}\geq 2\), add
\(m_{ab}-1\) new pendant leaves adjacent to \(a\), and \(m_{ab}-1\) new
pendant leaves adjacent to \(b\).  Denote
\[
\kappa
:=
\sum_{\{a,b\}:m_{ab}\geq 2}
(m_{ab}-1).
\] 
Then
\[
M_{\mathcal H}(P)
= (n^2p)^{-\kappa}
M_H(P)
\left[1+O(n^{-1})\right].
\]
The same conclusion holds if the sum in $M_\mathcal{H}(P)$ is taken over all
injective maps $\varphi$ satisfying $\varphi(o)=i$ for some fixed vertices
$o\in V(\mathcal{H})$ and $i\in[n]$.
\end{lemma}

\begin{proof}
The graph \(H\) contains the original vertices \(V(\mathcal{H})\) together with the new
pendant leaves added in the construction.  Let \(\mathcal L=V(H)\setminus V(\mathcal{H})\) denote the set of
all added leaves.  Since each surplus edge contributes two leaves, $|\mathcal L|=2\kappa$. For a leaf \(\ell\in\mathcal L\), let
\[
\pi(\ell)\in V(\mathcal{H})
\]
denote its parent vertex in the original graph $\mathcal{H}$.  Thus the corresponding
pendant edge in \(H\) is $\{\pi(\ell),\ell\}$. 
For a base vertex \(a\in V(\mathcal{H})\), denote the number of new leaves attached to \(a\) by
\[
\ell_a
:=
|\{\ell\in\mathcal L:\pi(\ell)=a\}| = \sum_{b\in V(\mathcal{H}):\, m_{ab}\geq 1}(m_{ab}-1).
\]
We now expand \(M_H(P)\) by first fixing the embedding of the original vertices $V(\mathcal{H})$.
Let $\varphi:V(\mathcal{H})\hookrightarrow[n]$
be an injective embedding.  Once \(\varphi\) is fixed,
the remaining sum is over injective images of the leaves into
\([n]\setminus\varphi(V_0)\).  Hence
\[
M_H(P)
=
\sum_{\varphi:V(\mathcal{H})\hookrightarrow[n]}
\left[
\prod_{\{a,b\}:m_{ab}\geq 1}
P_{\varphi(a)\varphi(b)}
\right]
S_\varphi,
\]
where
\[
S_\varphi
:=
\sum_{\substack{\psi:\mathcal L\hookrightarrow[n]\setminus\varphi(V(\mathcal{H}))}}
\prod_{\ell\in\mathcal L}
P_{\varphi(\pi(\ell))\psi(\ell)}.
\]
Here \(\psi:\mathcal L\hookrightarrow[n]\setminus\varphi(V(\mathcal{H}))\) means that the
leaf images are all distinct and none of them coincides with the images of original vertices.
Using the Chung--Lu edge probabilities $P_{ij}=p\theta_i\theta_j$ and the fact that $\ell_a$ new leaves are attached to $a$, 
we get
\[
\prod_{\ell\in\mathcal L}
P_{\varphi(\pi(\ell))\psi(\ell)}
=
p^{|\mathcal L|}
\prod_{\ell\in\mathcal L}
\theta_{\varphi(\pi(\ell))}
\prod_{\ell\in\mathcal L}
\theta_{\psi(\ell)}
=
p^{|\mathcal L|}
\prod_{a\in V(\mathcal{H})}
\theta_{\varphi(a)}^{\ell_a}
\prod_{\ell\in\mathcal L}
\theta_{\psi(\ell)}.
\]
Thus
\[
S_\varphi
=
p^{|\mathcal L|}
\prod_{a\in V(\mathcal{H})}
\theta_{\varphi(a)}^{\ell_a}
\sum_{\substack{\psi:\mathcal L\hookrightarrow[n]\setminus\varphi(V(\mathcal{H}))}}
\prod_{\ell\in\mathcal L}\theta_{\psi(\ell)}.
\]

We claim that, uniformly over the fixed base embedding \(\varphi\),
\[
\sum_{\substack{\psi:\mathcal L\hookrightarrow[n]\setminus\varphi(V(\mathcal{H}))}}
\prod_{\ell\in\mathcal L}\theta_{\psi(\ell)}=
n^{|\mathcal L|}\left[1+O(n^{-1})\right].
\]
Indeed, without the injectivity and exclusion restrictions, the sum would be $\left(\sum_{u=1}^n\theta_u\right)^{|\mathcal L|}
=
n^{|\mathcal L|}$.
The total weight of assignments excluded by the restrictions is
\(O(n^{|\mathcal L|-1})\).  To see this, first consider assignments in which a
specified leaf is mapped to one of the finitely many images of the original vertices
\(\varphi(V(\mathcal{H}))\).  Since the weights are bounded, $\sum_{u\in\varphi(V(\mathcal{H}))}\theta_u=O(1)$, so the total contribution of such assignments is
\[
O(1)\cdot n^{|\mathcal L|-1}.
\]
Next, consider assignments in which two specified leaves are mapped to the same
vertex.  Their contribution is bounded by
\[
\left(\sum_{u=1}^n\theta_u^2\right)n^{|\mathcal L|-2} = O(n)n^{|\mathcal L|-2}.
\]
Since
\(|\mathcal L|\) and \(|V(\mathcal{H})|\) are fixed, a union bound over the finitely many
possible violations gives the claimed estimate.
Since \(|\mathcal L|=2\kappa\), we obtain
\[
S_\varphi
=
(np)^{2\kappa}
\prod_{a\in V(\mathcal{H})}
\theta_{\varphi(a)}^{\ell_a}
\left[1+O(n^{-1})\right].
\]

Now consider the contribution corresponding to the fixed base embedding
$\varphi$. Let $G$ be the simple graph obtained from $\mathcal{H}$ by merging
all multiple edges into single edges; that is, $G$ has vertex set
$V(\mathcal{H})$ and edge set
\(
    E(G)=\{\{a,b\}: m_{ab}\ge 1\}
\).  Then
\[
\prod_{\{a,b\}:m_{ab}\geq 1}
P_{\varphi(a)\varphi(b)}
=
p^{e(G)}
\prod_{a\in V(\mathcal{H})}
\theta_{\varphi(a)}^{r_a(G)},
\]
where $e(G)$ is the number of edges of $G$ and $r_a(G)$ is the degree of vertex $a$ in $G$.
Multiplying by \(S_\varphi\), we get
\[
\left[
\prod_{\{a,b\}:m_{ab}\geq 1}
P_{\varphi(a)\varphi(b)}
\right]
S_\varphi
=
n^{2\kappa}p^{e(G)+2\kappa}
\prod_{a\in V(\mathcal{H})}
\theta_{\varphi(a)}^{r_a(G)+\ell_a}
\left[1+O(n^{-1})\right].
\]
By construction,
\[
r_a(G)+\ell_a
=
\sum_{b\in V(\mathcal{H})}m_{ab}
=
r_a(\mathcal H),
\]
the degree of \(a\) in the multigraph \(\mathcal H\), counting multiplicity.
Also, the total number of edges in \(\mathcal H\), counting multiplicity, is $e(G)+\kappa$. Therefore the multigraph contribution for the same base embedding \(\varphi\)
is
\[
\prod_{a<b}
P_{\varphi(a)\varphi(b)}^{m_{ab}}
=
p^{e(G)+\kappa}
\prod_{a\in V(\mathcal{H})}
\theta_{\varphi(a)}^{r_a(\mathcal H)}.
\]
Comparing the last two displays, we find
\[
\left[
\prod_{\{a,b\}:m_{ab}\geq 1}
P_{\varphi(a)\varphi(b)}
\right]
S_\varphi
=
(n^2p)^{\kappa}
\prod_{a<b}
P_{\varphi(a)\varphi(b)}^{m_{ab}}
\left[1+O(n^{-1})\right].
\]
The error term is uniform over \(\varphi\).  Summing over all injective
\(\varphi:V(\mathcal{H})\hookrightarrow[n]\) therefore gives
\[
M_H(P)
=
(n^2p)^{\kappa}
M_{\mathcal H}(P)
\left[1+O(n^{-1})\right].
\]
The proof is complete.
\end{proof}

\begin{corollary}[Plug-in stability for fixed multigraph means]
\label{cor:multigraph-stability}
Let $\mathcal{H}$ be any fixed connected multigraph and $M_\mathcal{H}(P)$ be the functional
defined in Lemma~\ref{lem:multigraph-to-simple-lifting}. Similarly, for fixed
vertices $o\in V(\mathcal{H})$ and $i\in[n]$, denote
\[
M_{\mathcal H,o}^{(i)}(P)
:=
\sum_{\substack{\varphi:V(\mathcal H)\hookrightarrow[n]\\
\varphi(o)=i}} \ 
\prod_{a<b}
P_{\varphi(a)\varphi(b)}^{m_{ab}}.
\]
Then, for the estimate $\widehat{P}$ defined in \eqref{eq: main parameter estimators}, we have 
\begin{eqnarray*}
M_{\mathcal H}(\widehat P)
&=&
M_{\mathcal H}(P)
\left[
1+O\left((np)^{-1}\right)
\right],\\
M_{\mathcal H,o}^{(i)}(\widehat P)
&=&
M_{\mathcal H,o}^{(i)}(P)
\left[
1+O\left((np)^{-1/2}\right)
\right].
\end{eqnarray*}
\end{corollary}

\begin{proof} Let $H$ be the lifted simple graph
constructed from $\mathcal{H}$ in Lemma~\ref{lem:multigraph-to-simple-lifting}. Then by Lemma~\ref{lem:multigraph-to-simple-lifting}, there exists a constant $\kappa$ depending only on $\mathcal{H}$ such that on the regular-degree event $\Omega_n$ defined in \eqref{eq:regular-degree-event}, we have
\begin{eqnarray*}
M_{\mathcal H}(P)
&=& (n^2p)^{-\kappa}
M_H(P)
\left[1+O(n^{-1})\right],\\
M_{\mathcal H}(\widehat{P})
&=& (n^2\widehat{p})^{-\kappa}
M_H(\widehat{P})
\left[1+O(n^{-1})\right].
\end{eqnarray*}
\[
M_{\mathcal H}(P)
= (n^2p)^{-\kappa}
M_H(P)
\left[1+O(n^{-1})\right].
\]
Since $\widehat{p} = p(1+O(np)^{-1})$ and $\mu_R(\widehat P)
=
\mu_R(P)
\left[
1+O\left((np)^{-1}\right)
\right]$ by Lemma~\ref{lem:fixed-simple-motif-stability}, it follows that 
$$
M_{\mathcal H}(\widehat P)
=
M_{\mathcal H}(P)
\left[
1+O\left((np)^{-1}\right)
\right].
$$
The proof for the rooted version follows the same argument. 
\end{proof}


\section{Proofs of the Main Results}
\label{sec:proofs of main results}

\begin{proof}[Proof of Proposition~\ref{prop:bias_general_subgraph_count_DCSBM}]
We first prove \eqref{eq:bias_order_general_subgraph_count}. Recall the bias
decomposition from \eqref{eq:bias decomposition global}:
\begin{equation*}
\mathrm{Bias}(\mu_R(\widehat{P}))=p^{e}\left(\Psi_1-\Psi_2\right)+p^{e-1}\Psi_3+O(n^{v-2}p^{e-2}).    
\end{equation*} 
By Lemma~\ref{lemma:derivative bounds for Psi_k}, we have
$|\Psi_1-\Psi_2|=O(n^{v-1})$. Regarding $\Psi_3$, since $R$ is connected and
has at least $v\ge 3$ vertices, it contains a vertex $a\in V(R)$ of degree at least two, or equivalently $\binom{r_a}{2}>0$.
Therefore, a direct calculation shows $\Psi_3\asymp n^{v-1}$. Under the assumption $p\prec 1$, we obtain
\begin{equation*}
\mathrm{Bias}(\mu_R(\widehat{P}))\asymp p^{e}n^{v-1}+p^{e-1}n^{v-1}+O(n^{v-2}p^{e-2})\asymp n^{v-1}p^{e-1}.   
\end{equation*} 
The proof of \eqref{eq:bias_order_general_rooted_subgraph_count} follows the
same argument, using the bias decomposition in
\eqref{eq:bias decomposition rooted} and the bound in
\eqref{eq:rooted-Psi-size}.
\end{proof}


\begin{proof}[Proof of Lemma~\ref{lemma:variance_of_hat_mu_R}]  
We first prove \eqref{eq:variance_of_hat_mu_R}. 
Note that $\mu_R(\widehat{P})=F(D)=G(E)$. From Lemma~\ref{lemma:equivalent variances}, for each $\kappa>0$, we have 
$$
|\Var(F(D))-\Var(\mathcal{F}(D))|\lesssim n^{-\kappa}\big[1+\Var(\mathcal{F}(D))\big].
$$
The Hoeffding decomposition of $\mathcal{F}(D)=\mathcal{G}(E)$ in \eqref{eq:Hoeffding decomposition of mathcal G} gives
$$
\mathcal{G}(E)-\E \mathcal{G}(E)=\sum_\alpha c_\alpha E_\alpha+\mathcal{G}^{(2+)}.
$$
Therefore by Lemma~\ref{lemma:Wge2},
$$
\Var(\mathcal{G}(E)) = \sum_\alpha c_\alpha^2\sigma_\alpha^2 + \Var(\mathcal{G}^{(2+)}) = \sum_\alpha c_\alpha^2\sigma_\alpha^2 + O(n^{2v-3}p^{2e-2}). 
$$
Using the approximation $(t+\varepsilon)^2 = t^2 + O(t\varepsilon+\varepsilon^2)$ and 
Lemma~\ref{lem:first-projection-error}, we get  
$$
\sum_\alpha c_\alpha^2\sigma_\alpha^2 = \sum_\alpha (\partial_\alpha\mathcal{G}(0))^2\sigma_\alpha^2 +O\left(n^{v-2}p^{e-3/2}\left[\sum_\alpha (\partial_\alpha\mathcal{G}(0))^2\sigma_\alpha^2\right]^{1/2}+
n^{2v-4}p^{2e-3}\right).
$$
Since $\partial_\alpha\mathcal{G}(0)=\sum_{i\in\alpha}\mathcal{F}_i(d) = \sum_{i\in\alpha}F_i(d)$ by \eqref{eq:localized-first-derivative-connection}, it follows from Lemma~\ref{lemma:scale of V_R} that  
\[
\sum_\alpha (\partial_\alpha\mathcal{G}(0))^2\sigma_\alpha^2 = \sum_{i<j}P_{ij}(1-P_{ij})(F_i(d)+F_j(d))^2\asymp n^{2v-2}p^{2e-1}.
\]
Therefore, 
$$
\Var(\mu_R(\widehat{P})) \asymp  n^{2v-2}p^{2e-1}+O(n^{2v-3}p^{2e-2}) + O(n^{-\kappa}) \asymp n^{2v-2}p^{2e-1},
$$
and \eqref{eq:variance_of_hat_mu_R} is proved. 

The proof of \eqref{eq:variance_of_hat_mu_R_i} follows from \eqref{eq:rooted variance order}. 
\end{proof}


\begin{proof}[Proof of Proposition~\ref{prop:bias_triangle_count_SVD}]
Let \(P\) be the loopless Erd\H{o}s--R\'enyi edge-probability matrix with entries 
$P_{ij}=p\mathbf 1_{\{i\neq j\}}$. 
Then $\lambda_1:=\lambda_1(P)=(n-1)p$. By Theorem 1 of \cite{furedi1981eigenvalues}, the leading eigenvalue $\widehat{\lambda}_1$ of adjacency matrix $A$ satisfies
\[
\widehat{\lambda}_1 - \lambda_1
\rightarrow N(1-p,2p(1-p)).
\]
Moreover, $\widehat{\lambda}_1$ has sub-Gaussian tail  around its median by \cite{alon2002concentration}, uniformly on $n$. Together with the weak convergence
above, it follows that   \(\widehat{\lambda}_1 - \lambda_1\) has bounded moments:
\[
\mathbb E\widehat{\lambda}_1 - \lambda_1=1-p+o(1), \quad
\operatorname{Var}(\widehat{\lambda}_1 - \lambda_1)=2p(1-p)+o(1), \quad
\mathbb E|\widehat{\lambda}_1 - \lambda_1|^q=O_q(1).
\]
We now introduce a closely related spectral statistic. For any symmetric
matrix \(Q\), define
\[
\widetilde\mu_\Delta(Q)
:=
\frac{1}{6}\operatorname{tr}(Q^3)
=
\frac{1}{3!}
\sum_{i,j,k=1}^n Q_{ij}Q_{jk}Q_{ik}.
\]
This statistic counts triangle-like closed walks of length three and therefore
allows repeated vertices. If \(Q\) has rank one with nonzero eigenvalue
\(\lambda_1(Q)\), then
$
\widetilde\mu_\Delta(Q)
=\lambda_1(Q)^3/6
$. 
In particular, for the rank-one matrices $\widehat{P}_{\mathrm{SVD}}$ and $\widetilde P:=p\mathbf 1_n\mathbf 1_n^\top$, we have 
\[
\widetilde\mu_\Delta(\widehat P_{\mathrm{SVD}})
=\frac{1}{6}\widehat\lambda_1^3, \qquad \widetilde\mu_\Delta(\widetilde P)
=
\frac{1}{6}(np)^3.
\]

We first estimate the mean and variance of $\widetilde\mu_\Delta(\widehat P_{\mathrm{SVD}})$. Expanding
\[
\frac{1}{6}\left[\widehat\lambda_1^3-\lambda_1^3\right]
=
\frac12\lambda_1^2(\widehat{\lambda}_1 - \lambda_1)
+
\frac12\lambda_1(\widehat{\lambda}_1 - \lambda_1)^2
+
\frac16(\widehat{\lambda}_1 - \lambda_1)^3,
\]
and using the boundedness of the moments of $\widehat{\lambda}_1 - \lambda_1$ give
\[
\mathbb E_P\left[
\widetilde\mu_\Delta(\widehat P_{\mathrm{SVD}})
\right]
-
\frac{1}{6}\lambda_1^3
=
\frac12\lambda_1^2\mathbb E\delta_n
+
\frac12\lambda_1\mathbb E\delta_n^2
+
\frac16\mathbb E\delta_n^3 = 
\frac12\lambda_1^2(1-p)[1+o(1)].
\]
Since $\mu_\Delta(P)=\binom n3 p^3$, it follows that 
$
\lambda_1^3/6
=
\mu_\Delta(P)+O(np^3)
$.
Therefore we obtain
\begin{equation}
\label{eq:mu tilde svd mean}
\mathbb E_P\left[
\widetilde\mu_\Delta(\widehat P_{\mathrm{SVD}})
\right]
-
\mu_\Delta(P)
\asymp
n^2p^2.    
\end{equation}
We now calculate the variance of $\widetilde\mu_\Delta(\widehat P_{\mathrm{SVD}})$. Using the expansion of $\widehat{\lambda}_1$ above with $\delta:=\widehat{\lambda}_1-\lambda_1$, we have
\[
\widetilde\mu_\Delta(\widehat P_{\mathrm{SVD}})
-
\mathbb E_P\widetilde\mu_\Delta(\widehat P_{\mathrm{SVD}})
=
\frac12\lambda_1^2(\delta-\mathbb E\delta)+
\frac12\lambda_1(\delta^2-\mathbb E\delta^2)
+
\frac16(\delta^3-\mathbb E\delta^3).
\]
Since all moments of $\delta$ are bounded, it follows that the variance of the last two terms on the right-hand side is at most $O(\lambda_1^2)$. Moreover,
\[
\operatorname{Var}\left(\lambda_1^2\delta
\right)
= \lambda_1^4\operatorname{Var}(\delta)
=2\lambda_1^4p(1-p)[1+o(1)].
\]
Therefore using $\lambda_1\asymp np$ and $p\in(0,\varepsilon)$, 
\begin{equation}
\label{eq:mu tilde svd var}
\operatorname{Var}_P\left(
\widetilde\mu_\Delta(\widehat P_{\mathrm{SVD}})
\right)
\asymp
n^4p^5.    
\end{equation}

It remains to transfer these estimates from the spectral statistic \(\widetilde\mu_\Delta\) to the loopless triangle count \(\mu_\Delta\).  Since $\widehat P_{\mathrm{SVD}}
=
\widehat\lambda_1\widehat v_1\widehat v_1^\top$, 
\[
\mu_\Delta(\widehat P_{\mathrm{SVD}})
=
\frac16\widehat\lambda_1^3
\sum_{\substack{i,j,k\\ i,j,k\ \mathrm{distinct}}}
\widehat v_{1,i}^2
\widehat v_{1,j}^2
\widehat v_{1,k}^2.
\]
Because \(\|\widehat v_1\|_2=1\),
\[
\sum_{i,j,k}
\widehat v_{1,i}^2
\widehat v_{1,j}^2
\widehat v_{1,k}^2
=
1.
\]
Therefore the difference between \(\widetilde\mu_\Delta\) and \(\mu_\Delta\)
comes only from triples with repeated indices.  By a union bound over the
cases \(i=j\), \(i=k\), and \(j=k\),
\begin{equation}
\label{eq:mu tilde svd}
\left|
\widetilde\mu_\Delta(\widehat P_{\mathrm{SVD}})
-
\mu_\Delta(\widehat P_{\mathrm{SVD}})
\right|
\lesssim
\widehat\lambda_1^3\|\widehat v_1\|_4^4.    
\end{equation}
It follows from \cite{furedi1981eigenvalues} and \cite{alon2002concentration} that for any $q\ge 1$, 
\[
\|A-P\|= O_{L^q}(\sqrt{np}),
\]
where the notation \(X_n=O_{L^q}(a_n)\) means $\left(\mathbb E|X_n|^q\right)^{1/q}
\lesssim a_n$.
By the Davis--Kahan theorem,
\[
\left\|\widehat v_1-u\right\|_2 = O_{L^q}((np)^{-1/2}), 
\]
where $u=(n)^{-1/2}\mathbf{1}_n$ is the constant leading eigenvector of $P$. 
By the triangle inequality,
\[
\|\widehat v_1\|_4
\le
\|u\|_4+\|\widehat v_1-u\|_4\le 
n^{-1/4} +\|\widehat v_1-u\|_2.
\]
Since $p$ is fixed, taking \(L^{16}\)-norms  gives
\[
\|\|\widehat v_1\|_4\|_{L^{16}}
\le
n^{-1/4}
+
\|\|\widehat v_1-u\|_2\|_{L^{16}}
=
O(n^{-1/4}).
\]
Therefore
\[
\left\|\|\widehat v_1\|_4^4\right\|_{L^4}
=
\left(\mathbb E\|\widehat v_1\|_4^{16}\right)^{1/4}
=
\left[
\left(\mathbb E\|\widehat v_1\|_4^{16}\right)^{1/16}
\right]^4
=
O(n^{-1}).
\]
We also need an \(L^{12}\)-bound for the top eigenvalue. Since
\[
\widehat\lambda_1
\le
\|P\|+\|A-P\|=(n-1)p + O_{L^{12}}(\sqrt{np}) =O_{L^{12}}(np),
\]
it follows that 
\[
\|\widehat\lambda_1^3\|_{L^4}
=
\|\widehat\lambda_1\|_{L^{12}}^3
=
O(n^3p^3).
\]
By Hölder's inequality,
\[
\left\|
\widehat\lambda_1^3\|\widehat v_1\|_4^4
\right\|_{L^2}
\le
\|\widehat\lambda_1^3\|_{L^4}
\left\|\|\widehat v_1\|_4^4\right\|_{L^4}
=
O(n^3p^3)\cdot O(n^{-1})
=
O(n^2p^3).
\]
Combining this with \eqref{eq:mu tilde svd}, we obtain
\[
\left\|
\widetilde\mu_\Delta(\widehat P_{\mathrm{SVD}})
-
\mu_\Delta(\widehat P_{\mathrm{SVD}})
\right\|_{L^2}
=
O(n^2p^3).
\]
Consequently,
\[
\mathbb E\left|
\widetilde\mu_\Delta(\widehat P_{\mathrm{SVD}})
-
\mu_\Delta(\widehat P_{\mathrm{SVD}})
\right|
=
O(n^2p^3),
\]
and
\[
\operatorname{Var}\left(
\widetilde\mu_\Delta(\widehat P_{\mathrm{SVD}})
-
\mu_\Delta(\widehat P_{\mathrm{SVD}})
\right)
=
O(n^4p^6).
\]
The claim of the proposition then follows from \eqref{eq:mu tilde svd mean} and \eqref{eq:mu tilde svd var}.  
\end{proof}

\begin{proof}[Proof of Proposition~\ref{prop:bias_correction_general_subgraph}]
For global subgraph counts, the bias-reduction result follows from
Lemma~\ref{lemma:bias correction by bootstrap} and
\eqref{eq:bias_order_general_subgraph_count}. For rooted subgraph counts, it
follows from \eqref{eq:rooted-G-stability-final} and
\eqref{eq:bias_order_general_rooted_subgraph_count}.
\end{proof}

\begin{proof}[Proof of Proposition~\ref{prop:asymptotic_normality_triangle_count}]
We prove the result for global subgraph counts using the method of cumulants;
the proof for rooted subgraph counts follows the same argument and is therefore
omitted.

Let $\mathcal I_R:=\{\phi:V(R)\hookrightarrow[n]\}$
be the set of injective embeddings of \(R\) into \([n]\). 
Then
\[
    T_R
    =
    \frac{1}{|\operatorname{Aut}(R)|}
    \sum_{\phi\in\mathcal I_R} I_\phi, \qquad I_\phi
    :=
    \prod_{\{a,b\}\in E(R)}
    A_{\phi(a)\phi(b)}.
\]
Since \(P_{ij}\asymp p\) uniformly,
it follows that $\mathbb E T_R \asymp n^v p^e$. 
For every nonempty subgraph \(H\subseteq R\) with at least one edge, denote
\[
    \Phi_H:=n^{v(H)}p^{e(H)}, \qquad \Phi_R :=
    \min_{\substack{H\subseteq R\\ e(H)\ge 1}}
    n^{v(H)}p^{e(H)}.
\]
The assumption \(p \succ n^{-1/m_R}\) implies $\Phi_R(n,p)\to\infty$. Indeed, for every \(H\subseteq R\) with \(e(H)\ge1\), we have $e(H)/v(H)\le m_R$, and therefore \(p\succ n^{-1/m_R}\) implies $n^{v(H)}p^{e(H)}\to\infty$. 

We first determine the variance scale.  Expanding the variance,
\[
    \operatorname{Var}(T_R)
    =
    \frac{1}{|\operatorname{Aut}(R)|^2}
    \sum_{\phi,\psi\in\mathcal I_R}
    \operatorname{Cov}(I_\phi,I_\psi).
\]
For two embeddings \(\phi,\psi\), let \(E_\phi\) and \(E_\psi\) denote the sets
of graph edges used by the two embedded copies.  If $E_\phi\cap E_\psi=\emptyset$
then \(I_\phi\) and \(I_\psi\) depend on disjoint sets of independent edge
variables, and hence $\operatorname{Cov}(I_\phi,I_\psi)=0$. 
Thus only pairs of copies sharing at least one edge contribute to the variance.
Let $H$ be the common edge subgraph (so $V(H)$ is the set of endpoints of the overlapping edges) of the two copies of $R$ with \(e(H)\ge1\).  Then the union of the two copies has $2v-v(H)$
vertices and $2e-e(H)$
edges.  Since every edge probability is comparable to \(p\), the total
contribution of overlap copies with the common edge subgraph isomorphic to $H$ is of order
\[
    n^{2v-v(H)}p^{2e-e(H)}
    =
    \frac{(n^v p^e)^2}{n^{v(H)}p^{e(H)}}\le     \frac{(n^v p^e)^2}{\Phi_R}\asymp \Var(T_R).
\]
Equivalently,
\begin{equation}
\label{eq:variance of T_R}
\operatorname{Var}(T_R)
    \asymp
    \max_{\substack{H\subseteq R\\ e(H)\ge1}}
    n^{2v-v(H)}p^{2e-e(H)}.
\end{equation}
We next bound higher cumulants.  For \(k\ge3\), let $\kappa_k(T_R)$
denote the \(k\)-th cumulant of \(T_R\).  Since \(T_R\) is a finite sum of the indicators \(I_\phi\),
\[
    \kappa_k(T_R)
    =
    \frac{1}{|\operatorname{Aut}(R)|^k}
    \sum_{\phi_1,\ldots,\phi_k\in\mathcal I_R}
    \kappa(I_{\phi_1},\ldots,I_{\phi_k}).
\]
If $I_{\phi_1},\ldots,I_{\phi_k}$ can be split into two groups depending on
disjoint sets of edges then
$$\kappa(I_{\phi_1},\ldots,I_{\phi_k})=0.$$ This is a standard property of joint
cumulants of functions of independent random variables.
Equivalently, the joint cumulant can be nonzero only if the copy-intersection
graph on \(\{1,\ldots,k\}\), where \(r\) and \(s\) are adjacent whenever
\(\phi_r(R)\) and \(\phi_s(R)\) share at least one edge, is connected.
Moreover, if \(G\) is the union graph of the \(k\) embedded copies, then by Lemma~\ref{lem:moment-cumulant-subgraph-bound}, 
\[
    \left|
    \kappa(I_{\phi_1},\ldots,I_{\phi_k})
    \right|
    \lesssim_k
    \prod_{e\in E(G)}P_e
    \lesssim_k
    p^{e(G)}.
\]
Hence, for a connected \(k\)-tuple with union graph \(G\), the total
contribution is at most
\[
    O_k\bigl(n^{v(G)}p^{e(G)}\bigr).
\]

We now bound \(n^{v(G)}p^{e(G)}\).  Since the copy-intersection graph is
connected, we may order the \(k\) copies so that the first copy is arbitrary
and each later copy shares at least one edge with the union of the preceding
copies.  When the \(t\)-th copy is added, let \(H_t\subseteq R\) denote the
edge subgraph through which this new copy overlaps the previous union.  Then
\(e(H_t)\ge1\).  Adding this copy contributes at most
\[
    n^{v-v(H_t)}p^{e-e(H_t)}
    =
    \frac{n^v p^e}{n^{v(H_t)}p^{e(H_t)}}
    \le
    \frac{n^v p^e}{\Phi_R}.
\]
The first copy contributes \(n^v p^e\).  Therefore, for every connected
\(k\)-tuple,
\[
    n^{v(G)}p^{e(G)}
    \lesssim_k
    n^v p^e
    \left(\frac{n^v p^e}{\Phi_R}\right)^{k-1}
    =
    \frac{(n^v p^e)^k}{\Phi_R^{k-1}}.
\]
Because \(R\) and \(k\) are fixed, there are only finitely many possible union
types.  Consequently,
\[
    |\kappa_k(T_R)|
    \lesssim_k
    \frac{(n^v p^e)^k}{\Phi_R^{k-1}}.
\]

Denote the normalized statistic by
\[
    Z_R
    :=
    \frac{T_R-\mathbb E T_R}
    {\sqrt{\operatorname{Var}(T_R)}}.
\]
Its first cumulant is zero and its second cumulant is one.  For \(k\ge3\), since $\Var(T_R)$ is of order $n^{2v}p^{2e}/\Phi_R$, the cumulant bound gives 
\[
    |\kappa_k(Z_R)|
    =
    \frac{|\kappa_k(T_R)|}
    {\operatorname{Var}(T_R)^{k/2}}
        \lesssim_k
    \frac{(n^v p^e)^k}{\Phi_R^{k-1}}    
    \frac{\Phi_R^{k/2}}
    {(n^vp^e)^k}
    =
    \Phi_R^{-(k-2)/2}.
\]
Because \(\Phi_R\to\infty\), it follows that $\kappa_k(Z_R)\to0$ for every $k\ge 3$
Thus, the cumulants of \(Z_R\) converge to the cumulants of a standard normal
random variable:
\[
    \kappa_1(Z_R)\to0,
    \qquad
    \kappa_2(Z_R)\to1,
    \qquad
    \kappa_k(Z_R)\to0
    \quad(k\ge3).
\]
By the method of cumulants, 
\[
    Z_R
    =
    \frac{T_R-\mathbb E T_R}
    {\sqrt{\operatorname{Var}(T_R)}}
    \Rightarrow N(0,1).
\]
The proof for rooted subgraph counts follows the same argument. 
\end{proof}

\begin{proof}[Proof of Proposition~\ref{prop:variance_approximation_general}]
Recall from \eqref{eq:global-variance-overlap-decomposition-careful} that 
\begin{equation*}
\Var_P(T_R)
=
\sum_{\omega\in\mathcal{O}_R}
b_\omega\left[\Psi_\omega^+(P)-\Psi_\omega^-(P)\right],
\end{equation*}
where $\mathcal{O}_R$ denotes the collection of overlap types of two copies of $R$, \(b_\omega>0\) depends only on \(R\), and 
\begin{eqnarray*}
\Psi_\omega^+(P)
&=&
\sum_{\varphi:V(R_\omega^+)\hookrightarrow[n]}
\prod_{e\in E(R_\omega^+)}P_{\varphi(e)},\\
\Psi_\omega^-(P)&=&\sum_{\varphi:V(R_\omega^+)\hookrightarrow[n]}
\prod_{e\in E(R_\omega^+)}P_{\varphi(e)}
\prod_{e\in E(R_\omega^-)}P_{\varphi(e)}.     
\end{eqnarray*}
By definition, $\Psi_\omega^+(P) = \mu_{R_\omega^+}(P)$. It then follows from   Lemma~\ref{lem:fixed-simple-motif-stability} that
\[
\Psi_\omega^+(\widehat P)
=
\Psi_\omega^+(P)
\left[
1+O\left((np)^{-1}\right)
\right].
\]
Let \(\mathcal{R}_\omega\) be the fixed multigraph obtained from \(R_\omega^+\) by adding one extra copy of every edge in
\(R_\omega^-\). By Corollary~\ref{cor:multigraph-stability},
\[
\Psi_{\omega}^-(\widehat P)
= M_{\mathcal{R}_\omega}(\widehat{P}) = 
M_{\mathcal{R}_\omega}(P)
\left[
1+O\left((np)^{-1}\right)
\right] = \Psi_{\omega}^-(P)\left[
1+O\left((np)^{-1}\right)
\right].
\]
Since \(R_\omega^-\) contains at least one edge and \(p\to0\),
\[
\Psi_\omega^-(P)
\lesssim
p^{|E(R_\omega^-)|}\Psi_\omega^+(P)
=
o\left(\Psi_{R_\omega^+}(P)\right).
\]
It follows that
\[
\Psi_\omega^+(\widehat P) - \Psi_\omega^-(\widehat P)
=\left[\Psi_\omega^+(P) - \Psi_\omega^-(P)\right]
\left[
1+O\left((np)^{-1}\right)
\right].
\]
Using
\eqref{eq:global-variance-overlap-decomposition-careful}, we obtain
\[
\Var_{\widehat P}(T_R)
=
\Var_P(T_R)
\left[
1+O\left((np)^{-1}\right)
\right].
\]
The proof for rooted variances follows the same argument.
\end{proof}

\begin{proof}[Proof of Proposition~\ref{prop:normal_convergence_expectation_of_subgraph_count}]
We first consider the plugin estimator for the global mean $\mu_R(\widehat{P})$. 
Recall from \eqref{eq:F definition} and \eqref{eq:untruncated G definition} that
$
\mu_R(\widehat{P}) = F(D) = G(E)
$,
where $D$ is the vector of node degrees and $E=A-\E A$. Also, recall from \eqref{eq:truncated FG} the degree-truncated versions of $F$ and $G$, denoted by $\mathcal{F}$ and $\mathcal{G}$. On the regular-degree event $\Omega_n$ defined in \eqref{eq:regular-degree-event}, we have 
$$
\mu_R(\widehat{P}) = F(D) = G(E) = \mathcal{F}(D) = \mathcal{G}(E).
$$
It is therefore enough to establish the central limit theorem first for \(\mathcal G(E)\), and then transfer the result back to
\(F(D)\).

By the Hoeffding decomposition in
\eqref{eq:Hoeffding decomposition of mathcal G}, we have
\begin{equation*}
\label{eq:CLT-Hoeffding-decomposition}
\mathcal G(E)-\mathbb E\mathcal G(E)
=
\sum_{\alpha}c_\alpha E_\alpha
+
\mathcal G^{(2+)},
\end{equation*}
where \(c_\alpha\) is the first-order Hoeffding coefficient and
\(\mathcal G^{(2+)}\) contains all Hoeffding components of order at least two. By Lemma~\ref{lem:first-projection-error},
\[
|c_\alpha-\partial_\alpha\mathcal G(0)|
\lesssim
n^{v-3}p^{e-2}, \qquad
\sum_\alpha
\sigma_\alpha^2(c_\alpha-h_\alpha)^2
\lesssim
n^{2v-4}p^{2e-3},
\]
where $\sigma_\alpha^2=P_\alpha(1-P_\alpha)$, and $\partial_\alpha \mathcal{G}(0)=\mathcal{F}_i(d)+\mathcal{F}_j(d)$ if $\alpha=\{i,j\}$. 
Consider the linearized statistic
\[
L_R
:=
\sum_{\alpha}h_\alpha E_\alpha
=
\sum_{i<j}
\{\mathcal{F}_i(d)+\mathcal{F}_j(d)\}(A_{ij}-P_{ij})=:\sum_{i<j} Y_{ij}.
\]
Its variance is
\[
\operatorname{Var}(L_R)
=
\sum_{i<j}
P_{ij}(1-P_{ij})
\{\mathcal{F}_i(d)+\mathcal{F}_j(d)\}^2
=:
V_R.
\]
From the Hoeffding decomposition, we have 
\[
\mathcal G(E)-\mathbb E\mathcal G(E)
= L_R + \sum_\alpha(c_\alpha-h_\alpha)E_\alpha
+
\mathcal G^{(2+)} =:
L_R+\Delta_R.
\]
By orthogonality of Hoeffding components, Lemma~\ref{lem:first-projection-error}, and Lemma~\ref{lemma:Wge2}, 
\[
\operatorname{Var}(\Delta_R)
=
\sum_\alpha
\sigma_\alpha^2(c_\alpha-h_\alpha)^2
+
\operatorname{Var}(\mathcal G^{(2+)})\lesssim n^{2v-4}p^{2e-3}
+
n^{2v-3}p^{2e-2}\lesssim n^{2v-3}p^{2e-2}.
\]
By Lemma~\ref{lemma:scale of V_R}, 
we have $V_R\asymp n^{2v-2}p^{2e-1}$, and therefore 
\[
\frac{\operatorname{Var}(\Delta_R)}{V_R}
\lesssim
\frac{1}{np}
\to 0.
\]
Therefore, $\Delta_R/V_R^{1/2}$ converges to zero in probability.
It remains to prove a central limit theorem for $L_R$ by verifying the
Lindeberg condition. Note that $L_R$ is a sum of independent, centered random
variables, each with magnitude at most $O(n^{v-2}p^{e-1})$ by
Lemma~\ref{lem:localized-derivative-bounds}. Since
\[
\frac{O(n^{v-2}p^{e-1})}{V_R^{1/2}}
\lesssim
\frac{1}{np^{1/2}}
\to 0,
\]
for every fixed \(\varepsilon>0\),
\[
\frac{1}{V_R}
\sum_{i<j}
\mathbb E
\left[
Y_{ij}^2
\mathbf 1\{|Y_{ij}|>\varepsilon V_R^{1/2}\}
\right]
\to 0.
\]
The Lindeberg condition holds, and the Lindeberg--Feller central limit theorem
yields 
\[
\frac{L_R}{V_R^{1/2}}
\Rightarrow
N(0,1).
\]
Combining this with $\Delta_R/V_R^{1/2}\to 0$
and applying Slutsky's theorem yields
\[
\frac{
\mathcal G(E)-\mathbb E\mathcal G(E)
}{
V_R^{1/2}
}
\Rightarrow
N(0,1).
\]

Finally, since \(\mathcal G(E)=G(E)=F(D)\) on the regular-degree event
\(\Omega_n\), and since the complement of \(\Omega_n\) has exponentially small
probability while \(F(D)\) and \(\mathcal F(D)\) are polynomially bounded, we have
\[
\frac{
\{F(D)-\mathbb EF(D)\}
-
\{\mathcal F(D)-\mathbb E\mathcal F(D)\}
}{
V_R^{1/2}
}
\to 0
\qquad\text{in probability}.
\]
A final application of Slutsky's theorem gives
\[
\frac{
F(D)-\mathbb EF(D)
}{
V_R^{1/2}
}
\Rightarrow
N(0,1).
\]
Since \(F(D)=\mu_R(\widehat P)\) and $\Var(\mu_R(\widehat{P}))/V_R\to 1$, we conclude that
\[
\frac{
\mu_R(\widehat P)-\mathbb E\mu_R(\widehat P)
}{
\Var(\mu_R(\widehat{P}))^{1/2}
}
\Rightarrow
N(0,1).
\]
The proof for the rooted statistic follows the same argument. 
\end{proof}

\begin{proof}[Proof of Proposition~\ref{prop:variance_approximation_mu_hat}]
The proof of the variance approximation via the second-level bootstrap follows
from Lemma~\ref{lemma:variance estimation via second layer boostrap global} for
global counts and from
\eqref{eq:variance estimate via second layer bootstrap rooted} for rooted
counts.      
\end{proof}

\begin{proof}[Proof of Proposition~\ref{prop:variance_comparison_general}]
We first consider the global counts. From \eqref{eq:variance of T_R} and
Lemma~\ref{lemma:scale of V_R},
\[
\frac{\Var(\mu_R(\widehat{P}))}{\Var(T_R)}
\asymp
\min_{\substack{H\subseteq R\\ e(H)\ge 1}}
n^{v(H)-2}p^{e(H)-1},
\]
where the minimum is taken over all subgraphs $H$ of $R$ with at least one edge;
in particular, $H$ contains no isolated vertices. Since $R$ is connected, it
contains at least one edge. Thus, by taking $H=K_2$, the right-hand side is at
most $O(1)$.

The two variances are of the same order exactly when
$n^{v(H)-2}p^{e(H)-1}\gtrsim 1$ for every subgraph $H\subseteq R$ with
$e(H)\ge 1$. This condition holds trivially for $H$ with $e(H)=1$, so it is
enough to consider subgraphs $H$ with $e(H)\ge 2$. Hence the condition is
equivalent to
\[
p
\gtrsim
\max_{\substack{H\subseteq R\\ e(H)\ge 2}}
n^{-\frac{v(H)-2}{e(H)-1}}.
\]

The proof for rooted counts follows the same argument. We use the bound
\eqref{eq:rooted variance order} for $\mu_{R,r}^{(i)}(\widehat{P})$ and the following bound
for rooted counts:
\[
\Var(T_{R,r}^{(i)})
\asymp
\max_{\substack{H\subseteq R, \ 
e(H)\ge 1\\
H \text{ has no isolated vertices}}}
n^{2(v-1)-v(H\setminus\{r\})}p^{2e-e(H)}.
\]
This bound can be obtained by following the argument used in the proof of
Proposition~\ref{prop:asymptotic_normality_triangle_count} for global counts.
We omit the details.
\end{proof}

\begin{lemma}[Cumulant bound]
\label{lem:moment-cumulant-subgraph-bound}
Let \(A\) be the adjacency matrix of an inhomogeneous Erd\H{o}s--R\'enyi graph with independent
edges and edge probabilities \(P_e\).  For embeddings $\phi_1,\ldots,\phi_k:V(R)\hookrightarrow[n]$, denote 
$I_{\phi_i}
:=
\prod_{e\in E_{\phi_i}}A_e,
$
where \(E_{\phi_j}\) is the set of graph edges used by the embedded copy \(\phi_i(R)\).  Let  $E=E_{\phi_1}\cup\cdots\cup E_{\phi_k}$. 
Then there exists a constant \(C_k\) depending only on \(k\) such that  

\[
\left|
\kappa(I_{\phi_1},\ldots,I_{\phi_k})
\right|
\le
C_k
\prod_{e\in E}P_e.
\]
\end{lemma}

\begin{proof}[Proof of Lemma~\ref{lem:moment-cumulant-subgraph-bound}]
We use the following moment--cumulant formula. Let $X_1,\ldots,X_k$ be random
variables. Then
\[
\kappa(X_1,\ldots,X_k)
=
\sum_{\pi\in\Pi_k}
(|\pi|-1)!(-1)^{|\pi|-1}
\prod_{B\in\pi}
\mathbb E\Bigg[\prod_{i\in B}X_i\Bigg],
\]
where $\Pi_k$ denotes the set of all partitions of $\{1,2,\ldots,k\}$. 
We apply this formula with $X_i=I_{\phi_i}$. 
Fix a partition \(\pi\in\Pi_k\).  For a block \(B\in\pi\), let $E_B
:=
\bigcup_{i\in B}E_{\phi_j}$ be the set of all edges appearing in $I_{\phi_i}$, $i\in B$.
Since each \(A_e\) is Bernoulli, \(A_e^m=A_e\) for every positive integer
\(m\).  Therefore
by independence of the graph edges,
\[
\prod_{B\in\pi}\mathbb E\Bigg[\prod_{i\in B}I_{\phi_i}\Bigg]
= \prod_{B\in\pi}\E\Bigg[\prod_{i\in B}\prod_{e\in E_{\phi_i}}A_e\Bigg]=
\prod_{B\in\pi}\mathbb E\left[\prod_{e\in E_B}A_e\right]
=
\prod_{B\in\pi}\prod_{e\in E_B}P_e.
\]
Now we compare this product with the product over the full union graph \(G\).  Each
edge \(e\in E(G)\) appears in at least one of the embedded copies
\(\phi_i(R)\).  Therefore, in the partition \(\pi\), the edge \(e\) appears in
at least one block \(B\).  Hence
\[
\prod_{B\in\pi}
\prod_{e\in E_B}P_e
=
\prod_{e\in E(G)}P_e^{m_e(\pi)}
\]
for some integers \(m_e(\pi)\ge 1\).  Since $P_e^{m_e(\pi)}\le P_e$, 
\[
\prod_{B\in\pi}
\mathbb E\Bigg[\prod_{j\in B}I_{\phi_j}\Bigg]
\le
\prod_{e\in E(G)}P_e.
\]
Using the moment--cumulant formula, we obtain
\[
\left|
\kappa(I_{\phi_1},\ldots,I_{\phi_k})
\right|
\le
\left[
\sum_{\pi\in\Pi_k}(|\pi|-1)!
\right]
\prod_{e\in E(G)}P_e = 
C_k\prod_{e\in E(G)}P_e.
\]
This proves the lemma.
\end{proof}


\end{appendix}

\bibliographystyle{alpha}
\bibliography{main}       

\end{document}